\documentclass[11pt]{article}

\usepackage[utf8]{inputenc}
\usepackage{lipsum}

\usepackage{amsmath, amssymb,amsfonts,amsthm,mathtools,nicefrac}
\usepackage{xspace,graphicx,relsize,bm,xcolor}
\usepackage{soul} 

\usepackage{parskip}  

\usepackage{libertinus}  
\usepackage[zerostyle=a,scaled=.97]{newtxtt}  
\usepackage[T1]{fontenc}
\usepackage[libertine]{newtxmath}  
\usepackage{dsfont}  
\usepackage[cal=stix, calscaled=.97]{mathalpha}  
\usepackage{anyfontsize}

\usepackage{url}
\usepackage[
    backend=biber,
    style=alphabetic,
    sorting=ynt,
    backref=true,
    maxbibnames=99
    ]{biblatex}
\DefineBibliographyStrings{english}{%
  backrefpage = {page},
  backrefpages = {pages},
  mathesis = {MSc thesis},
}

\usepackage[linktocpage=true]{hyperref}
\usepackage[margin=1in]{geometry}
\usepackage[singlespacing]{setspace}
\definecolor{linkcol}{rgb}{0.0,0.55,0.7}
\definecolor{citecol}{rgb}{0.0, 0.6, 0.45}
\definecolor{urlcol}{rgb}{0.7, 0.0, 0.55}
\hypersetup{
	colorlinks,
	linkcolor={linkcol},
	citecolor={citecol},
	urlcolor={urlcol}
}

\usepackage{subcaption}
\usepackage{mleftright}
\usepackage{tipa}
\usepackage{multirow}
\usepackage{physics}  
\usepackage{authblk}  

\usepackage{subfiles} 

\usepackage{algorithm}
\usepackage{algpseudocode}

\usepackage{multirow}

\usepackage{pifont}

\usepackage{float}

\usepackage{chngpage}

\usepackage[capitalize,noabbrev]{cleveref}  

\def\01{\{0,1\}}
\newcommand{\mc}[1]{\mathcal{#1}}

\renewcommand{\O}{\ensuremath{\mathcal{O}}}

\newcommand{\poly}{\mathrm{poly}}

\renewcommand{\norm}[1]{\lVert{#1}\rVert}

\DeclarePairedDelimiter\lrangle{\langle}{\rangle}
\makeatletter
\let\oldlrangle\lrangle
\def\lrangle{\@ifstar{\oldlrangle}{\oldlrangle*}}
\makeatother
\newcommand{\verteq}{\rotatebox{90}{$=$}}

\newcommand{\xmark}{\ding{55}}
\newtheoremstyle{mydefinitionsty}
{10pt}
{10pt}
{}
{}
{}
{}
{.5em}
{\textbf{\thmname{#1}~\thmnumber{#2}:  }\thmnote{(#3)}}
\theoremstyle{mydefinitionsty}
\newtheorem{definition}{Definition}
\newtheorem{remark}{Remark}

\newtheoremstyle{mythmsty}
{10pt}
{10pt}
{\itshape}
{}
{}
{}
{.5em}
{\textbf{\thmname{#1}~\thmnumber{#2}:  }\thmnote{(#3)}}
\theoremstyle{mythmsty}

\newtheorem{theorem}{Theorem}
\newtheorem{lemma}{Lemma}

\newtheorem{corollary}{Corollary}

\newtheorem{proposition}{Proposition}

\title{Classical Verification of Quantum Learning}
\author[1,2]{Matthias C.~Caro\thanks{
E-mail addresses: \{matthias.caro, m.hinsche, marios.ioannou, a.nietner\}@fu-berlin.de, ryan.sweke@ibm.com}}
\newcommand\CoAuthorMark{\footnotemark[\arabic{footnote}]}
\author[2]{Marcel Hinsche\protect\CoAuthorMark}
\author[2]{Marios Ioannou\protect\CoAuthorMark}
\author[2]{\\Alexander Nietner\protect\CoAuthorMark}
\author[2,3]{Ryan Sweke\protect\CoAuthorMark}
\affil[1]{Institute for Quantum Information and Matter, Caltech, Pasadena, CA, USA}
\affil[2]{Dahlem Center for Complex Quantum Systems, Freie Universit\"at Berlin, Berlin, Germany}
\affil[3]{IBM Quantum, Almaden Research Center, San Jose, CA, USA}
\date{}

\begin{document}
\maketitle

\begin{abstract}
    Quantum data access and quantum processing can make certain classically intractable learning tasks feasible.
    However, quantum capabilities will only be available to a select few in the near future. 
    Thus, reliable schemes that allow classical clients to delegate learning to untrusted quantum servers are required to facilitate widespread access to quantum learning advantages.
    Building on a recently introduced framework of interactive proof systems for classical machine learning, 
    we develop a framework for classical verification of quantum learning.
    We exhibit learning problems that a classical learner cannot efficiently solve on their own, but that they can efficiently and reliably solve when interacting with an untrusted quantum prover.
    Concretely, we consider the problems of agnostic learning parities and Fourier-sparse functions with respect to distributions with uniform input marginal.
    We propose a new quantum data access model that we call ``mixture-of-superpositions'' quantum examples, based on which we give efficient quantum learning algorithms for these tasks.
    Moreover, we prove that agnostic quantum parity and Fourier-sparse learning can be efficiently verified by a classical verifier with only random example or statistical query access.
    Finally, we showcase two general scenarios in learning and verification in which quantum mixture-of-superpositions examples do not lead to sample complexity improvements over classical data.
    Our results demonstrate that the potential power of quantum data for learning tasks, while not unlimited, can be utilized by classical agents through interaction with untrusted quantum entities.
\end{abstract}

\newpage
\tableofcontents
\newpage

\section{Introduction}

For many learning problems, the amount and type of data to which we have access determine our ability to obtain a good hypothesis. Unfortunately, in practical settings there is often a cost associated with collecting high quality data, and this cost prohibits us from solving a learning problem of interest. In light of this, it would be desirable to delegate learning problems to \textit{untrusted} servers with access to more or higher-quality data than ourselves.  
Ideally we would like such ``data-rich'' servers to efficiently solve the learning problem, and we would like to efficiently verify, using both the limited data available to us and interaction with the server, that the server has indeed provided a sufficiently good hypothesis and thus successfully solved the learning problem.
Recently, a formal framework  -- \textit{interactive proofs for the verification of machine learning} -- has been introduced to explore when, and to which extent, such delegation of learning tasks is possible~\cite{goldwasser2021interactive}. 

In this work, we are interested in verifying learning with untrusted \emph{quantum} servers, with access to some type of quantum data. Indeed, there is a rich history of work on quantum learning theory~\cite{arunachalam2017survey}, aimed at rigorously understanding the potential advantages and limitations of quantum learning algorithms with access to different types of quantum data oracles. 
Notably, there do exist learning problems which are intractable for classical algorithms, but which can be efficiently solved by quantum learning algorithms with quantum data access. 
However, the most realistic future scenario is that quantum devices will be accessed remotely, and that only certain parties have access to hard-to-prepare and hard-to-store quantum data. 
Therefore, to realize the advantages of quantum learning algorithms, it becomes crucial that classical clients (verifiers) can delegate learning problems to untrusted quantum servers (provers) and efficiently verify the provided hypotheses, using only interaction with the server and the classical data that is readily available.

In order to explore the setting just described, it is necessary to fix a formal learning problem. In the case of supervised learning, \cite{goldwasser2021interactive} showed that for standard Probably Approximately Correct (PAC) learning there exist trivial techniques for the verification of hypotheses (requiring only one round of one-way communication between prover and verifier) and as such the verification problem is only non-trivial for \textit{agnostic} PAC learning. 
In addition to being the natural and interesting setting for exploring the delegation and verification of learning, agnostic learning also captures an important feature of modern machine learning in practice: Often, one has few or no promises on the structure of the data, and one attempts to do the best possible by optimizing over a chosen model class (such as a particular neural network architecture). 
Given the necessity of working within the framework of agnostic learning, the question of whether or not it is possible for classical clients to delegate learning problems to untrusted quantum servers is only interesting if there exist agnostic learning problems in which the amount of resources required for classical learning exceeds that sufficient for quantum learners with access to quantum data.
Unfortunately, however, little is known about the power of quantum learning algorithms for agnostic learning.

In light of the above, the main contributions in this work are two-fold: Firstly, we identify and motivate a novel quantum oracle model for agnostic learning and, with respect to this oracle, provide the first efficient fully agnostic quantum learning algorithms for parities and Fourier-sparse functions. To the best of our knowledge, these are the first agnostic quantum learning algorithms for any model class for the problem of \textit{distributional} agnostic learning. Secondly, we leverage these positive agnostic quantum learning results to give a concrete example of an agnostic learning problem which is classically intractable, but can be efficiently and reliably delegated to an untrusted quantum server. More specifically, we provide an explicit interactive verification protocol which, despite the classical intractability of the learning problem, allows the classical client to efficiently verify the hypothesis provided by a potentially dishonest quantum server. 
This result serves as a proof-of-principle demonstration that classical clients can indeed reap the benefits of quantum advantages in learning, in the realistic setting where learning needs to be delegated to untrusted servers. Our hope is that these results provide new tools and insights for agnostic quantum learning, as well as motivation for the development of further techniques for the secure delegation of learning problems to quantum servers.

\subsection{Framework}\label{sbsct:framework}

\paragraph{Agnostic learning:}
When formalizing a learning task in which there may be a fundamental mismatch between the model used by the learner and the data-generating process, a so-called \emph{agnostic} learning task \cite{haussler1992decision, kearns1994toward}, there are two canonical choices:
\begin{itemize}
    \item In \emph{functional agnostic learning} w.r.t.~uniformly random inputs, we assume that the data consists of labeled examples $(x_i,f(x_i))$, with the $x_i$ drawn i.i.d.~uniformly at random from $\{0,1\}^n$ and with $f:\{0,1\}^n\to\{0,1\}$ an arbitrary unknown Boolean function. In this case, we denote the data-generating distribution as $\mathcal{D}=(\mathcal{U}_n, f)$.
    \item In \emph{distributional agnostic learning} w.r.t.~uniformly random inputs, we drop the assumption of a deterministic function that perfectly describes the data. That is, we assume labeled examples $(x_i,y_i)$ drawn i.i.d.~from some distribution $\mathcal{D}$ over $\{0,1\}^n\times\{0,1\}$ with uniform marginal over $\{0,1\}^n$. We denote this as $\mathcal{D}=(\mathcal{U}_n, \varphi)$ with conditional label expectation $\varphi:\{0,1\}^n\to [0,1]$, $\varphi(z) = \mathbb{E}_{(x,y)\sim \mathcal{D}} [y | x=z]$. 
\end{itemize}
Whereas in functional agnostic learning there is a ``correct'' label for every input, this is no longer true in the distributional agnostic setting. In particular, in the latter case data could contain conflicting labels for the same input.
Nevertheless, in both the functional and the distributional case, the goal is to learn an almost-optimal approximating function compared to a benchmark class $\mathcal{B}$: Given an accuracy parameter $\varepsilon$, a confidence parameter $\delta$, and access to a training data set generated i.i.d.~from $\mathcal{D}$, an $\alpha$-agnostic learner has to output, with success probability $\geq 1-\delta$, a hypothesis $h$ such that
\begin{equation}
    \mathbb{P}_{(x,y)\sim\mathcal{D}}[h(x)\neq y]
    \leq \alpha\cdot\inf_{b\in\mathcal{B}} \mathbb{P}_{(x,y)\sim\mathcal{D}}[b(x)\neq y] + \varepsilon .
\end{equation}
Note that here we do not necessarily require that $h\in\mathcal{B}$. If we add this requirement, we speak of \emph{proper} learning, otherwise the learner can be \emph{improper}.
Also, we recover the scenario of \emph{realizable} PAC learning when assuming that $\inf_{b\in\mathcal{B}} \mathbb{P}_{(x,y)\sim\mathcal{D}}[b(x)\neq y]=0$.

\paragraph{Learning classical functions from quantum data:}
In quantum learning theory, a learner can have access to $\mathcal{D}$ via a potentially more powerful resource than classical i.i.d.~examples.
Quantum training data for $\mathcal{D}$ is canonically taken to consist of copies of the \emph{quantum superposition example state} \cite{bshouty1998learning}
\begin{equation}
    \ket{\psi_\mathcal{D}} 
    = \sum_{(x,y)\in\{0,1\}^n\times\{0,1\}} \sqrt{\mathcal{D}(x,y)}\ket{x,y} .
\end{equation}
Such quantum data is at least as powerful as its classical counterpart, since the former can simulate the latter via computational basis measurements.
In fact, these quantum examples have proven to be useful for realizable learning and, to some degree, functional agnostic learning w.r.t.~the uniform distribution. However, it is unknown how to use copies of $\ket{\psi_\mathcal{D}}$ to improve upon classical distributional agnostic learning.

Therefore, we propose a different quantum resource for distributional agnostic learning. Our starting point is that a distribution $\mathcal{D}=(\mathcal{U}_n, \varphi)$ induces a distribution $F_\mathcal{D}$ over the set of all functions mapping $\{0,1\}^n$ to $\{0,1\}$. Namely, $F_\mathcal{D}$ is defined by taking the probability that $f(x)$ equals $1$ to be $\varphi(x)$ independently for each $x$, see \Cref{eq:induced-distribution-over-functions}.
We then consider quantum training data for $\mathcal{D}=(\mathcal{U}_n, \varphi)$ to consist of copies of the \emph{mixture-of-superpositions example state} (\Cref{definition:mixture-of-superpositions-quantum-example})
\begin{equation}
    \rho_\mathcal{D} = \mathbb{E}_{f\sim F_\mathcal{D}} \left[ \ket{\psi_{(\mathcal{U}_n, f)}}\bra{\psi_{(\mathcal{U}_n, f)}} \right] .
\end{equation}
Note that this kind of quantum data still reproduces classical training data upon computational basis measurements and is thus a consistent quantum generalization of the classical notion of training data.

\paragraph{Interactive verification of agnostic learning:}
If quantum processing and quantum data are only available to a select few, enabling widespread use of quantum learning requires classical verification procedures. 
Extending the framework of \cite{goldwasser2021interactive}, who formalized interactive verification of classical learning, we consider interactive classical verification of quantum learning. 
Here, an efficient classical verifier with classical data access, via random examples or statistical queries (SQs), interacts with an efficient quantum prover with mixture-of-superpositions quantum example or quantum SQ (QSQ) access. The goal of the verifier is twofold: On the one hand, when interacting with an honest quantum prover, the verifier should, with high probability, produce a hypothesis that satisfies the agnostic learning requirement. On the other hand, even when interacting with an arbitrarily powerful dishonest prover, the verifier should only accept the interaction and output a faulty hypothesis with small probability. If these two requirements are satisfied, the classical verifier can reliably profit from potential quantum advantages in learning.

\subsection{Overview of the Main Results}

Our first contribution is proposing mixture-of-superpositions states $\rho_\mathcal{D} = \mathbb{E}_{f\sim F_\mathcal{D}} \left[ \ket{\psi_{(\mathcal{U}_n, f)}}\bra{\psi_{(\mathcal{U}_n, f)}} \right]$ (see \cref{definition:mixture-of-superpositions-quantum-example}) as a resource for agnostic quantum learning. 
With this proposal, we return to the fundamental question of quantum learning theory: Do quantum versions of classical data access models enlarge the class of feasible learning problems? In particular, while quantum superposition examples have been widely adopted as the canonical “quantization” of classical random examples, it is of fundamental interest to understand what other consistent quantizations of classical data oracles exist, and how access to such oracles influences the complexity of different learning problems. To this end, we note that our mixture-of-superpositions examples are indeed \textit{consistent}, in the sense that they reduce to classical random examples upon measurements in the computational basis, and to the established quantum superposition examples in the functional agnostic case. 
Additionally, our definition is well-motivated by a natural operational interpretation of classical random examples for arbitrary distributions, which has previously been used to provide reductions from classical distributional to functional agnostic learning (see the discussion in \Cref{appendix:classical-distributional-to-agnostic}). 
More specifically, each time a mixture-of-superpositions oracle for the distribution $\mathcal{D}$ is queried, it responds by first choosing a random function $f:\{0,1\}^n\to\{0,1\}$ according to the distribution $F_\mathcal{D}$ induced by $\mathcal{D}$ and then sending a copy of $\ket{\psi_{(\mathcal{U}_n, f)}}$. Finally, our mixture-of-superpositions examples can be viewed as enriching quantum learning by an analogue of randomized quantum oracles, which, as discussed in \Cref{ss:related_work}, have recently received attention in quantum complexity theory~\cite{harrow2014uselessness, fefferman2018quantum-vs-classical, natarajan2022distribution, bassirian2022power}. Indeed, our motivation here is similar to these recent works -- namely to understand the effect of different oracle models on the landscape of quantum sample/query complexity. 

Quantum Fourier sampling \cite{bernstein1997quantum} is a central subroutine in most existing quantum learning algorithms. However, while it is known how to do quantum Fourier sampling from quantum superposition examples for functional agnostic learning, it is unknown whether quantum superposition examples suffice to perform quantum Fourier sampling in the distributional agnostic setting. 
Our first main result shows that mixture-of-superpositions examples allow for an approximate version of quantum Fourier sampling in the distributional agnostic setting and are thus a valuable resource for distributional agnostic quantum learning algorithms:

\begin{theorem}[Distributional agnostic approximate quantum Fourier sampling and learning -- Informal]\label{theorem:main-result-agnostic-quantum-fourier-sampling-learning}
    Let $\mathcal{D}=(\mathcal{U}_n, \varphi)$ be an unknown probability distribution over $\{0,1\}^n\times\{0,1\}$, with (known) uniform marginal over $\{0,1\}^n$ and with (unknown) conditional label expectation $\varphi:\{0,1\}^n\to [0,1]$.
    \begin{enumerate}
        \item \textbf{Distributional agnostic quantum Fourier sampling:} There is an efficient quantum algorithm that, given a single copy of $\rho_\mathcal{D}$, with success probability $\nicefrac{1}{2}$ outputs a sample from a probability distribution over $\{0,1\}^n$ that is inverse-exponentially close to the squares of the Fourier coefficients of $\phi = 1-2\varphi$. 
        \item \textbf{Distributional agnostic proper quantum parity learning:} There is an efficient quantum algorithm that properly $1$-agnostically learns parities from an efficient number of copies of $\rho_\mathcal{D}$.
        \item \textbf{Distributional 2-agnostic improper quantum Fourier-sparse learning:} There is an efficient quantum algorithm that improperly $2$-agnostically learns Fourier-sparse functions from an efficient number of copies of $\rho_\mathcal{D}$. 
    \end{enumerate}
\end{theorem}

\Cref{theorem:main-result-agnostic-quantum-fourier-sampling-learning}, proved in \Cref{section:distributional-agnostic-quantum-learning}, constitutes the first general progress on distributional agnostic quantum learning w.r.t.~uniform input marginal. It achieves this by generalizing quantum Fourier sampling from the functional to the distributional setting (see \Cref{theorem:agnostic-quantum-fourier-sampling}). In proving \Cref{theorem:main-result-agnostic-quantum-fourier-sampling-learning}, we establish agnostic learning guarantees from Fourier spectrum approximation that, to the best of our knowledge, also improve upon the best known analogous classical result in terms of the achieved $\alpha$.
Moreover, we prove that, based on a version of the Goldreich-Levin/Kushilevitz-Mansour algorithm \cite{goldreich1989hard, kushilevitz1993learning}, agnostic parity and Fourier-sparse learning remain possible efficiently even in a weaker data access model of distributional agnostic quantum statistical queries, which we introduce as an extension of the classical statistical query model \cite{kearns1998efficient} and its functional quantum variant \cite{arunachalam2020qsq}.
In addition, we provide a variety of results establishing the feasibility of Fourier sampling, finding heavy Fourier coefficients, and agnostic learning in the functional setting, when given access to different types of quantum oracles. While many of these results follow from existing techniques used in the realizable PAC setting, we present them to provide a complete picture of the status quo in agnostic learning from quantum resources and to highlight open questions. These additional results, as well as the main results from Theorem~\ref{theorem:main-result-agnostic-quantum-fourier-sampling-learning}, are summarized in \Cref{tab:results_table}. 

\begin{table}[h!]
\centering
\begin{adjustwidth}{-.4in}{-.5in} 
\begin{tabular}{ |c|c||c|c|c|c| } 
\hline
\multirow{4}{*}{} & \multirow{4}{*}{Oracle type} & \multicolumn{4}{c|}{Problem type} \\ \cline{3-6}
 &  & $\begin{matrix}
 \text{Fourier}\\
 \text{sampling}
 \end{matrix}$  & $\begin{matrix}
 \text{Heavy Fourier}\\
 \text{coefficients}\\
 \text{(à la GL/KM)}
 \end{matrix}$   & $\begin{matrix}
 \text{1-agnostic}\\
 \text{parity learning}
 \end{matrix}$  & $\begin{matrix}
 \text{2-agnostic}\\
 \text{Fourier-sparse}\\
 \text{learning}
 \end{matrix}$  \\ 
\hline
\multirow{4}{*}{Functional} & $\begin{matrix}
\text{superposition examples} \\
\verteq \\
\text{mixture-of-superpositions} \\
\end{matrix}$  & 
$\begin{matrix}
\text{\checkmark}\\\text{(\Cref{lemma:functional-quantum-Fourier-sampling})}
\end{matrix}$  
& 
$\begin{matrix}
\text{\checkmark}\\\text{(\cref{corollary:quantum-approximation-fourier-spectrum})}
\end{matrix}$   & 
$\begin{matrix}
\text{\checkmark}\\\text{(\Cref{corollary:quantum-functional-agnostic-parity-learning})}
\end{matrix}$   &
$\begin{matrix}
\text{\checkmark}\\\text{(\Cref{corollary:quantum-functional-agnostic-fourier-sparse-learning})}
\end{matrix}$  \\ \cline{2-6}
& $\begin{matrix}
\text{superposition QSQ} \\
\verteq \\
\text{mixture-of-superpositions QSQ} \\
\end{matrix}$  &
$\begin{matrix}
 \text{Probably \xmark}\\\text{(\Cref{remark:qsq-no-efficient-fourier-sampling})}
\end{matrix}$  & 
$\begin{matrix}
\text{\checkmark}\\\text{(\Cref{theorem:functional-agnostic-qsq-GL})}
\end{matrix}$  &
$\begin{matrix}
\text{\checkmark}\\\text{(\Cref{subsection:functional-qsq-learning})}
\end{matrix}$   & 
$\begin{matrix}
\text{\checkmark}\\\text{(\Cref{subsection:functional-qsq-learning})}
\end{matrix}$  \\ \hline
\multirow{6}{*}{Distributional} & superposition examples  & ?  & ? & ? & ? \\\cline{2-6}
 & superposition QSQ  & $\begin{matrix}
 \text{Probably \xmark}\\\text{(\Cref{remark:qsq-no-efficient-fourier-sampling})}
\end{matrix}$  & ? & ? & ? \\\cline{2-6}
 & mixture-of-superpositions  &
 $\begin{matrix}
\text{\checkmark}\\\text{(\Cref{theorem:agnostic-quantum-fourier-sampling})}
\end{matrix}$   &
 $\begin{matrix}
\text{\checkmark}\\\text{( \Cref{corollary:distributional-agnostic-quantum-approximation-fourier-spectrum})}
\end{matrix}$ & 
 $\begin{matrix}
\text{\checkmark}\\\text{(\Cref{corollary:agnostic-quantum-parity-learning})}
\end{matrix}$ &
 $\begin{matrix}
\text{\checkmark}\\\text{(\Cref{corollary:agnostic-quantum-fourier-sparse-learning})}
\end{matrix}$ \\\cline{2-6}
 & mixture-of-superpositions QSQ  & 
 $\begin{matrix}
 \text{Probably \xmark}\\\text{(\Cref{remark:qsq-no-efficient-fourier-sampling})}
\end{matrix}$ & 
 $\begin{matrix}
\text{\checkmark}\\\text{(\Cref{theorem:distributional-agnostic-qsq-GL})}
\end{matrix}$  &
 $\begin{matrix}
\text{\checkmark}\\\text{(\Cref{corollary:agnostic-qsq-learning})}
\end{matrix}$  & 
 $\begin{matrix}
\text{\checkmark}\\\text{(\Cref{corollary:agnostic-qsq-learning})}
\end{matrix}$\\
\hline
\end{tabular}
\end{adjustwidth}
\caption{\textbf{Quantum oracles and the feasibility of agnostic learning:} An overview of the different quantum oracles studied in this work, and the protocols that they allow for. For all problems, in both the functional and distributional case, we assume a uniform input marginal. A check mark (\checkmark) indicates an efficient algorithm, ``Probably \xmark" indicates that there exists evidence against the existence of an efficient algorithm, and a question mark indicates an open question. Notably, the mixture-of-superpositions examples that we introduce are the only form of quantum data access currently known to enable quantum Fourier sampling and agnostic learning in the distributional setting (as highlighted in \Cref{theorem:main-result-agnostic-quantum-fourier-sampling-learning}).}
\label{tab:results_table}
\end{table}

In our second main result, we identify an agnostic learning problem that a classical learner cannot solve on their own, but that becomes feasible for a classical verifier interacting with a quantum prover who has access to mixture-of-superpositions examples. 

\begin{theorem}[Verifying distributional agnostic quantum learning -- Informal]\label{theorem:main-result-verification}
    There is a class $\mathfrak{D}$ of probability distributions over $\{0,1\}^n\times\{0,1\}$ with (known) uniform marginal over $\{0,1\}^n$ such that:
    \begin{itemize}
        \item[(a)] Distributional $1$-agnostic parity learning is classically hard from SQs or random examples, even if the unknown distribution is promised to lie in  $\mathfrak{D}$.
        \item[(b)] When promised that the unknown distribution lies in $\mathfrak{D}$, there is an efficient interactive verification procedure that allows a classical verifier, with SQ or random example access, to verify a distributional $1$-agnostic quantum parity learner, who has mixture-of-superpositions example or QSQ access. 
        \item[(c)] When promised that the unknown distribution lies in $\mathfrak{D}$, there is an efficient interactive verification procedure that allows a classical verifier, with SQ or random example access, to verify a distributional $2$-agnostic quantum Fourier-sparse learner, who has mixture-of-superpositions example or QSQ access.
    \end{itemize}
\end{theorem}

\Cref{theorem:main-result-verification}, which collects the statements of \Cref{theorem:functional-agnostic-quantum-parity-verification-qsq-no-small-non-zero-Fourier-coeff,theorem:functional-agnostic-quantum-fourier-sparse-verification-qsq-no-small-non-zero-Fourier-coeff,theorem:distributional-agnostic-qsq-parity-verification-no-small-non-zero-Fourier-coeff,theorem:distributional-agnostic-quantum-parity-verification-no-small-non-zero-Fourier-coeff,theorem:distributional-agnostic-qsq-fourier-sparse-verification-no-small-non-zero-Fourier-coeff,theorem:distributional-agnostic-quantum-fourier-sparse-verification-no-small-non-zero-Fourier-coeff}, shows that our new notion of quantum data not only enables distributional agnostic quantum learning but does so in a classically efficiently verifiable manner. 
Thereby, \Cref{theorem:main-result-verification} establishes a separation between what a classical learner can achieve on their own and what they can achieve when interacting with an untrusted quantum prover. 
This separation is unconditional for SQ access and conditional on the hardness of Learning Parity with Noise (LPN) for random example access.
Moreover, we show that the distribution class $\mathfrak{D}$ used in  \Cref{theorem:main-result-verification} cannot be meaningfully enlarged without significant losses in the efficiency of the classical verifier.
All of this is proved in \Cref{section:verification}.

\Cref{theorem:main-result-agnostic-quantum-fourier-sampling-learning,theorem:main-result-verification} show that mixture-of-superpositions examples serve as a powerful resource that can change the learning landscape in a positive way, by allowing us to solve learning problems for which we have so far been lacking quantum learners.
Crucially, however, our proposed model of quantum data access is not all-powerful: Just like their established superposition counterpart, mixture-of-superpositions examples do not allow for relevant sample complexity advantages over classical learners when considering distribution-independent agnostic learning.

\begin{theorem}[Sample Complexity Lower Bound for Distribution-Independent Distributional Agnostic Quantum Learning -- Informal Version]\label{theorem:main-result-distribution-independent-lower-bound} 
    The quantum sample complexity of distribution-independent distributional agnostic learning a function class $\mathcal{F}\subseteq\{0,1\}^{\{0,1\}^n}$ from mixture-of-superpositions examples does not improve upon the classical sample complexity, up to logarithmic factors.
\end{theorem}

Classically, it is well-established that the sample complexity of distribution-independent distributional agnostic learning $\mathcal{F}$ behaves as $\Theta\left(\tfrac{\operatorname{VCdim}(\mathcal{F}) + \log(\nicefrac{1}{\delta})}{\varepsilon^2}\right)$ \cite{vapnik1971uniform, blumer1989learnability, talagrand1994sharper}. 
Here, the VC-dimension $\operatorname{VCdim}(\mathcal{F})$ is a combinatorial complexity measure for the function class $\mathcal{F}$ \cite{vapnik1971uniform}.
While we prove a quantum sample complexity lower bound that matches the classical upper bound up to factors logarithmic in $\operatorname{VCdim}(\mathcal{F})$, we in fact conjecture that quantum and classical sample complexities for distribution-independent learning coincide up to constant factors.
In addition, we show that also the optimal sample complexity lower bound for verifying distribution-independent agnostic classical learning from \cite{mutreja2022pac-verification} carries over to agnostic quantum learning with mixture-of-superpositions examples.
Thus, whereas \Cref{theorem:main-result-agnostic-quantum-fourier-sampling-learning} and \Cref{theorem:main-result-verification} exhibit the power of our newly proposed quantum resource, \Cref{theorem:main-result-distribution-independent-lower-bound} demonstrates that, from an information-theoretic perspective, mixture-of-superpositions examples do not change the landscape of distribution-independent learning.
Our detailed results for distribution-independent learning and verification can be found in \Cref{section:distribution-independent-quantum-limitations}.

\subsection{Related Work}\label{ss:related_work}

\paragraph{Verification and testing of learning:}
Learning problems are often studied in the framework of computational learning theory \cite{valiant1984theory}.
Our fundamental motivation is to understand the extent to which classical clients can efficiently delegate learning tasks to untrusted quantum servers.
A formal framework for reasoning about such settings -- i.e. the delegation of learning tasks to servers with more or better quality data via interactive proofs for PAC learning -- was only recently introduced in the classical case by \cite{goldwasser2020interactive-full-version,goldwasser2021interactive}.
They considered clients (verifiers) with random example access interacting with untrusted servers (provers) who have the ability to make membership queries.  Since then, this work has been extended to both statistical query learning algorithms~\cite{mutreja2022pac-verification} and the setting of limited communication complexity~\cite{oconnor2021delegating} between prover and verifier. Our work initiates the study of the natural setting in which the untrusted prover is quantum, with access to a quantum data oracle. 
As discussed in~\cite{goldwasser2021interactive}, interactive proofs for the verification of PAC learning are closely related to interactive proofs for distribution testing~\cite{chiesa2018proofs, herman2022verifying}, with the important difference that in the learning setting we insist on efficient honest provers. Additionally, we have already alluded to the main difficulty in developing interactive proofs for PAC learning: finding certificates for the quality of a hypothesis relative to the optimal achievable performance in the agnostic learning setting.
Addressing a similar issue, but from a different perspective, \cite{rubinfeld2022testing} recently initiated the study of \textit{testable} agnostic learning, in which the learning algorithm combines with a distribution testing algorithm, and is only required to succeed when the data passes the testing algorithm. This has since attracted significant attention and sparked the development of a variety of testable learning algorithms for a wide range of model classes~\cite{gollakota2022momentmatching, gollakota2023efficient,gollakota2023testerlearners,diakonikolas2023efficient}.

\paragraph{Fourier-based and agnostic learning:}
As discussed in \Cref{ss:techniques}, both our learning and verification algorithms rely heavily on the use of Fourier analytic tools. In this sense, our results and techniques extend a long line of work on Fourier-based learning algorithms, pioneered by the Goldreich-Levin/Kushilevitz-Mansour algorithm \cite{goldreich1989hard, kushilevitz1993learning} for learning Fourier-sparse functions, as well as the low-degree algorithm \cite{linial1993constant}, which has only recently been significantly improved~\cite{eskenazis2022learning}. 
Additionally, there is also a long history of work aimed at understanding the complexity of classical agnostic learning and its relation to standard realizable PAC learning~\cite{gopalan2008agnostically, kalai2008agnostic, kanade2009potential,feldman2009agnostic, feldman2009power, feldman2010agnostic-boosting}. 
Of particular interest from the perspective of our work is \cite{gopalan2008agnostically}, who use the ``distribution over functions" interpretation of a distributional oracle -- which is the starting point for the formulation of our quantum \textit{mixture-of-superpositions} oracle -- to provide reductions from distributional agnostic to functional agnostic learning. Additionally, we note that our work touches on the relatively underdeveloped setting of \textit{improper} agnostic learning, whose foundations have very recently been developed under the name of \textit{comparative learning}~\cite{hu2022comparative}. Finally, we remark that agnostic learning is of significant cryptographic relevance. Indeed, both the LWE and LPN assumptions, of key importance for a wide variety of cryptographic protocols, are fundamentally assumptions on the hardness of specific agnostic learning problems~\cite{regev2009lattices,pietrzak2012cryptography}. 

\paragraph{Quantum learning theory:} 
The field of quantum learning theory was initiated by \cite{bshouty1998learning}, who formulated the notion of \textit{quantum examples} and provided an efficient quantum algorithm for learning DNFs with respect to the uniform distribution from such examples. 
Since then, as surveyed in \cite{arunachalam2017survey}, there has been a variety of work in multiple directions. One series of works has developed the foundations of quantum learning theory, providing characterizations of the quantum query complexity for PAC and agnostic learning, which culminated in a proof that quantum examples cannot offer more than a polynomial advantage in query complexity in the distribution-independent setting~\cite{atici2005improved, zhang2010improved, arunachalam2018optimal}. Simultaneously, several works focused on the further development of explicit quantum algorithms for learning from quantum examples in the distribution-dependent setting. Examples include an improved quantum procedure for learning DNFs~\cite{jackson2002quantum}, a quantum learning algorithm for learning juntas~\cite{atici2007quantum}, as well as first quantum algorithms for learning with respect to non-uniform distributions~\cite{kanade2019learning, caro2020quantum}. 
Another series of works has focused on broadening the scope of quantum learning theory, introducing the notions of quantum membership queries~\cite{servedio2004equivalences, montanaro2012quantum, arunachalam2021twonewresults}, quantum statistical queries~\cite{arunachalam2020qsq}, and quantum distribution learning~\cite{sweke2021quantum}. 
Especially closely related to our work are previous efforts to develop quantum learning algorithms in the non-realizable setting. In particular, a variety of quantum algorithms have been given for learning parities with respect to various notions of \textit{noisy} quantum examples~\cite{cross2015quantum, grilo2019learning, caro2020quantum}, and recently \cite{bera2022efficient} provided a \textit{functional} agnostic quantum learning algorithm for decision trees. 
Against this backdrop, our work makes a variety of contributions. Firstly, we broaden the scope of quantum learning theory through both the introduction of the mixture-of-superpositions quantum example for agnostic learning, as well as the initiation of delegated quantum learning. Additionally, by using the mixture-of-superpositions oracle, we give the first efficient \textit{distributional} quantum agnostic learning algorithms. We achieve this by developing the toolbox for quantum Fourier sampling~\cite{bernstein1997quantum}, which is the key subroutine underlying quantum learning algorithms.

\paragraph{Randomized quantum oracles:}
One of the primary conceptual contributions of our work is to expose the impact of different quantum oracles on the quantum query complexity of agnostic learning.  The mixture-of-superpositions oracle that we introduce is similar in spirit to a variety of ``non-standard" randomized quantum oracles that have recently been introduced and used to enrich and explore the landscape of quantum complexity theory more broadly. More specifically, similar randomized quantum oracles have recently been used both to better understand the impact of oracle design on the quantum query complexity of a variety of testing problems and to provide oracle separations between $\mathsf{QMA}$ and $\mathsf{QCMA}$, as partial progress towards the long-standing goal of separating these classes~\cite{harrow2014uselessness,fefferman2018quantum-vs-classical,natarajan2022distribution,bassirian2022power}.
 
\paragraph{Verification of quantum computation:}
This work initiates and studies the question of whether and to which extent it is possible for classical clients to efficiently delegate learning problems to untrusted quantum servers. 
In particular we ask whether there exists efficient interactive protocols via which efficient classical clients, with access to one type of data, can verify the hypothesis provided by an efficient quantum learning algorithm with access to a quantum data oracle. 
This is a learning-theoretic analogue to a long line of work aimed at providing protocols via which efficient classical verifiers ($\mathsf{BPP}$ machines)  can verify the results of efficient quantum provers ($\mathsf{BQP}$ machines) \cite{gheorghiu2019verification}, which culminated in Mahadev's breakthrough verification protocol~\cite{Mahadev}. We note, however, that the relation between verifying computation and verifying learning is non-trivial, as discussed in \cite{goldwasser2021interactive}.
Additionally, there is a large body of work aimed at understanding the extent to which classical clients can delegate computations to quantum servers in a way that ensures the \textit{privacy} of the clients inputs, outputs and desired computation~\cite{Broadbent_2009, fitzsimons2017private}. While also similar in spirit to our work in some ways, we do not enforce any notion of privacy. However, we note that \cite{canetti2021covert} recently introduced the notion of \textit{covert} learning, a classical framework for exploring the possibility of private delegation of learning.

\subsection{Techniques and Proof Overview}\label{ss:techniques}

\paragraph{Distributional agnostic quantum learning:}
It is not known how to use the conventional superposition quantum examples $\ket{\psi_{\mathcal{D}}}$ to speed up agnostic learning in the distributional setting. The key advantage of our newly introduced mixture-of-superpositions quantum examples $\rho_\mathcal{D} = \mathbb{E}_{f\sim F_\mathcal{D}} \left[ \ket{\psi_{(\mathcal{U}_n, f)}}\bra{\psi_{(\mathcal{U}_n, f)}} \right]$ is that they enable approximate quantum Fourier sampling in the distributional setting under uniform input marginal. In particular, to achieve this approximate Fourier sampling, we use the same simple, standard quantum algorithm that is known to work in the realizable setting: applying a layer of single-qubit Hadamard gates to a single copy of $\rho_\mathcal{D}$ followed by a measurement in the computational basis and post-selecting on the outcome 1 in the last qubit. 
In \Cref{subsection:distributional-agnostic-quantum-fourier-sampling}, we show that, with post-selection probability $\nicefrac{1}{2}$, this procedure results in sampling strings from $\{0,1\}^n$ according to the probability distribution
\begin{equation}
    \mathrm{Pr} (s) = \frac{1}{2^n}\left(1 - \mathbb{E}_{x\sim\mathcal{U}_n}[(\phi(x))^2]\right) + (\hat{\phi}(s))^2 . \label{eq:approximate_Fourier_sampling}
\end{equation}
Here, $\phi = 1- 2\varphi$ is the conditional $\{-1,1\}$-label expectation whose Fourier spectrum $\{\hat{\phi}(s)\}_s$ we are interested in. Hence, by \Cref{eq:approximate_Fourier_sampling}  we can sample strings $s \in \{0,1\}^n$ essentially with probability proportional to the squared Fourier weight $(\hat{\phi}(s))^2$ up to an exponentially small correction. By the Dvoretzky-Kiefer-Wolfowitz theorem \cite{dvoretzky1956asymptotic, massart1990tight}, this is sufficient to yield a succinct $\infty$-norm approximation to the Fourier spectrum of $\phi$ and hence identify its heaviest Fourier coefficients. 
We connect these findings to quantum agnostic learning with a careful, purely classical analysis of how knowledge about the heaviest Fourier coefficients leads to distributional $\alpha$-agnostic learners. This analysis is the content of \Cref{appendix:useful}.

\paragraph{Classical verification of quantum agnostic learning:}
In contrast to the realizable setting, verifying the quality of a hypothesis in the agnostic setting is non-trivial. This is because the optimal performance with respect to the benchmark class is unknown. To appreciate this, consider the setting of verifying agnostic learning parities: As discussed above, a quantum prover is able to learn the string $s$ corresponding to the heaviest Fourier coefficient $\hat{\phi}(s)$ given access to mixture-of-superpositions examples. It could then send the string $s$ to the verifier. The verifier can use its classical data access to produce a good estimate of the Fourier weight $|\hat{\phi}(s)|$ of the string. However, the verifier has no efficient means of knowing if there are other strings $s'$ with even more Fourier weight and so it cannot rule out that the prover cheated by sending a non-optimal string.

To address such lack of soundness, in our protocol the prover ought to send not only the single, heaviest Fourier coefficient, but a list $L=\{s_1, \dots, s_{|L|}\}$ corresponding to all non-negligible Fourier coefficients $\hat{\phi}(s)$. 
Given such a list, the verifier can independently from the prover estimate the total squared Fourier weight of the list, namely $\sum_{\ell=1}^{|L|} (\hat{\phi}(s_\ell))^2$, using classical data access only. If additionally we restrict the agnostic learning task to distributions $\mathcal{D}=(\mathcal{U}_n, \varphi)$ where the total Fourier weight is fixed, or at least a priori known to lie within a small interval,
\begin{equation}
    \sum_{s \in \{0,1\}^n} (\hat{\phi}(s))^2 \in \left[ a^2, b^2 \right], 
\end{equation}
then the verifier can find out if the prover cheated by checking whether its estimate of the total Fourier weight of the list deviates too much from the a priori known total Fourier weight of the distribution. For the protocol to be efficient, we need to require that the underlying distribution $\mathcal{D}$ is well-described only by a few heavy Fourier coefficients such that the list $L$ is of size at most polynomial in $n$.
This means that we consider a distributional agnostic learning task over a restricted set of distributions rather than the set of all distributions with uniform input marginal. We emphasize that despite the restriction of the distribution class, the learning task is still classically hard and so the verifier crucially requires the prover to succeed. This classical hardness is based on the Learning Parity with Noise (LPN) problem, which is a special case of our more general distributional agnostic learning task.

Lastly, we note that in the functional agnostic learning setting, we provide an alternative approach to classically verifying quantum learning. 
This approach is based on the interactive Goldreich-Levin algorithm laid out in \cite{goldwasser2021interactive} for agnostic verification of Fourier-sparse learning. 
Said verification scheme requires classical membership query access for the prover in order to answer the queries sent by the verifier.
We observe that under certain Fourier-sparsity assumptions on the unknown function, the quantum prover can emulate this membership query access.

\paragraph{Distribution-independent agnostic quantum learning and its verification:}
To show that mixture-of-superpositions examples do not lead to a significant advantage in a distribution-independent agnostic setting, we provide a lower bound on the sample complexity required to learn a benchmark class of VC-dimension $d$. To prove this lower bound, we adapt the Fano-method based proof strategy from \cite{arunachalam2018optimal} to our mixture-of-superpositions quantum examples.
Similarly, to show that our mixture-of-superpositions examples do not lead to an advantage in the setting of verification of distribution-independent learning, we provide a lower bound on the sample complexity of the verifier for verifying a benchmark class of VC-dimension $d$. To prove this lower bound, we adopt the proof strategy from \cite{mutreja2022pac-verification}, which is based on a reduction from a testing task to the verification task. We find that the reduction applies even to our setting with a quantum prover since we focus purely on sample complexity. 
The sample complexity lower bound for the testing task was given in \cite{mutreja2022pac-verification} and carries over directly.

\subsection{Directions for Future Work}\label{subsection:future-work}

Our work opens up several directions for future research. 
Firstly, we have demonstrated that mixture-of-superpositions examples enable quantum Fourier sampling-based distributional agnostic learning -- in the distribution-dependent setting -- by giving explicit learning algorithms for parities and Fourier-sparse functions.
At the same time, we have shown that mixture-of-superpositions examples \textit{do not} give a sample complexity advantage over classical random examples in the distribution-independent setting.
Thus, it is of natural interest to go beyond these initial results and to further understand both the potential and the limitations of mixture-of-superpositions examples for agnostic learning, for example exploring their use for other model classes. 
Additionally, one of the primary motivations for the mixture-of-superpositions examples introduced here is the difficulty in developing techniques for Fourier sampling from standard quantum superposition examples in the distributional agnostic setting.
However, while no such techniques have been developed to date, there are no established hardness results.  
As such, it remains unclear whether mixture-of-superpositions examples are indeed strictly more powerful than standard quantum superposition examples.
In light of this, it would be interesting to understand whether there is a separation between the power of the two oracle models.

Additionally, in this work we have explored the extent to which one can delegate problems of \textit{supervised learning of Boolean functions} to untrusted quantum servers.  
However, there is a plethora of other learning problems whose delegation to quantum algorithms would be desirable to investigate. A natural first example would be the delegation and verification of \textit{distribution learning}~\cite{kearns1994learnability} problems to quantum servers. In particular, we note that, unlike for supervised learning, in the distribution learning context even the realizable setting seems non-trivial. 
Alternatively, there is a multitude of learning and testing problems for which the object to be learned or tested is inherently quantum. Examples include testing or learning to predict properties of quantum states \cite{aaronson2007learnability, aaronson2019shadow, huang2020predicting}, quantum measurements \cite{cheng2016learnability}, or quantum processes \cite{chung2021sample, caro2021binary, fanizza2022learning, caro2022learning, huang2023learning}. For many of these problems there are known exponential separations between what can be achieved by quantum algorithms with or without access to a quantum memory (see~\cite{huang2021information, aharonov2022quantum, chen2022exponential, huang2022quantum-advantage, caro2022learning, chen2022unitarity}). This prompts a natural question: Can quantum learning algorithms without a quantum memory efficiently delegate such learning or testing problems to untrusted quantum algorithms with access to a quantum memory?

Moreover, from a more technical perspective, there are concrete ways in which both our learning algorithms and verification procedures might be improved. Firstly, for the problem of distributional agnostic learning Fourier-sparse functions, our learning algorithms are 2-agnostic -- i.e., they yield hypotheses whose risk is guaranteed to be at most twice the risk of the optimal model, plus some desired tolerance $\varepsilon$. Ideally, however, one would like to give 1-agnostic learning algorithms. 
Secondly, our verification procedures do not work for arbitrary unknown functions or distributions, but require prior assumptions. The learning problems we consider remain classically hard under these assumptions, and are therefore still sufficient for demonstrating the existence of problems which can be efficiently delegated/verified by classical learning algorithms, although not efficiently solved without delegation.
Nevertheless, it seems interesting to understand the extent to which the assumptions we use here are truly necessary.

Finally, our work is motivated by a desire to understand the potential for classical clients to profit from the advantages of quantum learning algorithms in a (realistic) world where quantum computations are delegated to untrusted quantum servers with access to proprietary quantum data resources. However, at least currently and for the intermediate-term future, any quantum server will only have access to ``noisy intermediate scale quantum" (NISQ) devices \cite{preskill2018quantum}. 
As such, to bring our results closer to immediate practical relevance, it is of interest to explore the extent to which classical clients can verify untrusted NISQ-friendly quantum machine learning algorithms based on the variational optimization of parameterized quantum circuits. 
Indeed, there has recently been progress on the statistical foundations of such hybrid quantum-classical learning algorithms~\cite{caro2020pseudo, abbas2021power, banchi2021generalization, caro2021encodingdependent, du2022efficient, caro2022generalization, caro2022out-of-distribution}, and it would be of significant interest to enrich this developing understanding with insight into the complexity of classical verification. 
Additionally, in our work we have explored the setting in which a classical client interacts with a quantum server, with access to a quantum data oracle. However, one may also consider quantum clients of limited complexity (e.g., NISQ clients) that interact with more powerful quantum servers. 
This would serve to enrich our growing understanding of the capability of NISQ algorithms from a complexity-theoretic perspective~\cite{chen2022complexity}.

\subsection{Structure of the Paper}

The remainder of this work is structured as follows.
\Cref{section:preliminaries} recalls standard notions from the analysis of Boolean functions, from classical and quantum learning theory, and the framework of \cite{goldwasser2021interactive} for interactive verification of machine learning.
In \Cref{section:mixture-of-superpositions-examples}, we propose mixture-of-superpositions quantum examples as a new resource for quantum learning algorithms.
\Cref{section:functional-agnostic-quantum-learning} gives an overview over the power of quantum data access in functional agnostic learning.
This section may be viewed as a prelude to our distributional agnostic quantum learning results in \Cref{section:distributional-agnostic-quantum-learning}.
Nevertheless, \Cref{section:distributional-agnostic-quantum-learning} can be read mostly independently of \Cref{section:functional-agnostic-quantum-learning}.
We establish that our agnostic quantum learning procedures can be classically verified in \Cref{section:verification}.
Finally, \Cref{section:distribution-independent-quantum-limitations} contains our results about the limitations of quantum data for distribution-independent agnostic learning and its verification.
\Cref{appendix:useful} contains relevant tools and results in classical computational learning theory.
In \Cref{appendix:classical-distributional-to-agnostic}, we present a classical reduction from distributional to functional agnostic learning that serves as motivation for our notion of mixture-of-superpositions examples. 
\Cref{appendix:proofs} contains auxiliary results and proofs. 

\section{Notation and Preliminaries}\label{section:preliminaries}

In this section, we review standard notions from the analysis of Boolean functions as well as the underlying framework of computational learning theory, both classical and quantum. 
Moreover, we recall the framework for interactive verification of learning introduced in \cite{goldwasser2021interactive}.
Along the way, we fix notational conventions that will be used throughout the paper.

\subsection{Boolean Fourier Analysis}\label{subsection:fourier-analysis}

Throughout, we work with functions defined on the Boolean hypercube $\mathcal{X}_n=\{0,1\}^n$, for some $n\in\mathbb{N}_{>0}$.
We will be particularly interested in Boolean functions, which we denote by $f:\{0,1\}^n \to\{0,1\}$. When convenient, we redefine the binary labels according to $0\mapsto 1$, $1\mapsto -1$, and then consider the function $g :\{0,1\}^n \to\{-1,1\}$, $g(x)=(-1)^{f(x)}=1-2f(x)$ instead of $f$.
We will denote probability distributions over $\mathcal{X}_n\times\{0,1\}$ by $\mathcal{D}$. The marginal of $\mathcal{D}$ over the first $n$ bits is denoted by $\mathcal{D}_{\mathcal{X}_n}$. If this marginal is the uniform distribution over $\mathcal{X}_n$, then we denote that by $\mathcal{D}_{\mathcal{X}_n} = \mathcal{U}_n$. 
The conditional expectation of the $\{0,1\}$-label given the input is denoted by $\varphi:\mathcal{X}_n\to [0,1]$, $\varphi(z) = \mathbb{E}_{(x,y)\sim \mathcal{D}} [y | x=z]$.
In other words, under the probability distribution $\mathcal{D}$, conditioned on the input being $x$, the associated label equals $1$ with probability $\varphi(x)$ and equals $0$ with probability $1 - \varphi(x)$.
Again, whenever convenient, we relabel the target space from $\{0,1\}$ to $\{-1,1\}$, then replacing $(\mathcal{D}_{\mathcal{X}_n}, \varphi)$ by $(\mathcal{D}_{\mathcal{X}_n}, \phi)$ with $\phi:\mathcal{X}_n\to [-1,1]$, $\phi(z) = 1-2\varphi(z)$.

We use the standard framework of Fourier analysis for $\mathbb{R}$-valued functions defined on the Boolean hypercube $\mathcal{X}_n$, compare \cite{odonnell2014analysis}.
In particular, we define the Fourier coefficients w.r.t.~the uniform distribution over $\mathcal{X}_n$ of such as function as follows:

\begin{definition}[Fourier coefficients]\label{definition:fourier-coefficients}
    Let $\phi:\mathcal{X}_n\to \mathbb{R}$.
    Then, for any $s\in \{0,1\}^n = \mathcal{X}_n^\ast = \mathcal{X}_n$, we define the Fourier coefficient $\hat{\phi}(s)\in [-\norm{\phi}_\infty, \norm{\phi}_\infty]$ as
    \begin{equation}\label{eq:fourier-coefficient}
        \hat{\phi}(s)
        \coloneqq \mathbb{E}_{x\sim\mathcal{U}_n} \left[ \phi(x) \chi_s(x) \right] ,
    \end{equation}
    with the parity functions $\chi_s : \mathcal{X}_n\to \{-1,1\}$ defined as $\chi_s(x) = (-1)^{s\cdot x}$. Here, the inner product $s\cdot x$ is taken modulo $2$, i.e., $s\cdot x\coloneqq \sum_{i=1}^n s_i x_i \mod 2$.
\end{definition}

We can view $\{\chi_s\}_{s\in\mathcal{X}_n}$ as an orthonormal basis (ONB) for the space of functions $\mathbb{R}^{\mathcal{X}_n}$ with respect to the inner product $\langle \phi, \Tilde{\phi}\rangle_{\mathcal{U}_n} = \mathbb{E}_{x\sim\mathcal{U}_n} [\phi(x)\tilde{\phi}(x)]$. Then, the Fourier coefficients of \Cref{definition:fourier-coefficients} simply become the ONB expansion coefficients, $\phi = \sum_{s\in\mathcal{X}_n}\hat{\phi}(s) \chi_s$, and they satisfy the Parseval identity $\sum_{s\in\mathcal{X}_n} (\hat{\phi}(s))^2 = \mathbb{E}_{x\sim\mathcal{U}_n} [(\phi(x))^2]$. In particular, if $\phi$ is $\{-1,1\}$-valued, then the squares of its Fourier coefficients form a probability distribution over $\mathcal{X}_n$.

By orthonormality, the parity functions $\chi_s$ have exactly one non-zero Fourier coefficient. More generally, we will consider functions with few non-zero Fourier coefficients:

\begin{definition}[Fourier-sparse functions]
    Let $\phi:\mathcal{X}_n\to \mathbb{R}$. We denote the set of non-zero Fourier coefficients of $\phi$ by
    \begin{equation}
        \operatorname{supp}(\phi)
        \coloneqq \{s\in\mathcal{X}_n~|~ \hat{\phi}(s)\neq 0\}.
    \end{equation}
    If $\lvert \operatorname{supp}(\phi)\rvert = k$, then we say that $\phi$ is Fourier-$k$-sparse.
\end{definition}

When we speak of Fourier-sparse functions, we think of functions that are Fourier-$k$-sparse for some $k\ll 2^n$, for example $k\leq\mathcal{O}(1)$ or $k\leq\mathcal{O}(\poly (n))$.

\subsection{Agnostic Learning}

In agnostic learning the goal is to output the best possible approximation within some class of possible solutions, without prior structural assumptions on the observed data. 
We begin with a general definition for the agnostic learning framework, which reduces to all relevant special cases treated in this paper. For this we partially adopt the nomenclature of \cite{hu2022comparative}.

\begin{definition}[$\alpha$-agnostic learning $\mathcal{B}$ via $\mathcal{M}$ with respect to $\mathcal{D}$ from oracle $\mathsf{O}$]\label{def:agnostic-learning}
  We say that a learning algorithm $\mathcal{A}$ is an $\alpha$-agnostic learner for the benchmark class $\mathcal{B}$, via the model class $\mathcal{M}$, with respect to the distribution class $\mathfrak{D}$ over $\mathcal{X}_n\times\{0,1\}$, from oracle $\mathsf{O}$ if: 
    For all $\mathcal{D}\in\mathfrak{D}$, when given access to oracle $\mathsf{O}(\mathcal{D})$, as well as some $(\epsilon,\delta)\in (0,1)$, algorithm $\mathcal{A}$ outputs, with probability at least $1-\delta$, some model $m\in\mathcal{M}$ such that  
    \begin{equation}\label{eq:learning_criterion}
        \mathrm{err}_{\mathcal{D}}(m)\leq \alpha\cdot\mathrm{opt}_{\mathcal{D}}(\mathcal{B}) +\epsilon.
    \end{equation}
\end{definition}
Here we define the error of $m\in \mathcal{M}$ as $ \mathrm{err}_{\mathcal{D}}(m)=\mathbb{P}_{(x,y)\sim \mathcal{D}}[m(x)\neq y]$ and similarly for $b\in \mathcal{B}$.
The optimal error with respect to the benchmark  class $\mathcal{B}$ is defined as 
\begin{equation}
    \mathrm{opt}_{\mathcal{D}}(\mathcal{B})=\min_{b\in \mathcal{B}}\mathrm{err}_{\mathcal{D}}(b).
\end{equation}

In words, the criterion for learning as given in \Cref{eq:learning_criterion} says that the error achieved by the output model should be close to a multiple of the best achievable error by any model in the benchmark class. 
Usually $\mathcal{B}$ and $\mathcal{M}$ are taken to be subsets of $\{0,1\}^{\mathcal{X}_n}$, but we will also consider model classes $\mathcal{M}$ containing randomized hypotheses. (Note: We use the terms ``model'' and ``hypothesis'' interchangeably.) In this case, the definition of $\mathrm{err}_{\mathcal{D}}(m)$ is naturally adapted by considering the probability over both the randomness of $(x,y)$ and the randomness of $m(x)$, that is $ \mathrm{err}_{\mathcal{D}}(m)=\mathbb{P}_{(x,y)\sim \mathcal{D}, b\sim m(x)}[b\neq y]$ and similarly for $b\in \mathcal{B}$.
We also make a basic distinction between \emph{proper} and \emph{improper} learners. In the context of \Cref{def:agnostic-learning} we call a learner proper if $\mathcal{M}\subseteq\mathcal{B}$ and improper otherwise.

In any of the learning frameworks covered by \Cref{def:agnostic-learning}, we say that a learner is query/sample-efficient if it requires at most $\mathcal{O}(\poly (n, \nicefrac{1}{\varepsilon}, \nicefrac{1}{\delta}))$ oracle calls. Accordingly, the learner is said to be computationally efficient if it it requires computation time at most $\mathcal{O}(\poly (n, \nicefrac{1}{\varepsilon}, \nicefrac{1}{\delta}))$. 

We can furthermore distinguish between learners with access to random examples, learners with statistical query access, and learners access to quantum examples. This corresponds to specific oracles $\mathsf{O}$ in the above definition. On the classical side, the \emph{random examples oracle}, when queried, returns a randomly drawn example $(x,y)\sim \mathcal{D}$. 
The \emph{statistical query (SQ) oracle} accepts two inputs, a bounded function $g$ defined on $\mathcal{X}_n\times\{0,1\}$ and a tolerance parameter $\tau$, and returns $\mu$ such that
\begin{equation}
        \left\lvert \mu-\mathbb{E}_{(x,y)\sim \mathcal{D}} \left[ g\left(x,y\right)\right]\right\rvert\leq \tau.
\end{equation}
We postpone the discussion of quantum example oracles to \Cref{subsection:quantum-oracles} and \Cref{section:mixture-of-superpositions-examples}. 
Quantum versions of statistical queries are discussed in \Cref{subsection:functional-qsq-learning} and \Cref{subsection:distributional-qsq-learning}.

By specifying the classes $ \mathfrak{D}$ we define some special cases, which we will refer to as follows:
\begin{itemize}
    \item \textbf{Distributional agnostic learning:} A learning algorithm $\mathcal{A}$ is said to be a \emph{distributional} $\alpha$-agnostic learner if the criteria of \cref{def:agnostic-learning} are met without further assumptions on $\mathfrak{D}$. 
    \item \textbf{Functional agnostic learning:} A learning algorithm $\mathcal{A}$ is said to be a \emph{functional} $\alpha$-agnostic learner if the criteria of \cref{def:agnostic-learning} are met under the assumption that $\mathfrak{D}$ only contains distributions $\mathcal{D}:\mathcal{X}_n \times \{0,1\} \to [0,1]$ of the form $\mathcal{D}=(\mathcal{D}_{\mathcal{X}_n}, f)$ for any Boolean function $f:\mathcal{X}_n\to \{0,1\}$.
    \item \textbf{Realizable PAC learning:}  A learning algorithm $\mathcal{A}$ is said to be a realizable PAC-learner if the criteria of definition \cref{def:agnostic-learning} are met under the assumption that $\mathfrak{D}$ only contains distributions $\mathcal{D}:\mathcal{X}_n \times \{0,1\} \to [0,1]$ of the form $\mathcal{D}=(\mathcal{D}_{\mathcal{X}_n}, f)$ for any Boolean function $f\in\mathcal{B}$.
\end{itemize}

Finally, in this work we are mostly interested in the more restrictive setting where the marginal distribution over the first $n$ bits of $\mathcal{D}$ is the uniform distribution. In this case, we speak of (distributional agnostic, functional agnostic, or realizable) learning w.r.t.~uniformly random inputs. 
If we make now assumptions on the input marginal, we speak of distribution-independent learning.

\subsection{Quantum Data Oracles}\label{subsection:quantum-oracles}

Data oracles $\mathsf{O}$ can go beyond random examples and statistical queries. In this section we will consider different types of quantum oracles. We begin with the \emph{pure superposition oracle}.

\begin{definition}[Pure superposition oracle]\label{definition:functional-superposition-quantum-examples}
    Let $\mathcal{D}_{\mathcal{X}_n}$ be some distribution over $\mathcal{X}_n$ and  $f:\mathcal{X}_n\to\{0,1\}$. The pure superposition oracle, when queried, will return a copy of the pure superposition state  given by
    \begin{equation}\label{eq:superposition-functional-example}
        \ket{\psi_{(\mathcal{D}_{\mathcal{X}_n}, f)}}
        = \sum_{x\in \mathcal{X}_n} \sqrt{\mathcal{D}_{\mathcal{X}_n}(x)} \ket{x,f(x)}.
    \end{equation}
\end{definition}
It can be verified that examples of the form \cref{eq:superposition-functional-example} give rise to the random examples we considered in PAC and functional agnostic learning when the state is measured in the computational basis. This amounts to projecting onto one of the states in the superposition with probability given by $\mathcal{D}_{\mathcal{X}_n}(x)$. 
As in the case with classical random examples, we can move towards the distributional agnostic setting by considering random classification noise. There is no unique way to quantize noisy classical examples and therefore we provide the following three alternative definitions.

\begin{definition}[Noisy functional quantum examples]\label{definition:noisy-functional-quantum-examples}
    Let $\mathcal{D} = (\mathcal{D}_{\mathcal{X}_n}, f)$ be a probability distribution over $\mathcal{X}_n\times \{0,1\}$ with deterministic labeling function $f:\mathcal{X}_n\to\{0,1\}$. Let $0\leq \eta < \nicefrac{1}{2}$ be a noise parameter. 
    \begin{itemize}
        \item[(i)] A \emph{mixed $\eta$-noisy functional quantum example for $\mathcal{D} = (\mathcal{D}_{\mathcal{X}_n}, f)$} is a copy of the mixed $(n+1)$-qubit state
        \begin{equation}\label{eq:mixed-noisy-functional-quantum-example}
            \rho_{(\mathcal{D}_{\mathcal{X}_n}, f),\eta}
            \coloneqq (1-\eta)\ket{\psi_{(\mathcal{D}_{\mathcal{X}_n},f)}}\bra{\psi_{(\mathcal{D}_{\mathcal{X}_n},f)}} + \eta\ket{\psi_{(\mathcal{D}_{\mathcal{X}_n},f\oplus 1)}}\bra{\psi_{(\mathcal{D}_{\mathcal{X}_n},f\oplus 1)}} ,
        \end{equation}
        where $\oplus$ denotes additional modulo $2$.
        \item[(ii)] A \emph{pure $\eta$-noisy functional quantum example for $\mathcal{D} = (\mathcal{D}_{\mathcal{X}_n}, f)$} is a pure superposition quantum example $\ket{\psi_{\mathcal{D}_\eta}}$ according to \Cref{definition:superposition-quantum-examples}, where $\mathcal{D}_\eta$ is the probability distribution over $\mathcal{X}_n\times \{0,1\}$ defined as
        \begin{equation}
            \mathcal{D}_\eta (x,f(x)) = (1-\eta)\mathcal{D}_{\mathcal{X}_n}(x) ,~
            \mathcal{D}_\eta (x,f(x)\oplus 1) = \eta\mathcal{D}_{\mathcal{X}_n}(x),\quad\forall x\in\mathcal{X}_n .
        \end{equation}
        That is, such a quantum example is a copy of the pure $(n+1)$-qubit state
        \begin{equation}\label{eq:pure-noisy-functional-quantum-example}
            \ket{\psi_{\mathcal{D}_\eta}}
            \coloneqq \frac{1}{\sqrt{2^n}}\sum_{x\in\{0,1\}^n} \left(\sqrt{1-\eta}\ket{x,f(x)}+\sqrt{\eta}\ket{x,f(x)\oplus 1}\right) ,
        \end{equation}
        where $\oplus$ denotes additional modulo $2$.
        \item[(iii)] A \emph{mixture-of-superpositions $\eta$-noisy functional quantum example for $\mathcal{D} = (\mathcal{D}_{\mathcal{X}_n}, f)$} is a copy of the mixed $(n+1)$-qubit state
        \begin{equation}\label{eq:mixture-of-superpositions-noisy-functional-quantum-example}
            \rho_{\mathcal{D}_\eta}
            \coloneqq \mathbb{E}_{\{e_x\}_{x\in\mathcal{X}_n}} \left[ \ket{\psi_{(\mathcal{D}_{\mathcal{X}_n},f \oplus e)}}\bra{\psi_{(\mathcal{D}_{\mathcal{X}_n},f\oplus e)}} \right] ,
        \end{equation}
        where the $e_x$, $x\in\mathcal{X}_n$, are i.i.d.~random variables with $\mathbb{P}[e_x = 0 ]=1-\eta=1-\mathbb{P}[e_x = 1 ]$ for all $x\in\mathcal{X}_n$.
    \end{itemize}
\end{definition}

The mixed noisy quantum example in \Cref{eq:mixed-noisy-functional-quantum-example} can be interpreted as a version of the noiseless quantum superposition example $\ket{\psi_{(\mathcal{D}_{\mathcal{X}_n},f)}}$ subject to bit-flip noise and was suggested in \cite{cross2015quantum}.
The pure noisy quantum example in \Cref{eq:pure-noisy-functional-quantum-example} is the notion of quantum example obtained directly from \Cref{definition:superposition-quantum-examples}.
Finally, \Cref{eq:mixture-of-superpositions-noisy-functional-quantum-example} is the notion of noisy quantum example considered in \cite{grilo2019learning}.
Note that all three notions reproduce classical random examples with label noise when performing computational basis measurements. In contrast to the classical case where random classification noise makes learning parities hard, this does not happen in the case of noisy quantum examples \cite{cross2015quantum, grilo2019learning}.

The noisy quantum examples from \Cref{definition:noisy-functional-quantum-examples} can be viewed as a first step towards quantum data for distributional agnostic quantum learning: While there still is an underlying function, no deterministic function describes the data perfectly. Starting from the noisy case, one may attempt to extend the definitions to obtain fully distributional agnostic quantum examples. In fact taking \Cref{eq:pure-noisy-functional-quantum-example} as a starting point, one would arrive at the following definition for such examples: 

\begin{definition}[Pure superposition quantum examples]\label{definition:superposition-quantum-examples}
    Let $\mathcal{D}$ be a probability distribution over $\mathcal{X}_n\times\{0,1\}$.
    Then, a \emph{pure superposition example for $\mathcal{D}$} is a copy of the pure $(n+1)$-qubit state 
    \begin{equation}\label{eq:superposition-quantum-examples}
        \ket{\psi_{\mathcal{D}}}
        = \ket{\psi_{(\mathcal{D}_{\mathcal{X}_n}, \varphi)}}
        \coloneqq \sum_{x\in\{0,1\}^n}\sum_{y\in\{0,1\}} \sqrt{\mathcal{D}(x,y)}\ket{x,y}
        = \sum_{x\in\{0,1\}^n} \sqrt{\mathcal{D}\rvert_{\mathcal{X}_n}(x) }\left(\sqrt{1-\varphi(x)}\ket{x,0} + \sqrt{\varphi(x)}\ket{x,1} \right) .
    \end{equation}
    Accordingly, a \emph{pure superposition quantum example oracle for $\mathcal{D}$} is an oracle that when queried outputs a copy of $\ket{\psi_{\mathcal{D}}}$.
\end{definition}

While the literature contains many examples for the power of functional superposition examples (\Cref{definition:functional-superposition-quantum-examples}) in distribution-dependent learning, no comparable results exist for the distributional superposition examples \Cref{definition:superposition-quantum-examples} in the distributional agnostic case.

\subsection{Interactive Verification of Learning}

\cite{goldwasser2021interactive} recently introduced a formal framework for verification of machine learning in the spirit of interactive proofs. 
An important aspect of this framework is that a difference in power between verifier and prover can arise from them having data access via different oracles. 
Thus, with the following small modification of \cite[Definition 4]{goldwasser2021interactive}, we give a definition for interactively verifying learning in which we allow for classical as well as quantum data access and processing.

\begin{definition}[Interactive verification of $\alpha$-agnostic learning -- Classical and/or quantum]\label{definition:interactive-verification-of-learning}
    Let $\mathcal{F}\subseteq\{0,1\}^{\mathcal{X}_n}$ be a benchmark class.
    Let $\mathfrak{D}$ be a family of probability distributions over $\mathcal{X}_n\times\{0,1\}$.
    Let $\alpha\geq 1$.
    We say that \emph{$\mathcal{F}$ is $\alpha$-agnostic verifiable with respect to $\mathfrak{D}$ using classical or quantum oracles $\mathsf{O}_V$ and $\mathsf{O}_P$} if there exists a pair of classical or quantum algorithms $(V,P)$ that satisfy the following conditions for every input accuracy parameter $\varepsilon\in (0,1)$ and for every confidence parameter $\delta\in (0,1)$:
    \begin{itemize}
        \item \textbf{Completeness:} For any $\mathcal{D}\in\mathfrak{D}$, the random hypothesis $h:\mathcal{X}_n\to \{0,1\}$ that $V(\varepsilon,\delta)$ outputs after interacting with the honest prover $P$, assuming $V(\varepsilon,\delta)$ has access to $\mathsf{O}_V (\mathcal{D})$ and $P(\varepsilon,\delta)$ has access to $\mathsf{O}_P (\mathcal{D})$, satisfies
        \begin{equation}
            \mathbb{P}\left[h\neq\mathrm{reject}~\wedge ~ \left(\operatorname{err}_{\mathcal{D}}(h)\leq \alpha\cdot \operatorname{opt}_{\mathcal{D}}(\mathcal{F}) + \varepsilon\right)\right] 
            \geq 1-\delta .
        \end{equation}
        \item \textbf{Soundness:} For any $\mathcal{D}\in\mathfrak{D}$ and for any (possibly unbounded) dishonest prover $P'$, the random hypothesis $h:\mathcal{X}_n\to \{0,1\}$ that $V(\varepsilon,\delta)$ outputs after interacting with $P'$, assuming $V(\varepsilon,\delta)$ has access to $\mathsf{O}_V (\mathcal{D})$ and $P'(\varepsilon,\delta)$ has access to $\mathsf{O}_P (\mathcal{D})$, satisfies
        \begin{equation}
            \mathbb{P}\left[h\neq\mathrm{reject}~\wedge ~ \left(\operatorname{err}_{\mathcal{D}}(h)> \alpha\cdot \operatorname{opt}_{\mathcal{D}}(\mathcal{F}) + \varepsilon\right)\right] 
            \leq \delta .
        \end{equation}
    \end{itemize}
    Moreover:
    \begin{itemize}
        \item If the above can be achieved with a pair $(V,P)$ such that $V(\varepsilon,\delta)$ makes at most $\mathcal{O}(\poly (n, \nicefrac{1}{\varepsilon}, \nicefrac{1}{\delta}))$ queries to $\mathsf{O}_V (\mathcal{D})$, such that $P(\varepsilon,\delta)$ makes at most $\mathcal{O}(\poly (n, \nicefrac{1}{\varepsilon}, \nicefrac{1}{\delta}))$ queries to $\mathsf{O}_P (\mathcal{D})$, and such that both $V(\varepsilon,\delta)$ and $P(\varepsilon,\delta)$ have running time $\mathcal{O}(\poly (n, \nicefrac{1}{\varepsilon}, \nicefrac{1}{\delta}))$, then we say that \emph{$\mathcal{F}$ is efficiently $\alpha$-agnostic verifiable with respect to $\mathfrak{D}$ using oracles $\mathsf{O}_V$ and $\mathsf{O}_P$}.
        \item If additionally one can ensure that the random hypothesis $h:\mathcal{X}_n\to \{0,1\}$ that $V(\varepsilon,\delta)$ outputs after interacting with any (possibly unbounded) dishonest prover $P'$, assuming $V(\varepsilon,\delta)$ has access to $\mathsf{O}_V (\mathcal{D})$ and $P'(\varepsilon,\delta)$ has access to $\mathsf{O}_P (\mathcal{D})$, satisfies $h\in \{\mathrm{reject}\}\cup \mathcal{F}$ almost surely, then we say that \emph{$\mathcal{F}$ is proper $\alpha$-agnostic verifiable with respect to $\mathfrak{D}$ using oracles $\mathsf{O}_V$ and $\mathsf{O}_P$}.
    \end{itemize}
\end{definition}

\begin{figure}[h]
    \centering
    \includegraphics[width=0.75\textwidth]{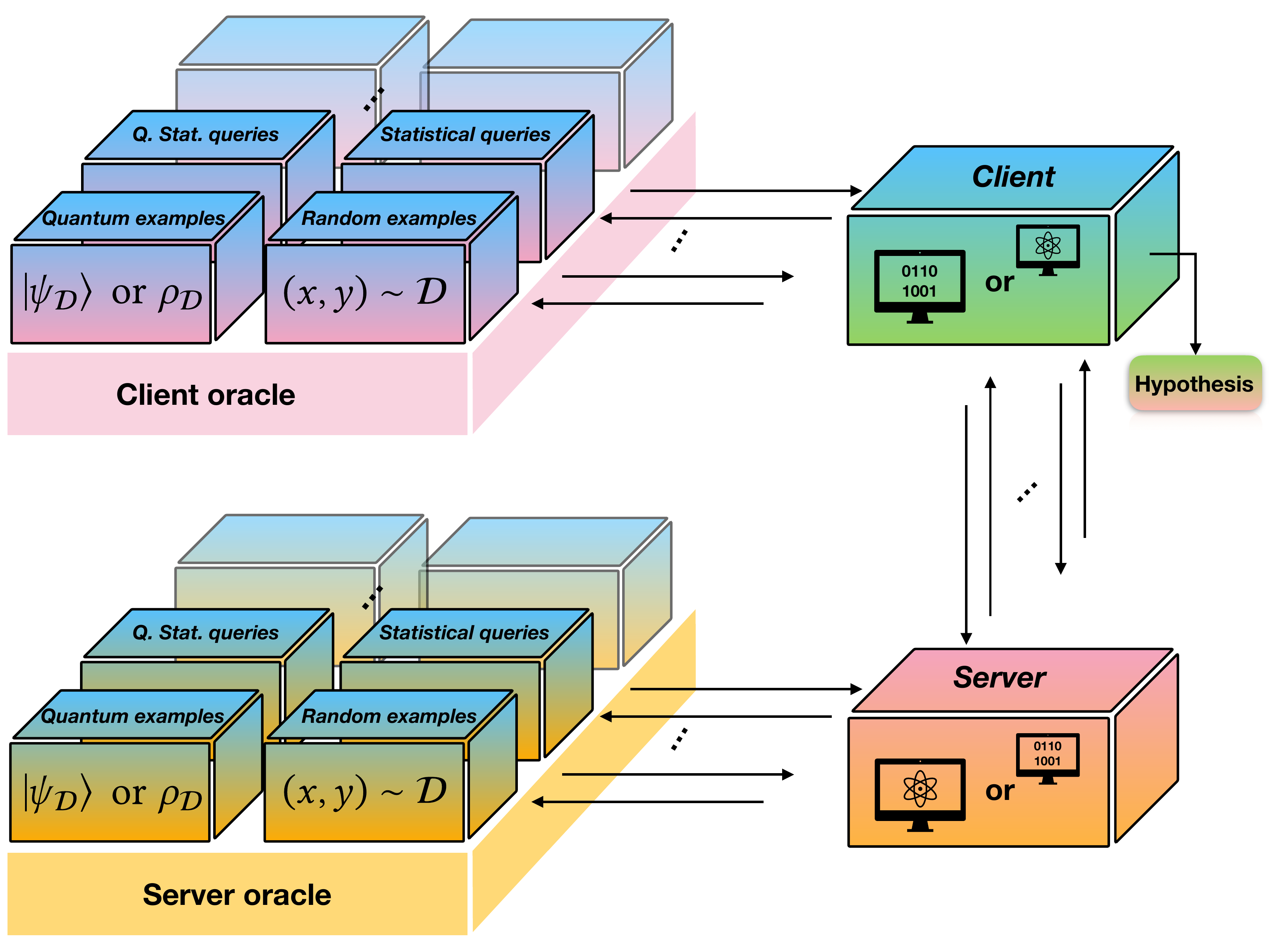}
    \caption{\textbf{Interactive verification of learning:} 
As per \Cref{definition:interactive-verification-of-learning}, we consider the setting in a which a client, with some type of oracle access, interacts with an untrusted server with access to a different oracle. The goal of the client is to solve a learning problem via interaction with the untrusted server. In general, both the client and the server could have access to either a classical or quantum computer, and one could consider any well-defined oracles. 
In this work we consider the setting in which the client only has access to a classical computer and classical data oracle, but the server has access to both a quantum computer and some sort of quantum data oracle.
}
    \label{fig:verification-framework}
\end{figure}

Our formulation of \Cref{definition:interactive-verification-of-learning}, which we illustrate in \Cref{fig:verification-framework}, is quite general in that it covers a variety of scenarios when focusing on specific choices of oracles.
For instance, if we consider only classical oracles and algorithms, then we recover \cite[Definition 4]{goldwasser2021interactive}. 
Here, \cite{goldwasser2021interactive} explored mainly two scenarios: On the one hand, they considered a random example oracle $\mathsf{O}_V$ and a membership query oracle $\mathsf{O}_P$. On the other hand, they explored the case in which both $\mathsf{O}_V$ and $\mathsf{O}_P$ are random example oracles. Albeit not done in prior work, one might also consider classical SQ oracles.

For the purposes of our work, the case of interest is that of a classical verifier and a quantum prover. That is, we will take $V$ to be a classical algorithm, with access to a classical data oracle $\mathsf{O}_V$, whereas the quantum prover $P$ (as well as dishonest provers $P'$) has access to a quantum data oracle $\mathsf{O}_P$.
Concretely, the following two scenarios will be our main focus: First, we are interested in the capabilities of a classical verifier $V$, with access to a classical SQ oracle $\mathsf{O}_V$, when interacting with a quantum prover, who has access to a quantum statistical query oracle $\mathsf{O}_P$. Second, we study a classical verifier $V$, with access to a classical random example oracle $\mathsf{O}_V$, when interacting with a quantum prover, who has access to a quantum example oracle $\mathsf{O}_P$.
We note, however, that it may also be interesting to investigate scenarios in which both $\mathsf{O}_V$ and $\mathsf{O}_P$ are (possibly different) quantum data oracles. 

\section{Mixture-of-Superpositions Examples}\label{section:mixture-of-superpositions-examples}

As mentioned when introducing the three versions of noisy functional quantum examples, there is no canonical way to quantize any given classical example. Here, we propose a new type of quantum examples which we call \emph{mixture-of-superpositions} quantum examples. This can be seen as extending the noisy example defined in \Cref{eq:mixture-of-superpositions-noisy-functional-quantum-example}.

\begin{definition}[Mixture-of-superpositions quantum examples for distributional agnostic learning]\label{definition:mixture-of-superpositions-quantum-example}
    Let $\mathcal{D}$ be a probability distribution over $\mathcal{X}_n\times\{0,1\}$.
    Let $F_\mathcal{D}$ be the probability distribution over $\{0,1\}^{\mathcal{X}_n}$ defined by sampling $f(x)$ from the conditional label distribution independently for each $x\in\mathcal{X}_n$. 
    That is, for any $\Tilde{f}:\mathcal{X}_n\to \{0,1\}$,
    \begin{equation}\label{eq:induced-distribution-over-functions}
        \mathbb{P}_{f\sim F_\mathcal{D}} [f=\Tilde{f}]
        = \prod_{z\in\mathcal{X}_n} \mathbb{P}_{(x,y)\sim \mathcal{D}}[ \tilde{f}(z)=y | x=z ]
        = \prod_{z\in\mathcal{X}_n} \left( (1-\varphi(z))(1-\tilde{f}(z)) + \varphi(z) \tilde{f}(z) \right) .
    \end{equation}
    Then, a \emph{mixture-of-superpositions quantum example for $\mathcal{D}$} is a copy of the mixed $(n+1)$-qubit state
    \begin{equation}\label{eq:mixture-of-superpositions-quantum-example}
        \rho_\mathcal{D}
        = \mathbb{E}_{f\sim F_\mathcal{D}} \left[ \ket{\psi_{(\mathcal{D}_{\mathcal{X}_n}, f)}}\bra{\psi_{(\mathcal{D}_{\mathcal{X}_n}, f)}} \right] .
    \end{equation}
    Accordingly, a \emph{mixture-of-superpositions quantum example oracle for $\mathcal{D}$} is an oracle that when queried outputs a copy of $\rho_\mathcal{D}$.
\end{definition}

Randomized quantum oracles similar in spirit to \Cref{definition:mixture-of-superpositions-quantum-example} have previously appeared in a complexity-theoretic context, see for example \cite{harrow2014uselessness, fefferman2018quantum-vs-classical, natarajan2022distribution, bassirian2022power}.
While standard quantum oracles (as for example that of \Cref{definition:superposition-quantum-examples}) can be viewed in terms of black box unitaries, randomized quantum oracles correspond to black box mixed unitary channels.
Note that \Cref{definition:mixture-of-superpositions-quantum-example} reproduces the definition of a mixture-of-superpositions noisy functional quantum example from \Cref{definition:noisy-functional-quantum-examples} when applied to a distribution $\mathcal{D}_\eta$ that is obtained by adding i.i.d.~label noise of strength $\eta \geq0$ to a distribution $(\mathcal{D}_{\mathcal{X}_n}, f)$ with a deterministic labeling function $f:\mathcal{X}_n\to\{0,1\}$. In particular, \Cref{definition:mixture-of-superpositions-quantum-example} reproduces the functional superposition example of \Cref{definition:functional-superposition-quantum-examples} for distributions of the form $(\mathcal{D}_{\mathcal{X}_n}, f)$ with Boolean $f$.
Importantly, however, \Cref{definition:mixture-of-superpositions-quantum-example} covers more general distributions, for example distributions arising from adding correlated labeling noise to a deterministic labeling.
\Cref{definition:mixture-of-superpositions-quantum-example} also reproduces the standard notion of a classical distributional agnostic random example under computational basis measurements:

\begin{lemma}\label{lemma:sanity-check-reproduce-classical-sample}
    Let $\mathcal{D}$ be a probability distribution over $\mathcal{X}_n\times\{0,1\}$.
    Performing computational basis measurements on all $n+1$ qubits of a copy of $\rho_\mathcal{D}$ produces a sample from $\mathcal{D}$.
\end{lemma}
\begin{proof}
    By definition of $\rho_\mathcal{D}$, the probability of observing an output string $(x,b)\in\mathcal{X}_n\times \{0,1\}$ when measuring all $n+1$ qubits in the computational basis is given by
    \begin{align}
        \bra{x,b}\rho_\mathcal{D}\ket{x,b}
        &= \mathbb{E}_{f\sim F_\mathcal{D}} \left[ \left\lvert \braket{x,b}{\psi_{(\mathcal{D}_{\mathcal{X}_n}, f)}} \right\rvert^2\right]\\
        &= \mathbb{E}_{f\sim F_\mathcal{D}} \left[ \mathcal{D}_{\mathcal{X}_n}(x) \delta_{b, f(x)}\right]\\
        &= \mathcal{D}_{\mathcal{X}_n}(x)\mathbb{P}_{f\sim F_\mathcal{D}} \left[ f(x)=b\right]\\
        &= \mathcal{D} (x,b),
    \end{align}
    as claimed.
\end{proof}

Thus, our mixture-of-superpositions quantum examples constitute a generalization of established definitions, both classical and quantum. 
Moreover, the probability distribution $F_\mathcal{D}$ over functions is an object that naturally appears in classical learning-theoretic proofs of reductions between distributional and functional agnostic learning, compare \cite[Appendix A]{gopalan2008agnostically} and our presentation in \Cref{appendix:classical-distributional-to-agnostic}. 

In the next sections, we investigate how the different notions of quantum data access discussed so far can be used for functional and distributional agnostic quantum learning.

\section{Functional Agnostic Quantum Learning}\label{section:functional-agnostic-quantum-learning}

In preparation for our main results and their proofs, this subsection contains a compilation of results on functional agnostic parity and Fourier-sparse learning. 
These results can be obtained by combining known results from prior work with the classical learning theory toolkit of \Cref{appendix:useful}. 
While these results are straightforward to obtain with known techniques, we provide them here to give a systematic approach to agnostic quantum learning and to prepare for the analysis of the following sections.

\subsection{Noiseless Functional Agnostic Quantum Learning}

We first recall the standard procedure of quantum Fourier sampling with pure functional superposition examples as in \Cref{eq:superposition-functional-example}, which goes back to \cite{bernstein1997quantum}:

\begin{lemma}\label{lemma:functional-quantum-Fourier-sampling}
    Let $\mathcal{D}$ be a probability distribution over $\mathcal{X}_n\times\{0,1\}$ with $\mathcal{D}=(\mathcal{U}_n, f)$ for some deterministic labeling function $f:\{0,1\}^n\to\{0,1\}$. 
    Consider the following quantum algorithm: 
    Given a single copy of $\ket{\psi_{(\mathcal{U}_n, f)}}$, first apply the unitary $H^{\otimes (n+1)}$, then measure all $n+1$ qubits in the computational basis.
    The measurement outcomes of this procedure satisfy the following:
    \begin{itemize}
        \item[(i)] The computational basis measurement on the last qubit gives outcome $0$ with probability $\nicefrac{1}{2}$ and outcome $1$ with probability $\nicefrac{1}{2}$.
        \item[(ii)] Conditioned on having observed outcome $1$ for the last qubit, the computational basis measurement on the first $n$ qubits outputs a string $s\in\{0,1\}^n$ with probability $(\hat{g}(s))^2$, with $g=(-1)^f$.
    \end{itemize}
\end{lemma}
\begin{proof}
    We do not give a proof of \Cref{lemma:functional-quantum-Fourier-sampling} at this point, since it can be shown by a standard computation and since it will follow as a special case of the more general quantum Fourier sampling results \Cref{lemma:noisy-functional-quantum-Fourier-sampling-v1}, \Cref{lemma:noisy-functional-quantum-Fourier-sampling-v2}, \Cref{lemma:noisy-functional-quantum-Fourier-sampling-v3}, and \Cref{theorem:agnostic-quantum-fourier-sampling} that we present in the next sections.
\end{proof}

The ability to sample from a probability distribution allows for an efficient construction of a succinct $\infty$-norm approximation to it. This is a consequence of the Dvoretzky-Kiefer-Wolfowitz (DKW) Theorem \cite{dvoretzky1956asymptotic, massart1990tight}. More precisely, we use a variant of the DKW Theorem for probability distributions over a discrete set, see \cite[Theorem 11.6]{kosorok2008introduction}.
This  version of empirical approximations to an unknown probability distribution can be found in \cite[Lemma 4]{kanade2019learning}, we present it here in a more detailed form:

\begin{lemma}\label{lemma:distribution-estimation-infty-norm}
    Let  $q:\{0,1\}^n\to [0,1]$ be a probability distribution and let $\tau,\delta \in (0,1)$.
    Then, there exists a classical algorithm that, given $m=\mathcal{O}\left(\tfrac{\log(1/\delta)}{\tau^2}\right)$ i.i.d.~samples from $q$, uses classical computation time $\mathcal{O}\left(n\tfrac{\log(1/\delta)}{\tau^2}\log\left(\tfrac{\log(1/\delta)}{\tau^2}\right)\right) = \tilde{\mathcal{O}}\left(n \tfrac{ \log(1/\delta)}{\tau^2}\right)$ and classical memory of size $\tilde{\mathcal{O}}\left(n\tfrac{\log(1/\delta)}{\tau^2}\right)$, and outputs, with success probability $\geq 1-\delta$, a succinctly represented empirical estimate $\tilde{q}_m : \{0,1\}^n\to [0,1]$ such that $\norm{q-\tilde{q}_m}_\infty\leq\tau$ and $\norm{\tilde{q}_m}_0\leq \mathcal{O}\left(\tfrac{\log(1/\delta)}{\tau^2}\right)$.
\end{lemma}

We note that the $\tilde{\mathcal{O}}$ hides factors logarithmic in $\nicefrac{1}{\tau}$ and doubly logarithmic in $\nicefrac{1}{\delta}$. However, the $\tilde{\mathcal{O}}$ does not hide any $n$-dependent terms here.
The same will be true for the results derived from \Cref{lemma:distribution-estimation-infty-norm}.

\begin{proof}
    Fix some ordering $e_1 \prec \ldots \prec e_{2^n}$ for the elements of $\mathcal{X}_n$. According to this ordering, define the cumulative distribution function $Q:\mathcal{X}_n\to [0,1]$ as $Q(e_k) = \sum_{\ell=1}^{k} q(e_\ell)$ for $1\leq k\leq 2^n$. Then, by definition, $q(e_1)=Q(e_1)$ and $q(e_k) = Q(e_{k}) - Q(e_{k-1})$ for all $2\leq k\leq 2^n$. 
    Given $m$ i.i.d.~examples $x_1,\ldots,x_m$ from $q$, define the empirical cumulative distribution function $\tilde{Q}_m:\mathcal{X}_n\to [0,1]$ as $\tilde{Q}_m(e_k)=\tfrac{1}{m}\sum_{i=1}^m \mathds{1}_{\{ x_i\preceq e_k\}}$ for $1\leq k\leq 2^n$ and the empirical distribution as $\tilde{q}_m:\mathcal{X}_n\to [0,1]$ as $\tilde{q}_m(e_k) = \tfrac{1}{m}\sum_{i=1}^m \mathds{1}_{\{ x_i =e_k\}}= \tilde{Q}_m(e_k) - \tilde{Q}_m(e_{k-1})$ for $1\leq k\leq 2^n$.
    
    Now, using the definition of the $\infty$-norm followed by a triangle inequality, we obtain
    \begin{align}
        \norm{q-\tilde{q}_m}_\infty
        &= \max_{1\leq k\leq 2^n} \left\lvert q(e_k) - \tilde{q}_m(e_k) \right\rvert\\
        &= \max_{1\leq k\leq 2^n} \left\lvert \left(Q(e_{k}) - \tilde{Q}_m(e_k)\right) +  \left(\tilde{Q}_m(e_{k-1}) - Q(e_{k-1})\right) \right\rvert\\
        &\leq 2 \norm{Q-\tilde{Q}_m}_\infty .
    \end{align}
    According to the DKW Theorem for cumulative distribution functions with countably many discontinuities (compare \cite[Theorem 11.6]{kosorok2008introduction}), we know that 
    \begin{equation}
        \mathbb{P}\left[\norm{Q-\tilde{Q}_m}_\infty > \frac{\tau}{2}\right]
        \leq 2\exp\left(-\frac{m\tau^2}{2}\right).
    \end{equation}
    The DKW Theorem, both its continuous and its discrete version, can be proved via VC-theory \cite{vapnik1971uniform, talagrand1994sharper, dudley1999uniform} since the class of threshold functions has VC-dimension equal to $1$.
    Combining this concentration inequality with our above reasoning, we conclude that
    \begin{equation}
        \mathbb{P}\left[\norm{q-\tilde{q}_m}_\infty > \tau\right]
        \leq 2\exp\left(-\frac{m\tau^2}{2}\right).
    \end{equation}
    Setting the right hand side equal to $\delta$ and rearranging, we see that a sample size of $m=\tfrac{2\log(\nicefrac{2}{\delta})}{\tau^2} = \mathcal{O}\left(\tfrac{\log(1/\delta)}{\tau^2}\right)$ suffices. 
    Given the sample $x_1,\ldots,x_m$, we can succinctly represent $\tilde{q}_m$ via the list $\{(x, \tilde{q}_m(x))~|~x\in \{x_1,\ldots,x_m\}\}$, which has at most $m$ entries. This is the claimed succinctness bound.
    This representation can be obtained as follows: 
    First, sort the samples according to the chosen ordering of $\{0,1\}^n$. As we have $m$ samples and determining the relative order of two samples takes time $n$, this sorting overall takes time $\mathcal{O}(n m\log (m))$.
    Second, run over the sorted list once to compute the empirical CDF for each of the samples in time $\mathcal{O}(n m\log (m))$.
    Finally, use time $\mathcal{O}(m)$ to get from the empirical CDF values to the empirical distribution values of the samples.
    This procedure overall takes time $\mathcal{O}(n m\log (m))$ and requires $\mathcal{O}(nm + m\log m) = \tilde{\mathcal{O}}(nm)$ bits of memory to store $m$ $n$-bit vectors and $m$ natural numbers in $\{0,\ldots,m\}$.
    This gives the claimed time and memory bounds.
\end{proof}

Combining this approximation result with quantum Fourier sampling, we can efficiently find a succinct approximation to the Fourier spectrum of an unknown function $f:\mathcal{X}_n\to\{0,1\}$ if we have access to copies of the corresponding pure superposition quantum example $\ket{\psi_{(\mathcal{U}_n, f)}}$. This has previously been observed in {\cite[Theorem 6]{kanade2019learning}}, where it is proven even for Fourier coefficients w.r.t.~non-uniform product distributions. Here, we restate that result for the special case of the uniform distribution, with a slightly improved copy complexity analysis:

\begin{corollary}\label{corollary:quantum-approximation-fourier-spectrum}
    Let $\mathcal{D}$ be a probability distribution over $\mathcal{X}_n\times\{0,1\}$ with $\mathcal{D}=(\mathcal{U}_n, f)$ for some deterministic labeling function $f:\{0,1\}^n\to\{0,1\}$. 
    Let $\delta,\varepsilon\in (0,1)$ and write $g=(-1)^f$.
    Then, there exists a quantum algorithm that, given $\mathcal{O}\left(\tfrac{\log(\nicefrac{1}{\delta\varepsilon^2})}{\varepsilon^4}\right)$ copies of $\ket{\psi_{\mathcal{D}}}$, uses $\mathcal{O}\left(n\tfrac{\log(\nicefrac{1}{\delta\varepsilon^2})}{\varepsilon^4}\right)$ single-qubit gates, classical computation time $\tilde{\mathcal{O}}\left(n\tfrac{\log(\nicefrac{1}{\delta \varepsilon^2})}{\varepsilon^4}\right)$, and classical memory of size $\tilde{\mathcal{O}}\left(n\tfrac{\log(\nicefrac{1}{\delta\varepsilon^2})}{\varepsilon^4}\right)$, and outputs, with success probability $\geq 1-\delta$, a succinctly represented $\Tilde{g}:\mathcal{X}_n\to [-1,1]$ such that $\norm{\Tilde{g}-\hat{g}}_\infty\leq\varepsilon$ and $\norm{\Tilde{g}}_0\leq\tfrac{4}{\varepsilon^2}$. 
\end{corollary}
\begin{proof}
    We do not present a separate proof here since this is essentially a special case of \Cref{corollary:distributional-agnostic-quantum-approximation-fourier-spectrum}. 
    In particular, from the proof of \Cref{corollary:distributional-agnostic-quantum-approximation-fourier-spectrum}, it is straightforward to see that we do not need to assume a lower bound on the accuracy $\varepsilon$ in the noiseless functional case.
\end{proof}

It is well known that functional agnostic parity learning is equivalent to identifying the heaviest Fourier coefficient of the unknown function (see \Cref{appendix:useful} for details). 
Thus, \Cref{corollary:quantum-approximation-fourier-spectrum} implies that pure superposition quantum examples allow for efficient functional agnostic quantum parity learning w.r.t.~uniformly random inputs:

\begin{corollary}\label{corollary:quantum-functional-agnostic-parity-learning}
    Let $\mathcal{D}$ be a probability distribution over $\mathcal{X}_n\times\{0,1\}$ with $\mathcal{D}=(\mathcal{U}_n, f)$ for some deterministic labeling function $f:\{0,1\}^n\to\{0,1\}$. 
    Let $\delta,\varepsilon\in (0,1)$.
    Then, there exists a quantum algorithm that, given $\mathcal{O}\left(\tfrac{\log(\nicefrac{1}{\delta\varepsilon^2})}{\varepsilon^4}\right)$ copies of $\ket{\psi_{\mathcal{D}}}$, uses $\mathcal{O}\left(n\tfrac{\log(\nicefrac{1}{\delta\varepsilon^2})}{\varepsilon^4}\right)$ single-qubit gates, classical computation time $\tilde{\mathcal{O}}\left(n\tfrac{\log(\nicefrac{1}{\delta \varepsilon^2})}{\varepsilon^4}\right)$, and classical memory of size $\tilde{\mathcal{O}}\left(n\tfrac{\log(\nicefrac{1}{\delta\varepsilon^2})}{\varepsilon^4}\right)$, and outputs, with success probability $\geq 1-\delta$, a bit string $s\in\{0,1\}^n$ such that
    \begin{equation}
        \mathbb{P}_{x\sim\mathcal{U}_n}[f(x)\neq s\cdot x]
        \leq \min\limits_{t\in\{0,1\}^n} \mathbb{P}_{x\sim\mathcal{U}_n}[f(x)\neq t\cdot x] + \varepsilon\, .
    \end{equation}
    Thus, this quantum algorithm is a functional agnostic proper quantum parity learner, assuming a uniform marginal over inputs.
\end{corollary}
\begin{proof}
    Starting from \Cref{corollary:quantum-approximation-fourier-spectrum}, this can be proven with the same steps that we later use to derive \Cref{corollary:agnostic-quantum-parity-learning} from \Cref{corollary:distributional-agnostic-quantum-approximation-fourier-spectrum}.
\end{proof}

In a similar vein, approximating the Fourier spectrum is also a powerful approach towards functional $\alpha$-agnostic learning Fourier-sparse functions. (The interested reader is again referred to \Cref{appendix:useful} for details.)
Therefore, we can also use pure superposition quantum examples for Fourier-sparse learning:

\begin{corollary}\label{corollary:quantum-functional-agnostic-fourier-sparse-learning}
    Let $\mathcal{D}$ be a probability distribution over $\mathcal{X}_n\times\{0,1\}$ with $\mathcal{D}=(\mathcal{U}_n, f)$ for some deterministic labeling function $f:\{0,1\}^n\to\{0,1\}$. 
    Let $\delta,\varepsilon\in (0,1)$.
    Then, there is a quantum algorithm that, given $\mathcal{O}\left(\tfrac{k^4\log(\nicefrac{k^2}{\delta\varepsilon^2})}{\varepsilon^4}\right)$ copies of $\rho_\mathcal{D}$, uses $\mathcal{O}\left(n\tfrac{k^4\log(\nicefrac{1}{\delta})}{\varepsilon^4}\right)$ single-qubit gates, classical computation time $\tilde{\mathcal{O}}\left(n\tfrac{k^4\log(\nicefrac{k^2}{\delta\varepsilon^2})}{\varepsilon^4}\right)$, and classical memory of size $\tilde{\mathcal{O}}\left(n\tfrac{k^4\log(\nicefrac{k^2}{\delta\varepsilon^2})}{\varepsilon^4}\right)$, and outputs, with success probability $\geq 1-\delta$, a randomized hypothesis $h:\mathcal{X}_n\to\{0,1\}$ such that
    \begin{equation}
        \mathbb{P}_{(x,b)\sim\mathcal{D}} \left[b\neq h(x)\right]
        \leq 2 \min\limits_{\substack{\tilde{f}:\mathcal{X}_n\to\{0,1\}\\\mathrm{Fourier-}k\mathrm{-sparse}}} \mathbb{P}_{(x,b)\sim\mathcal{D}}[b\neq \tilde{f}(x)] + \varepsilon\, .
    \end{equation}
    In particular, this quantum algorithm is a functional $2$-agnostic improper quantum Fourier-sparse learner, assuming a uniform marginal over inputs.
\end{corollary}
\begin{proof}
    Starting from \Cref{corollary:quantum-approximation-fourier-spectrum}, this can be proven with the same steps that we later use to derive \Cref{corollary:agnostic-quantum-fourier-sparse-learning} from \Cref{corollary:distributional-agnostic-quantum-approximation-fourier-spectrum}.
\end{proof}

Note that, since \Cref{corollary:quantum-functional-agnostic-fourier-sparse-learning} is an improper functional $2$-agnostic learning result, it immediately implies analogous guarantees for improper functional $2$-agnostic learning w.r.t.~any subclass of Fourier-sparse functions. Potential subclasses of interest may be decision trees of small depth or DNFs of small size.

In summary: The fundamental subroutine of quantum Fourier sampling is well established in the functional agnostic case. There, it leads to positive agnostic quantum learning results that go beyond the best classical algorithms for learning from random examples.

As a final point on noiseless Fourier-sparse learning, we note that \Cref{corollary:quantum-approximation-fourier-spectrum} also gives rise to an exact Fourier-sparse quantum learner in the realizable case, using $\mathcal{O}\left(k^4\log(\nicefrac{k^2}{\delta})\right)$ copies of the unknown state, see \Cref{corollary:functional-exact-fourier-sparse-learning} in \Cref{appendix:proofs} for more details.
However, this quantum sample complexity upper bound is worse than the $\mathcal{O}(k^{3/2} \log^2 (k))$ proved in \cite{arunachalam2021twonewresults}.
A potential appeal of our alternative quantum exact learning procedure, despite its worse sample complexity, may be that it can easily be adapted to the noisy case when replacing \Cref{corollary:quantum-approximation-fourier-spectrum} by one of its noisy analogues, which we discuss in the next subsection.

\subsection{Noisy Functional Agnostic Quantum Learning}\label{subsection:noisy-functional-agnostic}

Classically even basic problems of learning from random examples are believed to become hard when random label noise is added \cite{regev2009lattices,pietrzak2012cryptography}.
In contrast, quantum learning is remarkably robust to noise in the underlying quantum examples. In the special case of quantum LPN and quantum LWE, this has been proven by \cite{cross2015quantum, grilo2019learning, caro2020quantum}.
Moreover, \cite{arunachalam2020qsq} showed that certain classes of interest (parities, juntas, and DNFs) admit quantum learners with quantum statistical query access and therefore also noise-robust learning from quantum examples.
In this section, we present a general exposition of this noise-robustness for functional agnostic learning from quantum examples. For this, we will use the noisy functional agnostic quantum examples introduced in \Cref{definition:noisy-functional-quantum-examples}.

The central observation of this subsection is that all three types of noisy quantum examples still allow a quantum learner to (at least approximately) sample from the probability distribution formed by the squares of the Fourier coefficients of $g=(-1)^f$. We distribute this observation over the next three results. 
First, we note that performing the quantum Fourier sampling algorithm with a mixed noisy quantum example $\rho_{(\mathcal{U}_n, f),\eta}$ instead of the pure quantum example $\ket{\psi_{(\mathcal{U}_n,f)}}$ leads to the same outcome distribution:

\begin{lemma}\label{lemma:noisy-functional-quantum-Fourier-sampling-v1}
    Let $\mathcal{D}$ be a probability distribution over $\mathcal{X}_n\times\{0,1\}$ with $\mathcal{D}=(\mathcal{U}_n, f)$ for some deterministic labeling function $f:\{0,1\}^n\to\{0,1\}$. 
    Let $0\leq \eta <\nicefrac{1}{2}$.
    Consider the following quantum algorithm: 
    Given a single copy of $\rho_{(\mathcal{U}_n, f),\eta}$, first apply (the unitary channel for) the unitary $H^{\otimes (n+1)}$, then measure all $n+1$ qubits in the computational basis.
    The measurement outcomes of this procedure satisfy the following:
    \begin{itemize}
        \item[(i)] The computational basis measurement on the last qubit gives outcome $0$ with probability $\nicefrac{1}{2}$ and outcome $1$ with probability $\nicefrac{1}{2}$.
        \item[(ii)] Conditioned on having observed outcome $1$ for the last qubit, the computational basis measurement on the first $n$ qubits outputs a string $s\in\{0,1\}^n$ with probability $(\hat{g}(s))^2$, with $g=(-1)^f$.
    \end{itemize}
\end{lemma}
\begin{proof}
    This can be seen by a direct computation, which we show in \Cref{appendix:proofs}.
\end{proof}

In contrast, when replacing $\ket{\psi_{(\mathcal{U}_n,f)}}$ by $\ket{\psi_{\mathcal{D}_\eta}}$, the outcome distribution of quantum Fourier sampling does change. Fortunately, this change is not major and can easily be controlled if the noise strength $\eta$ is known:

\begin{lemma}\label{lemma:noisy-functional-quantum-Fourier-sampling-v2}
    Let $\mathcal{D}$ be a probability distribution over $\mathcal{X}_n\times\{0,1\}$ with $\mathcal{D}=(\mathcal{U}_n, f)$ for some deterministic labeling function $f:\{0,1\}^n\to\{0,1\}$. 
    Let $0\leq \eta <\nicefrac{1}{2}$.
    Consider the following quantum algorithm: 
    Given a single copy of $\ket{\psi_{\mathcal{D}_\eta}}$, first apply (the unitary channel for) the unitary $H^{\otimes (n+1)}$, then measure all $n+1$ qubits in the computational basis.
    The measurement outcomes of this procedure satisfy the following:
    \begin{itemize}
        \item[(i)] The computational basis measurement on the last qubit gives outcome $0$ with probability $\tfrac{1}{2}+ \sqrt{(1-\eta)\eta}$ and outcome $1$ with probability $\tfrac{1}{2} - \sqrt{(1-\eta)\eta}$.
        \item[(ii)] Conditioned on having observed outcome $1$ for the last qubit, the computational basis measurement on the first $n$ qubits outputs a string $s\in\{0,1\}^n$ with probability $(\hat{g}(s))^2$, with $g=(-1)^f$.
    \end{itemize}
\end{lemma}
\begin{proof}
    This can be seen by a direct computation, which we show in \Cref{appendix:proofs}.
\end{proof}

Note that in particular the conditional distribution of the outcome for the first $n$ qubits, conditioned on observing outcome $1$ for the last qubit, is the same as in the noiseless case. Only the probabilty of success for quantum Fourier sampling decreases from $\tfrac{1}{2}$ to $\tfrac{1}{2} - \sqrt{(1-\eta)\eta}$.

Finally, using a copy of $\rho_{\mathcal{D}_\eta}$ instead of $\ket{\psi_{(\mathcal{U}_n,f)}}$ changes the outcome distribution of quantum Fourier sampling as follows:

\begin{lemma}\label{lemma:noisy-functional-quantum-Fourier-sampling-v3}
    Let $\mathcal{D}$ be a probability distribution over $\mathcal{X}_n\times\{0,1\}$ with $\mathcal{D}=(\mathcal{U}_n, f)$ for some deterministic labeling function $f:\{0,1\}^n\to\{0,1\}$. 
    Let $0\leq \eta <\nicefrac{1}{2}$.
    Consider the following quantum algorithm: 
    Given a single copy of $\rho_{\mathcal{D}_\eta}$, first apply (the unitary channel for) the unitary $H^{\otimes (n+1)}$, then measure all $n+1$ qubits in the computational basis.
    The measurement outcomes of this procedure satisfy the following:
    \begin{itemize}
        \item[(i)] The computational basis measurement on the last qubit gives outcome $0$ with probability $\tfrac{1}{2}$ and outcome $1$ with probability $\tfrac{1}{2}$.
        \item[(ii)] Conditioned on having observed outcome $1$ for the last qubit, the computational basis measurement on the first $n$ qubits outputs a string $s\in\{0,1\}^n$ with probability $\tfrac{4\eta - 4\eta^2}{2^n} + (1-2\eta)^2\hat{g}(s)^2$, with $g=(-1)^f$.
    \end{itemize}
\end{lemma}
\begin{proof}
    This can be seen by a computation similar to that in the proof of Theorem 4.1 in the arXiv Version 1 of \cite{grilo2019learning}. 
    It can also be obtained as a special case of \Cref{theorem:agnostic-quantum-fourier-sampling}.
\end{proof}

Thus, in the case of a mixture-of-superpositions noisy functional quantum example, the sampling still succeeds with probability $\tfrac{1}{2}$. However, after conditioning on that success, the probability distribution that is being sampled from is a perturbed version of the probability distribution of interest. Fortunately, the perturbation is exponentially-in-$n$ small in $\infty$-norm, and moreover is known exactly if the noise strength $\eta$ is known.

We can now combine \Cref{lemma:distribution-estimation-infty-norm} with either \Cref{lemma:noisy-functional-quantum-Fourier-sampling-v1}, \Cref{lemma:noisy-functional-quantum-Fourier-sampling-v2}, or \Cref{lemma:noisy-functional-quantum-Fourier-sampling-v3} and get corresponding analogues of \Cref{corollary:quantum-approximation-fourier-spectrum} (as well as of \Cref{corollary:quantum-functional-agnostic-parity-learning} and \Cref{corollary:quantum-functional-agnostic-fourier-sparse-learning}). 
In the case of mixed $\eta$-noisy functional quantum examples $\rho_{(\mathcal{D}_{\mathcal{X}_n}, f),\eta}$, we obtain exactly the same guarantees and complexity bounds as in \Cref{corollary:quantum-approximation-fourier-spectrum}. 
For pure $\eta$-noisy functional quantum examples $\ket{\psi_{\mathcal{D}_\eta}}$, the guarantees are the same as in \Cref{corollary:quantum-approximation-fourier-spectrum}, but the all asymptotic complexity bounds increase by a factor of $(\tfrac{1}{2} - \sqrt{(1-\eta)\eta})^{-4}\leq (\tfrac{1}{2} - \eta)^{-4}$.
Finally, if we work with mixture-of-superpositions $\eta$-noisy functional quantum examples $\rho_{\mathcal{D}_\eta}$, we have to change the statement of \Cref{corollary:quantum-approximation-fourier-spectrum} by adding the assumption that $\varepsilon > 2^{-(\tfrac{n}{2}-3)}\sqrt{\eta (1-\eta)}$ and replacing the sparsity guarantee by $\norm{\tilde{g}}_0 \leq \tfrac{16}{\varepsilon^2}$. (That way, it becomes a special case of \Cref{corollary:distributional-agnostic-quantum-approximation-fourier-spectrum} , which we establish later.)
Note that in these noisy learning settings, the goal remains the same as in noiseless functional agnostic quantum learning. That is, despite the training data being noisy, the performance is still measured according to the probability of misclassifying a fresh noiseless sample: 
If we consider a noisy functional distribution $\mathcal{D}=(\mathcal{U}_n, \eta + (1-2\eta) f)$ for some $f:\mathcal{X}_n\to\{0,1\}$ and noise strength $\eta\in [0,\nicefrac{1}{2})$, then the model $m$ produced by a (classical or quantum) $\alpha$-agnostic learner relative to the benchmark class $\mathcal{B}$ should satisfy
\begin{equation}
    \mathrm{err}_{(\mathcal{U}_n, f)}(m)
    \leq \alpha\cdot\mathrm{opt}_{(\mathcal{U}_n, f)}(\mathcal{B}) +\epsilon
\end{equation}
with success probability $\geq 1-\delta$. Notice that errors are measured w.r.t.~the noiseless distribution $(\mathcal{U}_n, f)$, rather than w.r.t.~the noisy data-generating distribution $\mathcal{D}$.

To conclude our discussion of noisy functional agnostic learning, let us remark that for mixed functional quantum examples or for mixture-of-superpositions noisy functional quantum examples, the respective procedures for (approximate) quantum Fourier sampling and the quantum learners derived from it do not require explicit knowledge of the noise strength $\eta$ in advance. In fact, when working with $\rho_{(\mathcal{D}_{\mathcal{X}_n}, f),\eta}$, the noise strength does not matter at all.
In the case of learning from copies of $\rho_{\mathcal{D}_\eta}$, it suffices to know an upper bound $0\leq\eta \leq\eta_b<\nicefrac{1}{2}$ on the noise strength, at least if that upper bound satisfies $\lvert \eta_b - \tfrac{1}{2} \rvert\geq \Omega (\tfrac{1}{\poly (n)})$.
In contrast, our proposed quantum learning procedure for noisy functional agnostic learning from copies of $\ket{\psi_{\mathcal{D}_\eta}}$ does rely on knowing $\eta$ at least approximately. More precisely, we need to approximately know $\sqrt{(1-\eta)\eta}$. Fortunately, based on \Cref{lemma:noisy-functional-quantum-Fourier-sampling-v2}, we can easily obtain approximations to this quantity. Moreover, if a noise strength upper bound $0\leq\eta \leq\eta_b<\nicefrac{1}{2}$ is known in advance, we can also estimate $\eta$:

\begin{corollary}\label{corollary:quantum-noise-rate-learning-v2}
    Let $\mathcal{D}$ be a probability distribution over $\mathcal{X}_n\times\{0,1\}$ with $\mathcal{D}=(\mathcal{U}_n, f)$ for some deterministic labeling function $f:\{0,1\}^n\to\{0,1\}$. 
    Let $\delta,\varepsilon, \tilde{\varepsilon}\in (0,1)$.
    Then, there is a quantum algorithm that, given an upper bound $0\leq\eta\leq \eta_b\in [0,\tfrac{1}{2})$ on the noise strength and $\mathcal{O}\left(\tfrac{\log (\nicefrac{1}{\delta})}{\tilde{\varepsilon}^2}\right)$ copies of $\ket{\psi_{\mathcal{D}_\eta}}$, uses $\mathcal{O}\left(n \tfrac{\log (\nicefrac{1}{\delta})}{\tilde{\varepsilon}^2}\right)$ single-qubit gates, classical computation time $\tilde{\mathcal{O}}\left(n\tfrac{\log (\nicefrac{1}{\delta})}{\tilde{\varepsilon}^2}\right)$, and classical memory of size $\tilde{\mathcal{O}}\left(n\tfrac{\log (\nicefrac{1}{\delta})}{\tilde{\varepsilon}^2}\right)$, and outputs, with success probability $\geq 1-\delta$, an estimate $\hat{\xi}\in [0, \sqrt{(1-\eta_b)\eta_b}]$ such that $\lvert \hat{\xi}-\sqrt{(1-\eta)\eta}\rvert\leq\tilde{\varepsilon}$.
    
    Moreover, there is a quantum algorithm that, given an upper bound $0\leq\eta\leq \eta_b\in [0,\tfrac{1}{2})$ on the noise
    strength and $\mathcal{O}\left(\tfrac{\log (\nicefrac{1}{\delta})}{\varepsilon^2}\cdot \tfrac{\eta_b}{\left(\eta_b-\tfrac{1}{2}\right)^2}\right)$ copies of $\ket{\psi_{\mathcal{D}_\eta}}$, uses $\mathcal{O}\left(n \tfrac{\log (\nicefrac{1}{\delta})}{\varepsilon^2}\cdot \tfrac{\eta_b}{\left(\eta_b-\tfrac{1}{2}\right)^2}\right)$ single-qubit gates, classical computation time $\tilde{\mathcal{O}}\left(n\tfrac{\log (\nicefrac{1}{\delta})}{\varepsilon^2}\cdot \tfrac{\eta_b}{\left(\eta_b-\tfrac{1}{2}\right)^2}\right)$, and classical memory of size $\tilde{\mathcal{O}}\left(n\tfrac{\log (\nicefrac{1}{\delta})}{\varepsilon^2}\cdot \tfrac{\eta_b}{\left(\eta_b-\tfrac{1}{2}\right)^2}\right)$, and outputs, with success probability $\geq 1-\delta$, an estimate $\hat{\eta}\in [0,\eta_b]$ such that $\lvert\hat{\eta}-\eta\rvert\leq \varepsilon$.
\end{corollary}
\begin{proof}
    We give a complete proof in \Cref{appendix:proofs}.
\end{proof}

The results of this section demonstrate that the use of quantum training data goes beyond functional agnostic learning. In fact, learning problems that are, at least with current methods, classically hard to solve from noisy data remain efficiently solvable for quantum learners with access to noisy quantum data. 

\subsection{Functional Agnostic Quantum Statistical Query Learning}\label{subsection:functional-qsq-learning}

We already mentioned \cite{arunachalam2020qsq} because of its implications for learning from noisy examples. In this subsection, we consider agnostic quantum statistical query learning as an interesting problem in itself.
First, we recall the already established definition of QSQs in the functional case, derived from superposition examples:

\begin{definition}[Functional quantum statistical queries {\cite{arunachalam2020qsq}}]\label{definition:functional-qsq}
    Let $\mathcal{D}$ be a probability distribution over $\mathcal{X}_n\times\{0,1\}$ with $\mathcal{D}=(\mathcal{D}_{\mathcal{X}_n}, f)$ for some deterministic labeling function $f:\{0,1\}^n\to\{0,1\}$.
    A \emph{(functional) quantum statistical query (QSQ) oracle for $\mathcal{D}$} produces, when queried with a bounded $(n+1)$-qubit observable $M=M^\dagger\in\mathcal{B}((\mathbb{C}^2)^{\otimes (n+1)})$ satisfying $\norm{M}\leq 1$ and with a tolerance parameter $\tau> 0$, a number $\mu\in\mathbb{R}$ such that
    \begin{equation}
        \left\lvert \mu - \bra{\psi_{(\mathcal{D}_{\mathcal{X}_n}, f)}} M \ket{\psi_{(\mathcal{D}_{\mathcal{X}_n}, f)}} \right\rvert
        \leq \tau ,
    \end{equation}
    where $\ket{\psi_{(\mathcal{D}_{\mathcal{X}_n}, f)}}$ is a (functional) pure superposition example as in \Cref{definition:functional-superposition-quantum-examples}.
\end{definition}

When interested in matters of quantum computational efficiency (rather than query complexity), one may additionally require that the observables $M$ used in the QSQ queries are efficiently implementable. This is, for instance, important when considering reductions between noisy quantum learning and QSQ learning.

Next, we recall that functional agnostic QSQs are sufficient to perform the Goldreich-Levin algorithm for finding a list of heavy Fourier coefficients. This has previously been observed in\cite[Theorem 4.4]{arunachalam2020qsq}, the following version is a variant that can be obtained as special case of our later \Cref{theorem:distributional-agnostic-qsq-GL}.

\begin{theorem}[{\cite[Theorem 4.4]{arunachalam2020qsq}}]\label{theorem:functional-agnostic-qsq-GL}
    Let $\mathcal{D}$ be a probability distribution over $\mathcal{X}_n\times\{0,1\}$ with $\mathcal{D}=(\mathcal{U}_n, f)$ for some deterministic labeling function $f:\{0,1\}^n\to\{0,1\}$. 
    Let $\varepsilon\in (0,1)$. 
    Then, there exists an algorithm that, using $\mathcal{O}\left(\tfrac{n}{\varepsilon^2}\right)$ functional QSQs of tolerance $\nicefrac{\varepsilon^2}{8}$ for observables that can be implemented with $\mathcal{O}(n)$ single-qubit gates, classical computation time $\tilde{\mathcal{O}}\left(\tfrac{n}{\varepsilon^2}\right)$, and classical memory of size $\tilde{\mathcal{O}}\left(\tfrac{n^2}{\varepsilon^2}\right)$, and outputs a list $L=\{s_1,\ldots,s_{\lvert L\rvert}\}\subseteq\{0,1\}^n$ such that \begin{itemize}
        \item[(i)] if $\lvert \hat{g}(s)\rvert\geq \varepsilon$, then $s\in L$, and
        \item[(ii)] if $s\in L$, then $\lvert \hat{g}(s)\rvert\geq \nicefrac{\varepsilon}{2}$.
    \end{itemize}
    This list has length $\lvert L\rvert\leq \nicefrac{4}{\varepsilon^2}$ by Parseval.
\end{theorem}

Using \Cref{theorem:functional-agnostic-qsq-GL}, \cite{arunachalam2020qsq} for instance gave QSQ learners for parities, juntas, and DNFs.
More generally, using a simpler version of the reasoning worked out in detail in \Cref{subsection:distributional-qsq-learning}, \Cref{theorem:functional-agnostic-qsq-GL} gives rise to a functional $1$-agnostic QSQ learner for parities and functional $2$-agnostic QSQ learner for Fourier-sparse functions.

\section{Distributional Agnostic Quantum Learning}\label{section:distributional-agnostic-quantum-learning}

In this section, we go beyond (noisy) functional agnostic quantum learning. Namely, we show that changing the underlying quantum data resource to mixture-of-superpositions examples gives quantum learners the ability to solve distributional agnostic learning problems.

\subsection{Distributional Agnostic Approximate Quantum Fourier Sampling}\label{subsection:distributional-agnostic-quantum-fourier-sampling}

Recall that quantum Fourier sampling is a crucial subroutine for quantum learning algorithms in the functional case. However, no successful variant of this method for the distributional agnostic case was known.
Here, we show that mixture-of-superpositions quantum examples allow for an approximate version of this crucial tool. This demonstrates the usefulness of our new notion of quantum example in a distributional agnostic scenario:

\begin{theorem}[Formal statement of \Cref{theorem:main-result-agnostic-quantum-fourier-sampling-learning}, Point 1]\label{theorem:agnostic-quantum-fourier-sampling}
    Let $\mathcal{D}$ be a probability distribution over $\mathcal{X}_n\times\{0,1\}$ with $\mathcal{D}_{\mathcal{X}_n}=\mathcal{U}_n$.
    Consider the following quantum algorithm:
    Given a single copy of $\rho_\mathcal{D}$, first apply (the unitary channel for) the unitary $H^{\otimes (n+1)}$, then measure all $n+1$ qubits in the computational basis. 
    The measurement outcomes of this procedure satisfy the following:
    \begin{itemize}
        \item[(i)] The computational basis measurement on the last qubit gives outcome $0$ with probability $\nicefrac{1}{2}$ and outcome $1$ with probability $\nicefrac{1}{2}$.
        \item[(ii)] Conditioned on having observed outcome $1$ for the last qubit, the computational basis measurement on the first $n$ qubits outputs a string $s\in\{0,1\}^n$ with probability 
        \begin{equation}
            \frac{1}{2^n}\left(1 - \mathbb{E}_{x\sim\mathcal{U}_n}[(\phi(x))^2]\right) + (\hat{\phi}(s))^2 .
        \end{equation}
    \end{itemize}
\end{theorem}

The squares of the Fourier coefficients of $\phi$ in general do not form a probability distribution, because in general $\mathbb{E}_{x\sim\mathcal{U}_n}[(\phi(x))^2]< 1$. Thus, it does not make sense to speak of exact sampling from the ``distribution formed by squares of Fourier coefficients'' in this distributional agnostic case. 
However, by Parseval, we know that $\{\tfrac{1}{2^n}\left(1 - \mathbb{E}_{x\sim\mathcal{U}_n}[(\phi(x))^2]\right) + (\hat{\phi}(s))^2\}_{s\in\{0,1\}^n}$ does form a probability distribution. It is exactly this probability distribution that \Cref{theorem:agnostic-quantum-fourier-sampling} allows us to sample from (with success probability $\nicefrac{1}{2}$).

\begin{proof}
    As $\rho_\mathcal{D}$ is a probabilistic mixture, we have, for any $s\in\{0,1\}^n$ and $b\in\{0,1\}$,
    \begin{equation}
        \bra{s,b} H^{\otimes (n+1)} \rho_\mathcal{D} H^{\otimes (n+1)} \ket{s,b}
        = \mathbb{E}_{f\sim F_{\mathcal{D}}} \left[\left\lvert \bra{s,b} H^{\otimes (n+1)}\ket {\psi_{(\mathcal{U}_n, f)}}\right\rvert^2\right] .
    \end{equation}
    Thus, \Cref{lemma:functional-quantum-Fourier-sampling} immediately gives (i).
    Also, \Cref{lemma:functional-quantum-Fourier-sampling} tells us that, conditioned on having observed outcome $1$ for the computational basis measurement on the last qubit, the computational basis measurement on the first $n$ qubits produces string $s\in\{0,1\}^n$ with probability $\mathbb{E}_{f\sim F_{\mathcal{D}}} [( \hat{g}_f (s))^2]$, where $g_f(s) = (-1)^f$. 
    Using the definition of $F_\mathcal{D}$ via an independent sampling of labels, we can rewrite this quantity as 
    \begingroup
    \allowdisplaybreaks
    \begin{align}
        \mathbb{E}_{f\sim F_{\mathcal{D}}} [( \hat{g}_f (s))^2]
        &= \mathbb{E}_{f\sim F_{\mathcal{D}}} \left[ \left(\frac{1}{2^n} \sum_{x\in\{0,1\}^n} (-1)^{f(x)} \chi_s(x)\right)^2 \right]\\
        &= \mathbb{E}_{f\sim F_{\mathcal{D}}} \left[ \frac{1}{4^n} \sum_{x,y\in\{0,1\}^n} \chi_s(x) (-1)^{f(x)} \chi_s(y) (-1)^{f(y)} \right]\\
        &= \frac{1}{4^n} \sum_{x,y\in\{0,1\}^n} \chi_s(x)  \chi_s(y) \mathbb{E}_{f\sim F_{\mathcal{D}}} \left[ (-1)^{f(x)} (-1)^{f(y)} \right]\\
        &= \frac{1}{4^n} \sum_{x,y\in\{0,1\}^n} \chi_s(x)  \chi_s(y)\cdot \begin{cases} \mathbb{E}_{f\sim F_{\mathcal{D}}} \left[ (-1)^{f(x)}\right]\cdot \mathbb{E}_{f\sim F_{\mathcal{D}}} \left[ (-1)^{f(y)}\right]\quad &\textrm{ if } x\neq y\\ 1 &\textrm{ if } x=y \end{cases} .
    \end{align}
    \endgroup
    Next, recall that $\mathbb{E}_{f\sim F_{\mathcal{D}}} \left[ (-1)^{f(x)}\right] = 1-2\varphi(x)=\phi(x)$ holds by definition of $F_{\mathcal{D}}$. Using that $\chi_s^2=1$ holds for any $s\in\mathcal{X}_n$, this allows us to further rewrite 
    \begingroup
    \allowdisplaybreaks
    \begin{align}
        \mathbb{E}_{f\sim F_{\mathcal{D}}} [( \hat{g}_f (s))^2]
        &= \frac{1}{4^n} \sum_{x,y\in\{0,1\}^n} \chi_s(x)  \chi_s(y)\cdot \begin{cases} \phi(x) \phi(y)\quad &\textrm{ if } x\neq y\\ 1 &\textrm{ if } x=y \end{cases}\\
        &= \frac{1}{4^n} \sum_{x\in\{0,1\}^n} (\chi_s(x))^2 + \frac{1}{4^n} \sum_{\substack{x,y\in\{0,1\}^n\\ x\neq y}}  \chi_s(x) \phi(x) \chi_s(y) \phi(y)\\
        &= \frac{1}{2^n} + \frac{1}{4^n}\sum_{x,y\in\{0,1\}^n}  \chi_s(x) \phi(x) \chi_s(y) \phi(y) - \frac{1}{4^n} \sum_{x\in\{0,1\}^n} (\chi_s(x))^2 (\phi(x))^2\\
        &= \frac{1}{2^n} + \left(\frac{1}{2^n} \sum_{x\in\{0,1\}^n}\chi_s(x) \phi(x) \right)^2 - \frac{1}{2^n} \mathbb{E}_{x\sim \mathcal{U}_n}[(\phi(x))^2]\\
        &= \frac{1}{2^n}\left(1 - \mathbb{E}_{x\sim\mathcal{U}_n}[(\phi(x))^2]\right) + (\hat{\phi}(s))^2 .
    \end{align}
    \endgroup
    This finishes the proof.
\end{proof}

To see that \Cref{theorem:agnostic-quantum-fourier-sampling} indeed implies Point 1 of \Cref{theorem:main-result-agnostic-quantum-fourier-sampling-learning}, note that $\phi$ is $[-1,1]$-valued, which in particular implies $0\leq  1 - \mathbb{E}_{x\sim\mathcal{U}_n}[(\phi(x))^2]\leq 1$.
Therefore, we can, with success probability $\nicefrac{1}{2}$, produce a sample from a distribution that is $(\nicefrac{1}{2^n})$-close in $\infty$-norm to the sub-normalized distribution formed by the squares of the Fourier coefficients of $\phi$ as follows:
First, perform $n+1$ single-qubit Hadamard gates on $\rho_\mathcal{D}$.  
Second, measure the last qubit in the computational basis. If the outcome is $0$, the sampling attempt fails. If the outcome is $1$, then measure the first $n$ qubits in the computational basis and output the observed string of bits.

Now equipped with a distributional agnostic analogue of quantum Fourier sampling, we can appeal to \Cref{lemma:distribution-estimation-infty-norm} to approximate the Fourier spectrum of the conditional label expectation:

\begin{corollary}\label{corollary:distributional-agnostic-quantum-approximation-fourier-spectrum}
    Let $\mathcal{D}$ be a probability distribution over $\mathcal{X}_n\times\{0,1\}$ with $\mathcal{D}_{\mathcal{X}_n}=\mathcal{U}_n$.
    Let $\delta,\varepsilon\in (0,1)$. Assume that $\varepsilon > 2^{-(\tfrac{n}{2} - 2)}$.
    Then, there exists a quantum algorithm that, given $\mathcal{O}\left(\tfrac{\log(\nicefrac{1}{\delta\varepsilon^2})}{\varepsilon^4}\right)$ copies of $\rho_\mathcal{D}$, uses $\mathcal{O}\left(n\tfrac{\log(\nicefrac{1}{\delta\varepsilon^2})}{\varepsilon^4}\right)$ single-qubit gates, classical computation time $\tilde{\mathcal{O}}\left(n\tfrac{\log(\nicefrac{1}{\delta \varepsilon^2})}{\varepsilon^4}\right)$, and classical memory of size $\tilde{\mathcal{O}}\left(n\tfrac{\log(\nicefrac{1}{\delta\varepsilon^2})}{\varepsilon^4}\right)$, and outputs, with success probability $\geq 1-\delta$, a succinctly represented $\Tilde{\phi}:\mathcal{X}_n\to [-1,1]$ such that $\norm{\Tilde{\phi}-\hat{\phi}}_\infty\leq\varepsilon$ and $\norm{\Tilde{\phi}}_0\leq\tfrac{16\mathbb{E}_{x\sim\mathcal{U}_n} [(\phi(x))^2]}{\varepsilon^2}\leq \tfrac{16}{\varepsilon^2}$.
\end{corollary}

Note that \Cref{corollary:distributional-agnostic-quantum-approximation-fourier-spectrum} imposes an additional assumption compared to the noiseless functional case, namely a lower bound on the desired accuracy $\varepsilon$. However, this lower bound is inverse-exponential in $n$ and thus satisfied (for large enough $n$) for the inverse-polynomial accuracies that are usually of interest.

\begin{proof}
    Our proof is similar to that of \cite[Theorem 5]{kanade2019learning}.
    \Cref{theorem:agnostic-quantum-fourier-sampling} gives a procedure that, using a single copy of $\rho_\mathcal{D}$ and $n+1$ single-qubit quantum gates, produces a sample from the probability distribution $q:\{0,1\}^{n+1}\to [0,1]$ defined via
    \begin{equation}
        q(s,1)
        =\frac{1}{2}\left(\frac{1}{2^n}\left(1 - \mathbb{E}_{x\sim\mathcal{U}_n}[(\phi(x))^2]\right) + (\hat{\phi}(s))^2\right),\quad 
        q(0^n,0)
        =\frac{1}{2} .
    \end{equation}
    Hence, according to \Cref{lemma:distribution-estimation-infty-norm} applied for the probability distribution $q$, confidence $\delta >0$ and accuracy $\tau = \nicefrac{\varepsilon^2}{8}$, we see that $m=\mathcal{O}\left(\tfrac{\log(1/\delta)}{\varepsilon^4}\right)$ copies of $\rho_\mathcal{D}$ are sufficient to obtain, with success probability $\geq 1-\tfrac{\delta}{2}$, a succinctly represented estimate $\tilde{q}_m$ such that $\norm{\tilde{q}_m}_0 \leq \mathcal{O}\left(\tfrac{\log(1/\delta)}{\varepsilon^4}\right)$ and $\norm{q - \tilde{q}_m}_\infty\leq \frac{\varepsilon^2}{8}$. Moreover, the estimate $\tilde{q}_m$ can be obtained using $\mathcal{O}(nm)=\mathcal{O}\left(n\tfrac{\log(1/\delta)}{\varepsilon^4}\right)$ single-qubit Hadamard gates, classical computation time $\tilde{\mathcal{O}}\left(n\tfrac{\log(1/\delta)}{\varepsilon^4}\right)$, and classical memory of size $\tilde{\mathcal{O}}\left(n\tfrac{\log(1/\delta)}{\varepsilon^4}\right)$.
    
    Starting from the estimate $\tilde{q}_m$ for $q$, we output a list $L$ of strings $s\in\{0,1\}^n$ such that $\tilde{q}_m(s,1) \geq  \nicefrac{\varepsilon^2}{4}$. 
    As we have a succinct representation of $\tilde{q}_m$ with at most $\mathcal{O}\left(\tfrac{\log(1/\delta)}{\varepsilon^4}\right)$ non-zero entries, the list $L$ can be compiled by brute-force search in classical computation time $\mathcal{O}\left(n\tfrac{\log(1/\delta)}{\varepsilon^4}\right)$.
    By our approximation guarantee, with success probability $\geq 1-\tfrac{\delta}{2}$, we have the following:
    \begin{itemize}
        \item If $\lvert \hat{\phi}(s)\rvert\geq\varepsilon$, then $\tilde{q}(s,1)\geq \tfrac{\varepsilon^2}{2} - \tfrac{\varepsilon^2}{8}\geq \tfrac{\varepsilon^2}{4}$. So, if $s$ is an $\varepsilon$-heavy Fourier coefficient of $\hat{\phi}$, then $s\in L$.
        \item If $s\in L$, that is, if $\tilde{q}_m(s,1)\geq \tfrac{\varepsilon^2}{4}$, then $(\hat{\phi}(s))^2 \geq 2\left( \tfrac{\varepsilon^2}{8} - \frac{1}{2^n}\left(1 - \mathbb{E}_{x\sim\mathcal{U}_n}[(\phi(x))^2]\right)\right) \geq \tfrac{\varepsilon^2}{16}$, thus also $\lvert \hat{\phi}(s)\rvert \geq \tfrac{\varepsilon}{4}$ and $s$ is an $(\tfrac{\varepsilon}{4})$-heavy Fourier coefficient of $\hat{\phi}$. In particular, combining this with Parseval's equality and the fact that $\mathbb{E}_{x\sim\mathcal{U}_n} [(\phi(x))^2]\leq 1$, we see that $\lvert L\rvert\leq \tfrac{16 \mathbb{E}_{x\sim\mathcal{U}_n} [(\phi(x))^2]}{\varepsilon^2} \leq \tfrac{16}{\varepsilon^2}$.
    \end{itemize} 
    
    Now, for each of the at most $\tfrac{16}{\varepsilon^2}$ strings in $L$, we estimate the corresponding Fourier coefficient. 
    For any single such string $s$, by Hoeffding's inequality, we know that $\mathcal{O}\left(\tfrac{\log (\nicefrac{1}{\delta})}{\varepsilon^2}\right)$ classical samples from $\mathcal{D}$ suffice produce an empirical estimate $\tilde{\phi}(s)$ that matches $\hat{\phi}(s)$ up to accuracy $\varepsilon$, with success probability $\geq 1-\tfrac{\delta}{2}$.
    By a union bound over $L$, this implies that $m_2 = \mathcal{O}\left(\lvert L\rvert\tfrac{\log (\nicefrac{\lvert L\rvert }{\delta})}{\varepsilon^2}\right) = \mathcal{O}\left(\tfrac{\log (\nicefrac{1}{\delta\varepsilon^2})}{\varepsilon^4}\right)$ classical samples from $\mathcal{D}$ suffice to estimate all $\hat{\phi}(s)$ with $s\in L$ simultaneously up to accuracy $\varepsilon$, with success probability $\geq 1-\tfrac{\delta}{2}$. As $\hat{\phi}(s)\in [-1,1]$ for all $s\in\{0,1\}^n$, these estimates can only improve if we project them to $[-1,1]$. 
    Moreover, building these empirical estimates can be done using classical computation time $\tilde{\mathcal{O}}\left(n\tfrac{\log(\nicefrac{1}{\delta \varepsilon^2})}{\varepsilon^4}\right)$ and classical memory of size $\tilde{\mathcal{O}}\left(n\tfrac{\log(\nicefrac{1}{\delta\varepsilon^2})}{\varepsilon^4}\right)$.
    It remains to observe that a single copy of $\rho_\mathcal{D}$ can be measured in the computational basis to obtain a sample from $\mathcal{D}$ (recall \Cref{lemma:sanity-check-reproduce-classical-sample}), and that, by one more union bound, the produced estimate $\tilde{\phi}$ has the desired properties with probability $\geq 1-\delta$.
\end{proof}

\begin{remark}\label{remark:approximate-fourier-sampling-fourier-approximation}
    We note that the proof of \Cref{corollary:distributional-agnostic-quantum-approximation-fourier-spectrum} did not use the specific form of the additive ``perturbation'' term $\frac{1}{2^n}\left(1 - \mathbb{E}_{x\sim\mathcal{U}_n}[(\phi(x))^2]\right)$ coming from \Cref{theorem:agnostic-quantum-fourier-sampling}, we only used the fact that this perturbation lies in $[0,\tfrac{1}{2^n}]$. Therefore, this proof immediately extends to a more general statement of the form `$\norm{\cdot}_\infty$-approximate Fourier sampling enables succinct approximation of the Fourier spectrum.''
    Also, we see that the assumed lower bound of $\varepsilon > 2^{-(\tfrac{n}{2} - 2)}$ may be relaxed when given prior knowledge about $\mathbb{E}_{x\sim\mathcal{U}_n}[(\phi(x))^2]$.
\end{remark}

With this subroutine for obtaining a succinctly represented approximation to the Fourier spectrum of interest, we can now obtain distributional agnostic quantum learning algorithms. These are presented in the next subsections.

\subsection{Distributional Agnostic Quantum Learning Parities and Fourier-Sparse Functions}\label{subsection:distributional-agnostic-quantum-learning}

First, we show how to apply \Cref{corollary:distributional-agnostic-quantum-approximation-fourier-spectrum} as a subroutine for distributional agnostic quantum parity learning.

\begin{corollary}[Formal statement of \Cref{theorem:main-result-agnostic-quantum-fourier-sampling-learning}, Point 2]\label{corollary:agnostic-quantum-parity-learning}
    Let $\mathcal{D}$ be a probability distribution over $\mathcal{X}_n\times\{0,1\}$ with $\mathcal{D}_{\mathcal{X}_n}=\mathcal{U}_n$.
    Let $\delta,\varepsilon\in (0,1)$. Assume that $\varepsilon > 2^{-(\tfrac{n}{2} - 2)}$.
    There is a quantum algorithm that, given $\mathcal{O}\left(\tfrac{\log(\nicefrac{1}{\delta\varepsilon^2})}{\varepsilon^4}\right)$ copies of $\rho_\mathcal{D}$, uses $\mathcal{O}\left(n\tfrac{\log(\nicefrac{1}{\delta \varepsilon^2})}{\varepsilon^4}\right)$ single-qubit gates, classical computation time $\tilde{\mathcal{O}}\left(n\tfrac{\log(\nicefrac{1}{\delta\varepsilon^2})}{\varepsilon^4}\right)$, and classical memory of size $\tilde{\mathcal{O}}\left(n\tfrac{\log(\nicefrac{1}{\delta\varepsilon^2})}{\varepsilon^4}\right)$, and outputs, with success probability $\geq 1-\delta$, a bit string $s\in\{0,1\}^n$ such that 
    \begin{equation}
        \mathbb{P}_{(x,b)\sim\mathcal{D}}[b\neq s\cdot x]
        \leq \min\limits_{t\in\{0,1\}^n} \mathbb{P}_{(x,b)\sim\mathcal{D}}[b\neq t\cdot x] + \varepsilon\, .
    \end{equation}
    Thus, this quantum algorithm is a distributional agnostic proper quantum parity learner up to inverse-exponen\-tial\-ly small accuracies, assuming a uniform marginal over inputs.
\end{corollary}
\begin{proof}
    By \Cref{lemma:parity-learning-via-heaviest-Fourier-coefficient} (and the discussion on complexity bounds thereafter), it suffices to show that there is a quantum algorithm with the claimed complexity bounds that, with success probability $\geq 1-\delta$, outputs a $(2\varepsilon)$-approximately-largest Fourier coefficient of $\phi$.
    To achieve this, first run the procedure from \Cref{corollary:distributional-agnostic-quantum-approximation-fourier-spectrum} to obtain, with probability $\geq 1-\delta$, a succinctly represented $\Tilde{\phi}$ such that $\norm{\Tilde{\phi}-\hat{\phi}}_\infty\leq\varepsilon$ and $\norm{\Tilde{\phi}}_0\leq\tfrac{16}{\varepsilon^2}$. 
    Now, let $s\in\operatorname{argmax}_{t\in\{0,1\}^n} \Tilde{\phi}(t)$. Note that such an $s$ can be found in time $\mathcal{O}\left(\tfrac{n}{\varepsilon^2}\right)$ since $\norm{\Tilde{\phi}}_0\leq\tfrac{16}{\varepsilon^2}$.
    This $s$ now satisfies
    \begin{align}
        \max_{t\in\{0,1\}^n} \hat{\phi}(t) - \hat{\phi}(s)
        &= \max_{t\in\{0,1\}^n} \hat{\phi}(t) - \tilde{\phi}(t) + \tilde{\phi}(t) - \tilde{\phi}(s) + \tilde{\phi}(s) - \hat{\phi}(s)\\
        &\leq \norm{\Tilde{\phi}-\hat{\phi}}_\infty + 0 + \norm{\Tilde{\phi}-\hat{\phi}}_\infty\\
        &\leq 2\varepsilon,
    \end{align}
    as needed. 
    The bounds on copy complexity, quantum gate complexity, classical runtime, and classical memory are all inherited from \Cref{corollary:distributional-agnostic-quantum-approximation-fourier-spectrum}.
\end{proof}

In particular, \Cref{corollary:agnostic-quantum-parity-learning} gives an efficient procedure for agnostic quantum parity learning with inverse-polynomial accuracy parameter $\varepsilon$ and with inverse-exponential confidence parameter $\delta$.
In contrast, by reduction to the widely believed hardness of LPN, we do not expect an efficient classical procedure for the corresponding classical agnostic learning problem to exist.

In a similar vein, \Cref{corollary:distributional-agnostic-quantum-approximation-fourier-spectrum} can serve as a subroutine for distributional agnostic quantum learning of Fourier-sparse functions: 

\begin{corollary}[Formal statement of \Cref{theorem:main-result-agnostic-quantum-fourier-sampling-learning}, Point 3]\label{corollary:agnostic-quantum-fourier-sparse-learning}
    Let $\mathcal{D}$ be a probability distribution over $\mathcal{X}_n\times\{0,1\}$ with $\mathcal{D}_{\mathcal{X}_n}=\mathcal{U}_n$.
    Let $\delta,\varepsilon\in (0,1)$. Assume that $\varepsilon > 2^{-(\tfrac{n}{2} - 2)}$.
    Then, there is a quantum algorithm that, given $\mathcal{O}\left(\tfrac{k^4\log(\nicefrac{k^2}{\delta\varepsilon^2})}{\varepsilon^4}\right)$ copies of $\rho_\mathcal{D}$, uses $\mathcal{O}\left(n\tfrac{k^4\log(\nicefrac{1}{\delta\varepsilon^2})}{\varepsilon^4}\right)$ single-qubit gates, classical computation time $\tilde{\mathcal{O}}\left(n\tfrac{k^4\log(\nicefrac{k^2}{\delta\varepsilon^2})}{\varepsilon^4}\right)$, and classical memory of size $\tilde{\mathcal{O}}\left(n\tfrac{k^4\log(\nicefrac{k^2}{\delta\varepsilon^2})}{\varepsilon^4}\right)$, and outputs, with success probability $\geq 1-\delta$, a randomized hypothesis $h:\mathcal{X}_n\to\{0,1\}$ such that
    \begin{equation}
        \mathbb{P}_{(x,b)\sim\mathcal{D}} \left[b\neq h(x)\right]
        \leq 2 \min\limits_{\substack{\tilde{f}:\mathcal{X}_n\to\{0,1\}\\\mathrm{Fourier-}k\mathrm{-sparse}}} \mathbb{P}_{(x,b)\sim\mathcal{D}}[b\neq \tilde{f}(x)] + \varepsilon\, .
    \end{equation}
    In particular, this quantum algorithm is a distributional $2$-agnostic improper quantum Fourier-sparse learner up to inverse-exponentially small accuracies, assuming a uniform marginal over inputs.
\end{corollary}
\begin{proof}
    By \Cref{lemma:Fourier-sparse-learning-via-heaviest-Fourier-coefficients} and the accompanying discussion on complexity bounds, it suffices to show that there is a quantum algorithm with the claimed complexity bounds that, with success probability $\geq 1-\delta$, outputs $(\nicefrac{\varepsilon}{2k})$-accurate estimates of $k$ $(\nicefrac{\varepsilon}{2k})$-approximately-heaviest Fourier coefficients of $\phi$.
    To achieve this, let $\tilde{\varepsilon} = \nicefrac{\varepsilon}{4k}$ and run the procedure from \Cref{corollary:distributional-agnostic-quantum-approximation-fourier-spectrum} to obtain, with probability $\geq 1-\delta$, a succinctly represented $\Tilde{\phi}:\mathcal{X}_n\to [-1,1]$ such that $\norm{\Tilde{\phi}-\hat{\phi}}_\infty\leq\tilde{\varepsilon}$ and $\norm{\Tilde{\phi}}_0\leq\tfrac{16}{\tilde{\varepsilon}^2}$. 
    Let $s_1\in\operatorname{argmax}_{t\in\{0,1\}^n} \lvert \Tilde{\phi}(t)\rvert$ and, for $2\leq \ell\leq k$, let $s_\ell \in \operatorname{argmax}_{t\in\{0,1\}^n\setminus\{s_1,\ldots,s_{\ell -1}\}} \lvert \Tilde{\phi}(t)\rvert$. 
    Note that such $s_1,\ldots,s_k$ can be found in time $\mathcal{O}\left(\tfrac{n k^4}{\varepsilon^2}\right)$ since $\norm{\Tilde{\phi}}_0\leq\tfrac{16}{\tilde{\varepsilon}^2}\leq \mathcal{O}\left(\tfrac{k^4}{\varepsilon^2}\right)$.
    As in \Cref{lemma:Fourier-sparse-learning-via-heaviest-Fourier-coefficients}, let $t_1\in \operatorname{argmax}_{t\in\{0,1\}^n} \lvert \hat{\phi}(t)\rvert$, and for $2\leq \ell\leq k$, let $t_\ell\in \operatorname{argmax}_{t\in\{0,1\}^n\setminus\{t_1,\ldots,t_{\ell -1}\}} \lvert \hat{\phi}(t)\rvert$.
    By the technical \Cref{lemma:technical}, $\norm{\lvert\Tilde{\phi}\rvert-\lvert\hat{\phi}\rvert}_\infty\leq\norm{\Tilde{\phi}-\hat{\phi}}_\infty\leq\tilde{\varepsilon}$ implies that, for every $1\leq \ell\leq k$, $\left\lvert  \lvert\hat{\phi}(t_\ell)\rvert - \lvert\hat{\phi}(s_\ell)\rvert\right\rvert\leq 2\tilde{\varepsilon} \leq \nicefrac{\varepsilon}{2k} $, so we can apply \Cref{lemma:Fourier-sparse-learning-via-heaviest-Fourier-coefficients}.
    The bounds on copy complexity, quantum gate complexity, classical runtime, and classical memory are all inherited from \Cref{corollary:distributional-agnostic-quantum-approximation-fourier-spectrum}.
\end{proof}

As $1$-agnostic Fourier-sparse learning is at least as hard as $1$-agnostic parity learning, which in turn is at least as hard as LPN, this task is widely believed to be classically intractable from random examples. 
To the best of our knowledge, currently there are also no classical algorithms for $2$-agnostic Fourier-sparse learning from examples. 
Thus, while \Cref{corollary:agnostic-quantum-fourier-sparse-learning} does not achieve $1$-agnostic quantum Fourier-sparse learning, it serves as an indication for the power of mixture-of-superpositions examples in learning Fourier-sparse functions w.r.t.~uniformly random inputs.

In \Cref{corollary:agnostic-quantum-fourier-sparse-learning}, we aimed to achieve a small misclassification probability.
If instead we focus our attention on the $L_2$-error directly, then we can achieve $1$-agnostic learning w.r.t.~this performance measure:

\begin{corollary}\label{corollary:agnostic-quantum-fourier-sparse-learning-L2}
    Let $\mathcal{D}$ be a probability distribution over $\mathcal{X}_n\times\{0,1\}$ with $\mathcal{D}_{\mathcal{X}_n}=\mathcal{U}_n$.
    Let $\delta,\varepsilon\in (0,1)$. Assume that $\varepsilon > 2^{-(\tfrac{n}{2} - 2)}$.
    There is a quantum algorithm that, given $\mathcal{O}\left(\tfrac{k^4\log(\nicefrac{k^2}{\delta\varepsilon^2})}{\varepsilon^4}\right)$ copies of $\rho_\mathcal{D}$, uses $\mathcal{O}\left(n\tfrac{k^4\log(\nicefrac{1}{\delta\varepsilon^2})}{\varepsilon^4}\right)$ single-qubit gates, classical computation time $\tilde{\mathcal{O}}\left(n\tfrac{k^4\log(\nicefrac{k^2}{\delta\varepsilon^2})}{\varepsilon^4}\right)$, and classical memory of size $\tilde{\mathcal{O}}\left(n\tfrac{k^4\log(\nicefrac{k^2}{\delta\varepsilon^2})}{\varepsilon^4}\right)$, and outputs, with success probability $\geq 1-\delta$, a Fourier-$k$-sparse function $g:\mathcal{X}_n\to\mathbb{R}$ such that 
    \begin{equation}
        \mathbb{E}_{(x,b)\sim\mathcal{D}}[(y-g(x))^2]
        \leq \min\limits_{\substack{\tilde{f}:\mathcal{X}_n\to\{-1,1\}\\\mathrm{Fourier-}k\mathrm{-sparse}}} \mathbb{E}_{(x,b)\sim\mathcal{D}}[(y-\tilde{f}(x))^2] + \varepsilon\, .
    \end{equation}
\end{corollary}
\begin{proof}
    We start as in the proof of \Cref{corollary:agnostic-quantum-fourier-sparse-learning}, using \Cref{corollary:distributional-agnostic-quantum-approximation-fourier-spectrum} to find approximately largest Fourier coefficients. 
    Then, we reason as in \Cref{lemma:Fourier-sparse-learning-via-heaviest-Fourier-coefficients}.
    Recalling from \Cref{lemma:misclassification-probability-bound-L2} that $\mathbb{E}_{(x,b)\sim\mathcal{D}}[(y-g(x))^2] = \sum_{s\in\{0,1\}^n} \left(\hat{\phi}(s)-\hat{g}(s)\right)^2 + \left(1 - \mathbb{E}_{x\sim \mathcal{U}_n} \left[(\phi(x))^2\right]\right)$, we can extract the following bounds from our proof of \Cref{lemma:Fourier-sparse-learning-via-heaviest-Fourier-coefficients}, using the same notation: On the one hand,
    \begin{align}
        \mathbb{E}_{(x,b)\sim\mathcal{D}}[(y-g(x))^2]
        \leq 1 +  k\tilde{\varepsilon}^2 -   \sum_{\ell = 1}^k  \left(\hat{\phi}(s_\ell)\right)^2.
    \end{align}
    On the other hand, 
    \begin{align}
        \min\limits_{\substack{\tilde{f}:\mathcal{X}_n\to\{-1,1\}\\\mathrm{Fourier-}k\mathrm{-sparse}}} \mathbb{E}_{(x,b)\sim\mathcal{D}}[(y-\tilde{f}(x))^2]
        &\geq 1 -   \sum_{\ell = 1}^k  \left(\hat{\phi}(s_\ell)\right)^2 - 2k\tilde{\varepsilon}.
    \end{align}
    Altogether, we have shown
    \begin{align}
        \mathbb{E}_{(x,b)\sim\mathcal{D}}[(y-g(x))^2]
        &\leq 1 - \sum_{\ell =1}^k (\hat{\phi}(s_\ell))^2 + 2 k \tilde{\varepsilon}\\
        &\leq \min\limits_{\substack{\tilde{f}:\mathcal{X}_n\to\{-1,1\}\\\mathrm{Fourier-}k\mathrm{-sparse}}} \mathbb{E}_{(x,b)\sim\mathcal{D}}[(y-\tilde{f}(x))^2] + 2k\tilde{\varepsilon} + k\tilde{\varepsilon}^2\\
        &= \min\limits_{\substack{\tilde{f}:\mathcal{X}_n\to\{-1,1\}\\\mathrm{Fourier-}k\mathrm{-sparse}}} \mathbb{E}_{(x,b)\sim\mathcal{D}}[(y-\tilde{f}(x))^2] + \varepsilon,
    \end{align}
    if we choose $\tilde{\varepsilon}=\nicefrac{\varepsilon}{4k}$.
    The bounds on copy complexity, quantum gate complexity, classical runtime, and classical memory are all inherited from \Cref{corollary:distributional-agnostic-quantum-approximation-fourier-spectrum}.
\end{proof}

We make two short remarks about \Cref{corollary:agnostic-quantum-fourier-sparse-learning-L2}.
On the one hand, in the proof of \Cref{corollary:agnostic-quantum-fourier-sparse-learning-L2} we implicitly demonstrate a $1$-agnostic analogue of \Cref{lemma:Fourier-sparse-learning-via-heaviest-Fourier-coefficients} when performance is measured according to the $L_2$-error.
On the other hand, while the function $g$ in \Cref{corollary:agnostic-quantum-fourier-sparse-learning-L2} is Fourier-$k$-sparse, it is not $\{-1,1\}$-valued. Therefore, also the quantum learning procedure of \Cref{corollary:agnostic-quantum-fourier-sparse-learning-L2} is improper.

\subsection{Distributional Agnostic Quantum Statistical Query Learning}\label{subsection:distributional-qsq-learning}

The previous two subsections considered quantum agnostic learning when given access to mixture-of-superpositions examples. 
In classical learning theory as well as in quantum learning theory from superposition examples, statistical query access serves as an important way of weakening the learner's data access.
Here, we present a statistical query relaxation of mixture-of-superpositions quantum data access, demonstrate that this relaxed access still allow for a variant of the Goldreich-Levin/Kushilevitz-Mansour algorithm, and argue that it therefore enables distributional agnostic quantum learning.

In a similar vein to \Cref{definition:functional-qsq}, we define QSQs for distributional agnostic learning, based on \Cref{definition:mixture-of-superpositions-quantum-example}:

\begin{definition}[Distributional quantum statistical queries]\label{definition:distributional-agnostic-qsq}
    Let $\mathcal{D}$ be a probability distribution over $\mathcal{X}_n\times\{0,1\}$.
    A \emph{(distributional) quantum statistical query (QSQ) oracle for $\mathcal{D}$} produces, when queried with a bounded $(n+1)$-qubit observable $M=M^\dagger\in\mathcal{B}((\mathbb{C}^2)^{\otimes (n+1)})$ satisfying $\norm{M}\leq 1$ and with a tolerance parameter $\tau > 0$, a number $\mu\in\mathbb{R}$ such that
    \begin{equation}
        \left\lvert \mu - \tr[ M \rho_{\mathcal{D}}] \right\rvert
        \leq \tau ,
    \end{equation}
    where $\rho_{\mathcal{D}}$ is a (distributional) mixture-of-superpositions examples as in \Cref{definition:mixture-of-superpositions-quantum-example}.
\end{definition}

Again, one may additionally impose that $M$ be efficiently implementable. In fact, the observables used in our QSQs below all satisfy this additional requirement.

For the functional case with uniform marginal $\mathcal{D}_{\mathcal{X}_n} = \mathcal{U}_n$, as shown in \cite[Theorem 4.4]{arunachalam2020qsq}, the classical Goldreich-Levin (GL) algorithm \cite{goldreich1989hard, kushilevitz1993learning} for finding a list of heavy Fourier coefficients of an unknown function assuming query access admits a quantum counterpart that only needs QSQ access. 
This relies on the fact that, for any subset $S\subseteq\{0,1\}^n$, when choosing a suitable (efficiently implementable) observable $M = M_S$ inspired by quantum Fourier sampling, the expectation value $\bra{\psi_{(\mathcal{D}_{\mathcal{X}_n}, f)}} M \ket{\psi_{(\mathcal{D}_{\mathcal{X}_n}, f)}}$ exactly equals $\sum_{s\in S} (\hat{f}(s))^2$. 
Therefore, QSQ access in particular suffices to estimate quantities of the form $\sum_{t\in\{0,1\}^{n-k}} (\hat{f}(st))^2 $ for fixed $s\in\{0,1\}^k$, which is exactly what the classical GL procedure uses (classical) membership queries for.

This approach does not immediately work in the distributional setting. To see this, recall from \Cref{theorem:agnostic-quantum-fourier-sampling} that a copy $\rho_{\mathcal{D}} = \rho_{(\mathcal{U}_n, \varphi)}$ does not allow for exact but only for approximate quantum Fourier sampling from $\phi$. 
Therefore, when imitating the reasoning of \cite[Lemma 4.1 and Theorem 4.4]{arunachalam2020qsq}, we only obtain: 
For any subset $S\subseteq\{0,1\}^n$, choosing a suitable (efficiently implementable) observable $M = M_S$ inspired by quantum Fourier sampling, the expectation value $\tr[ M \rho_{\mathcal{D}}]$ equals
\begin{equation}
    \sum_{s\in S} \left( \frac{1}{2^n}\left(1 - \mathbb{E}_{x\sim\mathcal{U}_n}[(\phi(x))^2]\right) + (\hat{\phi}(s))^2\right)
    = \frac{\lvert S\rvert }{2^n}\left(1 - \mathbb{E}_{x\sim\mathcal{U}_n}[(\phi(x))^2]\right) + \sum_{s\in S} (\hat{\phi}(s))^2 .
\end{equation}
Therefore, for sets $S$ with large cardinality, there is an in general non-negligible additive perturbation to the quantity $\sum_{s\in S} (\hat{\phi}(s))^2$ that we need for the GL procedure. 
Such large sets $S$ appear towards the beginning of the GL algorithm.
In our next result, we show how to circumvent this issue with a small modification of the standard GL procedure, namely by allowing for a ``warm-up'' phase. That is, whereas the standard GL iteration is initialized at the empty string, we perform it starting from strings of a carefully chosen length. 

\begin{theorem}[Goldreich-Levin from distributional agnostic QSQs]\label{theorem:distributional-agnostic-qsq-GL}
    Let $\mathcal{D} = (\mathcal{U}_n, \varphi)$ be a probability distribution over $\mathcal{X}_n\times\{0,1\}$ with uniform marginal over $\mathcal{X}_n$.
    Let $\varepsilon\in (0,1)$. 
    Then, there exists an algorithm that, using $\mathcal{O}\left(\tfrac{n}{\varepsilon^2}\right)$ distributional QSQs of tolerance $\nicefrac{\varepsilon^2}{8}$ for observables that can be implemented with $\mathcal{O}(n)$ single-qubit gates, classical computation time $\tilde{\mathcal{O}}\left(\tfrac{n}{\varepsilon^2}\right)$, and classical memory of size $\tilde{\mathcal{O}}\left(\tfrac{n^2}{\varepsilon^2}\right)$, and outputs a list $L=\{s_1,\ldots,s_{\lvert L\rvert}\}\subseteq\{0,1\}^n$ such that \begin{itemize}
        \item[(i)] if $\lvert \hat{\phi}(s)\rvert\geq \varepsilon$, then $s\in L$, and
        \item[(ii)] if $s\in L$, then $\lvert \hat{\phi}(s)\rvert\geq \nicefrac{\varepsilon}{2}$.
    \end{itemize}
    This list has length $\lvert L\rvert\leq \nicefrac{4 \mathbb{E}_{x\sim\mathcal{U}_n}[(\phi(x))^2]}{\varepsilon^2}\leq \nicefrac{4}{\varepsilon^2}$ by Parseval.
\end{theorem}

\begin{remark}\label{remark:qsq-no-efficient-fourier-sampling}
    The QSQ complexity bounds of \Cref{theorem:distributional-agnostic-qsq-GL} are worse than the quantum sample complexity bounds in \Cref{subsection:distributional-agnostic-quantum-learning} in terms of $n$-dependence: The former depend linearly on $n$, the latter are $n$-independent. This discrepancy between QSQs and quantum examples is unavoidable since \cite[Lemma 5.3]{arunachalam2020qsq} gives an $\Omega(n)$ QSQ complexity lower bound  for exact parity learning. As two distinct parities disagree on half of all inputs, this lower bound carries over to realizable parity learning with a constant accuracy ($< \nicefrac{1}{2}$) from QSQs and to the more general tasks of functional and distributional agnostic parity and Fourier-sparse learning from (distributional) QSQs.
    Moreover, as a consequence of our proofs in \Cref{section:functional-agnostic-quantum-learning} and \Cref{subsection:distributional-agnostic-quantum-learning} (see also \Cref{remark:approximate-fourier-sampling-fourier-approximation}), this shows that a sublinear-in-$n$ number of QSQs is not sufficient to produce a single approximate Fourier sample.
\end{remark}

The proof of our distributional QSQ GL theorem relies on the following observation:

\begin{lemma}\label{lemma:distributional-qsq-GL-auxiliary}
    Let $\mathcal{D} = (\mathcal{U}_n, \varphi)$ be a probability distribution over $\mathcal{X}_n\times\{0,1\}$ with uniform marginal over $\mathcal{X}_n$.
    Let $\varepsilon\in (0,1)$.
    Let $k\geq \lceil \log (\tfrac{2}{\varepsilon})\rceil$ and let $s\in \{0,1\}^k$. 
    We can estimate $\sum_{t\in \{0,1\}^{n-k}} (\hat{\phi}(st))^2$ up to additive error $\varepsilon$ using a single distributional QSQ with tolerance $\nicefrac{\varepsilon}{2}$.
    Moreover, the observable used in the QSQ can be implemented with $\mathcal{O}(n)$ single-qubit gates.
\end{lemma}
\begin{proof}
    Similarly to \cite[Lemma 4.1]{arunachalam2020qsq}, consider the observable 
    \begin{equation}
        M
        = H^{\otimes (n+1)}\cdot\left(\sum_{t\in \{0,1\}^{n-k}} \ket{st}\bra{st} \otimes \ket{1}\bra{1}\right)\cdot H^{\otimes (n+1)} .
    \end{equation}
    By definition of $M$ and $\rho_{\mathcal{D}}$, we have
    \begin{equation}
        \tr[ M \rho_{\mathcal{D}}]
        = \sum_{t\in \{0,1\}^{n-k}} \bra{st, 1} H^{\otimes (n+1)} \rho_{\mathcal{D}} H^{\otimes (n+1)} \ket{st, 1}
        = \sum_{t\in \{0,1\}^{n-k}} \mathbb{E}_{f\sim F_{\mathcal{D}}} \left[ \left\lvert \bra{st, 1}H^{\otimes (n+1)} \ket{\psi_{(\mathcal{U}_n, f)}} \right\rvert^2\right]
    \end{equation}
    Using that $H^{\otimes (n+1)} \ket{\psi_{(\mathcal{U}_n, f)}} = \tfrac{1}{\sqrt{2}}\left(\ket{0}^{\otimes n+1} + \sum_{s\in\{0,1\}^n} \hat{g}_f (s)\ket{s,1}\right)$ with $g_f=(-1)^f$, we see that
    \begin{equation}
        \tr[ M \rho_{\mathcal{D}}]
        = \sum_{t\in \{0,1\}^{n-k}} \mathbb{E}_{f\sim F_{\mathcal{D}}} \left[(\hat{g}_f (st))^2\right] .
    \end{equation}
    Recalling the computation in the proof of \Cref{theorem:agnostic-quantum-fourier-sampling}, we can further rewrite this as
    \begin{equation}
        \tr[ M \rho_{\mathcal{D}}]
        = \sum_{t\in \{0,1\}^{n-k}} \left(\frac{1}{2^n} (1 - \mathbb{E}_{x\sim\mathcal{U}_n}[(\phi(x))^2]) + (\hat{\phi}(st))^2\right) .
    \end{equation}
    In particular, since $0\leq(1-\mathbb{E}_x[\phi^2(x)])\leq1$ this implies
    \begin{equation}
        \left\lvert \tr[ M \rho_{\mathcal{D}}] - \sum_{t\in \{0,1\}^{n-k}}(\hat{\phi}(st))^2\right\rvert
        \leq\tr[ M \rho_{\mathcal{D}}] - \sum_{t\in \{0,1\}^{n-k}}(\hat{\phi}(st))^2
        \leq 2^{-k}
        \leq \frac{\varepsilon}{2}\,,
    \end{equation}
    where the last inequality is by the assumption that $k\geq \lceil \log (\tfrac{2}{\varepsilon})\rceil$.
    As the output $\mu$ of a distributional QSQ with tolerance $\nicefrac{\varepsilon}{2}$ satisfies $\left\lvert \mu - \tr[ M \rho_{\mathcal{D}}] \right\rvert\leq \nicefrac{\varepsilon}{2}$, we conclude that
    \begin{equation}
        \left\lvert \mu - \sum_{t\in \{0,1\}^{n-k}}(\hat{\phi}(st))^2 \right\rvert
        \leq \varepsilon
    \end{equation}
    by triangle inequality.
    
    It remains to argue that measuring the observable $M$ can indeed be implemented using $\mathcal{O}(n)$ single-qubit gates. Clearly, it is sufficient to show that this holds for $\sum_{t\in \{0,1\}^{n-k}} \ket{st}\bra{st}$.
    To see this, it suffices to note that the unitary $\mathds{1}^{\otimes k}\otimes H^{\otimes (n-k)}$ transforms the computational basis into the eigenbasis of $\sum_{t\in \{0,1\}^{n-k}} \ket{st}\bra{st}$.
\end{proof}

We can now prove \Cref{theorem:distributional-agnostic-qsq-GL}:

\begin{proof}[Proof of \Cref{theorem:distributional-agnostic-qsq-GL}]
    Our procedure is very similar to the standard GL procedure (pedagogically presented in \cite[Section 3.5]{odonnell2014analysis}), with two differences: 
    On the one hand, we use distributional agnostic QSQs, instead of classical membership queries, to produce $(\nicefrac{\varepsilon^2}{4})$-accurate estimates for quantities of the form $\sum_{t\in\{0,1\}^{n-k}} (\hat{\phi}(st))^2$ for fixed $s\in\{0,1\}^k$. 
    On the other hand, informed by \Cref{lemma:distributional-qsq-GL-auxiliary}, we initialize the branch-and-prune procedure from the classical GL algorithm not at the empty string ($k=0$) but at strings of length $k=  \lceil \log (\tfrac{16}{\varepsilon^2})\rceil$, for which QSQs of tolerance $\nicefrac{\varepsilon^2}{8}$ are indeed sufficient to get $(\nicefrac{\varepsilon^2}{4})$-accurate estimates for our quantities of interest.
    The analysis then follows along the same lines as that of the standard GL procedure, compare again \cite[Section 3.5]{odonnell2014analysis}.
    Two aspects to note for the complexity analysis: On the one hand, by our choice of $k=  \lceil \log (\tfrac{16}{\varepsilon^2})\rceil$, there are at most $2^k \leq \tfrac{32}{\varepsilon^2}$ different strings at the level at which we start our branching-and-pruning. This matches the at most $\mathcal{O}(\nicefrac{1}{\varepsilon^2})$ strings that the standard GL maintains in each iteration and thereby leads to the claimed time and memory bounds.
    On the other hand, whereas the standard GL analysis incurs a factor of $\tfrac{\log(\nicefrac{n}{\delta})}{\varepsilon^2}$ from a Hoeffding bound when estimating $\sum_{t\in\{0,1\}^{n-k}} (\hat{\phi}(st))^2$, interpreted as expectation values, from samples, we do not incur this factor since QSQs already give estimates for the relevant quantities.
\end{proof}

Given the distributional agnostic QSQ GL procedure of \Cref{theorem:distributional-agnostic-qsq-GL}, we can now obtain distributional agnostic learning results by following a reasoning analogous to that behind \Cref{corollary:agnostic-quantum-parity-learning}, \Cref{corollary:agnostic-quantum-fourier-sparse-learning}, and \Cref{corollary:agnostic-quantum-fourier-sparse-learning-L2}.
However, as the QSQ complexity bound in \Cref{theorem:distributional-agnostic-qsq-GL} is worse than the quantum sample complexity bound of \Cref{corollary:distributional-agnostic-quantum-approximation-fourier-spectrum} in terms of $n$-dependence, this will carry over.
We collect these QSQ variants of our agnostic learning results in the following, slightly informally stated, corollary.

\begin{corollary}\label{corollary:agnostic-qsq-learning}
    Let $\delta,\varepsilon\in (0,1)$. 
    \begin{itemize}
        \item There is an algorithm that uses $\mathcal{O}\left(\tfrac{n}{\varepsilon^2}\right)$ distributional agnostic QSQs with tolerance $\mathcal{O}(\varepsilon^2)$, classical computation time $\mathcal{O}\left(\tfrac{n}{\varepsilon^2}\right)$, and classical memory of size $\tilde{\mathcal{O}}\left(\tfrac{n^2}{\varepsilon^2}\right)$, and achieves distributional $1$-agnostic proper parity learning w.r.t.~uniformly random inputs.
        \item There is an algorithm that uses $\mathcal{O}\left(\tfrac{nk^2}{\varepsilon^2}\right)$ distributional agnostic QSQs with tolerance $\mathcal{O}\left(\tfrac{\varepsilon^2}{k^2}\right)$, classical computation time $\mathcal{O}\left(\tfrac{n k^2}{\varepsilon^2}\right)$, and classical memory of size $\tilde{\mathcal{O}}\left(\tfrac{n^2 k^2}{\varepsilon^2}\right)$, and achieves distributional $2$-agnostic improper Fourier-sparse learning w.r.t.~uniformly random inputs.
    \end{itemize}
\end{corollary}

To conclude this section, we comment on how the complexities of \Cref{corollary:agnostic-qsq-learning} compare with those obtained in \Cref{subsection:distributional-agnostic-quantum-learning}. 
At first glance, \Cref{corollary:agnostic-qsq-learning} seems to achieve a better $\varepsilon$- and $k$-dependence than the results based on mixture-of-superpositions examples access. This is partially due to the tolerance of the QSQ oracle. Namely, if we naively were to use mixture-of-superpositions examples to simulate distributional agnostic QSQs with the desired tolerances (via Hoeffding) and then follow the procedure of \Cref{theorem:distributional-agnostic-qsq-GL} and \Cref{corollary:agnostic-qsq-learning}, we would actually end up with a worse $\varepsilon$- and $k$-dependence than in \Cref{subsection:distributional-agnostic-quantum-learning}.

\section{Classical Verification of Agnostic Quantum Learning}\label{section:verification}

\Cref{section:functional-agnostic-quantum-learning} and \Cref{section:distributional-agnostic-quantum-learning} have demonstrated the power of quantum data for agnostic parity learning and Fourier-sparse learning. 
Here, we show that classical verifiers interacting with quantum provers can make use of this power to solve similar learning problems.

\subsection{Verifying Functional Agnostic Quantum Learning}\label{subsection:verification-functional} 

The quantum learning algorithms presented in \Cref{section:functional-agnostic-quantum-learning} and \Cref{section:distributional-agnostic-quantum-learning} worked without any prior assumptions on the unknown function/distribution. 
For our classical verification protocols, however, we rely on an additional assumption:

\begin{definition}[Functional distributions with no small non-zero Fourier coefficients]\label{definition:functional-no-small-fourier-coefficients}
    Let $\vartheta\in (0,1)$.
    We denote the class of probability distributions $\mathcal{D} = (\mathcal{U}_n, f)$ over $\mathcal{X}_n\times\{0,1\}$ that have a uniform marginal over $\mathcal{X}_n$ and whose $\{-1,1\}$-labels are given by a deterministic function $g=(-1)^f$ that has no non-zero Fourier coefficients of magnitude $<\vartheta$ by 
    \begin{equation}
        \mathfrak{D}_{ \mathcal{U}_n; \geq \vartheta}^{\mathrm{func}}
        \coloneqq \left\{ (\mathcal{U}_n, f)~|~ f:\mathcal{X}_n\to\{0,1\}~\wedge (\hat{g}\neq 0~\Rightarrow \lvert \hat{g}\rvert \geq \vartheta ) \right\} .
    \end{equation}
\end{definition}

Using known results about granularity of Fourier-sparse Boolean functions, we can reinterpret the \Cref{definition:functional-no-small-fourier-coefficients} in terms of Fourier-sparsity requirements: 

\begin{lemma}\label{lemma:no-small-nonzero-Fourier-coeff-versus-sparsity}
    Let $\vartheta\in (0,1)$. 
    Then we have the inclusions
    \begin{equation}
        \left\{(\mathcal{U}_n, f)~|~ f:\mathcal{X}_n\to\{0,1\}~\textrm{s.t.}~g=(-1)^f~\textrm{is Fourier}~\left\lfloor\tfrac{2}{\vartheta}\right\rfloor\textrm{-sparse}\right\}
        \subseteq\mathfrak{D}_{ \mathcal{U}_n; \geq \vartheta}^{\mathrm{func}}    
    \end{equation}
    and 
    \begin{equation}
        \mathfrak{D}_{ \mathcal{U}_n; \geq \vartheta}^{\mathrm{func}}
        \subseteq \left\{(\mathcal{U}_n, f)~|~ f:\mathcal{X}_n\to\{0,1\}~\textrm{s.t.}~g=(-1)^f~\textrm{is Fourier}~\left(\tfrac{1}{\vartheta^2}\right)\textrm{-sparse}\right\}.
    \end{equation}
\end{lemma}
\begin{proof}
    We first show that $(\mathcal{U}_n, f)\in \mathfrak{D}_{ \mathcal{U}_n; \geq \vartheta}^{\mathrm{func}}$ implies that $g=(-1)^f$ is $(\nicefrac{1}{\vartheta^2})$-sparse. This is an easy consequence of Parseval and the definition of $\mathfrak{D}_{ \mathcal{U}_n; \geq \vartheta}^{\mathrm{func}}$, and can be seen as follows:
    \begin{align}
        1 
        &= \sum_{s\in\{0,1\}^n} \left(\hat{g}(s)\right)^2
        = \sum_{s\in\{0,1\}^n: \lvert \hat{g}(s)\rvert\geq \vartheta} \left(\hat{g}(s)\right)^2
        \geq \vartheta^2\cdot \left\lvert \left\{ s\in\{0,1\}^n~|~\lvert \hat{g}(s)\rvert\geq \varepsilon \right\}\right\rvert
        = \vartheta^2\cdot\left\lvert \left\{ s\in\{0,1\}^n~|~\hat{g}(s)\neq 0 \right\}\right\rvert,
    \end{align}
    so $g$ is $(\tfrac{1}{\vartheta^2})$-sparse.

    Now for the converse: Suppose $f:\mathcal{X}_n\to\{0,1\}$ is such that $g=(-1)^f$ is $\left\lfloor\nicefrac{2}{\vartheta}\right\rfloor$-sparse. According to \cite[Theorem 8.1]{gopalan2011testing}, for $k\geq 2$, a $k$-sparse $\{-1,1\}$-valued function on $\mathcal{X}_n$ is $(\lfloor \log(k)\rfloor - 1)$-granular. That is, every Fourier coefficient is an integer multiple of $2^{1 - \lfloor \log(k)\rfloor }$.
    This in particular means that every non-zero Fourier coefficient has absolute value $\geq \tfrac{1}{2^{\lfloor \log(k)\rfloor - 1}}\geq \tfrac{2}{k}$.
    In our case, that means: If $\hat{g}(s)\neq 0$, then $\lvert \hat{g}(s)\rvert\geq \tfrac{2}{\left\lfloor\nicefrac{2}{\vartheta}\right\rfloor}\geq \vartheta$. This gives the claimed inclusion.
\end{proof}

According to \Cref{lemma:no-small-nonzero-Fourier-coeff-versus-sparsity}, functional agnostic learning under the promise of no small non-zero Fourier coefficients is essentially equivalent to a promise that the unknown labeling function is Fourier-sparse. 
Therefore, functional agnostic learning parities or Fourier-sparse functions under such a promise can be viewed as starting from a Fourier-sparsity assumption and aiming to find a best-approximating function relative to an even Fourier-sparser benchmark class. 
Note that these comparative learning problems (borrowing language from \cite{hu2022comparative}) are unconditionally hard for a classical learner with SQ access, this is inherited immediately from the unconditional hardness of parity learning from SQs \cite{kearns1998efficient}. (This already gives the SQ-part of \Cref{theorem:main-result-verification} (a).)
Therefore, the positive results about classical-quantum verification given below establish an unconditional SQ-query complexity separation between what a classical verifier can learn on their own versus what they can learn when interacting with a quantum prover.
If the classical learner has random example access, we are not aware of any computational hardness results for the noiseless learning problems considered here. (In particular, the random example part of \Cref{theorem:main-result-verification} (a) will only follow from the results in the next two subsections.) In the context of our work, it would be interesting to investigate the classical computational resources required to solve these agnostic tasks under a sparsity promise.

In the following, we give two different approaches to interactive classical verification of functional agnostic learning under a promise as in \Cref{definition:functional-no-small-fourier-coefficients}. 
In the first approach, the verifier asks the prover to make use of their quantum data acces to obtain and send a list of relevant Fourier coefficients. To notice any potential dishonesty (or incompetence) on the part of the prover, the verifier then uses their classical data access to ensure that they have received a list whose accumulated Fourier weight is sufficiently large. 
For our second approach, the verifier asks the prover to exactly learn the unknown function, relying on their quantum data access and the promised Fourier-sparsity, and to then simulate the action of a classical membership query oracle. Thus, the verifier can use the interactive Goldreich-Levin algorithm of \cite{goldwasser2021interactive}.
We begin with the formalization of the former approach and present the strongest version of our results for the functional agnostic case. Here, we consider a classical verifier with classical SQ access and a quantum prover with distributional QSQ access.

\begin{theorem}\label{theorem:functional-agnostic-quantum-parity-verification-qsq-no-small-non-zero-Fourier-coeff}
    Let $\delta, \varepsilon,\vartheta\in (0,1)$.
    The class of $n$-bit parities is efficiently proper $1$-agnostic verifiable w.r.t.~$\mathfrak{D}_{ \mathcal{U}_n; \geq \vartheta}^{\mathrm{func}}$ by a classical verifier $V$ with classical SQ access interacting with a quantum prover $P$ with QSQ access.
    There is a verifier-prover pair $(V,P)$ such that $P$ uses $\mathcal{O}\left(\tfrac{n}{\vartheta^2}\right)$ functional agnostic QSQs of tolerance at worst $\nicefrac{\vartheta}{8}$ for observables that can be implemented with $\mathcal{O}(n)$ single-qubit gates, a classical memory of size $\tilde{\mathcal{O}}\left(\tfrac{n^2}{\vartheta^2}\right)$, and classical running time $\tilde{\mathcal{O}}\left(\tfrac{n}{\vartheta^2}\right)$, and such that $V$ uses $\mathcal{O}\left(\tfrac{1}{\vartheta^2}\right)$ classical SQs of tolerance at worst $\nicefrac{(\varepsilon^2 \vartheta^2)}{64}$, $\tilde{\mathcal{O}}\left(\tfrac{n}{\vartheta^2}\right)$ classical running time, and a classical memory of size $\tilde{\mathcal{O}}\left(\tfrac{n}{\vartheta^2}\right)$.
    Moreover, this can be achieved by a pair $(V,P)$ that uses only a single round of communication consisting of at most $\mathcal{O}\left(\tfrac{n}{\vartheta^2}\right)$ classical bits.
\end{theorem}

Here and in the following results, note that the complexities of the honest quantum prover are independent of $\varepsilon$ but depend on $\vartheta$. This happens in our verification approach because, due to our assumption $\mathcal{D}\in \mathfrak{D}_{ \mathcal{U}_n; \geq \vartheta}^{\mathrm{func}}$, it is sufficient for the quantum prover to ``resolve'' the Fourier spectrum to accuracy $\sim\vartheta$.
In contrast, the verifier in our procedure does need to aim for the accuracy $\varepsilon$, which therefore appears in the tolerances of the verifier's SQs in \Cref{theorem:functional-agnostic-quantum-parity-verification-qsq-no-small-non-zero-Fourier-coeff} (and in the verifier's complexities in \Cref{theorem:functional-agnostic-quantum-parity-verification-no-small-non-zero-Fourier-coeff} below).

\begin{proof}
    Let $\delta,\varepsilon\in (0,1)$.
    Let $\mathcal{D}\in \mathfrak{D}_{ \mathcal{U}_n; \geq \vartheta}$.
    We begin the proof by describing the actions of the classical verifier $V$ and the honest quantum prover $P$:
    \begin{enumerate}
        \item $V$ asks $P$ to provide a list $L=\{s_1,\ldots ,s_{\lvert L\rvert}\}\subset \{0,1\}^n$ of length $\lvert L \rvert\leq \nicefrac{4}{\vartheta^2}$ consisting of pairwise distinct $n$-bit strings whose associated Fourier coefficients are non-zero.
        \item $P$ follows the procedure in \Cref{theorem:functional-agnostic-qsq-GL} to produce a list $L=\{s_1,\ldots,s_{\lvert L\rvert}\}\subseteq\{0,1\}^n$ such that
        (i) if $\lvert \hat{g}(s)\rvert\geq \vartheta$, then $s\in L$, and (ii) if $s\in L$, then $\lvert \hat{g}(s)\rvert\geq \nicefrac{\vartheta}{2}$.
        $P$ then sends the list $L$ to $V$.
        \item If $V$ receives a list $L$ that is of length $\lvert L\rvert> \nicefrac{4}{\vartheta^2} $, $V$ rejects the interaction. Otherwise, $V$ uses $\lvert L\rvert$ classical SQs of tolerance $\nicefrac{\varepsilon^2}{16 \lvert L\rvert}$ to obtain estimates $\hat{\gamma}(s)$ of $\hat{g}(s)$ for all $s\in L$. (For $t\not\in L$, the verifier's estimate $\hat{\gamma}(t)$ for $\hat{g}(t)$ is just $0$.)
        \item If $\sum_{\ell = 1}^{\lvert L\rvert} \left( \hat{\gamma}(s_\ell) \right)^2 \geq 1 - \tfrac{\varepsilon^2}{8}$, then $V$ determines $s_{\rm out}\in\operatorname{argmax}_{1\leq \ell\leq \lvert L \rvert} \hat{\gamma}(s)$ and outputs the hypothesis $h:\mathcal{X}_n\to \{0,1\}$, $h(x)=s_{\rm out}\cdot x$. If $\sum_{\ell = 1}^{\lvert L\rvert} \left( \hat{\gamma}(s_\ell) \right)^2 < 1 - \tfrac{\varepsilon^2}{8}$, then $V$ outputs $\mathrm{reject}$.
    \end{enumerate}
    We now show that the pair $(V,P)$ has the desired completeness and soundness properties. As a first step towards this goal, we show that $V$ accepts an interaction with $P$ with high probability. To this end, observe that, according to \Cref{theorem:functional-agnostic-qsq-GL}, the list produced by $P$ has length $\lvert L \rvert\leq\nicefrac{4}{\vartheta^2}$, $V$ never rejects $P$ in Step 3. Moreover, we have
    \begin{align}
        \sum_{\ell = 1}^{\lvert L\rvert} \left( \hat{\gamma}(s_\ell) \right)^2
        &\geq \sum_{\ell = 1}^{\lvert L\rvert} \left( \hat{g}(s_\ell) \right)^2 - 2\sum_{\ell = 1}^{\lvert L\rvert} \lvert\hat{g}(s_\ell)-\hat{\gamma}(s_\ell)\rvert\\
        &\geq \sum_{s:\hat{g}(s)\neq 0} \lvert\hat{g}(s)\rvert^2 - 2\sum_{\ell = 1}^{\lvert L\rvert} \lvert\hat{g}(s_\ell)-\hat{\gamma}(s_\ell)\rvert\\
        &\geq 1 - 2\lvert L\rvert \cdot \frac{\varepsilon^2}{16 \lvert L\rvert}\\
        &= 1 - \frac{\varepsilon^2}{8}\, ,
    \end{align}
    where the first step used that the function $[-1,1]\ni \xi\to \xi^2$ is $2$-Lipschitz, the second step used that $\mathcal{D}\in \mathfrak{D}_{ \mathcal{U}_n; \geq \vartheta}^{\mathrm{func}}$ implies $(\hat{g}(s)\neq 0~ \Rightarrow ~ s\in L)$ according to Step 2, and the third step used the approximation guarantee in Step 3.
    
    Moreover, whenever $V$ not reject in Steps 3 or 4, then the output string $s_{\rm out}\in\operatorname{argmax}_{1\leq \ell\leq \lvert L \rvert} \gamma(s)$ of $V$ is as desired. This can be seen as follows: If $V$ does not reject in Steps 3 or 4, then this implies that for any $s\not\in L$,
    \begin{align}
        \left(\hat{g}(s)\right)^2
        &\leq \sum_{t\not\in L} \left(\hat{g}(t)\right)^2\\
        &= 1 -  \sum_{t\in L} \left(\hat{g}(t)\right)^2\\
        &\leq 1 - \sum_{t\in L} \left(\hat{\gamma}(t)\right)^2 + 2 \sum_{\ell = 1}^{\lvert L\rvert} \lvert\hat{g}(s_\ell)-\hat{\gamma}(s_\ell)\rvert\\
        &\leq \frac{\varepsilon^2}{8} + 2\lvert L\rvert \cdot \frac{\varepsilon^2}{16 \lvert L\rvert}\\
        &= \frac{\varepsilon^2}{4},
    \end{align}
    where we again used that $[-1,1]\ni \xi\to \xi^2$ is $2$-Lipschitz.
    This tells us that $\lvert \hat{g}(t)\rvert \leq\nicefrac{\varepsilon}{2}$ holds for every $t\not\in L$, which now allows us to show that the output $s_{\rm out}\in\operatorname{argmax}_{1\leq \ell\leq \lvert L \rvert} \gamma(s)$ of $V$ has the desired property:   
    \begingroup
    \allowdisplaybreaks
    \begin{align}
        \mathbb{P}_{x\sim\mathcal{U}_n}[f(x)\neq h(x)]
        &= \mathbb{P}_{x\sim\mathcal{U}_n}[g(x)\neq \chi_{s_{\rm out}}(x)]\\
        &= \frac{1-\hat{g}(s_{\rm out})}{2}\\
        &= \frac{1-\hat{\gamma}(s_{\rm out})}{2} + \frac{\hat{g}(s_{\rm out})-\hat{\gamma}(s_{\rm out})}{2}\\
        &\leq \frac{1-\max_{t\in\{0,1\}^n}\hat{\gamma}(t)}{2} + \frac{\varepsilon}{2}\\
        &= \min_{t\in\{0,1\}^n} \frac{1-\hat{\gamma}(t)}{2} + \frac{\varepsilon}{2}\\
        &= \min_{t\in\{0,1\}^n} \frac{1-\hat{g}(t)}{2} + \frac{\hat{\gamma}(t)-\hat{g}(t)}{2} + \frac{\varepsilon}{2}\\
        &\leq \min_{t\in\{0,1\}^n} \frac{1-\hat{g}(t)}{2} + \max_{v\in\{0,1\}^n}\frac{\hat{\gamma}(v)-\hat{g}(v)}{2}+ \frac{\varepsilon}{2}\\
        &\leq \min_{t\in\{0,1\}^n} \frac{1-\hat{g}(t)}{2} + \max\left\{\max_{v\in L}\frac{\hat{\gamma}(v)-\hat{g}(v)}{2}, \max_{v\not\in L}\frac{\hat{\gamma}(v)-\hat{g}(v)}{2}\right\}+ \frac{\varepsilon}{2}\\
        &\leq \min_{t\in\{0,1\}^n} \frac{1-\hat{g}(t)}{2} + \max\left\{\frac{\varepsilon}{2}, \max_{v\not\in L}\frac{-\hat{g}(v)}{2}\right\}+ \frac{\varepsilon}{2}\\
        &\leq \min_{t\in\{0,1\}^n} \frac{1-\hat{g}(t)}{2} + \max\left\{\frac{\varepsilon}{2}, \max_{v\not\in L}\frac{\lvert\hat{g}(v)\rvert}{2}\right\}+ \frac{\varepsilon}{2}\\
        &\leq \min_{t\in\{0,1\}^n} \frac{1-\hat{g}(t)}{2} + \frac{\varepsilon}{2}+ \frac{\varepsilon}{2}\\
        &= \min_{t\in\{0,1\}^n} \frac{1-\hat{g}(t)}{2} + \varepsilon\, .
    \end{align}
    \endgroup
    Note that this last part of our reasoning only relied on $V$ not rejecting in Steps 3 or 4, but was independent of the action of the quantum prover. Therefore, with this we have already established the desired soundness.

    It remains to bound the sample and time complexities for $V$ and $P$.
    \Cref{theorem:functional-agnostic-qsq-GL} directly tells us that $P$ uses $\mathcal{O}\left(\tfrac{n}{\vartheta^2}\right)$ QSQs of tolerance $\nicefrac{\tilde\varepsilon^2}{8}$ for observables that can be implemented with $\mathcal{O}(n)$ single-qubit gates as well as a classical memory of size $\tilde{\mathcal{O}}\left(\tfrac{n^2}{\vartheta^2}\right)$, and classical running time $\tilde{\mathcal{O}}\left(\tfrac{n}{\vartheta^2}\right)$.
    The classical SQ complexity of $V$ is $\mathcal{O}\left(\lvert L\rvert\right)\leq \mathcal{O}\left(\tfrac{1}{\vartheta^2}\right)$, where each SQ is with tolerance $\nicefrac{\varepsilon^2}{16 \lvert L\rvert} \geq \nicefrac{\varepsilon^2 \vartheta^2}{64}$. 
    The classical running time of $V$ is $\mathcal{O}\left(n\lvert L\rvert\right)\leq \mathcal{O}\left(\tfrac{n}{\vartheta^2}\right)$ for Step 3. 
    The classical memory used by $V$ for Step 3 is of size $\mathcal{O}\left(n\lvert L\rvert\right)\leq \mathcal{O}\left(\tfrac{n}{\vartheta^2}\right)$.
    The computational cost and memory required for Step 4 are negligible in comparison to those of Step 3.
    Finally, our protocol clearly requires only a single round of communication, and the communicated object is $L$, which consists of at most $\nicefrac{4}{\vartheta^2}$ strings of $n$ bits.
    This finishes the proof.
\end{proof}

The above result proves a SQ complexity separation between what is classically query-efficiently achievable versus what is classically query-efficiently achievable when interacting with an untrusted quantum prover with QSQ access. (Here, we consider the regime $\varepsilon, \vartheta\geq\Omega(\nicefrac{1}{\poly (n)})$, so that the verifier and prover in \Cref{theorem:functional-agnostic-quantum-parity-verification-qsq-no-small-non-zero-Fourier-coeff} are efficient.) Thus, \Cref{theorem:functional-agnostic-quantum-parity-verification-qsq-no-small-non-zero-Fourier-coeff} already establishes the SQ-parts of \Cref{theorem:main-result-verification}.

Next, we give a version of our interactive verification result in which we allow both the classical verifier and the quantum prover to have their respective classical and quantum example access. The underlying proof strategy is the same as for \Cref{theorem:functional-agnostic-quantum-parity-verification-qsq-no-small-non-zero-Fourier-coeff} and would even work for a classical verifier with SQ access, but here we give example access on both sides in the spirit of symmetry.

\begin{theorem}\label{theorem:functional-agnostic-quantum-parity-verification-no-small-non-zero-Fourier-coeff}
    Let $\delta, \varepsilon,\vartheta\in (0,1)$.
    The class of $n$-bit parities is efficiently proper $1$-agnostic verifiable w.r.t.~$\mathfrak{D}_{ \mathcal{U}_n; \geq \vartheta}^{\mathrm{func}}$ by a classical verifier $V$ with access to classical random examples interacting with a quantum prover $P$ with access to quantum superposition examples.
    There is a verifier-prover pair $(V,P)$ such that $P$ uses $\mathcal{O}\left(\tfrac{\log(\nicefrac{1}{\delta\vartheta^2})}{\vartheta^4}\right)$ copies of $\ket{\psi_{(\mathcal{U}_n, f)}}$ , $\mathcal{O}\left(n\tfrac{\log(\nicefrac{1}{\delta\vartheta^2})}{\vartheta^4}\right)$ single-qubit gates, a classical memory of size $\tilde{\mathcal{O}}\left(n\tfrac{\log(\nicefrac{1}{\delta\vartheta^2})}{\vartheta^4}\right)$, and classical running time $\tilde{\mathcal{O}}\left(n\tfrac{\log(\nicefrac{1}{\delta \vartheta^2})}{\vartheta^4}\right)$, and such that $V$ uses $\mathcal{O}\left(\tfrac{\log(\nicefrac{1}{\delta\vartheta^2})}{\varepsilon^4 \vartheta^4}\right)$ classical random examples, $\tilde{\mathcal{O}}\left(n\tfrac{\log(\nicefrac{1}{\delta\vartheta^2})}{\varepsilon^4 \vartheta^4}\right)$ classical running time, and a classical memory of size $\tilde{\mathcal{O}}\left(n\tfrac{\log(\nicefrac{1}{\delta\vartheta^2})}{\varepsilon^4 \vartheta^4}\right)$.
    Moreover, this can be achieved by a pair $(V,P)$ that uses only a single round of communication consisting of at most $\mathcal{O}\left(\tfrac{n}{\vartheta^2}\right)$ classical bits.
\end{theorem}
\begin{proof}
    Let $\delta,\varepsilon\in (0,1)$.
    Let $\mathcal{D}\in \mathfrak{D}_{ \mathcal{U}_n; \geq \vartheta}$.
    We begin the proof by describing the actions of the classical verifier $V$ and the honest quantum prover $P$:
    \begin{enumerate}
        \item $V$ asks $P$ to provide a list $L=\{s_1,\ldots ,s_{\lvert L\rvert}\}\subset \{0,1\}^n$ of length $\lvert L \rvert\leq \nicefrac{64}{\vartheta^2}$ consisting of pairwise distinct $n$-bit strings whose associated Fourier coefficients are non-zero.
        \item $P$ follows the procedure in \Cref{corollary:quantum-approximation-fourier-spectrum} to produce, with success probability $\geq 1-\tfrac{\delta}{2}$, a succinctly represented $\Tilde{g}:\mathcal{X}_n\to [-1,1]$ such that $\norm{\Tilde{g}-\hat{g}}_\infty\leq\nicefrac{\vartheta}{2}$ and $\norm{\Tilde{g}}_0\leq\nicefrac{16}{\vartheta^2}$. If $P$ obtains an output that violates the $\norm{\cdot}_0$-bound, then $P$ declares failure and the interaction aborts.
        Otherwise, $P$ then sends the list $L = \{s\in\{0,1\}^n~|~ \lvert \Tilde{g}(s)\rvert\geq \nicefrac{\vartheta}{2} \}$ to $V$.
        \item If $V$ receives a list $L$ that is of length $\lvert L\rvert> \nicefrac{16}{\vartheta^2} $, $V$ rejects the interaction. Otherwise, $V$ uses $\mathcal{O}\left(\tfrac{\lvert L\rvert^2\log(\nicefrac{\lvert L\rvert}{\delta})}{\varepsilon^4}\right)$ classical random examples from $\mathcal{D}$ to obtain simultaneously $(\nicefrac{\varepsilon^2}{16 \lvert L\rvert})$-accurate estimates $\hat{\gamma}(s)$ of $\hat{g}(s)$ for all $s\in L$, with success probability $\geq 1-\tfrac{\delta}{2}$, via Chernoff-Hoeffding combined with a union bound over $L$. (For $t\not\in L$, the verifier's estimate $\hat{\gamma}(t)$ for $\hat{g}(t)$ is just $0$.)
        \item If $\sum_{\ell = 1}^{\lvert L\rvert} \left( \hat{\gamma}(s_\ell) \right)^2 \geq 1 - \tfrac{\varepsilon^2}{8}$, then $V$ determines $s_{\rm out}\in\operatorname{argmax}_{1\leq \ell\leq \lvert L \rvert} \hat{\gamma}(s)$ and outputs the hypothesis $h:\mathcal{X}_n\to \{0,1\}$, $h(x)=s_{\rm out}\cdot x$. If $\sum_{\ell = 1}^{\lvert L\rvert} \left( \hat{\gamma}(s_\ell) \right)^2 < 1 - \tfrac{\varepsilon^2}{8}$, then $V$ outputs $\mathrm{reject}$.
    \end{enumerate}
    We can prove completeness and soundness of the pair $(V,P)$ with essentially the same reasoning as in the proof of \Cref{theorem:functional-agnostic-quantum-parity-verification-qsq-no-small-non-zero-Fourier-coeff}. The only difference in the analysis is that some simple union bounds are required.

    It remains to bound the sample and time complexities for $V$ and $P$.
    \Cref{corollary:quantum-approximation-fourier-spectrum} directly tells us that $P$ uses $\mathcal{O}\left(\tfrac{\log(\nicefrac{1}{\delta\vartheta^2})}{\vartheta^4}\right)$ copies of the quantum superposition state $\ket{\psi_{(\mathcal{U}_n, f)}}$ as well as $\mathcal{O}\left(n\tfrac{\log(\nicefrac{1}{\delta\vartheta^2})}{\vartheta^4}\right)$ single-qubit gates, a classical memory of size $\tilde{\mathcal{O}}\left(n\tfrac{\log(\nicefrac{1}{\delta\vartheta^2})}{\vartheta^4}\right)$, and classical running time $\tilde{\mathcal{O}}\left(n\tfrac{\log(\nicefrac{1}{\delta \vartheta^2})}{\vartheta^4}\right)$.
    The classical sample complexity of $V$ is $\mathcal{O}\left(\tfrac{\lvert L\rvert^2\log(\nicefrac{\lvert L\rvert}{\delta})}{\varepsilon^4}\right)\leq \mathcal{O}\left(\tfrac{\log(\nicefrac{1}{\delta\vartheta^2})}{\varepsilon^4 \vartheta^4}\right)$, as noted in Step 3. The classical running time of $V$ is $\tilde{\mathcal{O}}\left(n\tfrac{\lvert L\rvert^2\log(\nicefrac{\lvert L\rvert}{\delta})}{\varepsilon^4}\right)\leq \tilde{\mathcal{O}}\left(n\tfrac{\log(\nicefrac{1}{\delta\vartheta^2})}{\varepsilon^4 \vartheta^4}\right)$ for Step 3. 
    The classical memory used by $V$ for Step 3 is of size $\tilde{\mathcal{O}}\left(n\tfrac{\lvert L\rvert^2\log(\nicefrac{\lvert L\rvert}{\delta})}{\varepsilon^4}\right)\leq \tilde{\mathcal{O}}\left(n\tfrac{\log(\nicefrac{1}{\delta\vartheta^2})}{\varepsilon^4 \vartheta^4}\right)$.
    The computational cost and memory required for Step 4 are negligible in comparison to those of Step 3.
    Finally, our protocol clearly requires only a single round of communication, and the communicated object is $L$, which consists of at most $\nicefrac{64}{\vartheta^2}$ strings of $n$ bits.
    This finishes the proof.
\end{proof}

\Cref{theorem:functional-agnostic-quantum-parity-verification-no-small-non-zero-Fourier-coeff} again achieves a separation between the capabilities of a lone classical learner and those of a classical-quantum verifier-prover pair. 
Namely, classical exact parity learning requires linearly-in-$n$ many random examples (even in the realizable case, see \cite[Theorem 5]{caro2020quantum} for a simple information-theoretic proof), whereas the verifier in \Cref{theorem:functional-agnostic-quantum-parity-verification-no-small-non-zero-Fourier-coeff} only needs an $n$-independent number of copies. (While we do not explicitly consider exact learning in \Cref{theorem:functional-agnostic-quantum-parity-verification-no-small-non-zero-Fourier-coeff}, this can be obtained as a special case, namely by focusing on the realizable case with $\vartheta=\varepsilon=\nicefrac{1}{3}$, since two distinct parities differ on half of all inputs.)
However, we emphasize that \Cref{theorem:functional-agnostic-quantum-parity-verification-no-small-non-zero-Fourier-coeff} does not yet establish an analogous computational complexity separation conditioned on LPN. While functional agnostic parity learning is at least as hard as LPN and thus conjectured to be hard (compare the discussion in \cite{goldwasser2021interactive}), we do not believe that the problem of functional agnostic parity learning under the promise of no small non-zero Fourier coefficients has been considered in the literature, and we are not aware of any reductions establishing its hardness.
Therefore, to get a conditional separation specifically in terms of computational complexities, we extend our verification results in the next two subsections to also allow for noise and for the distributional agnostic case.

Next, we now show the Fourier-sparse $2$-agnostic analogues of \Cref{theorem:functional-agnostic-quantum-parity-verification-qsq-no-small-non-zero-Fourier-coeff} and \Cref{theorem:functional-agnostic-quantum-parity-verification-no-small-non-zero-Fourier-coeff}:

\begin{theorem}\label{theorem:functional-agnostic-quantum-fourier-sparse-verification-qsq-no-small-non-zero-Fourier-coeff}
    Let $\varepsilon,\vartheta\in (0,1)$.
    The class of Fourier-$k$-sparse functions on $n$ bits is efficiently improper $2$-agnostic verifiable w.r.t.~$\mathfrak{D}_{ \mathcal{U}_n; \geq \vartheta}^{\mathrm{func}}$ by a classical verifier $V$ with classical SQ access interacting with a quantum prover $P$ with QSQ access.
    There is a verifier-prover pair $(V,P)$ such that $P$ uses $\mathcal{O}\left(\tfrac{n}{\vartheta^2}\right)$ functional agnostic QSQs of tolerance at worst $\nicefrac{\vartheta}{8}$ for observables that can be implemented with $\mathcal{O}(n)$ single-qubit gates, a classical memory of size $\tilde{\mathcal{O}}\left(\tfrac{n^2 }{\vartheta^2}\right)$, and classical running time $\tilde{\mathcal{O}}\left(\tfrac{n }{\vartheta^2}\right)$, and such that $V$ uses $\mathcal{O}\left(\tfrac{1}{\vartheta^2}\right)$ classical SQs of tolerance at worst $\nicefrac{(\varepsilon^2 \vartheta^2)}{64 k^2}$, $\tilde{\mathcal{O}}\left(\tfrac{n }{\vartheta^2}\right)$ classical running time, and a classical memory of size $\tilde{\mathcal{O}}\left(\tfrac{n }{\vartheta^2}\right)$.
    Moreover, this can be achieved by a pair $(V,P)$ that uses only a single round of communication consisting of at most $\mathcal{O}\left(\tfrac{n}{\vartheta^2}\right)$ classical bits.
\end{theorem}
\begin{proof}
    Let $\delta,\varepsilon\in (0,1)$.
    Let $\mathcal{D}\in \mathfrak{D}_{ \mathcal{U}_n; \geq \vartheta}$.
    We begin the proof by describing the actions of the classical verifier $V$ and the honest quantum prover $P$:
    \begin{enumerate}
        \item $V$ asks $P$ to provide a list $L=\{s_1,\ldots ,s_{\lvert L\rvert}\}\subset \{0,1\}^n$ of length $\lvert L \rvert\leq \nicefrac{4}{\vartheta^2}$ consisting of pairwise distinct $n$-bit strings whose associated Fourier coefficients are non-zero.
        \item $P$ follows the procedure in \Cref{theorem:functional-agnostic-qsq-GL} to produce a list $L=\{s_1,\ldots,s_{\lvert L\rvert}\}\subseteq\{0,1\}^n$ such that
        (i) if $\lvert \hat{g}(s)\rvert\geq \vartheta$, then $s\in L$, and (ii) if $s\in L$, then $\lvert \hat{g}(s)\rvert\geq \nicefrac{\vartheta}{2}$.
        $P$ then sends the list $L$ to $V$.
        \item If $V$ receives a list $L$ of length $\lvert L\rvert> \nicefrac{4}{\vartheta^2} $, $V$ rejects the interaction. Otherwise, $V$ uses $\lvert L\rvert$ classical SQs of tolerance $\nicefrac{\varepsilon^2}{64 k^2 \lvert L\rvert}$ to obtain estimates $\hat{\gamma}(s)$ of $\hat{g}(s)$ for all $s\in L$. (For $t\not\in L$, the verifier's estimate $\hat{\gamma}(t)$ for $\hat{g}(t)$ is just $0$.)
        \item If $\sum_{\ell = 1}^{\lvert L\rvert} \left( \hat{\gamma}(s_\ell) \right)^2 \geq 1 - \tfrac{\varepsilon^2}{32 k^2}$, then $V$ determines the $k$ heaviest Fourier coefficients in $L$. That is, $V$ determines $s_1\in\operatorname{argmax}_{t\in L} \lvert \hat{\gamma}(t)\rvert$ and, for $2\leq \ell\leq k$, $s_\ell \in \operatorname{argmax}_{t\in L\setminus\{s_1,\ldots,s_{\ell -1}\}} \lvert \hat{\gamma}(t)\rvert$, and then outputs the randomized hypothesis  $h:\mathcal{X}_n\to \{0,1\}$, $h(x)=\sum_{\ell=1}^k \hat{\gamma}(s_\ell) \chi_{s_\ell}(x)$ from \Cref{lemma:Fourier-sparse-learning-via-heaviest-Fourier-coefficients}. (If $L$ has fewer than $k$ elements, $V$ just picks the strings $s_{\lvert L\rvert +1},\ldots,s_k$ at random from $\{0,1\}^n\setminus L$. These strings don't matter since $V$ estimates their Fourier weight as $0$.)
        If $\sum_{\ell = 1}^{\lvert L\rvert} \left( \hat{\gamma}(s_\ell) \right)^2 < 1 - \tfrac{\varepsilon^2}{32 k^2}$, then $V$ outputs $\mathrm{reject}$.
    \end{enumerate}
    We now show that the pair $(V,P)$ has the desired completeness and soundness properties. As a first step towards this goal, we show that $V$ accepts an interaction with $P$ with high probability. To this end, observe that, according to \Cref{theorem:functional-agnostic-qsq-GL}, the list produced by $P$ has length $\lvert L \rvert\leq\nicefrac{4}{\vartheta^2}$, $V$ never rejects $P$ in Step 3. Moreover, we have
    \begin{align}
        \sum_{\ell = 1}^{\lvert L\rvert} \left( \hat{\gamma}(s_\ell) \right)^2
        &\geq \sum_{\ell = 1}^{\lvert L\rvert} \left( \hat{g}(s_\ell) \right)^2 - 2\sum_{\ell = 1}^{\lvert L\rvert} \lvert\hat{g}(s_\ell)-\hat{\gamma}(s_\ell)\rvert\\
        &\geq \sum_{s:\hat{g}(s)\neq 0} \lvert\hat{g}(s)\rvert^2 - 2\sum_{\ell = 1}^{\lvert L\rvert} \lvert\hat{g}(s_\ell)-\hat{\gamma}(s_\ell)\rvert\\
        &\geq 1 - 2\lvert L\rvert \cdot \frac{\varepsilon^2}{64 k^2 \lvert L\rvert}\\
        &= 1 - \frac{\varepsilon^2}{32 k^2}\, ,
    \end{align}
    where the first step used that the function $[-1,1]\ni \xi\to \xi^2$ is $2$-Lipschitz, the second step used that $\mathcal{D}\in \mathfrak{D}_{ \mathcal{U}_n; \geq \vartheta}^{\mathrm{func}}$ implies $(\hat{g}(s)\neq 0~ \Rightarrow ~ s\in L)$ according to Step 2, and the third step used the approximation guarantee in Step 3.
    
    Moreover, whenever $V$ does not reject in Steps 3 or 4, then the output hypothesis of $V$ is as desired. This can be seen as follows: If $V$ does not reject in Steps 3 or 4, then this implies that
    \begin{align}
        \sum_{t\not\in L}\left(\hat{g}(t)\right)^2
        &= 1 -  \sum_{t\in L} \left(\hat{g}(t)\right)^2\\
        &\leq 1 - \sum_{t\in L} \left(\hat{\gamma}(t)\right)^2 + 2 \sum_{\ell = 1}^{\lvert L\rvert} \lvert\hat{g}(s_\ell)-\hat{\gamma}(s_\ell)\rvert\\
        &\leq \frac{\varepsilon^2}{32 k^2} + 2\lvert L\rvert \cdot \frac{\varepsilon^2}{64 k^2 \lvert L\rvert}\\
        &= \frac{\varepsilon^2}{16 k^2},
    \end{align}
    where we again used that $[-1,1]\ni \xi\to \xi^2$ is $2$-Lipschitz.
    This tells us that $\lvert \hat{g}(t)\rvert \leq\nicefrac{\varepsilon}{4k}$ holds for every $t\not\in L$, which now allows us to show that the output hypothesis of $V$ has the desired property.
    Let $t_1,\ldots,t_k\in \{0,1\}^n$ be $k$ heaviest Fourier coefficients of $g$. That is, let $t_1\in \operatorname{argmax}_{t\in \{0,1\}^n} \lvert \hat{g}(t)\rvert$, and for $2\leq \ell\leq k$, let $t_\ell\in \operatorname{argmax}_{t\in \{0,1\}^n\setminus\{t_1,\ldots,t_{\ell -1}\}} \lvert \hat{g}(t)\rvert$.
    By \Cref{lemma:technical}, we have 
    \begin{align}
        \max_{1\leq \ell\leq k} \lvert \hat{g}(t_\ell) - \hat{g}(s_\ell)\rvert
        &\leq 2 \max_{t\in\{0,1\}^n} \lvert \hat{g}(t) - \hat{\gamma}(t)\rvert\\
        &= 2\max\left\{\max_{t\in L}\lvert \hat{g}(t) - \hat{\gamma}(t)\rvert,\max_{t\not\in L}\lvert \hat{g}(t) - \hat{\gamma}(t)\rvert \right\}\\
        &= 2\max\left\{\max_{t\in L}\lvert \hat{g}(t) - \hat{\gamma}(t)\rvert,\max_{t\not\in L}\lvert \hat{g}(t)\rvert \right\}\\
        &\leq 2\max\left\{\frac{\varepsilon^2}{64 k^2 |L|}, \frac{\varepsilon}{4k}\right\}\\
        &= \frac{\varepsilon}{2k} .
    \end{align}
    As $\tfrac{\varepsilon^2}{64 k^2 |L|}\leq \tfrac{\varepsilon}{2k}$, this shows that the verifier has $(\nicefrac{\varepsilon}{2k})$-accurate estimates of $k$ $(\nicefrac{\varepsilon}{2k})$-approximately-heaviest Fourier coefficients of $\phi$. 
    Thus, by \Cref{lemma:Fourier-sparse-learning-via-heaviest-Fourier-coefficients}, the randomized hypothesis produced by $V$ is as desired.
    Note that this last part of our reasoning only relied on $V$ not rejecting in Steps 3 or 4, but was independent of the action of the quantum prover. Therefore, with this we have already established the desired soundness.

    It remains to bound the sample and time complexities for $V$ and $P$.
    \Cref{theorem:functional-agnostic-qsq-GL} directly tells us that $P$ uses $\mathcal{O}\left(\tfrac{n}{\vartheta^2}\right)$ QSQs of tolerance $\nicefrac{\tilde\varepsilon^2}{8}$ for observables that can be implemented with $\mathcal{O}(n)$ single-qubit gates as well as a classical memory of size $\tilde{\mathcal{O}}\left(\tfrac{n^2}{\vartheta^2}\right)$, and classical running time $\tilde{\mathcal{O}}\left(\tfrac{n}{\vartheta^2}\right)$.
    The classical SQ complexity of $V$ is $\mathcal{O}\left(\lvert L\rvert\right)\leq \mathcal{O}\left(\tfrac{1}{\vartheta^2}\right)$, where each SQ is with tolerance $\nicefrac{\varepsilon^2}{64 k^2 \lvert L\rvert} \geq \nicefrac{(\varepsilon^2 \vartheta^2)}{64 k^2}$. 
    The classical running time of $V$ is $\mathcal{O}\left(n\lvert L\rvert\right)\leq \mathcal{O}\left(\tfrac{n}{\vartheta^2}\right)$ for Step 3. 
    The classical memory used by $V$ for Step 3 is of size $\tilde{\mathcal{O}}\left(n\lvert L\rvert\right)\leq \tilde{\mathcal{O}}\left(\tfrac{n}{\vartheta^2}\right)$.
    The computational cost and memory required for Step 4 are negligible in comparison to those of Step 3.
    Finally, our protocol clearly requires only a single round of communication, and the communicated object is $L$, which consists of at most $\nicefrac{4}{\vartheta^2}$ strings of $n$ bits.
    This finishes the proof.
\end{proof}

\begin{theorem}\label{theorem:functional-agnostic-quantum-fourier-sparse-verification-no-small-non-zero-Fourier-coeff}
    Let $\delta, \varepsilon,\vartheta\in (0,1)$.
    The class of Fourier-$k$-sparse functions on $n$ bits is efficiently improper $2$-agnostic verifiable w.r.t.~$\mathfrak{D}_{ \mathcal{U}_n; \geq \vartheta}^{\mathrm{func}}$ by a classical verifier $V$ with access to classical random examples interacting with a quantum prover $P$ with access to quantum superposition examples.
    There is a verifier-prover pair $(V,P)$ such that $P$ uses $\mathcal{O}\left(\tfrac{\log(\nicefrac{1}{\delta\vartheta^2})}{\vartheta^4}\right)$ copies of $\ket{\psi_{(\mathcal{U}_n, f)}}$ , $\mathcal{O}\left(n\tfrac{\log(\nicefrac{1}{\delta\vartheta^2})}{\vartheta^4}\right)$ single-qubit gates, a classical memory of size $\tilde{\mathcal{O}}\left(n\tfrac{\log(\nicefrac{1}{\delta\vartheta^2})}{\vartheta^4}\right)$, and classical running time $\tilde{\mathcal{O}}\left(n\tfrac{\log(\nicefrac{1}{\delta \vartheta^2})}{\vartheta^4}\right)$, and such that $V$ uses $\mathcal{O}\left(\tfrac{k^4\log(\nicefrac{1}{\delta\vartheta^2})}{\varepsilon^4 \vartheta^4}\right)$ classical random examples, $\tilde{\mathcal{O}}\left(n\tfrac{k^4\log(\nicefrac{1}{\delta\vartheta^2})}{\varepsilon^4 \vartheta^4}\right)$ classical running time, and a classical memory of size $\tilde{\mathcal{O}}\left(n\tfrac{k^4\log(\nicefrac{1}{\delta\vartheta^2})}{\varepsilon^4 \vartheta^4}\right)$.
    Moreover, this can be achieved by a pair $(V,P)$ that uses only a single round of communication consisting of at most $\mathcal{O}\left(\tfrac{n}{\vartheta^2}\right)$ classical bits.
\end{theorem}
\begin{proof}
    This result can be proved by modifying the proof of \Cref{theorem:functional-agnostic-quantum-fourier-sparse-verification-qsq-no-small-non-zero-Fourier-coeff} analogously to how we got from \cref{theorem:functional-agnostic-quantum-parity-verification-qsq-no-small-non-zero-Fourier-coeff} to \Cref{theorem:functional-agnostic-quantum-parity-verification-no-small-non-zero-Fourier-coeff}.
\end{proof}

Before moving beyond the functional agnostic setting, we present our second approach towards classical verification of quantum learning.
In contrast to the first approach discussed above, which also worked with an SQ verifier and a QSQ prover, we currently only know how to execute this second strategy if the classical verifier has random example access and the quantum prover has superposition (or, equivalently for the functional case, mixture-of-superpositions) example access.

\begin{proposition}\label{proposition:functional-agnostic-quantum-fourier-sparse-verification-second-approach}
    Let $\delta, \varepsilon,\vartheta\in (0,1)$.
    The class of Fourier-$k$-sparse functions on $n$ bits is efficiently improper $2$-agnostic verifiable w.r.t.~$\mathfrak{D}_{ \mathcal{U}_n; \geq \vartheta}^{\mathrm{func}}$ by a classical verifier $V$ with access to classical random examples interacting with a quantum prover $P$ with access to quantum superposition examples.
    There is a verifier-prover pair $(V,P)$ such that $P$ uses $\mathcal{O}\left(\poly(\nicefrac{1}{\vartheta}, \log(\nicefrac{1}{\delta}))\right)$ copies of $\ket{\psi_{(\mathcal{U}_n, f)}}$ , $\mathcal{O}\left(n\cdot \poly(\nicefrac{1}{\vartheta}, \log(\nicefrac{1}{\delta}))\right)$ single-qubit gates, and classical computation time $\mathcal{O}\left(\poly(n, \nicefrac{1}{\vartheta}, \nicefrac{1}{\varepsilon}, \log(\nicefrac{1}{\delta}))\right)$, and such that $V$ uses $\mathcal{O}\left(\poly(n, k, \nicefrac{1}{\varepsilon}, \log(\nicefrac{1}{\delta}))\right)$ classical random examples and $\mathcal{O}\left(\poly(n, k, \nicefrac{1}{\varepsilon}, \log(\nicefrac{1}{\delta}))\right)$ classical running time.
\end{proposition}
\begin{proof}
    Let $\delta,\varepsilon\in (0,1)$.
    Let $\mathcal{D}\in \mathfrak{D}_{ \mathcal{U}_n; \geq \vartheta}^{\rm func}$.
    By \Cref{lemma:no-small-nonzero-Fourier-coeff-versus-sparsity}, there exists a function $f:\mathcal{X}_n\to\{0,1\}$ such that $g=(-1)^f$ is Fourier $(\tfrac{1}{\vartheta^2})$-sparse and $\mathcal{D} = (\mathcal{U}_n,f)$.
    Therefore, according to the exact learning result of \cite{arunachalam2021twonewresults}, an honest quantum prover can use $\mathcal{O}\left(\poly(\nicefrac{1}{\vartheta}, \log(\nicefrac{1}{\delta}))\right)$ copies of $\ket{\psi_{(\mathcal{U}_n, f)}}$ , $\mathcal{O}\left(n\cdot\poly(\nicefrac{1}{\vartheta}, \log(\nicefrac{1}{\delta}))\right)$ single-qubit gates, and classical computation time $\mathcal{O}\left(\poly(n, \nicefrac{1}{\vartheta}, \nicefrac{1}{\varepsilon}, \log(\nicefrac{1}{\delta}))\right)$
    to obtain, with success probability $\geq 1-\tfrac{\delta}{2}$, a succinct representation of $f$ and thus act as a membership query oracle for $f$.
    Therefore, we can use the interactive Goldreich-Levin algorithm of \cite[Lemma 8]{goldwasser2021interactive}, followed by our by now familiar analysis proving that sufficiently accurate approximations to sufficiently many of the heaviest Fourier coefficients suffice for $2$-agnostic Fourier-sparse learning.
\end{proof}

Some comments about the two approaches are in order.
First, we clarify that, while \cite{goldwasser2021interactive} proved interactive verification of $1$-agnostic Fourier-sparse learning in the classical case, a closer inspection of \cite{goldwasser2020interactive-full-version} shows that this is for the case in which performance is measured according to the $L_2$-error. (With this performance measure, also our first approach achieves $1$-agnostic Fourier-sparse verification.)
If we instead focus the misclassification probability, then \cite{goldwasser2021interactive} does not give $1$-agnostic interactive verification guarantees. However, the above proof implies that their approach gives rise to a $2$-agnostic Fourier-sparse verification of learning scheme in this setting. Moreover, based on our previous analyses, \cite{goldwasser2021interactive} also gives a $1$-agnostic interactive verification scheme for parities w.r.t.~the misclassification probability.

Second, we highlight an advantage of the second classical verification approach over the first one. Namely, in this second approach, only the prover needs to know the promised value $\vartheta$, whereas in the first approach both prover and verifier make use of the knowledge of $\vartheta$.
The verifier does not need the promise $\mathcal{D}\in \mathfrak{D}_{ \mathcal{U}_n; \geq \vartheta}$ in the second approach because the interactive GL procedure of \cite{goldwasser2021interactive} is designed to work without any assumptions on the unknown labeling function. Thus, it is conceivable that this second approach can be further improved by modifying the interactive GL procedure to take $\mathcal{D}\in \mathfrak{D}_{ \mathcal{U}_n; \geq \vartheta}$ into account.
However, such a modification would likely still lead to a verifier with $n$-dependent classical sample complexity, whereas the sample complexity for the verifier in the first approach is $n$-independent.

\subsection{Verifying Noisy Functional Agnostic Quantum Learning} 

In \Cref{subsection:noisy-functional-agnostic} we have seen that quantum data, in contrast to its classical counterpart, can be powerful even for learning from noisy data. Therefore, we next consider extensions of the verification results from the previous subsection to the noisy setting. For this, we use the following noisy variant of \Cref{definition:functional-no-small-fourier-coefficients}, which considers noisy versions of elements of $\mathfrak{D}_{ \mathcal{U}_n; \geq \varepsilon}^{\mathrm{func}}$:

\begin{definition}[Noisy functional distributions with no small non-zero Fourier coefficients]\label{{definition:noisy-functional-no-small-fourier-coefficients}}
    Let $\vartheta\in (0,1)$. Let $0\leq \eta <\nicefrac{1}{2}$.
    We denote the class of probability distributions $\mathcal{D} = (\mathcal{U}_n, \varphi)$ over $\mathcal{X}_n\times\{0,1\}$ that have a uniform marginal over $\mathcal{X}_n$ and whose $\{-1,1\}$-label is given as $\eta$-noisy version of deterministic function $g=(-1)^f$ that has no non-zero Fourier coefficients of magnitude $<\vartheta$ by 
    \begin{equation}
        \mathfrak{D}_{ \mathcal{U}_n; \geq \vartheta}^{\mathrm{func},\eta}
        \coloneqq \left\{ (\mathcal{U}_n, \varphi)~|~ \exists f:\mathcal{X}_n\to\{0,1\} \textrm{ s.t. } (\varphi = \eta + (1-2\eta)f)  \wedge ((\mathcal{U}_n,f)\in \mathfrak{D}_{ \mathcal{U}_n; \geq \vartheta}^{\mathrm{func}}) \right\} .
    \end{equation}
\end{definition}

If both the verifier $V$ and the honest quantum prover $P$ know the noise rate $\eta$ in advance, then our approach towards verified learning via checking accumulated Fourier weight can be applied. In fact, this works quite similarly to what we describe in \Cref{subsection:verification-functional}, with some small modifications, depending on which version of noisy quantum example from \Cref{definition:noisy-functional-quantum-examples} is used and relying on the corresponding noisy version of \Cref{corollary:quantum-approximation-fourier-spectrum} (compare the discussion in \Cref{subsection:noisy-functional-agnostic}) .
Doing so in particular relies on the fact that in the noise case $\phi = (1-2\eta) g$, so, knowing $\eta$, $V$ can estimate $\hat{g}$ via $\hat{\phi}$ from classical noisy SQs or noisy random examples.

In this noisy setting, we immediately obtain computational hardness of the learning task for the classical verifier without the help of the quantum prover conditioned on classical hardness of LPN.
Importantly, in contrast to the results in \Cref{subsection:verification-functional}, this gives us a separation between a lone classical learner and an interacting pair of classical verifier and quantum prover even if the classical resource are (noisy) random examples.
This already is the random example part of \Cref{theorem:main-result-verification} (a).

\subsection{Verifying Distributional Agnostic Quantum Learning} 

Finally, we discuss classical verification of quantum learning for the distributional agnostic case. 
The results in this subsection serve to fully establish \Cref{theorem:main-result-verification} when focusing on the regime $\varepsilon, \vartheta\geq\Omega(\nicefrac{1}{\poly (n)})$ and $\delta\geq\Omega(\nicefrac{
1}{\exp (n)})$.
Again, we require a suitable variant of \Cref{definition:functional-no-small-fourier-coefficients}.

\begin{definition}[Distributions with no small non-zero Fourier coefficients]\label{definition:distributional-no-small-fourier-coefficients}
    Let $\vartheta\in (0,1)$.
    We denote the class of probability distributions $\mathcal{D} = (\mathcal{U}_n, \varphi)$ over $\mathcal{X}_n\times\{0,1\}$ that have a uniform marginal over $\mathcal{X}_n$ and whose $\{-1,1\}$-label expectation $\phi$ has no non-zero Fourier coefficients of magnitude $<\vartheta$ by 
    \begin{equation}
        \mathfrak{D}_{ \mathcal{U}_n; \geq \vartheta}
        \coloneqq \left\{ (\mathcal{U}_n, \varphi)~|~ \hat{\phi}\neq 0~\Rightarrow \lvert \hat{\phi}\rvert \geq \vartheta \right\} .
    \end{equation}
\end{definition}

Directly from the definitions we get the inclusions $\mathfrak{D}_{ \mathcal{U}_n; \geq \varepsilon}\supset \mathfrak{D}_{ \mathcal{U}_n; \geq \varepsilon}^{\mathrm{func}}$ as well as $\mathfrak{D}_{ \mathcal{U}_n; \geq \varepsilon}\supset \mathfrak{D}_{ \mathcal{U}_n; \geq \nicefrac{\varepsilon}{(1-2\eta)}}^{\mathrm{func},\eta}$.
Noisy parity distributions are even contained in $\mathfrak{D}_{ \mathcal{U}_n; \geq (1-2\eta)}$.
So, when considering learning problems under the promise $\mathcal{D}\in \mathfrak{D}_{ \mathcal{U}_n; \geq \varepsilon}$, we in particular still allow scenarios in which the unknown distribution is a (noisy) parity or Fourier-sparse function.
However, for the distributional agnostic setting considered in this subsection, we rely on the following additional assumption.

\begin{definition}[Distributions with $L_2$-bounded bias]\label{definition:l2-bounded-bias}
    Let $0\leq a \leq b\leq 1$.
    We denote the class of probability distributions $\mathcal{D} = (\mathcal{U}_n, \varphi)$ over $\mathcal{X}_n\times\{0,1\}$ that have a uniform marginal over $\mathcal{X}_n$ and whose $\{-1,1\}$-label expectation $\phi$ has squared $L_2$ norm in $[a^2,b^2]$ by
    \begin{equation}
        \mathfrak{D}_{ \mathcal{U}_n; [a^2,b^2]}
        \coloneqq\left\{(\mathcal{U}_n, \varphi)~|~ \mathbb{E}_{x\sim\mathcal{U}_n}[(\phi(x))^2]\in [a^2,b^2]\right\} .
    \end{equation}
\end{definition}

The motivation behind \Cref{definition:l2-bounded-bias} is as follows: In our verification protocol, the verifier checks whether the prover has provided a list with sufficient accumulated Fourier weight. A promise as in \Cref{definition:l2-bounded-bias} ensures that the verifier knows what ``sufficient'' means. Without such a promise, the total Fourier weight in the distributional agnostic case may take any value between $0$ and $1$.
Fortunately, even with an added promise of this form, we still generalize beyond the noiseless and noisy functional agnostic cases. Namely, the noiseless functional case comes with the strong promise of $\mathcal{D}\in \mathfrak{D}_{ \mathcal{U}_n; [a^2, b^2]}$ for $a=b=1$, and for the noisy functional case we can take $a=b=(1-2\eta)$. 

With the relevant definitions established, we now state the distributional agnostic versions of classical verification for quantum parity learning.

\begin{theorem}\label{theorem:distributional-agnostic-qsq-parity-verification-no-small-non-zero-Fourier-coeff}
    Let $\vartheta\in (0,1)$.
    Let $0\leq a \leq b\leq 1$.
    Let $\varepsilon\geq 2\sqrt{b^2-a^2}$.
    The class of $n$-bit parities is efficiently proper $1$-agnostic verifiable w.r.t.~$\mathfrak{D}_{ \mathcal{U}_n; \geq \vartheta} \cap \mathfrak{D}_{ \mathcal{U}_n; [a^2,b^2]}$ by a classical verifier $V$ with access to classical SQs interacting with a quantum prover $P$ with distributional QSQ access.
    There is a verifier-prover pair $(V,P)$ such that $P$ uses $\mathcal{O}\left(\tfrac{n}{\vartheta^2}\right)$ distributional QSQs of tolerance at worst $\nicefrac{\vartheta}{8}$ for observables that can be implemented with $\mathcal{O}(n)$ single-qubit gates, a classical memory of size $\tilde{\mathcal{O}}\left(\tfrac{n^2}{\vartheta^2}\right)$, and classical running time $\tilde{\mathcal{O}}\left(\tfrac{n}{\vartheta^2}\right)$, and such that $V$ uses $\mathcal{O}\left(\tfrac{b^2}{\vartheta^2}\right)$ classical SQs of tolerance at worst $\nicefrac{(\varepsilon^2 \vartheta^2)}{16}$, $\tilde{\mathcal{O}}\left(\tfrac{n b^2}{\vartheta^2}\right)$ classical running time, and a classical memory of size $\tilde{\mathcal{O}}\left(\tfrac{n b^2}{\vartheta^2}\right)$.
    Moreover, this can be achieved by a pair $(V,P)$ that uses only a single round of communication consisting of at most $\mathcal{O}\left(\tfrac{n }{\vartheta^2}\right)$ classical bits.
\end{theorem}
\begin{proof}
    This proof differs from that of \Cref{theorem:functional-agnostic-quantum-parity-verification-qsq-no-small-non-zero-Fourier-coeff} only by a few technical adjustments to the procedure and its analysis, and by replacing \Cref{theorem:functional-agnostic-qsq-GL} with \Cref{theorem:distributional-agnostic-qsq-GL}.
    We present the relevant technical adjustments in detail in the proof of \Cref{theorem:distributional-agnostic-quantum-parity-verification-no-small-non-zero-Fourier-coeff}.
\end{proof}

\begin{theorem}\label{theorem:distributional-agnostic-quantum-parity-verification-no-small-non-zero-Fourier-coeff}
    Let $\vartheta\in (2^{-(\tfrac{n}{2} - 3)},1)$. Let $0\leq a \leq b\leq 1$.
    Let $\delta\in (0,1)$ and $\varepsilon\geq 2\sqrt{b^2-a^2}$.
    The class of $n$-bit parities is efficiently proper $1$-agnostic verifiable w.r.t.~$\mathfrak{D}_{ \mathcal{U}_n; \geq \vartheta} \cap \mathfrak{D}_{ \mathcal{U}_n; [a^2,b^2]}$ by a classical verifier $V$ with access to classical random examples interacting with a quantum prover $P$ with access to mixture-of-superpositions quantum examples.
    There is a verifier-prover pair $(V,P)$ such that $P$ uses $\mathcal{O}\left(\tfrac{\log(\nicefrac{1}{\delta\vartheta^2})}{\vartheta^4}\right)$ copies of $\rho_{\mathcal{D}}$, $\mathcal{O}\left(n\tfrac{\log(\nicefrac{1}{\delta\vartheta^2})}{\vartheta^4}\right)$ single-qubit gates, a classical memory of size $\tilde{\mathcal{O}}\left(n\tfrac{\log(\nicefrac{1}{\delta\vartheta^2})}{\vartheta^4}\right)$, and classical running time $\tilde{\mathcal{O}}\left(n\tfrac{\log(\nicefrac{1}{\delta \vartheta^2})}{\vartheta^4}\right)$, and such that $V$ uses $\mathcal{O}\left(\tfrac{b^4\log(\nicefrac{1}{\delta\vartheta^2})}{\varepsilon^4 \vartheta^4}\right)$ classical random examples, $\tilde{\mathcal{O}}\left(n\tfrac{b^4\log(\nicefrac{1}{\delta\vartheta^2})}{\varepsilon^4 \vartheta^4}\right)$ classical running time, and a classical memory of size $\tilde{\mathcal{O}}\left(n\tfrac{b^4\log(\nicefrac{1}{\delta\vartheta^2})}{\varepsilon^4 \vartheta^4}\right)$.
    Moreover, this can be achieved by a pair $(V,P)$ that uses only a single round of communication consisting of at most $\mathcal{O}\left(\tfrac{n}{\vartheta^2}\right)$ classical bits.
\end{theorem}
\begin{proof}
    This proof differs from that of \Cref{theorem:functional-agnostic-quantum-parity-verification-no-small-non-zero-Fourier-coeff} only by a few technical adjustments to the procedure and its analysis, and by replacing \Cref{corollary:quantum-approximation-fourier-spectrum} with \Cref{corollary:distributional-agnostic-quantum-approximation-fourier-spectrum}.
    Let $\delta,\varepsilon\in (0,1)$.
    Let $0\leq a \leq b\leq 1$.
    Let $\mathcal{D}\in \mathfrak{D}_{ \mathcal{U}_n; \geq \vartheta} \cap \mathfrak{D}_{ \mathcal{U}_n; [a^2,b^2]}$.
    Assume that $\varepsilon\geq 2\sqrt{b^2-a^2}$, with $\vartheta\in (2^{-(\tfrac{n}{2} - 3)},1)$.
    We begin the proof by describing the actions of the classical verifier $V$ and the honest quantum prover $P$:
    \begin{enumerate}
        \item $V$ asks $P$ to provide a list $L=\{s_1,\ldots ,s_{\lvert L\rvert}\}\subset \{0,1\}^n$ of length $\lvert L \rvert\leq \nicefrac{64 b^2}{\vartheta^2}$ consisting of pairwise distinct $n$-bit strings whose associated Fourier coefficients are non-zero.
        \item $P$ follows the procedure in \Cref{corollary:distributional-agnostic-quantum-approximation-fourier-spectrum} to produce, with success probability $\geq 1-\tfrac{\delta}{2}$, a succinctly represented $\Tilde{\phi}:\mathcal{X}_n\to [-1,1]$ such that $\norm{\Tilde{\phi}-\hat{\phi}}_\infty\leq\nicefrac{\vartheta}{2}$ and $\norm{\Tilde{\phi}}_0\leq\tfrac{64 b^2}{\vartheta^2}$. If $P$ obtains an output that violates the $\norm{\cdot}_0$-bound, then $P$ declares failure and the interaction aborts.
        Otherwise, $P$ then sends the list $L = \{s\in\{0,1\}^n~|~ \lvert \Tilde{\phi}(s)\rvert\geq \nicefrac{\vartheta}{2} \}$ to $V$.
        \item If $V$ receives a list $L$ of length $\lvert L\rvert> \nicefrac{64 b^2}{\vartheta^2} $, $V$ rejects the interaction. Otherwise, $V$ uses $\mathcal{O}\left(\tfrac{\lvert L\rvert^2\log(\nicefrac{\lvert L\rvert}{\delta})}{\varepsilon^4}\right)$ classical random examples from $\mathcal{D}$ to obtain simultaneously $(\nicefrac{\varepsilon^2}{16 \lvert L\rvert})$-accurate estimates $\hat{\xi}(s)$ of $\hat{\phi}(s)$ for all $s\in L$, with success probability $\geq 1-\tfrac{\delta}{2}$, via Chernoff-Hoeffding combined with a union bound over $L$. (For $t\not\in L$, the verifier's estimate $\hat{\gamma}(t)$ for $\hat{g}(t)$ is just $0$.)
        \item If $\sum_{\ell = 1}^{\lvert L\rvert} \left( \hat{\xi}(s_\ell) \right)^2 \geq a^2 - \tfrac{\varepsilon^2}{8}$, then $V$ determines $s_{\rm out}\in\operatorname{argmax}_{1\leq \ell\leq \lvert L \rvert} \hat{\xi}(s)$ and outputs the hypothesis $h:\mathcal{X}_n\to \{0,1\}$, $h(x)=s_{\rm out}\cdot x$. If $\sum_{\ell = 1}^{\lvert L\rvert} \left( \hat{\xi}(s_\ell) \right)^2 < a^2 - \tfrac{\varepsilon^2}{8}$, then $V$ outputs $\mathrm{reject}$.
    \end{enumerate}
    We now show that the pair $(V,P)$ has the desired completeness and soundness properties. As a first step towards this goal, we show that $V$ accepts an interaction with $P$ with high probability. To this end, observe that, conditioned on $P$ succeeding in Step 2, $V$ never rejects in Step 3. If we then further condition on $V$ succeeding in Step 3, we have
    \begin{align}
        \sum_{\ell = 1}^{\lvert L\rvert} \left( \hat{\xi}(s_\ell) \right)^2
        &\geq \sum_{\ell = 1}^{\lvert L\rvert} \left( \hat{\phi}(s_\ell) \right)^2 - 2\sum_{\ell = 1}^{\lvert L\rvert} \lvert\hat{\phi}(s_\ell)-\hat{\xi}(s_\ell)\rvert\\
        &\geq \sum_{s:\hat{\phi}(s)\neq 0} \lvert\hat{\phi}(s)\rvert^2 - 2\sum_{\ell = 1}^{\lvert L\rvert} \lvert\hat{\phi}(s_\ell)-\hat{\xi}(s_\ell)\rvert\\
        &\geq a^2 - 2\lvert L\rvert \cdot \frac{\varepsilon^2}{16 \lvert L\rvert}\\
        &= a^2 - \frac{\varepsilon^2}{8}\, ,
    \end{align}
    where the first step used that the function $[-1,1]\ni \xi\to \xi^2$ is $2$-Lipschitz, the second step used that $\mathcal{D}\in \mathfrak{D}_{ \mathcal{U}_n; \geq \vartheta}$ implies $(\hat{\phi}(s)\neq 0~ \Rightarrow ~ s\in L)$ if Step 2 succeeds, and the third step used the approximation guarantee in Step 3 as well as Parseval together with $\mathcal{D}\in \mathfrak{D}_{ \mathcal{U}_n; [a^2,b^2]}$.
    Thus, if both Step 2 and Step 3 succeed, which by a union bound happens with probability $\geq 1-\delta$, then $V$ accepts in Step 4. 
    
    Moreover, whenever Step 3 is successful and $V$ does not reject in Step 4, then the output string $s_{\rm out}\in\operatorname{argmax}_{1\leq \ell\leq \lvert L \rvert} \gamma(s)$ of $V$ is as desired. This can be seen as follows: If $V$ does not reject in Step 4 and if Step 3 was successful, then this implies that for any $s\not\in L$,
    \begingroup
    \allowdisplaybreaks
    \begin{align}
        \left(\hat{\phi}(s)\right)^2
        &\leq \sum_{t\not\in L} \left(\hat{\phi}(t)\right)^2\\
        &= \mathbb{E}_{x\sim\mathcal{U}_n}[(\phi(x))^2] -  \sum_{t\in L} \left(\hat{\phi}(t)\right)^2\\
        &\leq b^2 - \sum_{t\in L} \left(\hat{\xi}(t)\right)^2 + 2 \sum_{\ell = 1}^{\lvert L\rvert} \lvert\hat{\phi}(s_\ell)-\hat{\xi}(s_\ell)\rvert\\
        &\leq (b^2-a^2) + \frac{\varepsilon^2}{8} + 2\lvert L\rvert \cdot \frac{\varepsilon^2}{16 \lvert L\rvert}\\
        &= (b^2-a^2) + \frac{\varepsilon^2}{4},
    \end{align}
    \endgroup
    where we again used that $[-1,1]\ni \xi\to \xi^2$ is $2$-Lipschitz.
    This tells us that $\lvert \hat{\phi}(s)\rvert \leq \sqrt{(b^2-a^2) + \frac{\varepsilon^2}{4}}\leq \sqrt{b^2-a^2} + \nicefrac{\varepsilon}{2} \leq \varepsilon$ holds for every $s\not\in L$, which now allows us to show that the output $s_{\rm out}\in\operatorname{argmax}_{1\leq \ell\leq \lvert L \rvert} \gamma(s)$ of $V$ has the desired property with a reasoning analogous to that in the proof of \Cref{theorem:functional-agnostic-quantum-parity-verification-qsq-no-small-non-zero-Fourier-coeff}. Again, this last part of our reasoning only relies on $V$ succeeding in Step 3 and accepting in Step 4, but is independent of the action of the quantum prover. Therefore, we have also established the desired soundness.

    It remains to bound the sample and time complexities for $V$ and $P$.
    \Cref{corollary:distributional-agnostic-quantum-approximation-fourier-spectrum} directly tells us that $P$ uses $\mathcal{O}\left(\tfrac{\log(\nicefrac{1}{\delta\varepsilon^2})}{\vartheta^4}\right)$ copies mixture-of-superpositions state $\rho_{\mathcal{D}}$ as well as $\mathcal{O}\left(n\tfrac{\log(\nicefrac{1}{\delta\vartheta^2})}{\vartheta^4}\right)$ single-qubit gates, a classical memory of size $\tilde{\mathcal{O}}\left(n\tfrac{\log(\nicefrac{1}{\delta\vartheta^2})}{\vartheta^4}\right)$, and classical running time $\tilde{\mathcal{O}}\left(n\tfrac{\log(\nicefrac{1}{\delta \vartheta^2})}{\vartheta^4}\right)$.
    The classical sample complexity of $V$ is $\mathcal{O}\left(\tfrac{\lvert L\rvert^2\log(\nicefrac{\lvert L\rvert}{\delta})}{\varepsilon^4}\right)\leq \mathcal{O}\left(\tfrac{b^4\log(\nicefrac{1}{\delta\vartheta^2})}{\varepsilon^4 \vartheta^4}\right)$, as noted in Step 3. The classical running time of $V$ is $\tilde{\mathcal{O}}\left(n\tfrac{\lvert L\rvert^2\log(\nicefrac{\lvert L\rvert}{\delta})}{\varepsilon^4}\right)\leq \tilde{\mathcal{O}}\left(n\tfrac{b^4\log(\nicefrac{1}{\delta\vartheta^2})}{\varepsilon^4 \vartheta^4}\right)$. 
    The classical memory used by $V$ is of size $\mathcal{O}\left(n\tfrac{\lvert L\rvert^2\log(\nicefrac{\lvert L\rvert}{\delta})}{\varepsilon^4}\right)\leq \mathcal{O}\left(n\tfrac{b^4\log(\nicefrac{1}{\delta\vartheta^2})}{\varepsilon^4 \Tilde{\varepsilon}^4}\right)$.
    Finally, our protocol clearly requires only a single round of communication, and the communicated object is $L$, which consists of at most $\nicefrac{64 b^4}{\vartheta^2}$ strings of $n$ bits.
    This finishes the proof.
\end{proof}

In \Cref{theorem:distributional-agnostic-quantum-parity-verification-no-small-non-zero-Fourier-coeff}, the achievable accuracy is limited by $2\sqrt{b^2-a^2}$. Next, we show that such a limitation is necessary for interactive classical-quantum verification of learning with a sublinear-in-$n$ sample complexity for the classical verifier:

\begin{theorem}\label{theorem:limitation-improvement-distributional-agnostic-verification}
    Let $\eta\in [0,\nicefrac{1}{6})$. Define $a=0$ and $b=\vartheta=1-2\eta$.
    Let $\delta = \nicefrac{1}{3}$ and $\varepsilon = \nicefrac{(1-2\eta)}{3} = \tfrac{1}{3}\cdot \sqrt{b^2-a^2}$.
    Proper $1$-PAC verification for the class of $n$-bit parities w.r.t.~$\mathfrak{D}_{ \mathcal{U}_n; \geq \vartheta} \cap \mathfrak{D}_{ \mathcal{U}_n; [a^2,b^2]}$ by a classical verifier $V$ with access to classical random examples interacting with a quantum prover $P$ with access to mixture-of-superpositions quantum examples requires the verifier to use at least $\Omega (n)$ classical examples.
\end{theorem}

Here, we consider $\eta$ to be a constant and focus on the scaling with $n$.
\Cref{theorem:limitation-improvement-distributional-agnostic-verification} tells us that the accuracy lower bound $\varepsilon\geq 2\sqrt{b^2-a^2}$ in \Cref{theorem:distributional-agnostic-quantum-parity-verification-no-small-non-zero-Fourier-coeff} cannot be significantly improved without at the same time worsening the number of examples used by the classical verifier from $n$-independent to linear-in-$n$. 

\begin{proof}
    We adapt the proof strategy of \cite[Theorem 8]{mutreja2022pac-verification} to our setting. That is, we use the assumed pair $(V,P)$ of a classical verifier and a quantum prover to construct a testing algorithm $T$ that can distinguish between $\mathcal{D} = \mathcal{U}_{n+1}$ and $\mathcal{D}\in \{(\mathcal{U}_n, (1-2\eta)\chi_s)\}_{s\in\{0,1\}^n}$ using $m_T=m_V + \mathcal{O}(1)$ classical random examples of the unknown distribution. Then we appeal to the lower bound of \Cref{lemma:fully-uniform-vs-random-noisy-parity}.
    
    In more detail, we construct the tester $T$ as follows:
    \begin{enumerate}
        \item Let $m_V$ and $m_P$ be the classical and quantum sample complexities of $V$ and $P$, respectively. $T$ draws a sample $S_V\sim\mathcal{D}^{\otimes m_V}$ from the unknown distribution $\mathcal{D}$ and prepares $m_P$ (classical descriptions of) copies of the mixture-of-superpositions quantum example state $\rho_{\mathcal{U}_{n+1}}$. Note that the latter is possible because $\mathcal{U}_{n+1}$ is known to $T$.
        \item $T$ simulates $(V,P)$, where $V$ gets the classical data $S_V$ and $P$ gets the quantum data $\rho_{\mathcal{U}_{n+1}}^{\otimes m}$, to obtain the output $h\in\{\mathrm{reject}\}\cup \{\chi_s\}_{s\in\{0,1\}^n}$. Note that the simulation of the quantum prover $P$ via the classical tester $T$ may be computationally inefficient, but this is irrelevant for our purposes since we focus on sample complexity.
        \item $T$ draws a sample $S_{\mathrm{test}}\sim\mathcal{D}^{\otimes m_{\mathrm{test}}}$ from the unknown distribution $\mathcal{D}$, where $m_{\mathrm{test}} = \tfrac{1}{2}\cdot (\tfrac{18}{1+\eta})^2 \log(24)$.
        \item If $h=\mathrm{reject}$ or $(h\neq \mathrm{reject}~\wedge~\tfrac{1}{m_{\mathrm{test}}}\lvert\{ (x,y)\in S_{\mathrm{test}} ~|~ h(x)\neq y \}\rvert) \leq \tfrac{7(1+\eta)}{18}$, then $T$ outputs ``$\mathcal{D}\in \{(\mathcal{U}_n, (1-2\eta)\chi_s)\}_{s\in\{0,1\}^n}$.'' Otherwise, $T$ outputs ``$\mathcal{D} = \mathcal{U}_{n+1}$.''
    \end{enumerate}
    We claim that $T$ succeeds at the distinguishing task with probability $\geq \nicefrac{7}{12}$. This can be seen as follows:
    \begin{itemize}
        \item If $\mathcal{D} = (\mathcal{U}_n, (1-2\eta)\chi_t)$ for some $t\in\{0,1\}^n$, then $\min_{s\in\{0,1\}^n} \mathbb{P}_{(x,y)\sim\mathcal{D}}[y\neq \chi_s(x)]=\eta$. Hence, by the soundness of the interactive verifier-prover pair $(V,P)$, we see that $h=\mathrm{reject}$ or $(h\neq \mathrm{reject}~\wedge~\mathbb{P}_{(x,y)\sim\mathcal{D}}[y\neq h(x)] \leq \eta + \tfrac{1-2\eta}{3} = \tfrac{1+\eta}{3}$) holds with probability $\geq \nicefrac{2}{3}$.
        Conditioned on this event occurring, Hoeffding's inequality and our choice of $m_{\mathrm{test}}$ guarantee that $\tfrac{1}{m_{\mathrm{test}}}\lvert\{ (x,y)\in S_{\mathrm{test}} ~|~ h(x)\neq y \}\rvert) \leq \tfrac{1+\eta}{3} + \tfrac{1+\eta}{18} = \tfrac{7(1+\eta)}{18}$ holds with probability $\geq \nicefrac{11}{12}$. 
        By a union bound, we can therefore conclude: If $\mathcal{D} = (\mathcal{U}_n, (1-2\eta)\chi_t)$ for some $t\in\{0,1\}^n$, then $T$ outputs ``$\mathcal{D}\in \{(\mathcal{U}_n, (1-2\eta)\chi_s)\}_{s\in\{0,1\}^n}$'' with probability $\geq 1 - \nicefrac{1}{3} - \nicefrac{1}{12} = \nicefrac{7}{12}$.
        \item If $\mathcal{D} = \mathcal{U}_{n+1}$, then $\mathbb{P}_{(x,y)\sim\mathcal{D}}[y\neq \chi_s(x)]=\nicefrac{1}{2}$ holds for every $s\in\{0,1\}^n$. In particular, if $h\neq \mathrm{reject}$, then Hoeffding's inequality and our choice of $m_{\mathrm{test}}$ guarantee that $\tfrac{1}{m_{\mathrm{test}}}\lvert\{ (x,y)\in S_{\mathrm{test}} ~|~ h(x)\neq y \}\rvert)\geq \tfrac{1}{2} - \tfrac{1+\eta}{18} = \tfrac{7(1+\eta)}{18} + \tfrac{1-6\eta}{18}>\tfrac{7(1+\eta)}{18}$ holds with probability $\geq \nicefrac{11}{12}$.
        By completeness of the interactive verifier-prover pair $(V,P)$, we see that $h\neq \mathrm{reject}$ with probability $\geq \nicefrac{2}{3}$. By a union bound, we can therefore conclude: 
        If $\mathcal{D} = \mathcal{U}_{n+1}$, then $T$ outputs ``$\mathcal{D} = \mathcal{U}_{n+1}$'' with probability $\geq 1 - \nicefrac{1}{3} - \nicefrac{1}{12} = \nicefrac{7}{12}$.
    \end{itemize}
    As $T$ uses $m_V + m_{\mathrm{test}} = m_V + \mathcal{O}(1)$ random examples from the unknown distribution $\mathcal{D}$ to solve the distinguishing task, a comparison to the $\Omega (n)$ sample complexity lower bound of \Cref{lemma:fully-uniform-vs-random-noisy-parity} implies that $m_V\geq \Omega(n)$.
\end{proof}

\begin{remark}\label{remark:mutreja-shafer-computational}
    While \Cref{theorem:limitation-improvement-distributional-agnostic-verification} focuses on sample complexities, the proof has immediate computational complexity implications. 
    To see this, notice that (by \Cref{theorem:agnostic-quantum-fourier-sampling}) it is trivial to classically simulate distributional agnostic quantum Fourier sampling if $\mathcal{D}=\mathcal{U}_{n+1}$ and thus $\phi \equiv 0$. Namely, we first toss a fair coin to decide whether the sampling attempt succeeds or fails, and in the case of success we then sample a uniformly random $n$-bit string $s$.
    Thus, a classical $T$ can efficiently simulate the actions of a quantum $P$ with access to copies of $\rho_{\mathcal{U}_{n+1}}$.
    Consequently, with the same parameter choices as in \Cref{theorem:limitation-improvement-distributional-agnostic-verification}, a computationally classically efficient $V$ would lead to a computationally efficient classical tester $T$ able to distinguish between the uniform distribution and random noisy parities. 
    Therefore, assuming that this decision version of LPN is hard, we cannot meaningfully improve the the accuracy lower bound $\varepsilon\geq 2\sqrt{b^2-a^2}$ in \Cref{theorem:distributional-agnostic-quantum-parity-verification-no-small-non-zero-Fourier-coeff} without losing computational efficiency of $V$.
\end{remark}

As in \Cref{subsection:verification-functional}, we also explicitly state our results for classical verification of quantum learning Fourier-sparse functions in the distributional agnostic case.

\begin{theorem}\label{theorem:distributional-agnostic-qsq-fourier-sparse-verification-no-small-non-zero-Fourier-coeff}
    Let $\vartheta\in (0,1)$.
    Let $0\leq a \leq b\leq 1$.
    Let $\varepsilon\geq 4k\sqrt{b^2-a^2}$.
    The class of Fourier-$k$-sparse functions of $n$ bits is efficiently improper $2$-agnostic verifiable w.r.t.~$\mathfrak{D}_{ \mathcal{U}_n; \geq \vartheta} \cap \mathfrak{D}_{ \mathcal{U}_n; [a^2,b^2]}$ by a classical verifier $V$ with access to classical SQs interacting with a quantum prover $P$ with distributional QSQ access.
    There is a verifier-prover pair $(V,P)$ such that $P$ uses $\mathcal{O}\left(\tfrac{n}{\vartheta^2}\right)$ distributional QSQs of tolerance at worst $\nicefrac{\vartheta}{8}$ for observables that can be implemented with $\mathcal{O}(n)$ single-qubit gates, a classical memory of size $\tilde{\mathcal{O}}\left(\tfrac{n^2}{\vartheta^2}\right)$, and classical running time $\tilde{\mathcal{O}}\left(\tfrac{n}{\vartheta^2}\right)$, and such that $V$ uses $\mathcal{O}\left(\tfrac{b^2}{\vartheta^2}\right)$ classical SQs of tolerance at worst $\nicefrac{(\varepsilon^2 \vartheta^2)}{256 k^2}$, $\tilde{\mathcal{O}}\left(\tfrac{n b^2}{\vartheta^2}\right)$ classical running time, and a classical memory of size $\tilde{\mathcal{O}}\left(\tfrac{n b^2}{\vartheta^2}\right)$.
    Moreover, this can be achieved by a pair $(V,P)$ that uses only a single round of communication consisting of at most $\mathcal{O}\left(\tfrac{n}{\vartheta^2}\right)$ classical bits.
\end{theorem}
\begin{proof}
    This proof differs from that of \Cref{theorem:functional-agnostic-quantum-fourier-sparse-verification-qsq-no-small-non-zero-Fourier-coeff} only by a few technical adjustments to the procedure and its analysis, and by replacing \Cref{theorem:functional-agnostic-qsq-GL} with \Cref{theorem:distributional-agnostic-qsq-GL}.
    We present the relevant technical adjustments in detail in the proof of the next theorem.
\end{proof}

\begin{theorem}\label{theorem:distributional-agnostic-quantum-fourier-sparse-verification-no-small-non-zero-Fourier-coeff}
    Let $\vartheta\in (2^{-(\tfrac{n}{2} - 3)},1)$.
    Let $0\leq a \leq b\leq 1$.
    Let $\delta\in (0,1)$ and $\varepsilon\geq 4k\sqrt{b^2-a^2}$.
    The class of Fourier-$k$-sparse functions on $n$ bits is efficiently improper $2$-agnostic verifiable w.r.t.~$\mathfrak{D}_{ \mathcal{U}_n; \geq \vartheta} \cap \mathfrak{D}_{ \mathcal{U}_n; [a^2,b^2]}$ by a classical verifier $V$ with access to classical random examples interacting with a quantum prover $P$ with access to mixture-of-superpositions examples.
    There is a verifier-prover pair $(V,P)$ such that $P$ uses $\mathcal{O}\left(\tfrac{\log(\nicefrac{1}{\delta\vartheta^2})}{\vartheta^4}\right)$ copies of $\rho_{\mathcal{D}}$ , $\mathcal{O}\left(n\tfrac{\log(\nicefrac{1}{\delta\vartheta^2})}{\vartheta^4}\right)$ single-qubit gates, a classical memory of size $\tilde{\mathcal{O}}\left(n\tfrac{\log(\nicefrac{1}{\delta\vartheta^2})}{\vartheta^4}\right)$, and classical running time $\tilde{\mathcal{O}}\left(n\tfrac{\log(\nicefrac{1}{\delta \vartheta^2})}{\vartheta^4}\right)$, and such that $V$ uses $\mathcal{O}\left(\tfrac{b^4 k^4\log(\nicefrac{1}{\delta\vartheta^2})}{\varepsilon^4 \vartheta^4}\right)$ classical random examples, $\tilde{\mathcal{O}}\left(n\tfrac{b^4 k^4\log(\nicefrac{1}{\delta\vartheta^2})}{\varepsilon^4 \vartheta^4}\right)$ classical running time, and a classical memory of size $\tilde{\mathcal{O}}\left(n\tfrac{b^4 k^4\log(\nicefrac{1}{\delta\vartheta^2})}{\varepsilon^4 \vartheta^4}\right)$.
    Moreover, this can be achieved by a pair $(V,P)$ that uses only a single round of communication consisting of at most $\mathcal{O}\left(\tfrac{n}{\vartheta^2}\right)$ classical bits.
\end{theorem}
\begin{proof}
    This proof differs from that of \Cref{theorem:functional-agnostic-quantum-fourier-sparse-verification-no-small-non-zero-Fourier-coeff} only by a few technical adjustments to the procedure and its analysis, and by replacing \Cref{corollary:quantum-approximation-fourier-spectrum} with \Cref{corollary:distributional-agnostic-quantum-approximation-fourier-spectrum}.
    Let $\delta,\varepsilon\in (0,1)$.
    Let $0\leq a \leq b\leq 1$.
    Let $\mathcal{D}\in \mathfrak{D}_{ \mathcal{U}_n; \geq \vartheta} \cap \mathfrak{D}_{ \mathcal{U}_n; [a^2,b^2]}$.
    Assume that $\varepsilon\geq 4k\sqrt{b^2-a^2}$, with $\vartheta\in (2^{-(\tfrac{n}{2} - 3)},1)$.
    We begin the proof by describing the actions of the classical verifier $V$ and the honest quantum prover $P$:
    \begin{enumerate}
        \item $V$ asks $P$ to provide a list $L=\{s_1,\ldots ,s_{\lvert L\rvert}\}\subset \{0,1\}^n$ of length $\lvert L \rvert\leq \nicefrac{64 b^2}{\vartheta^2}$ consisting of pairwise distinct $n$-bit strings whose associated Fourier coefficients are non-zero.
        \item $P$ follows the procedure in \Cref{corollary:distributional-agnostic-quantum-approximation-fourier-spectrum} to produce, with success probability $\geq 1-\tfrac{\delta}{2}$, a succinctly represented $\Tilde{\phi}:\mathcal{X}_n\to [-1,1]$ such that $\norm{\Tilde{\phi}-\hat{\phi}}_\infty\leq\nicefrac{\vartheta}{2}$ and $\norm{\Tilde{\phi}}_0\leq\tfrac{64 b^2}{\vartheta^2}$. If $P$ obtains an output that violates the $\norm{\cdot}_0$-bound, then $P$ declares failure and the interaction aborts.
        Otherwise, $P$ then sends the list $L = \{s\in\{0,1\}^n~|~ \lvert \Tilde{\phi}(s)\rvert\geq \nicefrac{\vartheta}{2} \}$ to $V$.
        \item If $V$ receives a list $L$ of length $\lvert L\rvert> \nicefrac{64 b^2}{\vartheta^2} $, $V$ rejects the interaction. Otherwise, $V$ uses $\mathcal{O}\left(\tfrac{k^4 \lvert L\rvert^2\log(\nicefrac{\lvert L\rvert}{\delta})}{\varepsilon^4}\right)$ classical random examples from $\mathcal{D}$ to obtain simultaneously $(\nicefrac{\varepsilon^2}{256 k^2\lvert L\rvert})$-accurate estimates $\hat{\xi}(s)$ of $\hat{\phi}(s)$ for all $s\in L$, with success probability $\geq 1-\tfrac{\delta}{2}$, via Chernoff-Hoeffding combined with a union bound over $L$. (For $t\not\in L$, the verifier's estimate $\hat{\gamma}(t)$ for $\hat{g}(t)$ is just $0$.)
        \item If $\sum_{\ell = 1}^{\lvert L\rvert} \left( \hat{\xi}(s_\ell) \right)^2 \geq a^2 - \tfrac{\varepsilon^2}{128 k^2}$, then $V$ determines the $k$ heaviest Fourier coefficients in $L$. That is, $V$ determines $s_1\in\operatorname{argmax}_{t\in L} \lvert \hat{\xi}(t)\rvert$ and, for $2\leq \ell\leq k$, $s_\ell \in \operatorname{argmax}_{t\in L\setminus\{s_1,\ldots,s_{\ell -1}\}} \lvert \hat{\xi}(t)\rvert$, and then outputs the randomized hypothesis  $h:\mathcal{X}_n\to \{0,1\}$ from \Cref{lemma:Fourier-sparse-learning-via-heaviest-Fourier-coefficients}. (If $L$ has fewer than $k$ elements, $V$ just picks the strings $s_{\lvert L\rvert +1},\ldots,s_k$ at random from $\{0,1\}^n\setminus L$. These strings do not matter since $V$ estimates their Fourier weight as $0$.) If $\sum_{\ell = 1}^{\lvert L\rvert} \left( \hat{\xi}(s_\ell) \right)^2 < a^2 - \tfrac{\varepsilon^2}{128 k^2}$, then $V$ outputs $\mathrm{reject}$.
    \end{enumerate}
    We now show that the pair $(V,P)$ has the desired completeness and soundness properties. 
    As a first step towards this goal, we show that $V$ accepts an interaction with $P$ with high probability. To this end, observe that, conditioned on $P$ succeeding in Step 2, $V$ never rejects in Step 3. If we then further condition on $V$ succeeding in Step 3, we have
    \begingroup
    \allowdisplaybreaks
    \begin{align}
        \sum_{\ell = 1}^{\lvert L\rvert} \left( \hat{\xi}(s_\ell) \right)^2
        &\geq \sum_{\ell = 1}^{\lvert L\rvert} \left( \hat{\phi}(s_\ell) \right)^2 - 2\sum_{\ell = 1}^{\lvert L\rvert} \lvert\hat{\phi}(s_\ell)-\hat{\xi}(s_\ell)\rvert\\
        &\geq \sum_{s:\hat{\phi}(s)\neq 0} \lvert\hat{\phi}(s)\rvert^2 - 2\sum_{\ell = 1}^{\lvert L\rvert} \lvert\hat{\phi}(s_\ell)-\hat{\xi}(s_\ell)\rvert\\
        &\geq a^2 - 2\lvert L\rvert \cdot \frac{\varepsilon^2}{256 k^2 \lvert L\rvert}\\
        &= a^2 - \frac{\varepsilon^2}{128 k^2}\, ,
    \end{align}
    \endgroup
    where the first step used that the function $[-1,1]\ni \xi\to \xi^2$ is $2$-Lipschitz, the second step used that $\mathcal{D}\in \mathfrak{D}_{ \mathcal{U}_n; \geq \vartheta}$ implies $(\hat{\phi}(s)\neq 0~ \Rightarrow ~ s\in L)$ if Step 2 succeeds, and the third step used the approximation guarantee in Step 3 as well as Parseval together with $\mathcal{D}\in \mathfrak{D}_{ \mathcal{U}_n; [a^2,b^2]}$.
    Thus, if both Step 2 and Step 3 succeed, which by a union bound happens with probability $\geq 1-\delta$, then $V$ accepts in Step 4.
    
    Moreover, whenever Step 3 is successful and $V$ does not reject in Step 4, then the output of $V$ is as desired. This can be seen as follows: If $V$ does not reject in Step 4 and if Step 3 was successful, then this implies that
    \begingroup
    \allowdisplaybreaks
    \begin{align}
        \sum_{t\in L} \left(\hat{\phi}(t)\right)^2
        &=      \mathbb{E}_{x\sim\mathcal{U}_n}[(\phi(x))^2] -  \sum_{t\in L} \left(\hat{\phi}(t)\right)^2\\
        &\leq b^2 - \sum_{t\in L} \left(\hat{\xi}(t)\right)^2 + 2 \sum_{\ell = 1}^{\lvert L\rvert} \lvert\hat{\phi}(s_\ell)-\hat{\xi}(s_\ell)\rvert\\
        &\leq (b^2-a^2) + \frac{\varepsilon^2}{128 k^2} + 2\lvert L\rvert \cdot \frac{\varepsilon^2}{256 k^2\lvert L\rvert}\\
        &= (b^2-a^2) + \frac{\varepsilon^2}{64 k^2}\, ,    
    \end{align}
    \endgroup
    where we again used that $[-1,1]\ni \xi\to \xi^2$ is $2$-Lipschitz.
    This tells us that $\lvert \hat{\phi}(t)\rvert \leq \sqrt{(b^2-a^2) + \tfrac{\varepsilon^2}{64 k^2}}\leq \sqrt{b^2-a^2} + \tfrac{\varepsilon}{8 k} \leq \tfrac{\varepsilon}{4k}$ holds for every $t\not\in L$, which now allows us to show that the output hypothesis of $V$ has the desired property.
    Let $t_1,\ldots,t_k\in \{0,1\}^n$ be $k$ heaviest Fourier coefficients of $\phi$. That is, let $t_1\in \operatorname{argmax}_{t\in \{0,1\}^n} \lvert \hat{\phi}(t)\rvert$, and for $2\leq \ell\leq k$, let $t_\ell\in \operatorname{argmax}_{t\in \{0,1\}^n\setminus\{t_1,\ldots,t_{\ell -1}\}} \lvert \hat{\phi}(t)\rvert$.
    By \Cref{lemma:technical}, we have 
    \begin{align}
        \max_{1\leq \ell\leq k} \lvert \hat{\phi}(t_\ell) - \hat{\phi}(s_\ell)\rvert
        &\leq 2 \max_{t\in\{0,1\}^n} \lvert \hat{\phi}(t) - \hat{\xi}(t)\rvert\\
        &= 2\max\left\{\max_{t\in L}\lvert \hat{\phi}(t) - \hat{\xi}(t)\rvert,\max_{t\not\in L}\lvert \hat{\phi}(t) - \hat{\xi}(t)\rvert \right\}\\
        &= 2\max\left\{\max_{t\in L}\lvert \hat{\phi}(t) - \hat{\xi}(t)\rvert,\max_{t\not\in L}\lvert \hat{\phi}(t)\rvert \right\}\\
        &\leq 2\max\left\{\frac{\varepsilon^2}{256 k^2 |L|}, \frac{\varepsilon}{4k}\right\}\\
        &= \frac{\varepsilon}{2k}
    \end{align}
    As $\tfrac{\varepsilon^2}{256 k^2 |L|}\leq \tfrac{\varepsilon}{2k}$, this shows that the verifier has $(\nicefrac{\varepsilon}{2k})$-accurate estimates of $k$ $(\nicefrac{\varepsilon}{2k})$-approximately-heaviest Fourier coefficients of $\phi$. 
    Thus, by \Cref{lemma:Fourier-sparse-learning-via-heaviest-Fourier-coefficients}, the randomized hypothesis produced by $V$ is as desired.
    Note that this last part of our reasoning only relied on $V$ not rejecting in Steps 3 or 4, but was independent of the action of the quantum prover. Therefore, with this we have already established the desired soundness.
    The complexity bounds are obtained similarly to the proofs of \Cref{theorem:functional-agnostic-quantum-fourier-sparse-verification-qsq-no-small-non-zero-Fourier-coeff} and \Cref{theorem:distributional-agnostic-quantum-parity-verification-no-small-non-zero-Fourier-coeff}.
\end{proof}

To conclude our discussion of interactive verification of quantum learning, we extract the central routine underlying our verification protocols.
Namely, the above verification results rely on the fact that a classical verifier interacting with an untrusted quantum prover can construct an approximation to the Fourier spectrum of the unknown distribution. We make this explicit in the following result:

\begin{theorem}\label{proposition:interactive-verification-fourier-approximation}
    Let $\vartheta\in (2^{-(\tfrac{n}{2} - 3)},1)$.
    Let $0\leq a \leq b\leq 1$.
    Let $\delta\in (0,1)$ and $\varepsilon\geq 2\sqrt{b^2-a^2}$.
    There is a classical-quantum verifier-prover pair $(V,P)$ that achieves the following for any $\mathcal{D}=(\mathcal{U}_n, \varphi)\in \mathfrak{D}_{ \mathcal{U}_n; \geq \vartheta} \cap \mathfrak{D}_{ \mathcal{U}_n; [a^2,b^2]}$:
    \begin{itemize}
        \item $P$ uses $\mathcal{O}\left(\tfrac{\log(\nicefrac{1}{\delta\vartheta^2})}{\vartheta^4}\right)$ copies of $\rho_{\mathcal{D}}$ , $\mathcal{O}\left(n\tfrac{\log(\nicefrac{1}{\delta\vartheta^2})}{\vartheta^4}\right)$ single-qubit gates, a classical memory of size $\tilde{\mathcal{O}}\left(n\tfrac{\log(\nicefrac{1}{\delta\vartheta^2})}{\vartheta^4}\right)$, and classical running time $\tilde{\mathcal{O}}\left(n\tfrac{\log(\nicefrac{1}{\delta \vartheta^2})}{\vartheta^4}\right)$.
        \item $V$ uses $\mathcal{O}\left(b^4\tfrac{\log(\nicefrac{1}{\delta\vartheta^2})}{\varepsilon^4 \vartheta^4}\right)$ classical random examples, $\tilde{\mathcal{O}}\left(n\tfrac{b^4\log(\nicefrac{1}{\delta\vartheta^2})}{\varepsilon^4 \vartheta^4}\right)$ classical running time, and a classical memory of size $\tilde{\mathcal{O}}\left(n\tfrac{b^4 \log(\nicefrac{1}{\delta\vartheta^2})}{\varepsilon^4 \vartheta^4}\right)$.
        \item $(V,P)$ uses only a single round of communication consisting of at most $\mathcal{O}\left(\tfrac{n}{\vartheta^2}\right)$ classical bits.
        \item Completeness: If $V$ interacts with the honest prover $P$, then, with success probability $\geq 1-\delta$, $V$ accepts the interaction and outputs a succinctly represented $\tilde{\phi}$ such that $\norm{\tilde{\phi}-\hat{\phi}}_1\leq\varepsilon$ and $\norm{\tilde{\phi}}_0 \leq \mathcal{O}(\nicefrac{1}{\vartheta^2})$.
        \item Soundness: If $V$ interacts with any (possibly unbounded) prover $P'$, then, with failure probability at most $\delta$, $V$ accepts the interaction and outputs a $\tilde{\phi}$ such that $\norm{\tilde{\phi}-\hat{\phi}}_1 >\varepsilon$. 
    \end{itemize}
    In that sense, the Fourier spectrum of any distribution in $\mathfrak{D}_{ \mathcal{U}_n; \geq \vartheta} \cap \mathfrak{D}_{ \mathcal{U}_n; [a^2,b^2]}$ is efficiently verifiable through classical-quantum interactions. 
\end{theorem}
\begin{proof}
    This is a consequence of the proofs of \Cref{theorem:distributional-agnostic-quantum-parity-verification-no-small-non-zero-Fourier-coeff} and \Cref{theorem:distributional-agnostic-quantum-fourier-sparse-verification-no-small-non-zero-Fourier-coeff}.
\end{proof}

An analogous interactive verification protocol for Fourier spectrum approximation is also possible in the statistical query setting.

\begin{remark}
    We complement \Cref{theorem:distributional-agnostic-qsq-fourier-sparse-verification-no-small-non-zero-Fourier-coeff,theorem:distributional-agnostic-quantum-fourier-sparse-verification-no-small-non-zero-Fourier-coeff,proposition:interactive-verification-fourier-approximation} with a short discussion of an alternative, simplified verification procedure. Assuming $\varepsilon> \sqrt{b^2 - a^2}$, this alternative for a classical SQ verifier and a quantum mixture-of-superpositions example prover looks as follows in the case of Fourier spectrum approximation up to error $\varepsilon$ in $2$-norm:
    \begin{enumerate}
        \item  $V$ asks $P$ to provide a list $L=\{(s_1,\hat{\phi}'(s_1)),\ldots , (s_{\lvert L\rvert},\hat{\phi}'(s_{\lvert L\rvert}))\}\subset \{0,1\}^n$ of length $\lvert L \rvert\leq \nicefrac{64 b^2}{\vartheta^2}$ consisting of pairwise distinct $n$-bit strings $s_\ell$ whose associated Fourier coefficients are non-zero and sufficiently accurate estimates $\hat{\phi}'(s_\ell)$ of the associated Fourier coefficients $\hat{\phi}(s_\ell)$.
        \item $P$ follows the procedure in \Cref{corollary:distributional-agnostic-quantum-approximation-fourier-spectrum} to produce, with success probability $\geq 1-\delta$, a succinctly represented $\hat{\phi}':\mathcal{X}_n\to [-1,1]$ such that $\norm{\hat{\phi}'-\hat{\phi}}_\infty\leq\Tilde{\varepsilon}\coloneqq \tfrac{\vartheta}{8b}\cdot\sqrt{\tfrac{\varepsilon^2 - (b^2-a^2)}{2}}$ and $\norm{\hat{\phi}'}_0\leq\tfrac{64 b^2}{\vartheta^2}$. 
        While this $0$-norm bound is not immediate from a direct application of \Cref{corollary:distributional-agnostic-quantum-approximation-fourier-spectrum}, it can be obtained via Parseval when using that every non-zero Fourier coefficient of $\phi$ has absolute value $\geq\vartheta$, that the total Fourier weight of $\phi$ is at most $b^2$, and that $\Tilde{\varepsilon}\leq\nicefrac{\vartheta}{3}$.
        If $P$ obtains an output that violates the $\norm{\cdot}_0$-bound, then $P$ declares failure and the interaction aborts.
        Otherwise, $P$ then sends the list $L = \{(s,\hat{\phi}'(s))\in\{0,1\}^n~|~ \lvert \hat{\phi}'(s)\rvert\geq \vartheta (1 - \tfrac{1}{8b}\cdot\sqrt{\tfrac{\varepsilon^2 - (b^2-a^2)}{2}}) \}$ to $V$. 
        \item If $V$ receives a list $L$ of length $\lvert L\rvert> \nicefrac{64 b^2}{\vartheta^2} $, $V$ rejects the interaction. 
        Otherwise, $V$ uses a single classical SQ of tolerance $\tau = \tfrac{\varepsilon^2 - (b^2-a^2)}{8}$ to obtain an estimate $\iota$ with $\lvert \iota - \langle \hat{\phi}', \hat{\phi}\rangle\rvert\leq\tau$. 
        This quantity can be estimated via an SQ since, by Plancherel, it can be rewritten as $\langle \hat{\phi}', \hat{\phi}\rangle = \sum_{s}\hat{\phi}'(s) \hat{\phi}(s) = \mathbb{E}_{x\sim\mathcal{U}_n}[\phi'(x)\phi(x)] = \mathbb{E}_{(x,y)\sim\mathcal{D}}[\phi'(x) (1-2y)]$, an expectation value of a function (which is known to $V$) w.r.t.~$\mathcal{D}$.
        \item If $\sum_{\ell=1}^{\lvert L\rvert} \left(\hat{\phi}'(s_\ell)\right)^2 - 2\iota> \tfrac{3\varepsilon^2 - 3b^2 + a^2}{4}$, $V$ rejects the interaction. Otherwise, $V$ outputs $\hat{\phi}'$.
    \end{enumerate}
    This verification procedure has the same soundness and completeness guarantees as \Cref{proposition:interactive-verification-fourier-approximation}, only with $2$-norms instead of $1$-norms. As we are dealing with sparse Fourier spectra, we can easily translate from $2$- to $1$-norm guarantees. Concretely, to achieve $1$-norm accuracy $\varepsilon$, it suffices to ensure a $2$-norm accuracy of $\varepsilon (\mathrm{sparsity})^{-1/2}\leq \mathcal{O}(\nicefrac{\varepsilon\vartheta}{b})$.\\
    Compared to our previous protocol, this alternative has the advantage that $V$ uses only a single SQ of similar tolerance and little classical computation.
    This is achieved by ``outsourcing'' more work to the quantum prover, leading to an $\varepsilon$-dependence and thus an increase in the number of copies and single-qubit gates as well as in the classical memory size and running time used by the honest $P$.\\
    To see that the protocol is sound, assume that a dishonest prover provides $\hat{\phi}'$ with $\norm{\hat{\phi}'-\hat{\phi}}_2^2> \varepsilon^2$. Then, expanding $\norm{\hat{\phi}'-\hat{\phi}}_2^2 = \norm{\hat{\phi}'}_2^2 + \norm{\hat{\phi}}_2^2 - 2\langle \hat{\phi}', \hat{\phi}\rangle$, rearranging, using $\norm{\hat{\phi}}_2^2\leq b^2$, and using the definition of $\iota$ gives us $\norm{\hat{\phi}'}_2^2 - 2\iota > \varepsilon^2 - b^2 - 2\tau = \tfrac{3\varepsilon^2 - 3b^2 + a^2}{4}$. Thus, $V$ rejects the interaction.
    In contrast, if $P$ is honest and provides $\hat{\phi}'$ with $\norm{\hat{\phi}'-\hat{\phi}}_2^2\leq \tfrac{64}{\vartheta^2}\cdot\norm{\hat{\phi}'-\hat{\phi}}_\infty^2 \leq \tfrac{\varepsilon^2 - (b^2-a^2)}{2}$, then expanding $\norm{\hat{\phi}'-\hat{\phi}}_2^2$, rearranging, using $\norm{\hat{\phi}}_2^2\geq a^2$, and using the definition of $\iota$ gives us$\norm{\hat{\phi}'}_2^2 - 2\iota \leq \tfrac{\varepsilon^2 - (b^2-a^2)}{2} - a^2 + 2\tau = \tfrac{3\varepsilon^2 - 3b^2 + a^2}{4}$. Thus, $V$ accepts the interaction and outputs a $\hat{\phi}'$ that is sufficiently accurate in $2$-norm.
\end{remark}

Finally, we note that, similarly to \Cref{theorem:limitation-improvement-distributional-agnostic-verification}, the accuracy limitation $\varepsilon\geq 2\sqrt{b^2-a^2}$ in \Cref{proposition:interactive-verification-fourier-approximation} cannot be significantly improved while keeping the sample complexity of the verifier $n$-independent.

\begin{proposition}\label{proposition:limitation-interactive-fourier-approximation}
    Let $0<a\leq b\leq 1$.
    Let $\alpha = a$ and $\beta = \sqrt{b^2-a^2}$.
    Define $\vartheta = \alpha$, $\varepsilon = \nicefrac{\beta}{3} = \tfrac{1}{3}\cdot \sqrt{b^2 -a^2}$ and $\delta=\nicefrac{1}{3}$.
    Any protocol for interactive verification for Fourier spectrum approximation (in the sense of \Cref{proposition:interactive-verification-fourier-approximation}) with a classical verifier $V$ with access to classical random examples interacting with a quantum prover $P$ with access to mixture-of-superpositions quantum examples that succeeds for any $\mathcal{D}\in \mathfrak{D}_{ \mathcal{U}_n; \geq \vartheta} \cap \mathfrak{D}_{ \mathcal{U}_n; [a^2,b^2]}$ requires the verifier to use at least $\Omega (n)$ classical examples. 
\end{proposition}

Here, we consider $a$ and $b$ to be constant, and we focus on the scaling with $n$. 

\begin{proof}   
    Let $s\in \{0,1\}^n$ be arbitrary. 
    Define the class of distributions
    \begin{equation}
        \mathfrak{D}_{\alpha, \beta}^{s}
        =\left\{(\mathcal{U}_n, \varphi')~|~\exists t\in\{0,1\}^n\setminus\{s\}: \phi' =1-2\varphi' = \alpha\chi_s + \beta\chi_t\right\} .
    \end{equation}
    Note that $\mathfrak{D}_{\alpha, \beta}^{s}\subset \mathfrak{D}_{ \mathcal{U}_n; \geq \vartheta} \cap \mathfrak{D}_{ \mathcal{U}_n; [a^2,b^2]}$.
    Using once more the strategy of \cite[Theorem 8]{mutreja2022pac-verification} similarly to the proof of \Cref{theorem:limitation-improvement-distributional-agnostic-verification}, one can show that any classical-quantum interactive verification procedure for Fourier spectrum approximation (in the sense of \Cref{proposition:interactive-verification-fourier-approximation}) with parameters as in the statement of the proposition gives rise to a classical tester that can distinguish a uniformly random element of $\mathfrak{D}_{\alpha, \beta}^{s}$ from $(\mathcal{U}_n, \alpha\chi_s)$, with success probability $\geq 1-\delta$.
    In \Cref{lemma:fourier-modification-reduction}, we show that this latter problem is at least as hard as distinguishing a uniformly random $\eta$-noisy parity (acting on uniformly random inputs) from the uniform distribution $\mathcal{U}_{n+1}$, for any $\eta\in [0,\tfrac{1-\nicefrac{\beta}{(1-\alpha)}}{2}]$.
    This last problem requires at least $\Omega (n)$ classical examples (compare \Cref{lemma:fully-uniform-vs-random-noisy-parity}).
    Thus, by our reduction, so does the problem that we started from.
\end{proof}

\section{Distribution-Independent Agnostic Quantum Learning and its Verification}\label{section:distribution-independent-quantum-limitations}

So far, we have focused on agnostic learning under a promise on the input marginal. Namely, we assumed that $\mathcal{D}_{\mathcal{X}_n}=\mathcal{U}_n$ is the uniform distribution. 
While this focus on distribution-dependent learning is common in computational learning theory (to avoid computational infeasibility results), statistical learning theory also often considers a setting of distribution-independent learning, where no prior assumptions on the input marginal of the unknown distribution are made. 
The sample complexity of learning in this distribution-independent agnostic model has long been fully characterized in the classical case \cite{vapnik1971uniform, blumer1989learnability, talagrand1994sharper}.
Moreover, \cite{arunachalam2018optimal} recently established that the optimal quantum sample complexity when using superposition examples coincides with the classical one up to constant factors. 
Here, we demonstrate that such a limitation of quantum learning also applies to our mixture-of-superpositions examples.

\begin{theorem}[Formal statement of \Cref{theorem:main-result-distribution-independent-lower-bound}]\label{theorem:distribution-independent-lower-bound}
    Let $\mathcal{F}\subseteq\{0,1\}^{\mathcal{X}_n}$ be a benchmark class with VC-dimension $\operatorname{VC}(\mathcal{F})=d\geq 1$.
    Then at least
    \begin{equation}
        m\geq \tilde{\Omega} \left( \frac{d + \log(1/\delta)}{\varepsilon^2} \right)
    \end{equation}
    copies of $\rho_\mathcal{D}$, with $\mathcal{D}$ an unknown probability distribution over $\{0,1\}^n\times\{0,1\}$, are necessary for distribution-independent quantum agnostic learning of $\mathcal{F}$ with accuracy $\varepsilon\in (0,\tfrac{1}{4})$ and confidence parameter $\delta \in (0,\tfrac{1}{2})$. 
    Here, the $\tilde{\Omega}$ hides prefactors logarithmic in $d$.
\end{theorem}
\begin{proof}
    We first prove the lower bound in terms of $\delta$ and $\varepsilon$. In fact, we establish this lower bound for any non-trivial benchmark class $\mathcal{F}$, similarly to \cite[Lemma 12]{arunachalam2018optimal} and \cite[Lemma 5.1]{caro2021binary}. 
    As $\mathcal{F}$ is non-trivial, then there are two concepts $f_1,f_2\in\mathcal{F}$ and an input $x_0\in\mathcal{X}_n$ such that $f_1(x_0)\neq f_2(x_0)$. 
    Define two probability distributions $D_{\pm}$ over $\mathcal{X}_n\times\{0,1\}$ as follows: $D_{\pm}(x,y)=0$ if $x\neq x_0$, $D_{\pm}(x_0,f_1(x)) = \tfrac{1\pm\varepsilon}{2}$, and $D_{\pm}(x_0,f_2(x)) = \tfrac{1\mp\varepsilon}{2}$.
    It is easy to see that any quantum agnostic learning algorithm that, with success probability $\geq 1-\delta$, achieves an additive error $\leq \varepsilon$ compared to the optimal achievable risk in particular also distinguishes between $D_+$ and $D_-$ with success probability $\geq 1-\delta$, given access to the corresponding training data states $\rho_{\mathcal{D}_\pm}^{\otimes m}$.
    The optimal success probability for this distinguishing task is given by (compare, e.g., \cite{nielsen2000quantum}) $p_{\mathrm{opt}} = \tfrac{1}{2}(1 + \tfrac{1}{2}\norm{\rho_{\mathcal{D}_+}^{\otimes m} - \rho_{\mathcal{D}_-}^{\otimes m}}_1)$. With the fidelity defined as $F(\rho,\sigma)=\Tr[\sqrt{\rho^{\nicefrac{1}{2}}\sigma\rho^{\nicefrac{1}{2}}}]$, the Fuchs-van de Graaf inequalities \cite{fuchs1999cryptographic} now tell us
    \begin{equation}
        \frac{1}{2}\norm{\rho_{\mathcal{D}_+}^{\otimes m} - \rho_{\mathcal{D}_-}^{\otimes m}}_1
        \leq \sqrt{1 - F(\rho_{\mathcal{D}_+}^{\otimes m}, \rho_{\mathcal{D}_-}^{\otimes m})^2}
        = \sqrt{1 - F(\rho_{\mathcal{D}_+}, \rho_{\mathcal{D}_-})^{2m}} ,
    \end{equation}
    where the last step holds because the fidelity is multiplicative under tensor products.
    Combining this upper bound with the requirement $p_{\mathrm{opt}}\geq 1-\delta$, we see after a rearranging that $m\geq \tfrac{\log(4\delta(1-\delta))}{\log(F(\rho_{\mathcal{D}_+}, \rho_{\mathcal{D}_-})^2)}$. 
    Using that the fidelity is strongly concave \cite[Theorem 9.7]{nielsen2000quantum} and that $\rho_{\mathcal{D}_{\pm}} = \tfrac{1\pm\varepsilon}{2}\ket{x_0,f_1(x_0)}\bra{x_0,f_1(x_0)} + \tfrac{1\mp\varepsilon}{2}\ket{x_0,f_2(x_0)}\bra{x_0,f_2(x_0)}$, we get the fidelity lower bound
    \begin{align}
        F(\rho_{\mathcal{D}_+}, \rho_{\mathcal{D}_-})
        &\geq \sqrt{\frac{1+\varepsilon}{2}\cdot \frac{1-\varepsilon}{2}} F(\ket{x_0,f_1(x_0)}\bra{x_0,f_1(x_0)}, \ket{x_0,f_1(x_0)}\bra{x_0,f_1(x_0)}) \\
        &\hphantom{\geq}~ + \sqrt{\frac{1-\varepsilon}{2}\cdot \frac{1+\varepsilon}{2}} F(\ket{x_0,f_2(x_0)}\bra{x_0,f_2(x_0)}, \ket{x_0,f_2(x_0)}\bra{x_0,f_2(x_0)})\\
        &= \sqrt{1-\varepsilon^2} .
    \end{align}
    Plugging this lower bound on $F(\rho_{\mathcal{D}_+}, \rho_{\mathcal{D}_-})$ into our previous lower bound on $m$, we get
    \begin{equation}
        m
        \geq \frac{\log(4\delta(1-\delta))}{\log(F(\rho_{\mathcal{D}_+}, \rho_{\mathcal{D}_-})^2)}
        = \frac{\log(\tfrac{1}{4\delta(1-\delta)})}{\log(\tfrac{1}{F(\rho_{\mathcal{D}_+}, \rho_{\mathcal{D}_-})^2})}
        \geq \frac{\log(\tfrac{1}{4\delta(1-\delta)})}{\log(\tfrac{1}{1-\varepsilon^2})} .
    \end{equation}
    A Taylor expansion of the logarithm in the denominator now gives the lower bound
    \begin{equation}
        m\geq \Omega \left(\frac{\log(1/\delta)}{\varepsilon^2} \right) .
    \end{equation}
    
    Next, we prove the $d$-dependent part of the lower bound. For this, we adapt the information-theoretic proof strategy from \cite{arunachalam2018optimal}.
    Let $\varepsilon\in (0,\tfrac{1}{4})$.
    As $\operatorname{VC}(\mathcal{F})=d$, we can find a set $S=\{x_1,\ldots , x_d\}\subset\mathcal{X}_n$ of $d$ distinct points that is shattered by $\mathcal{F}$. That is, for every $a\in\{0,1\}^d$, there exists $f_a\in\mathcal{F}$ such that $f_a(x_i)=a_i$ holds for all $1\leq i\leq d$.
    Now, for each $a\in\{0,1\}^d$, we define the probability distribution $\mathcal{D}_a$ over $\mathcal{X}_n\times \{0,1\}$ as follows:
    \begin{equation}
        \mathcal{D}_a (x,b)
        = 
        \begin{cases}
            \frac{1}{2d}\left(1 + (-1)^{a_i + b}\cdot 4\varepsilon\right)\quad &\textrm{ if } x = x_i\\
            0 &\textrm{ else}
        \end{cases}.
    \end{equation}
    By construction, for every $a\in\{0,1\}^d$ and for every $f\in\mathcal{F}$, we have
    \begin{equation}
        \mathbb{P}_{(x,b)\sim \mathcal{D}_a} [b\neq f(x)]
        = \frac{1}{2d}\sum_{i=1}^d \left( \left(1 + 4\varepsilon\right) \delta_{f(x_i), a_i\oplus 1} + \left(1 - 4\varepsilon\right) \delta_{f(x_i), a_i} \right)
    \end{equation}
    Thus, $f\in\mathcal{F}$ is a minimum-error concept in $\mathcal{F}$ w.r.t.~$\mathcal{D}_a$ if an only if $f\rvert_{S}(x_i) = a_i$ holds for all $1\leq i\leq d$. Moreover, if $f\rvert_{S}(x_i) = c_i$ for some $c\in\{0,1\}^d$ with $c\neq a$, then such an $f$ incurs excess risk $\tfrac{4\varepsilon}{d}\cdot d_{\rm H}(a,c)$, where $d_{\rm H}(\cdot,\cdot)$ denotes the Hamming distance. 
    Accordingly, any quantum algorithm for distribution-independent quantum agnostic learning $\mathcal{F}$ from $m$ copies of $\rho_{\mathcal{D}_a}$, $a\in\{0,1\}^d$ unknown, has to output a hypothesis that, when restricted to $S$, becomes a $d$-bit string that is $\tfrac{d}{4}$-close to $a$ in Hamming distance, with success probability $\geq 1-\delta$. This still holds if the quantum learner is promised in advance that the unknown distribution is supported on $S$.

    Let us consider the CQ state
    \begin{equation}
        \rho
        \coloneqq \frac{1}{2^d} \sum_{a\in\{0,1\}^d} \ket{a}\bra{a} \otimes \rho_{\mathcal{D}_a}^{\otimes m},
    \end{equation}
    where $m$ is the training data size. We will refer to its classical subsystem as the $A$-subsystem and to the $m$ quantum registers as subsystems $B_1,\ldots ,B_m$.
    Using that the output of the quantum agnostic learner upon input of $\rho_{\mathcal{D}_a}^{\otimes m}$ is $\tfrac{d}{4}$-close to $a$, one can show that the mutual information between the classical subsystem and the quantum subsystems in $\rho$ satisfies $I(A;B_1,\ldots ,B_m)_\rho\geq \Omega (d)$, compare \cite[Proof of Theorem 12]{arunachalam2018optimal}.
    Next, since $\rho_{\mathcal{D}_a}^{\otimes m}$ is a tensor power for every $a$, we have $I(A;B_1,\ldots ,B_m)_\rho\leq m\cdot I(A;B_1)_\rho$, compare again \cite[Proof of Theorem 12]{arunachalam2018optimal}.
    Thus, the remainder of the proof is concerned with upper bounding $I(A;B_1)_\rho$. As $\rho_{A B_1}$ is a CQ-state, $I(A;B_1)_\rho$ equals the Holevo information of the ensemble $\{(\tfrac{1}{2^d}, \rho_{\mathcal{D}_a})\}_{a\in\{0,1\}^d}$, compare \cite[Exercise 11.6.9]{wilde2011bookarxiv}. That is,
    \begin{equation}
        I(A;B_1)_\rho
        = S\left(\frac{1}{2^d} \sum_{a\in\{0,1\}^d} \rho_{\mathcal{D}_a}\right) - \frac{1}{2^d} \sum_{a\in\{0,1\}^d} S\left(\rho_{\mathcal{D}_a}\right),
        \label{eq:mutual_info_holevo_information}
    \end{equation}
    where $S(\rho)$ denotes the von Neumann entropy of the state $\rho$.
    Henceforth, we refer to $\overline{\rho}:=\frac{1}{2^d} \sum_{a\in\{0,1\}^d} \rho_{\mathcal{D}_a}$ as the average state. 
    We can diagonalize both $\overline{\rho}$ as well as $\rho_{\mathcal{D}_a}$ to obtain their respective spectra and hence their von Neumann entropies.
    We start by writing out $\rho_{\mathcal{D}_a}$, which is the mixture-of-superpositions example state corresponding to the distribution $\mathcal{D}_a$. Note that the marginal distribution $\mathcal{D}_{a}\rvert_{ \mathcal{X}_n}$ is uniform over the set shattered $S=\{x_1, \dots, x_d\}$. We denote this marginal distribution by $\mathcal{U}_S$. Hence, we have
    \begin{equation}
        \rho_{\mathcal{D}_a} 
        = \mathbb{E}_{f\sim F_{\mathcal{D}_a}} \left[ \ket{\psi_{(\mathcal{U}_S, f)}}\bra{\psi_{(\mathcal{U}_S, f)}} \right]
        = \sum_{c\in\left\{ 0,1\right\} ^{d}}\left(\prod_{i=1}^{d}\frac{1}{2}\left(1+\left(-1\right)^{a_{i}+c_{i}}4\epsilon\right)\right)\ket{\psi_{\mathcal{U}_S,f_{c}}}\bra{\psi_{\mathcal{U}_S,f_{c}}} ,
    \end{equation}
    where we defined the pure states $     \ket{\psi_{\mathcal{U}_S,f_{c}}}=\frac{1}{\sqrt{d}}\sum_{i=1}^{d}\ket{x_{i},f_{c}\left(x_{i}\right)}=\frac{1}{\sqrt{d}}\sum_{i=1}^{d}\ket{x_{i},c_{i}}$.
    The average state $\overline{\rho}$ is given by
    \begin{align}
        \overline{\rho}&=\frac{1}{2^{d}}\sum_{a\in\left\{ 0,1\right\} ^{d}}\rho_{\mathcal{D}_{a}} \\
	 &=\frac{1}{2^{d}} \sum_{a\in\left\{ 0,1\right\} ^{d}} \sum_{c\in\left\{ 0,1\right\} ^{d}} \left(\prod_{i=1}^{d}\frac{1}{2}\left(1+\left(-1\right)^{a_{i}+c_{i}}4\epsilon\right)\right) \ket{\psi_{\mathcal{U}_S,f_{c}}}\bra{\psi_{\mathcal{U}_S,f_{c}}} \\
	&=\frac{1}{d2^{d}}\sum_{c\in\left\{ 0,1\right\} ^{d}}\sum_{i,j=1}^{d}\ket{x_{i},c_{i}}\bra{x_{j},c_{j}} \\
        &= \frac{1}{2d}\sum_{i=1}^{d}\sum_{c_{i}\in\left\{ 0,1\right\} }\ket{x_{i},c_{i}}\bra{x_{i},c_{i}}+\frac{1}{4d}\sum_{\substack{i,j\\i\neq j}}\sum_{c_{i}}\sum_{c_{j}}\ket{x_{i},c_{i}}\bra{x_{j},c_{j}} .
    \end{align}
    where in the third equality we have used that $\sum_{a_i\in \{0,1\}}\left(-1\right)^{a_{i}+c_{i}}=0$. 
    The average state $\overline{\rho}$ can be diagonalized as follows: 
    Its eigenvalues are given by $\lambda_{1}(\overline{\rho}) =1/2$, $\lambda_{2}(\overline{\rho}) =\frac{1}{2d}$, and $\lambda_{3}(\overline{\rho}) =0$.
    The $\lambda_{1}(\overline{\rho})$-eigenspace is $1$-dimensional and spanned by the uniform superposition state
    \begin{equation}
        \ket{\psi^{\left(1\right)}}=\frac{1}{\sqrt{2d}}\sum_{i=1}^{d}\sum_{b\in\left\{ 0,1\right\} }\ket{x_{i},b} \, .
    \end{equation}
    The $\lambda_{2}(\overline{\rho})$-eigenspace is $d$-dimensional and spanned by the eigenvectors
    \begin{equation}
        \ket{\psi_{i}^{\left(2\right)}}=\frac{1}{\sqrt{2}}\sum_{b\in\left\{ 0,1\right\} }\left(-1\right)^{b}\ket{x_{i},b}
        =\frac{1}{\sqrt{2}}\left(\ket{x_{i},0}-\ket{x_{i},1}\right)\text{ for }i=1,\dots,d \, .
    \end{equation}
    Lastly, the $\lambda_{3}(\overline{\rho})$-eigenspace is $\left(d-1\right)$-dimensional and spanned by the eigenvectors
    \begin{equation}
        \ket{\psi_{i}^{\left(3\right)}} =\frac{1}{2}\left(\sum_{b\in\left\{ 0,1\right\} }\ket{x_{1},b}-\sum_{b\in\left\{ 0,1\right\} }\ket{x_{i},b}\right)\text{ for }i=2,\dots,d \, .
    \end{equation}
    We thus find that the von Neumann entropy of the average state is given by
    \begin{equation}
        S\left(\overline{\rho}\right) 
        =\tfrac{1}{2}\log2+\frac{d}{2d}\,\log\left(2d\right)
        =\frac{1}{2}\left(1+\log\left(2d\right)\right)
        =1+\frac{1}{2}\log\left(d\right) . \label{eq:entropy_average_state}
    \end{equation}
    Next, we deal with diagonalizing the $\rho_{\mathcal{D}_a}$. We observe that the states $\rho_{\mathcal{D}_a}$ for different $a\in\left\{ 0,1\right\} ^{d}$ are just permuted versions of one another. To see this, consider the expression
    \begin{align}
        \rho_{a} & 
        =\frac{1}{d2^{d}}\sum_{c\in\left\{ 0,1\right\} ^{d}}\left(\prod_{i=1}^{d}\left(1+\left(-1\right)^{c_{i}+a_{i}}4\epsilon\right)\right)\sum_{i,j=1}^{d}\ket{x_{i},c_{i}}\bra{x_{j},c_{j}}
    \end{align}
    and note that all the $\rho_{a}$ are related to $\rho_{a=0^d}$
    by a basis permutation $P_{a}\ket{x_{i},c_{i}}\to\ket{x_{i},c_{i}\oplus a_{i}}$, which is a unitary transformation. That is, $\rho_{0\dots0}=P_{a}\rho_{a}P_{a}^{\dagger}$ holds for all $a\in\{0,1\}^d$. 
    Hence, the $\rho_{\mathcal{D}_a}$ share the same spectrum for all $a\in\left\{ 0,1\right\} ^{d}$ and it suffices to diagonalize $\rho_{\mathcal{D}_a}$ for a single $a$, say $a=0^d$.
    Now, writing out $\rho_{\mathcal{D}_{a=0^{d}}}$, we have
    \begingroup
    \allowdisplaybreaks
    \begin{align}
    \rho_{\mathcal{D}_{a=0^{d}}} & =\frac{1}{d2^{d}}\sum_{c\in\left\{ 0,1\right\} ^{d}}\left(\prod_{i=1}^{d}\left(1+\left(-1\right)^{c_{i}}4\epsilon\right)\right)\sum_{i,j=1}^{d}\ket{x_{i},c_{i}}\bra{x_{j},c_{j}}\\
     & =\frac{1}{d2^{d}}\sum_{c\in\left\{ 0,1\right\} ^{d}}\left(\prod_{i=1}^{d}\left(1+\left(-1\right)^{c_{i}}4\epsilon\right)\right)\left[\underbrace{\sum_{i=1}^{d}\ket{x_{i},c_{i}}\bra{x_{i},c_{i}}}_{i=j}+\sum_{\substack{i,j\\i\neq j}}^{d}\ket{x_{i},c_{i}}\bra{x_{j},c_{j}}\right]\\
     & =\frac{1}{2d}\sum_{i=1}^{d}\sum_{c_{i}\in\left\{ 0,1\right\} }\left(1+\left(-1\right)^{c_{i}}4\epsilon\right)\ket{x_{i},c_{i}}\bra{x_{i},c_{i}}\\
     &\hphantom{=}+\frac{1}{4d}\sum_{\substack{i,j\\i\neq j}}^{d}\sum_{c_{i},c_{j}}\left(1+\left(-1\right)^{c_{i}}4\epsilon\right)\left(1+\left(-1\right)^{c_{j}}4\epsilon\right)\ket{x_{i},c_{i}}\bra{x_{j},c_{j}}\\
     & = \frac{1}{2d} M \, .
    \end{align}
    \endgroup
    where in the last step, we defined the matrix $M$ in order to factor out $\tfrac{1}{2d}$.
    While a bit more tedious than for the average state, we can still diagonalize the matrix $M$ and hence $\rho_{\mathcal{D}_{a=0^{d}}}$ exactly.
    Letting $\tilde{\epsilon}=4\epsilon$, the eigenvalues of $M$ are 
    \begin{align}
        \lambda_{1}(M) & =\frac{1}{2}\left(d\left(1+\tilde{\epsilon}^{2}\right)+\left(1-\tilde{\epsilon}^{2}\right)+\sqrt{\left(d-1\right)^{2}+2\left(d^{2}+2d-1\right)\tilde{\epsilon}^{2}+\left(d-1\right)^{2}\tilde{\epsilon}^{4}}\right) ,\\
        \lambda_{2}(M) & =\frac{1}{2}\left(d\left(1+\tilde{\epsilon}^{2}\right)+\left(1-\tilde{\epsilon}^{2}\right)-\sqrt{\left(d-1\right)^{2}+2\left(d^{2}+2d-1\right)\tilde{\epsilon}^{2}+\left(d-1\right)^{2}\tilde{\epsilon}^{4}}\right) ,\\
        \lambda_{3}(M) & =1-\tilde{\epsilon}^{2} ,\\
        \lambda_{4}(M) & =0 .
    \end{align}
    The $\lambda_{1}(M)$-eigenspace is 1-dimensional and spanned by
    \begin{equation}
        \ket{v^{\left(1\right)}} =\sum_{i=1}^{d}\sum_{b\in\left\{ 0,1\right\} }(\alpha_{+})^{1-b}\ket{x_{i},b},
    \end{equation} 
    where
    \begin{equation}
        \alpha_{+}=\frac{\frac{2d}{d-1}\tilde{\epsilon}+\sqrt{\left(\frac{2d}{d-1}\tilde{\epsilon}\right)^{2}+\left(1-\tilde{\epsilon}^{2}\right)^{2}}}{1-\tilde{\epsilon}^{2}}.
    \end{equation}
    The $\lambda_{2}(M)$-eigenspace is 1-dimensional and spanned by
    \begin{equation}
        \ket{v^{\left(2\right)}} =\sum_{i=1}^{d}\sum_{b\in\left\{ 0,1\right\} }(\alpha_{-})^{1-b}\ket{x_{i},b} ,
    \end{equation}
    where 
    \begin{equation}
        \alpha_{-}=\frac{\frac{2d}{d-1}\tilde{\epsilon}-\sqrt{\left(\frac{2d}{d-1}\tilde{\epsilon}\right)^{2}+\left(1-\tilde{\epsilon}^{2}\right)^{2}}}{1-\tilde{\epsilon}^{2}}.
    \end{equation}
    The $\lambda_{3}(M)$-eigenspace is $\left(d-1\right)$-dimensional and
    spanned by
    \begin{align}
        \ket{v_{i}^{\left(3\right)}} & =\frac{1}{2}\left(\sum_{b\in\left\{ 0,1\right\} }\left(-1\right)^{b}\ket{x_{1},b}-\sum_{b\in\left\{ 0,1\right\} }\left(-1\right)^{1-b}\ket{x_{i},b}\right)\text{ for }i=2,\dots,d \, .\\
    \end{align}
    The $\lambda_{4}(M)$-eigenspace is $\left(d-1\right)$-dimensional and
    spanned by
    \begin{align}
        \ket{v_{i}^{\left(4\right)}} & =\frac{1}{2}\left(\sum_{b\in\left\{ 0,1\right\} }\ket{x_{1},b}-\sum_{b\in\left\{ 0,1\right\} }\ket{x_{i},b}\right)\text{ for }i=2,\dots,d \, .
    \end{align}
    Hence, we can calculate the von Neumann entropy of $\rho_{\mathcal{D}_a}$ for any $a$ directly. It is given by
     \begin{align}
    S\left( \rho_{\mathcal{D}_a} \right) & =-\frac{\lambda_{1}}{2d}\log\left(\frac{\lambda_{1}}{2d}\right)-\frac{\lambda_{2}}{2d}\log\left(\frac{\lambda_{2}}{2d}\right)-\frac{(d-1)\lambda_{3}}{2d}\log\left(\frac{\lambda_{3}}{2d}\right)\\
     & =1+\frac{1}{2}\log\left(d\right)-\frac{d}{2\left(d-1\right)}\log\left(d\right) \tilde{\epsilon}^{2}+O\left(\tilde{\epsilon}^{4}\right), \label{eq:entropy_rho_a}
    \end{align}
    where the last line was obtained via a series expansion in $\tilde{\epsilon}$.
    
    Now, plugging $S(\overline{\rho})$ from \cref{eq:entropy_average_state} and $S\left(\rho_{\mathcal{D}_a} \right)$  from \cref{eq:entropy_rho_a} into the expression for the mutual information  $I(A;B_1)_\rho$ in \cref{eq:mutual_info_holevo_information}, we obtain
    \begin{equation}
        I(A;B_1)_\rho  =\frac{d}{2\left(d-1\right)}\log\left(d\right)\tilde{\epsilon}^{2}+O\left(\tilde{\epsilon}^{4}\right)
        =O\left(\epsilon^{2}\log d \right)
    \end{equation}
    Hence, we have an overall upper bound on the mutual information 
    \begin{equation}
        I(A;B_1,\ldots ,B_m)_\rho\leq O( m \epsilon^{2} \log d) .
    \end{equation}
    Contrasting this with the lower bound found above that is required for agnostic learning, namely of $I(A;B_1,\ldots ,B_m)_\rho=\Omega(d)$, we hence find that the number of copies $m$ necessary to learn must be at least
    \begin{equation}
        m = \Omega\left( \frac{d}{\epsilon^{2} \log d}  \right) = \tilde{\Omega}\left( \frac{d}{\epsilon^{2}} \right).
    \end{equation}
    This proves the $d$-dependent part of the overall lower bound on the sample complexity of distribution-independent agnostic learning and hence completes the proof.
\end{proof}

\cite{hadiashar2023optimal} recently proposed an improved variant of the information-theoretic proof strategy of \cite{arunachalam2018optimal}. This variant allowed them to recover the optimal sample complexity lower bounds for (realizable and agnostic) quantum PAC learning from superposition examples. It would be interesting to explore whether ideas from \cite{hadiashar2023optimal} can also help to remove the logarithmic-in-$d$ factors from the sample complexity lower bound in \Cref{theorem:distribution-independent-lower-bound}.

With \Cref{theorem:distribution-independent-lower-bound}, we have seen that mixture-of-superpositions examples do not significantly impact the landscape of distribution-independent agnostic learning compared to their classical data counterpart when focusing on sample complexities.
Our next result is similar in spirit, it shows that mixture-of-superpositions examples are also not information-theoretically more powerful than classical examples for verification of learning. Namely, when we consider interactive verification of distribution-independent learning, we match the classical upper and lower bounds of \cite{mutreja2022pac-verification} for this task.

\begin{theorem}\label{theorem:distribution-independent-verification-lower-bound}
    Let $\mathcal{F}\subseteq\{0,1\}^{\mathcal{X}_n}$ be a benchmark class with VC-dimension $\operatorname{VC}(\mathcal{F})=d\geq 1$.
    Assume that $(V,P)$ is an interactive classical-quantum verifier-prover pair that $1$-agnostic verifies $\mathcal{F}$ with accuracy parameter $\varepsilon=\nicefrac{1}{3}$ and confidence parameter $\delta=\nicefrac{1}{3}$.
    If we assume that $V$ uses $m_V$ classical random examples and $P$ uses $m_P$ mixture-of-superpositions quantum examples, then $m_V\geq\Omega(\sqrt{d})$, independently of $m_P$.
\end{theorem}
\begin{proof}
    As already argued in the proof of \Cref{theorem:limitation-improvement-distributional-agnostic-verification}, the reduction strategy used in \cite[Theorem 8]{mutreja2022pac-verification} is also immediately applicable to a scenario with a quantum prover, because the relevant quantum mixture-of-superpositions state $\rho_{\mathcal{U}_{n+1}}$ is completely known and because we can, due to the focus on sample complexity, ignore computational efficiency issues arising from classically simulating a quantum computation.
    Thus, we get the lower bound exactly as in \cite[Theorem 8]{mutreja2022pac-verification} (with tiny corrections to the chosen constants, see the proof of \Cref{theorem:limitation-improvement-distributional-agnostic-verification}).
\end{proof}

We note that, as a consequence of \cite[Theorem 9]{mutreja2022pac-verification}, the lower bound of \Cref{theorem:distribution-independent-verification-lower-bound} cannot be further improved in general. 
To see this, note the following two facts:
On the one hand, recall that mixture-of-superpositions examples can simulate classical random examples via computational basis measurements, so the honest quantum prover can in particular play the role of the classical prover in \cite[Protocol 1]{mutreja2022pac-verification}.
On the other hand, the argument for soundness in proving \cite[Theorem 9]{mutreja2022pac-verification} did not rely on the prover being classical and also applies to quantum provers with quantum data access.

\section*{Acknowledgments}

First and foremost, we thank Jens Eisert for his valuable input to discussions on this project and for many helpful suggestions for improving the draft.
We also thank Srinivasan Arunachalam, Jack O'Connor, Yihui Quek, Jonathan Shafer, and Thomas Vidick for insightful discussions.

The authors gratefully acknowledge support from the BMWK (PlanQK, EniQmA), the BMBF (Hybrid), and the Munich Quantum Valley (K-8). This work has also been funded by the Deutsche Forschungsgemeinschaft (DFG) under Germany’s Excellence Strategy, The Berlin Mathematics Research Center MATH+ (EXC-2046/1, project ID: 390685689) as well as CRC 183 (B1).
MCC was supported by a DAAD PRIME fellowship.
The Institute for Quantum Information and Matter is an NSF Physics Frontiers Center.

\newpage

\printbibliography

@inproceedings{goldwasser2021interactive,
  title={Interactive proofs for verifying machine learning},
  author={Goldwasser, Shafi and Rothblum, Guy N and Shafer, Jonathan and Yehudayoff, Amir},
  booktitle={12th Innovations in Theoretical Computer Science Conference (ITCS 2021)},
  year={2021},
  organization={Schloss Dagstuhl-Leibniz-Zentrum f{\"u}r Informatik},
  doi={10.4230/LIPIcs.ITCS.2021.41}
}

@article{goldwasser2020interactive-full-version,
  author    = {Shafi Goldwasser and
               Guy N. Rothblum and
               Jonathan Shafer and
               Amir Yehudayoff},
  title     = {Interactive Proofs for Verifying Machine Learning},
  journal   = {Electron. Colloquium Comput. Complex.},
  volume    = {20},
  number    = {58},
  year      = {2020},
  url       = {https://eccc.weizmann.ac.il/report/2020/058}
}

@article{feldman2009agnostic,
  title={On agnostic learning of parities, monomials, and halfspaces},
  author={Feldman, Vitaly and Gopalan, Parikshit and Khot, Subhash and Ponnuswami, Ashok Kumar},
  journal={SIAM Journal on Computing},
  volume={39},
  number={2},
  pages={606--645},
  year={2009},
  publisher={SIAM},
  doi={10.1137/070684914}
}

@article{feldman2009power,
  author  = {Vitaly Feldman},
  title   = {On The Power of Membership Queries in Agnostic Learning},
  journal = {Journal of Machine Learning Research},
  year    = {2009},
  volume  = {10},
  number  = {7},
  pages   = {163--182},
  url     = {http://jmlr.org/papers/v10/feldman09a.html}
}

@article{kanade2019learning,
 author = {Kanade, Varun and Rocchetto, Andrea and Severini, Simone},
 year = {2019},
 title = {Learning DNFs under product distributions via {$\mu$}-biased quantum Fourier sampling},
 pages = {1261--1278},
 volume = {19},
 number = {15{\&}16},
 journal = {Quantum Information {\&} Computation}, 
 doi={10.26421/QIC19.15-16}
}

@article{caro2020pseudo,
  title={Pseudo-dimension of quantum circuits},
  author={Caro, Matthias C. and Datta, Ishaun},
  journal={Quantum Machine Intelligence},
  volume={2},
  pages={14},
  year={2020},
  publisher={Springer},
  doi={10.1007/s42484-020-00027-5}
}

@article{caro2021encodingdependent,
  doi = {10.22331/q-2021-11-17-582},
  title = {Encoding-dependent generalization bounds for parametrized quantum circuits},
  author = {Caro, Matthias C. and Gil-Fuster, Elies and Meyer, Johannes Jakob and Eisert, Jens and Sweke, Ryan},
  journal = {{Quantum}},
  publisher = {{Verein zur F{\"{o}}rderung des Open Access Publizierens in den Quantenwissenschaften}},
  volume = {5},
  pages = {582},
  year = {2021}
}

@article{caro2022generalization,
  title={Generalization in quantum machine learning from few training data},
  author={Caro, Matthias C and Huang, Hsin-Yuan and Cerezo, Marco and Sharma, Kunal and Sornborger, Andrew and Cincio, Lukasz and Coles, Patrick J},
  journal={Nature {C}ommunications},
  volume={13},
  eid={4919},
  year={2022},
  publisher={Nature Publishing Group},
  doi={10.1038/s41467-022-32550-3}
}

@article{caro2021binary,
	author={Caro, Matthias C.},
	title={Binary Classification with Classical Instances and Quantum Labels},
	journal={Quantum Machine Intelligence},
	year={2021},
	volume={3},
	eid={18},
	doi={10.1007/s42484-021-00043-z}
}

@book{nielsen2000quantum,
	author = {Michael A. Nielsen and Isaac L. Chuang},
	year = {2000},
	title = {Quantum Computation and Quantum Information},
	publisher = {Cambridge University Press}
}

@article{fuchs1999cryptographic,
  title={Cryptographic distinguishability measures for quantum-mechanical states},
  author={Fuchs, Christopher A and Van De Graaf, Jeroen},
  journal={IEEE Transactions on Information Theory},
  volume={45},
  number={4},
  pages={1216--1227},
  year={1999},
  publisher={IEEE},
  doi={10.1109/18.761271}
}

@article{caro2020quantum,
  title={Quantum learning Boolean linear functions w.r.t.~product distributions},
  author={Caro, Matthias C.},
  journal={Quantum Information Processing},
  volume={19},
  eid={172},
  year={2020},
  publisher={Springer},
  doi={10.1007/s11128-020-02661-1}
}

@article{cross2015quantum,
  title={Quantum learning robust against noise},
  author={Cross, Andrew W and Smith, Graeme and Smolin, John A},
  journal={Physical Review A},
  volume={92},
  number={1},
  pages={012327},
  year={2015},
  publisher={APS},
  doi={10.1103/PhysRevA.92.012327}
}

@article{grilo2019learning,
  title={Learning-with-errors problem is easy with quantum samples},
  author={Grilo, Alex B and Kerenidis, Iordanis and Zijlstra, Timo},
  journal={Physical Review A},
  volume={99},
  number={3},
  pages={032314},
  year={2019},
  publisher={APS},
  doi={10.1103/PhysRevA.99.032314}
}

@article{caro2022out-of-distribution,
      title={Out-of-distribution generalization for learning quantum dynamics}, 
      author={Caro, Matthias C. and Huang, Hsin-Yuan and Ezzell, Nicholas and Gibbs, Joe and Sornborger, Andrew T. and Cincio, Lukasz and Coles, Patrick J. and Holmes, Zoë},
      journal={Nature Communications},
      volume={14},
      eid={3751},
      year={2023},
      publisher={Nature Publishing Group},
      doi={10.1038/s41467-023-39381-w}
}

@article{banchi2021generalization,
  title={Generalization in Quantum Machine Learning: A Quantum Information Standpoint},
  author={Banchi, Leonardo and Pereira, Jason and Pirandola, Stefano},
  journal={PRX Quantum},
  volume={2},
  number={4},
  pages={040321},
  year={2021},
  publisher={APS},
  doi={10.1103/PRXQuantum.2.040321}
}

@online{wilde2011bookarxiv,
  title={From Classical to Quantum Shannon Theory},
  author={Mark M. Wilde},
  date={2019-07-14},
  version=8,
  eprinttype={arxiv},
  eprint={1106.1445},
  primaryClass={quant-ph}
}

@article{sweke2021quantum,
  title={On the quantum versus classical learnability of discrete distributions},
  author={Sweke, Ryan and Seifert, Jean-Pierre and Hangleiter, Dominik and Eisert, Jens},
  journal={Quantum},
  volume={5},
  pages={417},
  year={2021},
  publisher={Verein zur F{\"o}rderung des Open Access Publizierens in den Quantenwissenschaften},
  doi={10.22331/q-2021-03-23-417}
}

@article{haussler1992decision,
  title={Decision theoretic generalizations of the PAC model for neural net and other learning applications},
  author={Haussler, David},
  journal={Information and computation},
  volume={100},
  number={1},
  pages={78--150},
  year={1992},
  publisher={Elsevier},
  doi={10.1016/0890-5401(92)90010-D}
}

@article{vapnik1971uniform,
 author = {Vapnik, Vladimir N. and Chervonenkis, Alexei Ya.},
 year = {1971},
 title = {{On the Uniform Convergence of Relative Frequencies of Events to Their Probabilities}},
 pages = {264--280},
 volume = {16},
 number = {2},
 journal = {{Theory of Probability {\&} Its Applications}},
 doi = {10.1137/1116025}
}

@article{blumer1989learnability,
  title={Learnability and the Vapnik-Chervonenkis dimension},
  author={Blumer, Anselm and Ehrenfeucht, Andrzej and Haussler, David and Warmuth, Manfred K.},
  journal={Journal of the ACM (JACM)},
  volume={36},
  number={4},
  pages={929--965},
  year={1989},
  publisher={ACM New York, NY, USA},
  doi={10.1145/76359.76371}
}

@article{kearns1998efficient,
  title={Efficient noise-tolerant learning from statistical queries},
  author={Kearns, Michael},
  journal={Journal of the ACM (JACM)},
  volume={45},
  number={6},
  pages={983--1006},
  year={1998},
  publisher={ACM New York, NY, USA},
  doi={10.1145/293347.293351}
}

@article{du2022efficient,
  title={Efficient measure for the expressivity of variational quantum algorithms},
  author={Du, Yuxuan and Tu, Zhuozhuo and Yuan, Xiao and Tao, Dacheng},
  journal={Physical Review Letters},
  volume={128},
  number={8},
  pages={080506},
  year={2022},
  publisher={APS},
  doi={10.1103/PhysRevLett.128.080506}
}

@article{abbas2021power,
  title={The power of quantum neural networks},
  author={Abbas, Amira and Sutter, David and Zoufal, Christa and Lucchi, Aur{\'e}lien and Figalli, Alessio and Woerner, Stefan},
  journal={Nature Computational Science},
  volume={1},
  number={6},
  pages={403--409},
  year={2021},
  publisher={Nature Publishing Group},
  doi={10.1038/s43588-021-00084-1}
}

@misc{Mahadev,
   year={2018},
      eprint={1804.01082},
      archivePrefix={arXiv},
title={Classical verification of quantum computations},
author={U. Mahadev}}

@article{kearns1994toward,
author = {Kearns, Michael J. and Schapire, Robert E. and Sellie, Linda M.},
title = {Toward Efficient Agnostic Learning},
year = {1994},
issue_date = {Nov./Dec. 1994},
publisher = {Kluwer Academic Publishers},
address = {USA},
volume = {17},
number = {2–3},
issn = {0885-6125},
url = {https://doi.org/10.1007/BF00993468},
doi = {10.1007/BF00993468},
journal = {Mach. Learn.},
month = {11},
pages = {115–141},
numpages = {27}
}

@misc{arunachalam2020qsq,
      title={Quantum statistical query learning}, 
      author={Srinivasan Arunachalam and Alex B. Grilo and Henry Yuen},
      year={2020},
      eprint={2002.08240},
      archivePrefix={arXiv},
      primaryClass={quant-ph}
}

@inproceedings{rubinfeld2022testing,
author = {Rubinfeld, Ronitt and Vasilyan, Arsen},
title = {Testing Distributional Assumptions of Learning Algorithms},
year = {2023},
isbn = {9781450399135},
publisher = {Association for Computing Machinery},
address = {New York, NY, USA},
url = {https://doi.org/10.1145/3564246.3585117},
doi = {10.1145/3564246.3585117},
booktitle = {Proceedings of the 55th Annual ACM Symposium on Theory of Computing},
pages = {1643–1656},
numpages = {14},
location = {Orlando, FL, USA},
series = {STOC 2023}
}

@inproceedings{gollakota2022momentmatching,
author = {Gollakota, Aravind and Klivans, Adam R. and Kothari, Pravesh K.},
title = {A Moment-Matching Approach to Testable Learning and a New Characterization of Rademacher Complexity},
year = {2023},
isbn = {9781450399135},
publisher = {Association for Computing Machinery},
address = {New York, NY, USA},
url = {https://doi.org/10.1145/3564246.3585206},
doi = {10.1145/3564246.3585206},
booktitle = {Proceedings of the 55th Annual ACM Symposium on Theory of Computing},
pages = {1657–1670},
numpages = {14},
location = {Orlando, FL, USA},
series = {STOC 2023}
}

@article{gopalan2011testing,
  title={Testing Fourier dimensionality and sparsity},
  author={Gopalan, Parikshit and O'Donnell, Ryan and Servedio, Rocco A and Shpilka, Amir and Wimmer, Karl},
  journal={SIAM Journal on Computing},
  volume={40},
  number={4},
  pages={1075--1100},
  year={2011},
  publisher={SIAM},
  doi={10.1137/100785429}
}

@book{odonnell2014analysis,
  title={Analysis of boolean functions},
  author={O'Donnell, Ryan},
  year={2014},
  publisher={Cambridge University Press},
  doi={10.1017/CBO9781139814782}
}

@article{arunachalam2021twonewresults,
  doi = {10.22331/q-2021-11-24-587},
  url = {https://doi.org/10.22331/q-2021-11-24-587},
  title = {Two new results about quantum exact learning},
  author = {Arunachalam, Srinivasan and Chakraborty, Sourav and Lee, Troy and Paraashar, Manaswi and de Wolf, Ronald},
  journal = {{Quantum}},
  issn = {2521-327X},
  publisher = {{Verein zur F{\"{o}}rderung des Open Access Publizierens in den Quantenwissenschaften}},
  volume = {5},
  pages = {587},
  month = nov,
  year = {2021}
}

@inproceedings{gopalan2008agnostically,
  title={Agnostically learning decision trees},
  author={Gopalan, Parikshit and Kalai, Adam Tauman and Klivans, Adam R},
  booktitle={Proceedings of the fortieth annual ACM symposium on Theory of computing},
  pages={527--536},
  year={2008},
  doi={10.1145/1374376.1374451}
}

@incollection{mcdiarmid1989onthemethod,
	Author = {McDiarmid, Colin},
	Booktitle = {Surveys in combinatorics, 1989 ({N}orwich, 1989)},
	Pages = {148--188},
	Publisher = {Cambridge Univ. Press, Cambridge},
	Series = {London Math. Soc. Lecture Note Ser.},
	Title = {On the method of bounded differences},
	Volume = {141},
	Year = {1989},
	doi={10.1017/CBO9781107359949.008}
}

@article{valiant1984theory,
 author = {Valiant, Leslie G.},
 year = {1984},
 title = {{A Theory of the Learnable}},
 pages = {1134--1142},
 volume = {27},
 number = {11},
 issn = {00010782},
 journal = {{Communications of the ACM}},
 doi = {10.1145/1968.1972}
}

@article{talagrand1994sharper,
  title={Sharper bounds for Gaussian and empirical processes},
  author={Talagrand, Michel},
  journal={The Annals of Probability},
  pages={28--76},
  year={1994},
  publisher={Institute of Mathematical Statistics},
  doi={10.1214/aop/1176988847}
}

@inproceedings{kalai2008agnostic,
author = {Kalai, Adam Tauman and Mansour, Yishay and Verbin, Elad},
title = {On Agnostic Boosting and Parity Learning},
year = {2008},
isbn = {9781605580470},
publisher = {Association for Computing Machinery},
address = {New York, NY, USA},
url = {https://doi.org/10.1145/1374376.1374466},
doi = {10.1145/1374376.1374466},
booktitle = {Proceedings of the Fortieth Annual ACM Symposium on Theory of Computing},
pages = {629–638},
numpages = {10},
location = {Victoria, British Columbia, Canada},
series = {STOC '08}
}

@inproceedings{kanade2009potential,
 author = {Kanade, Varun and Kalai, Adam},
 booktitle = {Advances in Neural Information Processing Systems},
 editor = {Y. Bengio and D. Schuurmans and J. Lafferty and C. Williams and A. Culotta},
 pages = {},
 publisher = {Curran Associates, Inc.},
 title = {Potential-Based Agnostic Boosting},
 url = {https://proceedings.neurips.cc/paper/2009/file/13f9896df61279c928f19721878fac41-Paper.pdf},
 volume = {22},
 year = {2009}
}

@inproceedings{feldman2010agnostic-boosting,
author = {Feldman, Vitaly},
title = {Distribution-Specific Agnostic Boosting},
isbn = {9787302217527},
year = {2010},
publisher = {Tsinghua University Press},
address = {Tsinghua University, Beijing, China},
url = {https://arxiv.org/abs/0909.2927},
booktitle = {Proceedings Innovations in Computer Science - ICS 2010},
pages = {241-250},
numpages = {10},
location = {Tsinghua University, Beijing, China},
series = {ICS 2010}
}

@article{bshouty1998learning,
 author = {Bshouty, Nader H. and Jackson, Jeffrey C.},
 year = {1998},
 title = {Learning {DNF} over the Uniform Distribution Using a Quantum Example Oracle},
 pages = {1136--1153},
 volume = {28},
 number = {3},
 journal = {SIAM Journal on Computing},
 doi = {10.1137/S0097539795293123}
}

@article{arunachalam2017survey,
 author = {Arunachalam, Srinivasan and de Wolf, Ronald},
 year = {2017},
 title = {Guest Column: A Survey of Quantum Learning Theory},
 volume = {48},
 journal = {SIGACT News},
 doi = {10.1145/3106700.3106710}
}

@InProceedings{chung2021sample,
  author =	{Chung, Kai-Min and Lin, Han-Hsuan},
  title =	{{Sample Efficient Algorithms for Learning Quantum Channels in PAC Model and the Approximate State Discrimination Problem}},
  booktitle =	{16th Conference on the Theory of Quantum Computation, Communication and Cryptography (TQC 2021)},
  pages =	{3:1--3:22},
  series =	{Leibniz International Proceedings in Informatics (LIPIcs)},
  ISBN =	{978-3-95977-198-6},
  ISSN =	{1868-8969},
  year =	{2021},
  volume =	{197},
  editor =	{Hsieh, Min-Hsiu},
  publisher =	{Schloss Dagstuhl -- Leibniz-Zentrum f{\"u}r Informatik},
  address =	{Dagstuhl, Germany},
  doi =		{10.4230/LIPIcs.TQC.2021.3}
}

@misc{fanizza2022learning,
      title={Learning quantum processes without input control}, 
      author={Marco Fanizza and Yihui Quek and Matteo Rosati},
      year={2022},
      eprint={2211.05005},
      archivePrefix={arXiv},
      primaryClass={quant-ph}
}

@misc{bera2022efficient,
      title={Efficient Quantum Agnostic Improper Learning of Decision Trees}, 
      author={Debajyoti Bera and Sagnik Chatterjee},
      year={2022},
      eprint={2210.00212},
      archivePrefix={arXiv},
      primaryClass={quant-ph}
}

@misc{hadiashar2023optimal,
      title={Optimal lower bounds for Quantum Learning via Information Theory}, 
      author={Shima Bab Hadiashar and Ashwin Nayak and Pulkit Sinha},
      year={2023},
      eprint={2301.02227},
      archivePrefix={arXiv},
      primaryClass={quant-ph}
}

@book{kosorok2008introduction,
  title={Introduction to empirical processes and semiparametric inference.},
  author={Kosorok, Michael R},
  year={2008},
  publisher={Springer},
  doi={https://doi.org/10.1007/978-0-387-74978-5}
}

@article{dvoretzky1956asymptotic,
  title={Asymptotic minimax character of the sample distribution function and of the classical multinomial estimator},
  author={Dvoretzky, Aryeh and Kiefer, Jack and Wolfowitz, Jacob},
  journal={The Annals of Mathematical Statistics},
  volume={27},
  number={3},
  pages={642--669},
  year={1956},
  publisher={Institute of Mathematical Statistics},
  url={https://www.jstor.org/stable/2237374}
}

@article{preskill2018quantum,
  title={Quantum computing in the NISQ era and beyond},
  author={Preskill, John},
  journal={Quantum},
  volume={2},
  pages={79},
  year={2018},
  publisher={Verein zur F{\"o}rderung des Open Access Publizierens in den Quantenwissenschaften},
  doi={10.22331/q-2018-08-06-79}
}

@article{massart1990tight,
  title={The tight constant in the Dvoretzky-Kiefer-Wolfowitz inequality},
  author={Massart, Pascal},
  journal={The Annals of Probability},
  volume={18},
  number={3},
  pages={1269--1283},
  year={1990},
  publisher={Institute of Mathematical Statistics},
  url={https://www.jstor.org/stable/2244426}
}

@article{arunachalam2018optimal,
  author  = {Srinivasan Arunachalam and Ronald de Wolf},
  title   = {Optimal Quantum Sample Complexity of Learning Algorithms},
  journal = {Journal of Machine Learning Research},
  year    = {2018},
  volume  = {19},
  number  = {71},
  pages   = {1--36},
  url     = {http://jmlr.org/papers/v19/18-195.html}
}

@incollection{mansour1994learning,
  title={Learning Boolean functions via the Fourier transform},
  author={Mansour, Yishay},
  booktitle={Theoretical advances in neural computation and learning},
  pages={391--424},
  year={1994},
  publisher={Springer},
  doi={10.1007/978-1-4615-2696-4_11}
}

@article{kushilevitz1993learning,
author = {Kushilevitz, Eyal and Mansour, Yishay},
title = {Learning Decision Trees Using the Fourier Spectrum},
journal = {SIAM Journal on Computing},
volume = {22},
number = {6},
pages = {1331-1348},
year = {1993},
doi = {10.1137/0222080},
URL = {https://doi.org/10.1137/0222080},
eprint = {https://doi.org/10.1137/0222080}
}

@article{linial1993constant,
  title={Constant depth circuits, Fourier transform, and learnability},
  author={Linial, Nathan and Mansour, Yishay and Nisan, Noam},
  journal={Journal of the ACM (JACM)},
  volume={40},
  number={3},
  pages={607--620},
  year={1993},
  publisher={ACM New York, NY, USA},
  doi={10.1145/174130.174138}
}

@inproceedings{goldreich1989hard,
  title={A hard-core predicate for all one-way functions},
  author={Goldreich, Oded and Levin, Leonid A},
  booktitle={Proceedings of the twenty-first annual ACM symposium on Theory of computing},
  pages={25--32},
  year={1989},
  doi={10.1145/73007.73010}
}

@article{huang2020predicting,
  title={Predicting many properties of a quantum system from very few measurements},
  author={Huang, Hsin-Yuan and Kueng, Richard and Preskill, John},
  journal={Nature Physics},
  volume={16},
  number={10},
  pages={1050--1057},
  year={2020},
  publisher={Nature Publishing Group},
  doi={10.1038/s41567-020-0932-7}
}

@article{aaronson2019shadow,
  title={Shadow tomography of quantum states},
  author={Aaronson, Scott},
  journal={SIAM Journal on Computing},
  volume={49},
  number={5},
  pages={STOC18--368},
  year={2019},
  publisher={SIAM},
  doi={10.1137/18M120275X}
}

@article{cheng2016learnability,
author = {Cheng, Hao-Chung and Hsieh, Min-Hsiu and Yeh, Ping-Cheng},
title = {The Learnability of Unknown Quantum Measurements},
year = {2016},
issue_date = {May 2016},
publisher = {Rinton Press, Incorporated},
address = {Paramus, NJ},
volume = {16},
number = {7–8},
issn = {1533-7146},
journal = {Quantum Info. Comput.},
month = may,
pages = {615–656},
numpages = {42},
doi = {10.5555/3179466.3179470}
}

@article{servedio2004equivalences,
  title={Equivalences and Separations Between Quantum and Classical Learnability},
  author={Servedio, Rocco A. and Gortler, Steven J.},
  journal={SIAM Journal on Computing},
  volume={33},
  number={5},
  pages={1067--1092},
  year={2004},
  publisher={SIAM},
  doi={10.1137/S0097539704412910}
}

@article{montanaro2012quantum,
  title={The quantum query complexity of learning multilinear polynomials},
  author={Montanaro, Ashley},
  journal={Information Processing Letters},
  volume={112},
  number={11},
  pages={438--442},
  year={2012},
  publisher={Elsevier},
  doi={10.1016/j.ipl.2012.03.002}
}

@article{atici2007quantum,
 author = {At{\i}c{\i}, Alp and Servedio, Rocco A.},
 year = {2007},
 title = {Quantum Algorithms for Learning and Testing Juntas},
 journal={Quantum Information Processing},
 pages = {323--348},
 volume = {6},
 number = {5},
 issn = {1570-0755},
 doi = {10.1007/s11128-007-0061-6}
}

@inproceedings{jackson2002quantum,
  title={Quantum DNF learnability revisited},
  author={Jackson, Jeffrey C. and Tamon, Christino and Yamakami, Tomoyuki},
  booktitle={International Computing and Combinatorics Conference},
  pages={595--604},
  year={2002},
  organization={Springer},
  doi={10.1007/3-540-45655-4_63}
}

@article{atici2005improved,
  title={Improved Bounds on Quantum Learning Algorithms},
  author={At{\i}c{\i}, Alp and Servedio, Rocco A.},
  journal={Quantum Information Processing},
  volume={4},
  number={5},
  pages={355--386},
  year={2005},
  publisher={Springer},
  doi={10.1007/s11128-005-0001-2}
}

@article{zhang2010improved,
  title={An improved lower bound on query complexity for quantum PAC learning},
  author={Zhang, Chi},
  journal={Information Processing Letters},
  volume={111},
  number={1},
  pages={40--45},
  year={2010},
  publisher={Elsevier},
  doi={10.1016/j.ipl.2010.10.007}
}

@article{aaronson2007learnability,
author = {Aaronson, Scott},
title = {The learnability of quantum states},
journal = {Proceedings of the Royal Society A: Mathematical, Physical and Engineering Sciences},
volume = {463},
number = {2088},
pages = {3089-3114},
year = {2007},
doi = {10.1098/rspa.2007.0113},
URL = {https://royalsocietypublishing.org/doi/abs/10.1098/rspa.2007.0113}
}

@inproceedings{eskenazis2022learning,
author = {Eskenazis, Alexandros and Ivanisvili, Paata},
title = {Learning Low-Degree Functions from a Logarithmic Number of Random Queries},
year = {2022},
isbn = {9781450392648},
publisher = {Association for Computing Machinery},
address = {New York, NY, USA},
url = {https://doi.org/10.1145/3519935.3519981},
doi = {10.1145/3519935.3519981},
booktitle = {Proceedings of the 54th Annual ACM SIGACT Symposium on Theory of Computing},
pages = {203–207},
numpages = {5},
location = {Rome, Italy},
series = {STOC 2022}
}

@article{bernstein1997quantum,
author = {Bernstein, Ethan and Vazirani, Umesh},
title = {Quantum Complexity Theory},
journal = {SIAM Journal on Computing},
volume = {26},
number = {5},
pages = {1411-1473},
year = {1997},
doi = {10.1137/S0097539796300921}
}

@misc{huang2023learning,
      title={Learning to predict arbitrary quantum processes}, 
      author={Hsin-Yuan Huang and Sitan Chen and John Preskill},
      year={2023},
      eprint={2210.14894},
      archivePrefix={arXiv},
      primaryClass={quant-ph}
}

@misc{chen2022complexity,
      title={The Complexity of NISQ}, 
      author={Sitan Chen and Jordan Cotler and Hsin-Yuan Huang and Jerry Li},
      year={2022},
      eprint={2210.07234},
      archivePrefix={arXiv},
      primaryClass={quant-ph}
}

@misc{caro2022learning,
      title={Learning Quantum Processes and Hamiltonians via the Pauli Transfer Matrix}, 
      author={Matthias C. Caro},
      year={2022},
      eprint={2212.04471},
      archivePrefix={arXiv},
      primaryClass={quant-ph}
}

@article{huang2021information,
  title={Information-theoretic bounds on quantum advantage in machine learning},
  author={Huang, Hsin-Yuan and Kueng, Richard and Preskill, John},
  journal={Physical Review Letters},
  volume={126},
  number={19},
  pages={190505},
  year={2021},
  publisher={APS},
  doi={10.1103/PhysRevLett.126.190505}
}

@article{aharonov2022quantum,
  title={Quantum algorithmic measurement},
  author={Aharonov, Dorit and Cotler, Jordan and Qi, Xiao-Liang},
  journal={Nature Communications},
  volume={13},
  number={1},
  pages={1--9},
  year={2022},
  publisher={Nature Publishing Group},
  doi={10.1038/s41467-021-27922-0},
  url={https://www.nature.com/articles/s41467-021-27922-0}
}

@inproceedings{chen2022exponential,
  title={Exponential separations between learning with and without quantum memory},
  author={Chen, Sitan and Cotler, Jordan and Huang, Hsin-Yuan and Li, Jerry},
  booktitle={2021 IEEE 62nd Annual Symposium on Foundations of Computer Science (FOCS)},
  pages={574--585},
  year={2022},
  organization={IEEE},
  doi={10.1109/FOCS52979.2021.00063}
}

@article{huang2022quantum-advantage,
  title={Quantum advantage in learning from experiments},
  author={Huang, Hsin-Yuan and Broughton, Michael and Cotler, Jordan and Chen, Sitan and Li, Jerry and Mohseni, Masoud and Neven, Hartmut and Babbush, Ryan and Kueng, Richard and Preskill, John and others},
  journal={Science},
  volume={376},
  number={6598},
  pages={1182--1186},
  year={2022},
  publisher={American Association for the Advancement of Science},
  doi={10.1126/science.abn7293}
}

@book{dudley1999uniform, 
 title={Uniform central limit theorems}, 
 publisher={Cambridge University Press}, 
 author={Dudley, Richard M.}, 
 year={1999}
}

@inproceedings{blum1994weakly,
  title={Weakly learning DNF and characterizing statistical query learning using Fourier analysis},
  author={Blum, Avrim and Furst, Merrick and Jackson, Jeffrey and Kearns, Michael and Mansour, Yishay and Rudich, Steven},
  booktitle={Proceedings of the twenty-sixth annual ACM symposium on Theory of computing},
  pages={253--262},
  year={1994},
  doi={10.1145/195058.195147}
}

@article{bshouty1996fourier,
  title={On the Fourier spectrum of monotone functions},
  author={Bshouty, Nader H and Tamon, Christino},
  journal={Journal of the ACM (JACM)},
  volume={43},
  number={4},
  pages={747--770},
  year={1996},
  publisher={ACM New York, NY, USA},
  doi={10.1145/234533.234564}
}

@InProceedings{fefferman2018quantum-vs-classical,
  author =	{Bill Fefferman and Shelby Kimmel},
  title =	{{Quantum vs. Classical Proofs and Subset Verification}},
  booktitle =	{43rd International Symposium on Mathematical Foundations  of Computer Science (MFCS 2018)},
  pages =	{22:1--22:23},
  series =	{Leibniz International Proceedings in Informatics (LIPIcs)},
  ISBN =	{978-3-95977-086-6},
  ISSN =	{1868-8969},
  year =	{2018},
  volume =	{117},
  editor =	{Igor Potapov and Paul Spirakis and James Worrell},
  publisher =	{Schloss Dagstuhl--Leibniz-Zentrum fuer Informatik},
  address =	{Dagstuhl, Germany},
  URL =		{http://drops.dagstuhl.de/opus/volltexte/2018/9604},
  URN =		{urn:nbn:de:0030-drops-96040},
  doi =		{10.4230/LIPIcs.MFCS.2018.22}
}

@InProceedings{natarajan2022distribution,
  author =	{Natarajan, Anand and Nirkhe, Chinmay},
  title =	{{A Distribution Testing Oracle Separating QMA and QCMA}},
  booktitle =	{38th Computational Complexity Conference (CCC 2023)},
  pages =	{22:1--22:27},
  series =	{Leibniz International Proceedings in Informatics (LIPIcs)},
  ISBN =	{978-3-95977-282-2},
  ISSN =	{1868-8969},
  year =	{2023},
  volume =	{264},
  editor =	{Ta-Shma, Amnon},
  publisher =	{Schloss Dagstuhl -- Leibniz-Zentrum f{\"u}r Informatik},
  address =	{Dagstuhl, Germany},
  URL =		{https://drops.dagstuhl.de/opus/volltexte/2023/18292},
  URN =		{urn:nbn:de:0030-drops-182928},
  doi =		{10.4230/LIPIcs.CCC.2023.22}
}

@article{gheorghiu2019verification,
  title={Verification of quantum computation: An overview of existing approaches},
  author={Gheorghiu, Alexandru and Kapourniotis, Theodoros and Kashefi, Elham},
  journal={Theory of computing systems},
  volume={63},
  pages={715--808},
  year={2019},
  publisher={Springer},
  doi={10.1007/s00224-018-9872-3}
}

@article{chen2022unitarity,
  author={Chen, Kean and Wang, Qisheng and Long, Peixun and Ying, Mingsheng},
  journal={IEEE Transactions on Information Theory},
  title={Unitarity Estimation for Quantum Channels}, 
  year={2023},
  volume={69},
  number={8},
  pages={5116-5134},
  doi={10.1109/TIT.2023.3263645}
}

@article{harrow2014uselessness,
  author    = {Aram W. Harrow and
               David J. Rosenbaum},
  title     = {Uselessness for an Oracle model with internal randomness},
  journal   = {Quantum Inf. Comput.},
  volume    = {14},
  number    = {7-8},
  pages     = {608--624},
  year      = {2014},
  url       = {https://doi.org/10.26421/QIC14.7-8-5},
  doi       = {10.26421/QIC14.7-8-5}
}

@InProceedings{bassirian2022power,
  author =	{Bassirian, Roozbeh and Fefferman, Bill and Marwaha, Kunal},
  title =	{{On the Power of Nonstandard Quantum Oracles}},
  booktitle =	{18th Conference on the Theory of Quantum Computation, Communication and Cryptography (TQC 2023)},
  pages =	{11:1--11:25},
  series =	{Leibniz International Proceedings in Informatics (LIPIcs)},
  ISBN =	{978-3-95977-283-9},
  ISSN =	{1868-8969},
  year =	{2023},
  volume =	{266},
  editor =	{Fawzi, Omar and Walter, Michael},
  publisher =	{Schloss Dagstuhl -- Leibniz-Zentrum f{\"u}r Informatik},
  address =	{Dagstuhl, Germany},
  URL =		{https://drops.dagstuhl.de/opus/volltexte/2023/18321},
  URN =		{urn:nbn:de:0030-drops-183215},
  doi =		{10.4230/LIPIcs.TQC.2023.11}
}

@InProceedings{hu2022comparative,
  author =	{Hu, Lunjia and Peale, Charlotte},
  title =	{{Comparative Learning: A Sample Complexity Theory for Two Hypothesis Classes}},
  booktitle =	{14th Innovations in Theoretical Computer Science Conference (ITCS 2023)},
  pages =	{72:1--72:30},
  series =	{Leibniz International Proceedings in Informatics (LIPIcs)},
  ISBN =	{978-3-95977-263-1},
  ISSN =	{1868-8969},
  year =	{2023},
  volume =	{251},
  editor =	{Tauman Kalai, Yael},
  publisher =	{Schloss Dagstuhl -- Leibniz-Zentrum f{\"u}r Informatik},
  address =	{Dagstuhl, Germany},
  URL =		{https://drops.dagstuhl.de/opus/volltexte/2023/17575},
  URN =		{urn:nbn:de:0030-drops-175752},
  doi =		{10.4230/LIPIcs.ITCS.2023.72}
}

@InProceedings{mutreja2022pac-verification,
  title = 	 {PAC Verification of Statistical Algorithms},
  author =       {Mutreja, Saachi and Shafer, Jonathan},
  booktitle = 	 {Proceedings of Thirty Sixth Conference on Learning Theory},
  pages = 	 {5021--5043},
  year = 	 {2023},
  editor = 	 {Neu, Gergely and Rosasco, Lorenzo},
  volume = 	 {195},
  series = 	 {Proceedings of Machine Learning Research},
  month = 	 {Jul},
  publisher =    {PMLR},
  url = 	 {https://proceedings.mlr.press/v195/mutreja23a.html}
}

@inproceedings{canetti2021covert,
  title={Covert Learning: How to Learn with an Untrusted Intermediary},
  author={Canetti, Ran and Karchmer, Ari},
  booktitle={Theory of Cryptography Conference},
  pages={1--31},
  year={2021},
  organization={Springer},
  doi={10.1007/978-3-030-90456-2_1}
}

@inproceedings{herman2022verifying,
  title={Verifying the unseen: interactive proofs for label-invariant distribution properties},
  author={Herman, Tal and Rothblum, Guy N},
  booktitle={Proceedings of the 54th Annual ACM SIGACT Symposium on Theory of Computing},
  pages={1208--1219},
  year={2022},
  doi={10.1145/3519935.3519987}
}

@InProceedings{chiesa2018proofs,
  author =	{Alessandro Chiesa and Tom Gur},
  title =	{{Proofs of Proximity for Distribution Testing}},
  booktitle =	{9th Innovations in Theoretical Computer Science Conference (ITCS 2018)},
  pages =	{53:1--53:14},
  series =	{Leibniz International Proceedings in Informatics (LIPIcs)},
  ISBN =	{978-3-95977-060-6},
  ISSN =	{1868-8969},
  year =	{2018},
  volume =	{94},
  editor =	{Anna R. Karlin},
  publisher =	{Schloss Dagstuhl--Leibniz-Zentrum fuer Informatik},
  address =	{Dagstuhl, Germany},
  URL =		{http://drops.dagstuhl.de/opus/volltexte/2018/8311},
  URN =		{urn:nbn:de:0030-drops-83114},
  doi =		{10.4230/LIPIcs.ITCS.2018.53}
}

@mastersthesis{oconnor2021delegating,
  author = {Jack O'Connor},
  title = {Delegating Machine Learning with Succinct Proofs},
  year = {2021},
  school = {University of Warwick}
}

@article{regev2009lattices,
  title={On lattices, learning with errors, random linear codes, and cryptography},
  author={Regev, Oded},
  journal={Journal of the ACM (JACM)},
  volume={56},
  number={6},
  pages={1--40},
  year={2009},
  publisher={ACM New York, NY, USA}
}

@inproceedings{pietrzak2012cryptography,
  title={Cryptography from Learning Parity with Noise.},
  author={Pietrzak, Krzysztof},
  booktitle={SOFSEM},
  volume={12},
  pages={99--114},
  year={2012},
  organization={Springer}
}

@misc{gollakota2023efficient,
      title={An Efficient Tester-Learner for Halfspaces},
      author={Aravind Gollakota and Adam R. Klivans and Konstantinos Stavropoulos and Arsen Vasilyan},
      year={2023},
      eprint={2302.14853},
      archivePrefix={arXiv},
      primaryClass={cs.LG}
}

@misc{diakonikolas2023efficient,
      title={Efficient Testable Learning of Halfspaces with Adversarial Label Noise}, 
      author={Ilias Diakonikolas and Daniel M. Kane and Vasilis Kontonis and Sihan Liu and Nikos Zarifis},
      year={2023},
      eprint={2303.05485},
      archivePrefix={arXiv},
      primaryClass={cs.LG}
}

@misc{gollakota2023testerlearners,
      title={Tester-Learners for Halfspaces: Universal Algorithms}, 
      author={Aravind Gollakota and Adam R. Klivans and Konstantinos Stavropoulos and Arsen Vasilyan},
      year={2023},
      eprint={2305.11765},
      archivePrefix={arXiv},
      primaryClass={cs.LG}
}

@article{fitzsimons2017private,
  title={Private quantum computation: an introduction to blind quantum computing and related protocols},
  author={Fitzsimons, Joseph F},
  journal={npj Quantum Information},
  volume={3},
  number={1},
  pages={23},
  year={2017},
  publisher={Nature Publishing Group UK London}
}

@inproceedings{Broadbent_2009,
	doi = {10.1109/focs.2009.36},
	url = {https://doi.org/10.1109%2Ffocs.2009.36},
	year = 2009,
	month = {10},
	publisher = {{IEEE}},
	author = {Anne Broadbent and Joseph Fitzsimons and Elham Kashefi},
	title = {Universal Blind Quantum Computation},
	booktitle = {2009 50th Annual {IEEE} Symposium on Foundations of Computer Science}
}

@inproceedings{kearns1994learnability,
  title={On the learnability of discrete distributions},
  author={Kearns, Michael and Mansour, Yishay and Ron, Dana and Rubinfeld, Ronitt and Schapire, Robert E and Sellie, Linda},
  booktitle={Proceedings of the twenty-sixth annual ACM symposium on Theory of computing},
  pages={273--282},
  year={1994}
}

\newpage

\appendix

\section{Auxiliary Results in Classical Computational Learning Theory}\label{appendix:useful}

In this appendix, we compile mostly well known results on how to obtain computational learning guarantees from a Fourier analysis perspective.
We begin with a standard fact: If $\mathcal{D}_{\mathcal{X}_n}=\mathcal{U}_n$, then the misclassification probability for a binary-valued hypothesis is determined by an inner product of Fourier coefficients or, equivalently, by an $\ell_2$-distance between Fourier coefficients.

\begin{lemma}\label{lemma:general-loss-explicit-inner-product}
    Let $\mathcal{D}=(\mathcal{U}_n, \varphi)$ be a probability distribution over $\mathcal{X}_n\times\{0,1\}$ with uniform input marginal.
    Let $f:\{0,1\}^n\to\{0,1\}$ and $g=(-1)^f$. 
    Then,
    \begin{align}
        \mathbb{P}_{(x,b)\sim \mathcal{D}}[b\neq f(x)]
        &= \frac{1-\langle \phi,g\rangle_{\mathcal{U}_n}}{2}\\
        &= \frac{1-\sum_{s\in\{0,1\}^n}\hat{g}(s)\hat{\phi}(s)}{2}\\
        &= \frac{1}{4} \left(\sum_{s\in\{0,1\}^n} \left(\hat{\phi}(s)-\hat{g}(s)\right)^2 + \left(1 - \mathbb{E}_{x\sim \mathcal{U}_n} \left[(\phi(x))^2\right]\right)\right).
    \end{align}
\end{lemma}
\begin{proof}
    This is a proof by computation. With our notation of $\phi = 1-2\varphi$ and $g=(-1)^f$, we get:
    \begin{align}
        \mathbb{P}_{(x,b)\sim \mathcal{D}}[b\neq f(x)]
        &= \mathbb{P}_{(x,b)\sim (\mathcal{U}_n, \phi)}[b\neq g(x)]\\
        &= \frac{1}{2^n}\sum_{x\in\{0,1\}^n} \sum_{b\in\{-1,1\}} \frac{1+b\phi(x)}{2}\cdot \mathds{1}_{\{g(x)\neq b\}} \\
        &= \frac{1}{2^n}\sum_{x\in\{0,1\}^n} \sum_{b\in\{-1,1\}} \frac{1+b\phi(x)}{2}\cdot\frac{1-b g(x)}{2}\\
        &= \frac{1}{2} - \frac{1}{2}\mathbb{E}_{x\sim\mathcal{U}_n}[\phi(x) g(x)] + \frac{1}{4}\underbrace{\sum_{b\in\{-1,1\}} b \mathbb{E}_{x\sim\mathcal{U}_n}[\phi(x) - g(x)]}_{=0}\\
        &= \frac{1-\langle \phi,g\rangle_{\mathcal{U}_n}}{2} \\
        &= \frac{1-\sum_{s\in\{0,1\}^n}\hat{g}(s)\hat{\phi}(s)}{2}\, .
    \end{align}
    Here, the last step holds by Plancherel's theorem for the Boolean Fourier transform.
    The last equality in the statement of the lemma now follows by Parseval's equality.
\end{proof}

We highlight a useful special case of this lemma
Namely, if $\mathcal{D}_{\mathcal{X}_n}=\mathcal{U}_n$, then agnostic parity learning is equivalent to identifying the largest Fourier coefficient of $\phi$.

\begin{lemma}\label{lemma:parity-loss-explicit}
    Let $\mathcal{D}=(\mathcal{U}_n, \varphi)$ be a probability distribution over $\mathcal{X}_n\times\{0,1\}$ with uniform input marginal.
    Let $s\in\{0,1\}^n$. 
    Then,
    \begin{equation}
        \mathbb{P}_{(x,b)\sim \mathcal{D}}[b\neq s\cdot x]
        = \frac{1-\hat{\phi} (s)}{2}\, .
    \end{equation}
\end{lemma}
\begin{proof}
    This is a special case of \Cref{lemma:general-loss-explicit-inner-product}.
\end{proof}

This well known characterization of the misclassification probability of a parity (compare, e.g., \cite[Eq.~(1)]{feldman2009power}) makes it easy to see that parity learning is possible as soon as one can identify an approximately heaviest Fourier coefficient. We emphasize that this is true even in the distributional agnostic case:

\begin{lemma}\label{lemma:parity-learning-via-heaviest-Fourier-coefficient}
    Let $\mathcal{D}=(\mathcal{U}_n, \varphi)$ be a probability distribution over $\mathcal{X}_n\times\{0,1\}$ with uniform input marginal.
    Let $\varepsilon\in (0,1)$. 
    If $s\in\{0,1\}^n$ is such that $\max_{t\in\{0,1\}^n} \hat{\phi}(t) - \hat{\phi}(s) \leq 2\varepsilon$, then also
    \begin{equation}
        \mathbb{P}_{(x,b)\sim\mathcal{D}}[b\neq s\cdot x]
        \leq \min\limits_{t\in\{0,1\}^n} \mathbb{P}_{(x,b)\sim\mathcal{D}}[b\neq t\cdot x] + \varepsilon\, .
    \end{equation}
    In particular, any procedure that, given $\delta,\varepsilon\in (0,1)$, outputs, with success probability $\geq 1-\delta$, an $(2\varepsilon)$-approximately-heaviest Fourier coefficient of $\phi$, immediately gives rise to a distributional $1$-agnostic parity learner.
\end{lemma}
\begin{proof}
    Let $s\in\{0,1\}^n$ be such that $\max_{t\in\{0,1\}^n} \hat{\phi}(t) - \hat{\phi}(s) \leq 2\varepsilon$. Then we have (using \Cref{lemma:parity-loss-explicit}):
    \begin{align}
        \mathbb{P}_{(x,b)\sim\mathcal{D}}[b\neq s\cdot x]
        &= \frac{1-\hat{\phi}(s)}{2}\\
        &= \frac{1- \max_{t\in\{0,1\}^n} \hat{\phi}(t)}{2} + \frac{\max_{t\in\{0,1\}^n} \hat{\phi}(t) - \hat{\phi}(s)}{2}\\
        &\leq \min_{t\in\{0,1\}^n} \frac{1- \hat{\phi}(t)}{2} + \varepsilon\\
        &= \min_{t\in\{0,1\}^n} \mathbb{P}_{(x,b)\sim\mathcal{D}}[b\neq t\cdot x] + \varepsilon ,
    \end{align}
    as claimed.
\end{proof}

From the proof of \Cref{lemma:parity-learning-via-heaviest-Fourier-coefficient}, we also see the following: Let $\mathcal{A}$ be a procedure that, given $\delta,\varepsilon\in (0,1)$, outputs, with success probability $\geq 1-\delta$, an $(2\varepsilon)$-approximately-largest Fourier coefficient of $\phi$. Suppose the information-theoretic complexity of $\mathcal{A}$ is $m_\mathcal{A} (n,  \delta, \varepsilon)$ and the classical computational complexity of $\mathcal{A}$ is $t_\mathcal{A} (n,  \delta, \varepsilon)$. Then, the resulting distributional agnostic parity learner $\mathcal{A}'$ has the same information-theoretic and computational complexity as $\mathcal{A}$.

For the remainder of this section, we are concerned with extending the idea of \Cref{lemma:parity-learning-via-heaviest-Fourier-coefficient} beyond parities to Fourier-sparse function. Here, when building a hypothesis from an approximation to the Fourier spectrum, we do not necessarily obtain a Boolean function. Thus, we need results in the spirit of \Cref{lemma:general-loss-explicit-inner-product} but without the restriction to $\{0,1\}$-valued hypotheses. 
As a first step in this direction, we recall that the misclassification probability of a thresholded version of a general $\mathbb{R}$-valued function can be controlled in terms of an $L_2$-error, compare \cite{mansour1994learning}.

\begin{lemma}\label{lemma:misclassification-probability-bound-L2}
    Let $\mathcal{D}=(\mathcal{U}_n, \varphi)$ be a probability distribution over $\mathcal{X}_n\times\{0,1\}$ with uniform input marginal.
    Let $f:\{0,1\}^n\to\mathbb{R}$ and $g = 1-2f$. 
    Then,
    \begingroup
    \allowdisplaybreaks
    \begin{align}
        \mathbb{P}_{(x,b)\sim \mathcal{D}}\left[b\neq \mathds{1}_{\left\{f(x) \geq \tfrac{1}{2}\right\}}\right]
        &\leq \left(1 + 2\min_{x\in\{0,1\}^n} \lvert f(x) - \tfrac{1}{2}\rvert\right)^{-2} \mathbb{E}_{(x,b)\sim (\mathcal{U}_n, \phi)} \left[\left(b- g(x)\right)^2\right]\\
        &= \left(1 + 2\min_{x\in\{0,1\}^n} \lvert f(x) - \tfrac{1}{2}\rvert\right)^{-2}\left(\sum_{s\in\{0,1\}^n} \left(\hat{\phi}(s)-\hat{g}(s)\right)^2 + \left(1 - \mathbb{E}_{x\sim \mathcal{U}_n} \left[(\phi(x))^2\right]\right)\right)\\
        &= \left(1 + 2\min_{x\in\{0,1\}^n} \lvert f(x) - \tfrac{1}{2}\rvert\right)^{-2} \left(1 + \mathbb{E}_{x\sim \mathcal{U}_n} [(g(x))^2] - 2\sum_{s\in\{0,1\}^n}\hat{g}(s)\hat{\phi}(s)\right) .
    \end{align}
    \endgroup
\end{lemma}
\begin{proof}
    First, we change the label from $\{0,1\}$ to $\{-1,1\}$. Observe that with this change,
    \begin{equation}
        \mathbb{P}_{(x,b)\sim \mathcal{D}}\left[b\neq \mathds{1}_{\left\{f(x) \geq \tfrac{1}{2}\right\}}\right]
        = \mathbb{P}_{(x,b)\sim (\mathcal{U}_n, \varphi)}\left[b\neq \mathds{1}_{\left\{f(x) \geq \tfrac{1}{2}\right\}}\right]
        = \mathbb{P}_{(x,b)\sim (\mathcal{U}_n, \phi)}\left[b\neq \operatorname{sgn}(1-2f(x))\right] .
    \end{equation}
    Here, we defined the sign function as
    \begin{equation}
        \operatorname{sgn}:\mathbb{R}\to \{-1,1\},~
        \operatorname{sgn}(\alpha) = \begin{cases} 1 \quad &\textrm{ if }\alpha > 0\\ -1 &\textrm{ if }\alpha \leq 0 \end{cases}.
    \end{equation}
    Next, we write $g = 1-2f$ and note that $\mathds{1}_{\{b \neq \operatorname{sgn}(g(x))\}}\leq \left(1 + 2\min_{x\in\{0,1\}^n} \lvert f(x) - \tfrac{1}{2}\rvert\right)^{-2} (b-g(x))^2$ holds for any $(x,b)\in\mathcal{X}_n\times\{-1,1\}$ because $b \neq \operatorname{sgn}(g(x))$ implies $\lvert b - g(x)\rvert\geq 1 + 2\min_{x\in\{0,1\}^n} \lvert f(x) - \tfrac{1}{2}\rvert$.
    Therefore, we have shown
    \begin{equation}
        \mathbb{P}_{(x,b)\sim \mathcal{D}}\left[b\neq \mathds{1}_{\left\{f(x) \geq \tfrac{1}{2}\right\}}\right]
        = \mathbb{E}_{(x,b)\sim (\mathcal{U}_n, \phi)}\left[\mathds{1}_{\{b \neq \operatorname{sgn}(g(x))\}}\right]
        \leq \left(1 + 2\min_{x\in\{0,1\}^n} \lvert f(x) - \tfrac{1}{2}\rvert\right)^{-2}\mathbb{E}_{(x,b)\sim (\mathcal{U}_n, \phi)}\left[(b-g(x))^2\right] ,
    \end{equation}
    the claimed inequality.
    To see the claimed equalities, we rewrite
    \begin{align}
        \mathbb{E}_{(x,b)\sim (\mathcal{U}_n, \phi)}\left[(b-g(x))^2\right]
        &= \mathbb{E}_{x\sim \mathcal{U}_n} \left[ \left(\frac{1-\phi(x)}{2}\right) (-1 - g(x))^2  + \left(\frac{1+\phi(x)}{2}\right) (1 - g(x))^2\right]\\
        &= \mathbb{E}_{x\sim \mathcal{U}_n} \left[ 1 + (g(x))^2 - 2 g(x) \phi(x)\right]\\
        &= \mathbb{E}_{x\sim \mathcal{U}_n} \left[ (\phi(x) - g(x))^2\right] + \left(1 - \mathbb{E}_{x\sim \mathcal{U}_n} \left[(\phi(x))^2\right]\right)\\
        &= \sum_{s\in\{0,1\}^n} \left(\hat{\phi}(s)-\hat{g}(s)\right)^2 + \left(1 - \mathbb{E}_{x\sim \mathcal{U}_n} \left[(\phi(x))^2\right]\right)\\
        &= 1 + \mathbb{E}_{x\sim \mathcal{U}_n} [(g(x))^2] - 2\sum_{s\in\{0,1\}^n}\hat{g}(s)\hat{\phi}(s) ,
    \end{align}
    where we used Parseval in the second-to-last and the last step.
\end{proof}

If $f$ is $\{0,1\}$-valued, then the upper bound of \Cref{lemma:misclassification-probability-bound-L2} coincides with the exact expression from \Cref{lemma:general-loss-explicit-inner-product}. In that sense, it is a tight extension of \Cref{lemma:general-loss-explicit-inner-product} to non-Boolean hypotheses.

In \Cref{lemma:misclassification-probability-bound-L2}, we still insist on the hypothesis being given by a deterministic function.
If we allow for randomized hypotheses, there is also the following useful variant of \Cref{lemma:misclassification-probability-bound-L2}:

\begin{lemma}\label{lemma:misclassification-probability-bound-L2-probabilistic-hypothesis}
    Let $\mathcal{D}=(\mathcal{U}_n, \varphi)$ be a probability distribution over $\mathcal{X}_n\times\{0,1\}$ with uniform input marginal.
    Let $g:\{0,1\}^n\to\mathbb{R}$. 
    Define a randomized hypothesis $h:\{0,1\}^n\to\{0,1\}$ as follows: Independently for each $x\in\{0,1\}^n$, $h(x)=1$ with probability $p(x) = \tfrac{(1 - g(x))^2}{2 (1+ (g(x))^2)}$ and $h(x)=0$ with probability $1-p(x)$. 
    Then
    \begingroup
    \allowdisplaybreaks
    \begin{align}
        \mathbb{P}_{(x,b)\sim \mathcal{D};h}\left[b\neq h(x)\right]
        &\leq \left(  2 (1+ \min_{x\in\{0,1\}^n}(g(x))^2)\right)^{-1} \mathbb{E}_{(x,b)\sim (\mathcal{U}_n, \phi)} \left[\left(b- g(x)\right)^2\right]\\
        &= \left(  2 (1+ \min_{x\in\{0,1\}^n}(g(x))^2)\right)^{-1}\left(\sum_{s\in\{0,1\}^n} \left(\hat{\phi}(s)-\hat{g}(s)\right)^2 + \left(1 - \mathbb{E}_{x\sim \mathcal{U}_n} \left[(\phi(x))^2\right]\right)\right)\\
        &= \left(  2 (1+ \min_{x\in\{0,1\}^n}(g(x))^2)\right)^{-1} \left(1 + \mathbb{E}_{x\sim \mathcal{U}_n} [(g(x))^2] - 2\sum_{s\in\{0,1\}^n}\hat{g}(s)\hat{\phi}(s)\right) .
    \end{align}
    \endgroup
\end{lemma}
\begin{proof}
    First, we note that $p(x)\in [0,1]$ for all $x\in\{0,1\}^n$, so we indeed have well-defined probabilities.
    Using the given expression for $p(x)$, we get $1-p(x) = 1 - \tfrac{(1 - g(x))^2}{2 (1+ (g(x))^2)} = \tfrac{(-1-g(x))^2}{2 (1+ (g(x))^2)}$. With this, see that, for any fixed $x\in\{0,1\}^n$ and $b\in\{0,1\}$
    \begin{equation}
        \mathbb{P}_{h}\left[b\neq h(x)\right]
        = \mathbb{P}_{h}\left[(-1)^b\neq (-1)^{h(x)}\right]
        = \frac{((-1)^b - g(x))^2}{2 (1+ (g(x))^2)} .
    \end{equation}
    If we now additionally consider the random choice of $(x,b)\sim\mathcal{D}$, we obtain an overall misclassification probability of
    \begin{align}
        \mathbb{P}_{(x,b)\sim \mathcal{D};h}\left[b\neq h(x)\right]
        &= \mathbb{E}_{(x,b)\sim \mathcal{D}}\left[ \mathbb{P}_{h}\left[b\neq h(x)\right]\right]\\
        &= \mathbb{E}_{(x,b)\sim (\mathcal{U}_n, \phi)}\left[ \frac{(b - g(x))^2}{2 (1+ (g(x))^2)}\right] \\
        &\leq \left(  2 (1+ \min_{x\in\{0,1\}^n}(g(x))^2)\right)^{-1} \mathbb{E}_{(x,b)\sim (\mathcal{U}_n, \phi)} \left[\left(b- g(x)\right)^2\right] .
    \end{align}
    The remaining two equalities follow from rewritings for $\mathbb{E}_{(x,b)\sim (\mathcal{U}_n, \phi)} \left[\left(b- g(x)\right)^2\right]$ that have already been observed in \Cref{lemma:misclassification-probability-bound-L2}.
\end{proof}

\Cref{lemma:misclassification-probability-bound-L2-probabilistic-hypothesis} constitutes a slight generalization of \cite[Lemma 3]{blum1994weakly}, see also \cite[Section 5]{bshouty1996fourier}.
Similarly to how \Cref{lemma:parity-loss-explicit} implied \Cref{lemma:parity-learning-via-heaviest-Fourier-coefficient}, \Cref{lemma:misclassification-probability-bound-L2-probabilistic-hypothesis} tells us that finding $k$ approximately heaviest Fourier coefficients is sufficient for distributional $\alpha$-agnostic Fourier-$k$-sparse learning:

\begin{lemma}\label{lemma:Fourier-sparse-learning-via-heaviest-Fourier-coefficients}
    Let $\mathcal{D}=(\mathcal{U}_n, \varphi)$ be a probability distribution over $\mathcal{X}_n\times\{0,1\}$ with uniform input marginal.
    Let $\varepsilon\in (0,1)$ and let $k\in \{1,\ldots, 2^n\}$.
    Let $t_1,\ldots,t_k\in\{0,1\}^n$ be $k$ heaviest Fourier coefficients of $\phi$. That is, let $t_1\in \operatorname{argmax}_{t\in\{0,1\}^n} \lvert \hat{\phi}(t)\rvert$, and for $2\leq \ell\leq k$, let $t_\ell\in \operatorname{argmax}_{t\in\{0,1\}^n\setminus\{t_1,\ldots,t_{\ell -1}\}} \lvert \hat{\phi}(t)\rvert$. 
    Also, let $s_1,\ldots,s_k\in\{0,1\}^n$ be such that $\left\lvert  \lvert\hat{\phi}(t_\ell)\rvert - \lvert\hat{\phi}(s_\ell)\rvert\right\rvert \leq \nicefrac{\varepsilon}{2k} $ holds for every $1\leq \ell\leq k$ and let $\Tilde{\phi}(s_\ell)$ be $(\nicefrac{\varepsilon}{2k})$-accurate estimates of the respective Fourier coefficients, that is $\lvert \Tilde{\phi}(s_\ell) - \hat{\phi}(s_\ell)\rvert\leq \nicefrac{\varepsilon}{2k}$ holds for every $1\leq \ell\leq k$.
    Then, if we define 
    $g:\{0,1\}^n\to\mathbb{R}$, $g(x) = \sum_{\ell=1}^k \tilde{\phi}(s_\ell)\chi_{s_\ell}$, and take the randomized hypothesis $h:\{0,1\}^n\to \{0,1\}$ as in \Cref{lemma:misclassification-probability-bound-L2-probabilistic-hypothesis}, 
    \begin{equation}
        \mathbb{P}_{(x,b)\sim \mathcal{D};h}\left[b\neq h(x)\right]
        \leq \frac{2}{1+ \min_{x\in\{0,1\}^n}(g(x))^2} \min\limits_{\substack{\tilde{f}:\mathcal{X}_n\to\{0,1\}\\\mathrm{Fourier-}k\mathrm{-sparse}}} \mathbb{P}_{(x,b)\sim\mathcal{D}}[b\neq \tilde{f}(x)] + \varepsilon.
    \end{equation}
    In particular, any procedure that, given $\delta,\varepsilon\in (0,1)$, outputs, with success probability $\geq 1-\delta$, an $(\nicefrac{\varepsilon}{2k})$-accurate estimates of $k$ $(\nicefrac{\varepsilon}{2k})$-approximately-heaviest Fourier coefficients of $\phi$, immediately gives rise to a distributional $2$-agnostic Fourier-$k$-sparse learner.
\end{lemma}
\begin{proof}
    For notational convenience, we write $\Tilde{\varepsilon} = \nicefrac{\varepsilon}{2k}$.
    We first use \Cref{lemma:misclassification-probability-bound-L2-probabilistic-hypothesis} to upper bound the misclassification probability of $h$ as follows:
    \begin{align}
        &\mathbb{P}_{(x,b)\sim \mathcal{D};h}\left[b\neq h(x)\right]\\
        &\leq \left(  2 (1+ \min_{x\in\{0,1\}^n}(g(x))^2)\right)^{-1}\left(\sum_{s\in\{0,1\}^n} \left(\hat{\phi}(s)-\hat{g}(s)\right)^2 + \left(1 - \mathbb{E}_{x\sim \mathcal{U}_n} \left[(\phi(x))^2\right]\right)\right)\\
        &= \left(  2 (1+ \min_{x\in\{0,1\}^n}(g(x))^2)\right)^{-1}\left(1 +  \sum_{\ell = 1}^k \left(\hat{\phi}(s_\ell)-\tilde{\phi}(s_\ell)\right)^2 +  \sum_{s\in\{0,1\}^n\setminus\{s_1,\ldots,s_\ell\}} \left(\hat{\phi}(s)\right)^2  - \mathbb{E}_{x\sim \mathcal{U}_n} \left[(\phi(x))^2\right]\right)\\
        &\leq \left(  2 (1+ \min_{x\in\{0,1\}^n}(g(x))^2)\right)^{-1}\left(1 +  k\vartheta^2 -   \sum_{\ell = 1}^k  \left(\hat{\phi}(s_\ell)\right)^2\right) ,
    \end{align}    
    where the last step used that $\mathbb{E}_{x\sim \mathcal{U}_n} \left[(\phi(x))^2\right] = \sum_{s\in\{0,1\}^n} (\hat{\phi}(s))^2$ holds by Parseval as well as the guarantee that $\lvert \Tilde{\phi}(s_\ell) - \hat{\phi}(s_\ell)\rvert\leq \Tilde{\varepsilon}$ holds for every $1\leq \ell\leq k$.

    Next, we lower bound the optimal misclassification probability achievable by Fourier-$k$-sparse Boolean functions, using \Cref{lemma:general-loss-explicit-inner-product} and Cauchy-Schwarz:
    \begin{align}
        \min\limits_{\substack{\tilde{f}:\mathcal{X}_n\to\{0,1\}\\\mathrm{Fourier-}k\mathrm{-sparse}}} \mathbb{P}_{(x,b)\sim\mathcal{D}}[b\neq \tilde{f}(x)]
        &= \min\limits_{\substack{\tilde{f}:\mathcal{X}_n\to\{0,1\}\\\mathrm{Fourier-}k\mathrm{-sparse}}} \frac{1-\sum_{s\in\{0,1\}^n}\hat{\tilde{g}}(s)\hat{\phi}(s)}{2}\\
        &\geq \frac{1}{2}\left(1 - \max\limits_{\hat{h}\in [-1,1]^{2^n}:\norm{\hat{h}}_0\leq k~\wedge ~ \norm{\hat{h}}_2= 1}\sum_{s\in\{0,1\}^n}\hat{h}(s)\hat{\phi}(s) \right) \label{eq:inter_step}\\
        &\geq \frac{1}{2}\left(1 - \sqrt{\sum_{\ell =1}^k (\hat{\phi}(t_\ell))^2} \right)\label{eq:inter_step2}\\
        &\geq \frac{1}{4}\left(1 - \sum_{\ell =1}^k (\hat{\phi}(t_\ell))^2 \right) .
    \end{align}
    Here, the last step uses that $1-\sqrt{x}\geq \tfrac{1-x}{2}$ holds for all $x\in [0,1]$.
    Using that $\hat{\phi}(s)\in [-1,1]$ for all $s\in\{0,1\}^n$ and that $[-1,1]\ni\xi\mapsto \xi^2$ is $2$-Lipschitz as well as our assumption that $\left\lvert  \lvert\hat{\phi}(t_\ell)\rvert - \lvert\hat{\phi}(s_\ell)\rvert\right\rvert \leq \Tilde{\varepsilon} $ holds for every $1\leq \ell\leq k$, we see that 
    \begin{equation}\label{eq:lipschitz}
        \left\lvert \sum_{\ell = 1}^k  \left(\hat{\phi}(t_\ell)\right)^2 - \sum_{\ell = 1}^k  \left(\hat{\phi}(s_\ell)\right)^2 \right\rvert
        \leq 2 k \Tilde{\varepsilon},
    \end{equation}
so our lower bound becomes
    \begin{equation}
        \min\limits_{\substack{\tilde{f}:\mathcal{X}_n\to\{0,1\}\\\mathrm{Fourier-}k\mathrm{-sparse}}} \mathbb{P}_{(x,b)\sim\mathcal{D}}[b\neq \tilde{f}(x)]
        \geq \frac{1}{4}\left(1 -   \sum_{\ell = 1}^k  \left(\hat{\phi}(s_\ell)\right)^2 \right) - \frac{k}{2} \Tilde{\varepsilon} .
    \end{equation}

    Combining our upper and lower bounds, we have shown:
    \begin{align}
        \mathbb{P}_{(x,b)\sim \mathcal{D};h}\left[b\neq h(x)\right]
        &\leq \left(  2 (1+ \min_{x\in\{0,1\}^n}(g(x))^2)\right)^{-1}\left(1 +  k\tilde{\varepsilon}^2 -   \sum_{\ell = 1}^k  \left(\hat{\phi}(s_\ell)\right)^2\right)\\
        &\leq \left(  2 (1+ \min_{x\in\{0,1\}^n}(g(x))^2)\right)^{-1}\left( 4\left(\min\limits_{\substack{\tilde{f}:\mathcal{X}_n\to\{0,1\}\\\mathrm{Fourier-}k\mathrm{-sparse}}} \mathbb{P}_{(x,b)\sim\mathcal{D}}[b\neq \tilde{f}(x)] + \frac{k}{2}\Tilde{\varepsilon}\right) + k\Tilde{\varepsilon}^2\right)\\
        &\leq \frac{2}{1+ \min_{x\in\{0,1\}^n}(g(x))^2} \min\limits_{\substack{\tilde{f}:\mathcal{X}_n\to\{0,1\}\\\mathrm{Fourier-}k\mathrm{-sparse}}} \mathbb{P}_{(x,b)\sim\mathcal{D}}[b\neq \tilde{f}(x)] + 2k \Tilde{\varepsilon}\\
        &\leq \frac{2}{1+ \min_{x\in\{0,1\}^n}(g(x))^2} \min\limits_{\substack{\tilde{f}:\mathcal{X}_n\to\{0,1\}\\\mathrm{Fourier-}k\mathrm{-sparse}}} \mathbb{P}_{(x,b)\sim\mathcal{D}}[b\neq \tilde{f}(x)] + \varepsilon ,
    \end{align}
    where the second to last step uses that $\tfrac{2}{1+ \min_{x\in\{0,1\}^n}(g(x))^2}\leq 2$ and the last step is by our choice of $\Tilde{\varepsilon}$.
\end{proof}

From the proof of \Cref{lemma:Fourier-sparse-learning-via-heaviest-Fourier-coefficients}, we also see the following: Let $\mathcal{A}$ be a procedure that, given $\delta,\varepsilon\in (0,1)$, outputs, with success probability $\geq 1-\delta$, $(\nicefrac{\varepsilon}{2k})$-accurate estimates of $k$ $(\nicefrac{\varepsilon}{2k})$-approximately-heaviest Fourier coefficients of $\phi$. Suppose the information-theoretic complexity of $\mathcal{A}$ is $m_\mathcal{A} (n,k,  \delta, \varepsilon)$ and the classical computational complexity of $\mathcal{A}$ is $t_\mathcal{A} (n,k,  \delta, \varepsilon)$. Then, the resulting distributional $2$-agnostic Fourier-$k$-sparse learner $\mathcal{A}'$ has information-theoretic complexity $m_\mathcal{A'} (n,k,  \delta, \varepsilon) = m_\mathcal{A} (n,k,  \delta, \nicefrac{\varepsilon}{2k})$ and classical computational complexity $t_\mathcal{A} (n,k,  \delta, \varepsilon) = t_\mathcal{A} (n,k,  \delta, \nicefrac{\varepsilon}{2k}) + k$. This is only a minor increase in complexity.

Next, we note a variant of \Cref{lemma:Fourier-sparse-learning-via-heaviest-Fourier-coefficients} with a deterministic hypothesis, obtained by replacing \Cref{lemma:misclassification-probability-bound-L2-probabilistic-hypothesis} by \Cref{lemma:misclassification-probability-bound-L2} in the reasoning above.

\begin{lemma}\label{lemma:Fourier-sparse-learning-via-heaviest-Fourier-coefficients-no-randomization}
    Let $\mathcal{D}=(\mathcal{U}_n, \varphi)$ be a probability distribution over $\mathcal{X}_n\times\{0,1\}$ with uniform input marginal.
    Let $\varepsilon\in (0,1)$ and let $k\in \{1,\ldots, 2^n\}$.
    Let $t_1,\ldots,t_k\in\{0,1\}^n$ be $k$ heaviest Fourier coefficients of $\phi$. That is, let $t_1\in \operatorname{argmax}_{t\in\{0,1\}^n} \lvert \hat{\phi}(t)\rvert$, and for $2\leq \ell\leq k$, let $t_\ell\in \operatorname{argmax}_{t\in\{0,1\}^n\setminus\{t_1,\ldots,t_{\ell -1}\}} \lvert \hat{\phi}(t)\rvert$. 
    Also, let $s_1,\ldots,s_k\in\{0,1\}^n$ be such that $\left\lvert\lvert \hat{\phi}(t_\ell)\rvert - \lvert\hat{\phi}(s_\ell)\rvert \right\rvert\leq \nicefrac{\varepsilon}{3k} $ holds for every $1\leq \ell\leq k$ and let $\Tilde{\phi}(s_\ell)$ be $(\nicefrac{\varepsilon}{3k})$-accurate estimates of the respective Fourier coefficients, that is $\lvert \Tilde{\phi}(s_\ell) - \hat{\phi}(s_\ell)\rvert\leq \nicefrac{\varepsilon}{3k}$ holds for every $1\leq \ell\leq k$.
    Then, if we define $f:\{0,1\}^n\to\mathbb{R}$ as $f=\tfrac{1-g}{2}$, where $g:\{0,1\}^n\to\mathbb{R}$, $g(x) = \sum_{\ell=1}^k \tilde{\phi}(s_\ell)\chi_{s_\ell}$,
    \begin{equation}
        \mathbb{P}_{(x,b)\sim\mathcal{D}} \left[b\neq \mathds{1}_{\left\{f(x) \geq \tfrac{1}{2}\right\}}\right]
        \leq \frac{4}{\left(1 + 2\min_{x\in\{0,1\}^n} \lvert f(x) - \tfrac{1}{2}\rvert\right)^{2}} \min\limits_{\substack{\tilde{f}:\mathcal{X}_n\to\{0,1\}\\\mathrm{Fourier-}k\mathrm{-sparse}}} \mathbb{P}_{(x,b)\sim\mathcal{D}}[b\neq \tilde{f}(x)] + \varepsilon.
    \end{equation}
    In particular, any procedure that, given $\delta,\varepsilon\in (0,1)$, outputs, with success probability $\geq 1-\delta$, an $(\nicefrac{\varepsilon}{3k})$-accurate estimates of the $k$ $(\nicefrac{\varepsilon}{3k})$-approximately-heaviest Fourier coefficients of $\phi$ immediately gives rise to a distributional $2$-agnostic Fourier-$k$-sparse learner.
\end{lemma}
\begin{proof}
    Follow the reasoning used to prove \Cref{lemma:Fourier-sparse-learning-via-heaviest-Fourier-coefficients}, replacing \Cref{lemma:misclassification-probability-bound-L2-probabilistic-hypothesis} by \Cref{lemma:misclassification-probability-bound-L2}.
\end{proof}

The prefactor of $4\left(1 + 2\min_{x\in\{0,1\}^n} \lvert f(x) - \tfrac{1}{2}\rvert\right)^{-2}$ in front of the optimal achievable risk in \Cref{lemma:Fourier-sparse-learning-via-heaviest-Fourier-coefficients-no-randomization} is always upper bounded by $4$, but it may be strictly larger than $2$. Thus, this only leads to a distributional $4$-agnostic Fourier-$k$-sparse learner, instead of a $2$-agnostic one obtained from \Cref{lemma:Fourier-sparse-learning-via-heaviest-Fourier-coefficients}.  
A possible advantage of using \Cref{lemma:Fourier-sparse-learning-via-heaviest-Fourier-coefficients-no-randomization} instead over \Cref{lemma:Fourier-sparse-learning-via-heaviest-Fourier-coefficients} is that the resulting distributional $4$-agnostic Fourier-$k$-sparse learner outputs a deterministic hypothesis rather than a randomized one.

\section{Classical Distributional-to-Functional Agnostic Learning Reduction}\label{appendix:classical-distributional-to-agnostic}

As mentioned in the main text, our \Cref{definition:mixture-of-superpositions-quantum-example} is partially motivated by proofs of classical distributional-to-functional agnostic learning reductions. For instance, given a distribution $\mathcal{D}$ over $\mathcal{X}_n\times\{0,1\}$, the associated probability distribution $F_\mathcal{D}$ over $\{0,1\}^{\mathcal{X}_n}$ was used in \cite[Appendix A.1]{gopalan2008agnostically} to reduce general distributional agnostic learning to functional agnostic learning, working in a scenario of learning from membership queries. 
To connect this more closely to the scenario of learning from random examples that is our main focus, we next demonstrate that one can make use of similar ideas to employ $F_\mathcal{D}$ as an auxiliary tool for an analogous reduction in learning from random examples.

\begin{theorem}\label{theorem:distributional-to-functional-reduction-classical}
    Let $\mathcal{F}\subseteq\{0,1\}^{\{0,1\}^n}$ be a benchmark class, let $\delta,\varepsilon\in (0,1)$ and $m=m(n, \delta,\varepsilon)\in\mathbb{N}_{>0}$ .
    Suppose that there is a randomized algorithm $A$ that, given access to a set of $m$ random examples, $S^{f}_m=(x_1,f(x_1)),\ldots,(x_m,f(x_m))$ from $(\mathcal{U}_n,f)$ for any function $f:\{0,1\}^n\to\{0,1\}$, with success probability $\geq 1-\delta$, outputs a function $h:\{0,1\}^n\to\{0,1\}$ such that 
    \begin{equation}
        \mathbb{P}_{x\sim\mathcal{U}_n}[f(x)\neq h(x)]
        \leq \inf_{\Tilde{h}\in\mathcal{F}}\mathbb{P}_{x\sim\mathcal{U}_n}[f(x)\neq \tilde{h}(x)] + \varepsilon.
    \end{equation}
    Then there is a randomized algorithm $\tilde{\mathcal{A}}$ that, given access to a set of $\tilde{m}=m(n, \tfrac{\delta}{2},\tfrac{\varepsilon}{3})$ random examples, $S^{\mathcal{D}}_m=(x_1,y_1),\ldots,(x_{\tilde{m}},y_{\tilde{m}})$ from $\mathcal{D}$ for any distribution $\mathcal{D}$ over $\{0,1\}^n\times\{0,1\}$ with $\mathcal{D}\rvert_{\mathcal{X}_n}=\mathcal{U}_n$, with success probability $\geq 1- \tfrac{\delta}{2} - \mathcal{O}(\tilde{m}^2\cdot 2^{-n}) -  ( \lvert\mathcal{F}\rvert + 1) \cdot \exp\left( - \frac{2^{n+1} \varepsilon^2}{9}\right)$, outputs a function $h:\{0,1\}^n\to\{0,1\}$ such that
    \begin{equation}
        \mathbb{P}_{(x,y)\sim\mathcal{D}}[y\neq h(x)]
        \leq \inf_{\Tilde{h}\in\mathcal{F}} \mathbb{P}_{(x,y)\sim\mathcal{D}}[y\neq \tilde{h}(x)] + \varepsilon.
    \end{equation} 
    Moreover, if $\mathcal{A}$ has runtime $t=t(n,\varepsilon,\delta)$, then $\tilde{\mathcal{A}}$ has runtime $\tilde{t} = t(n,\tfrac{\delta}{2},\tfrac{\varepsilon}{3})$.
    Finally, if $\mathcal{A}$ is proper, then so is $\tilde{\mathcal{A}}$.
\end{theorem}
\begin{proof}
    For the purpose of the proof, we adopt the following notation: 
    For $f,h:\{0,1\}^n\to\{0,1\}$, we write $\operatorname{err}_f(h)\coloneqq \mathbb{P}_{x\sim\mathcal{U}_n}[f(x)\neq h(x)]$ and $\operatorname{opt}_f(\mathcal{F})=\inf_{\Tilde{h}\in\mathcal{F}}\operatorname{err}_f(\tilde{h})$.
    For a distribution $\mathcal{D}$ over $\{0,1\}^n\times\{0,1\}$ and for $h:\{0,1\}^n\to\{0,1\}$, we write $\operatorname{err}_{\mathcal{D}}(h)\coloneqq \mathbb{P}_{(x,y)\sim\mathcal{D}}[f(x)\neq h(x)]$ and $\operatorname{opt}_{\mathcal{D}}(\mathcal{F})=\inf_{\Tilde{h}\in\mathcal{F}}\operatorname{err}_{\mathcal{D}}(\tilde{h})$.
    
    The randomized algorithm $\tilde{\mathcal{A}}$ should simply consist in running $\mathcal{A}$ on the available data (possibly aborting if there is a collision). Note: Conditioned on there being no collision in the training data, it is statistically indistinguishable whether the training data was generated from $\mathcal{D}$ or from $(\mathcal{U}_n,f)$ for some randomly drawn $f\sim F(\mathcal{D})$. 
    Then, we can argue as follows:
    \begingroup
    \allowdisplaybreaks
    \small
    \begin{align}
        &\mathbb{P}_{S^{\mathcal{D}}_{\tilde{m}},\mathcal{A}}\left[\operatorname{err}_{\mathcal{D}}(h_{S^{\mathcal{D}}_{\tilde{m}},\mathcal{A}})> \operatorname{opt}_{\mathcal{D}}(\mathcal{F}) + \varepsilon\right]\\
        &\leq \mathcal{O}(\tilde{m}^2\cdot 2^{-n}) + \mathbb{P}_{S^{\mathcal{D}}_{\tilde{m}},\mathcal{A}}\left[\operatorname{err}_{\mathcal{D}}(h_{S^{\mathcal{D}}_m,\mathcal{A}})> \operatorname{opt}_{\mathcal{D}}(\mathcal{F}) + \varepsilon~|~\textrm{no collision}\right]\\
        &\leq \mathcal{O}(\tilde{m}^2\cdot 2^{-n}) + \mathbb{P}_{f\sim F(\mathcal{D}), S^{f}_{\Tilde{m}},\mathcal{A}}\left[\operatorname{err}_{\mathcal{D}}(h_{S^{f}_{\Tilde{m}},\mathcal{A}})> \operatorname{opt}_{\mathcal{D}}(\mathcal{F}) + \varepsilon\right]\\
        &\leq \mathcal{O}(\tilde{m}^2\cdot 2^{-n}) + \mathbb{P}_{f\sim F(\mathcal{D}), S^{f}_{\Tilde{m}},\mathcal{A}}\left[\operatorname{err}_{\mathcal{D}}(h_{S^{f}_{\Tilde{m}},\mathcal{A}})> \operatorname{opt}_{f}(\mathcal{F}) + \frac{2\varepsilon}{3}\right] + \mathbb{P}_{f\sim F(\mathcal{D})}\left[\operatorname{opt}_{f}(\mathcal{F}) > \operatorname{opt}_{\mathcal{D}}(\mathcal{F}) + \frac{\varepsilon}{3}\right]\\[20pt]
        \begin{split}
            &\leq \mathcal{O}(\tilde{m}^2\cdot 2^{-n}) + \mathbb{P}_{f\sim F(\mathcal{D}), S^{f}_{\Tilde{m}},\mathcal{A}}\left[\operatorname{err}_{f}(h_{S^{f}_{\Tilde{m}},\mathcal{A}})> \operatorname{opt}_{f}(\mathcal{F}) + \frac{\varepsilon}{3}\right] + \mathbb{P}_{f\sim F(\mathcal{D}), S^{f}_{\Tilde{m}},\mathcal{A}}\left[\operatorname{err}_{\mathcal{D}}(h_{S^{f}_{\Tilde{m}},\mathcal{A}}) > \operatorname{err}_{f}(h_{S^{f}_{\Tilde{m}},\mathcal{A}})  + \frac{\varepsilon}{3}\right] \\
            &\hphantom{\leq \mathcal{O}(\tilde{m}^2\cdot 2^{-n}) + } + \mathbb{P}_{f\sim F(\mathcal{D})}\left[\operatorname{opt}_{f}(\mathcal{F}) > \operatorname{opt}_{\mathcal{D}}(\mathcal{F}) + \frac{\varepsilon}{3}\right]
        \end{split}\\
        \begin{split}
            &= \mathcal{O}(\tilde{m}^2\cdot 2^{-n}) + \mathbb{E}_{f\sim F(\mathcal{D})}\left[\mathbb{P}_{ S^{f}_{\Tilde{m}},\mathcal{A}}\left[\operatorname{err}_{f}(h_{S^{f}_{\Tilde{m}},\mathcal{A}})> \operatorname{opt}_{f}(\mathcal{F}) + \frac{\varepsilon}{3}\right]\right] + \mathbb{P}_{f\sim F(\mathcal{D}), S^{f}_{\Tilde{m}},\mathcal{A}}\left[\operatorname{err}_{\mathcal{D}}(h_{S^{f}_{\Tilde{m}},\mathcal{A}}) > \operatorname{err}_{f}(h_{S^{f}_{\Tilde{m}},\mathcal{A}})  + \frac{\varepsilon}{3}\right]\\
            &\hphantom{\leq \mathcal{O}(\tilde{m}^2\cdot 2^{-n}) +} + \mathbb{P}_{f\sim F(\mathcal{D})}\left[\operatorname{opt}_{f}(\mathcal{F}) > \operatorname{opt}_{\mathcal{D}}(\mathcal{F}) + \frac{\varepsilon}{3}\right]
        \end{split}\\
        &\leq \mathcal{O}(\tilde{m}^2\cdot 2^{-n}) + \mathbb{E}_{f\sim F(\mathcal{D})}\left[\tfrac{\delta}{2}\right] + \mathbb{P}_{f\sim F(\mathcal{D}), S^{f}_{\Tilde{m}},\mathcal{A}}\left[\operatorname{err}_{\mathcal{D}}(h_{S^{f}_{\Tilde{m}},\mathcal{A}}) > \operatorname{err}_{f}(h_{S^{f}_{\Tilde{m}},\mathcal{A}})  + \frac{\varepsilon}{3}\right] + \mathbb{P}_{f\sim F(\mathcal{D})}\left[\operatorname{opt}_{f}(\mathcal{F}) > \operatorname{opt}_{\mathcal{D}}(\mathcal{F}) + \frac{\varepsilon}{3}\right]\\
        &= \mathcal{O}(\tilde{m}^2\cdot 2^{-n}) + \frac{\delta}{2} + \mathbb{P}_{f\sim F(\mathcal{D}), S^{f}_{\Tilde{m}},\mathcal{A}}\left[\operatorname{err}_{\mathcal{D}}(h_{S^{f}_{\Tilde{m}},\mathcal{A}}) > \operatorname{err}_{f}(h_{S^{f}_{\Tilde{m}},\mathcal{A}})  + \frac{\varepsilon}{3}\right] + \mathbb{P}_{f\sim F(\mathcal{D})}\left[\operatorname{opt}_{f}(\mathcal{F}) > \operatorname{opt}_{\mathcal{D}}(\mathcal{F}) + \frac{\varepsilon}{3}\right].
    \end{align}
    \normalsize
    \endgroup
    Next, observe that, since the expectation of an infimum is upper-bounded by the infimum over expectations, and using the definition of $F(\mathcal{D})$,
    \begin{align}
        \mathbb{E}_{f\sim F(\mathcal{D})}\left[\operatorname{opt}_{f}(\mathcal{F})\right]
        &= \mathbb{E}_{f\sim F(\mathcal{D})}\left[\inf_{\Tilde{h}\in\mathcal{F}}\operatorname{err}_f(\tilde{h})\right]\\
        &\leq \inf_{\Tilde{h}\in\mathcal{F}}\mathbb{E}_{f\sim F(\mathcal{D})}\left[\operatorname{err}_f(\tilde{h})\right]\\
        &= \inf_{\Tilde{h}\in\mathcal{F}}\operatorname{err}_{\mathcal{D}}(h)\\
        &= \operatorname{opt}_{\mathcal{D}}(\mathcal{F}).
    \end{align}
    This in particular tells us that
    \begin{equation}
        \mathbb{P}_{f\sim F(\mathcal{D})}\left[\operatorname{opt}_{f}(\mathcal{F}) > \operatorname{opt}_{\mathcal{D}}(\mathcal{F}) + \frac{\varepsilon}{3}\right]
        \leq \mathbb{P}_{f\sim F(\mathcal{D})}\left[\operatorname{opt}_{f}(\mathcal{F}) > \mathbb{E}_{f\sim F(\mathcal{D})}\left[\operatorname{opt}_{f}(\mathcal{F})\right] + \frac{\varepsilon}{3}\right].
    \end{equation}
    We claim that the latter probability is small by McDiarmid's bounded differences inequality \cite{mcdiarmid1989onthemethod}. Namely, we can write 
    \begin{equation}
        \operatorname{opt}_{f}(\mathcal{F})
        = \inf_{\Tilde{h}\in\mathcal{F}}\operatorname{err}_f(\tilde{h})
        = \inf_{\Tilde{h}\in\mathcal{F}} \frac{1}{2^n}\sum_{x\in\{0,1\}^n} (1-\delta_{f(x),\tilde{h}(x)})
        = \xi(\{Z_b\}_{b\in\{0,1\}^n}),
    \end{equation}
    where we defined the measurable function $\xi:\{0,1\}^{\{0,1\}^n}\to\mathbb{R}$, $\xi(\{z_b\}_{b\in\{0,1\}^n}) = \inf_{\Tilde{h}\in\mathcal{F}} \frac{1}{2^n}\sum_{x\in\{0,1\}^n} (1-\delta_{z_x,\tilde{h}(x)})$, and the $Z_b, b\in\{0,1\}^n$, are independent random variables with $\mathbb{P}[Z_b=0]=\mathbb{P}_{(x,y)\sim\mathcal{D}}[y=0|x=b]=1-\mathbb{P}[Z_b=1]$.
    Now, suppose that $\{z_b\}_{b\in\{0,1\}^n}$ and $\{z'_b\}_{b\in\{0,1\}^n}$ differ only in a single coordinate, $z_{b_0}\neq z'_{b_0}$ for some $b_0\in\{0,1\}^n$ but $z_b=z'_b$ for all $b\in\{0,1\}^n\setminus\{b_0\}$. Then, using that $\lvert \inf_A f - \inf_A g \rvert\leq \sup_A \lvert f-g\rvert$, we get
    \begin{align}
        \lvert \xi(\{z_b\}_{b\in\{0,1\}^n}) - \xi(\{z'_b\}_{b\in\{0,1\}^n})\rvert
        &= \left\lvert \inf_{\Tilde{h}\in\mathcal{F}} \frac{1}{2^n}\sum_{x\in\{0,1\}^n} (1-\delta_{z_x,\tilde{h}(x)}) - \inf_{\Tilde{h}\in\mathcal{F}} \frac{1}{2^n}\sum_{x\in\{0,1\}^n} (1-\delta_{z'_x,\tilde{h}(x)})\right\rvert\\
        &\leq \sup_{\Tilde{h}\in\mathcal{F}}\left\lvert \frac{1}{2^n}\sum_{x\in\{0,1\}^n} (1-\delta_{z_x,\tilde{h}(x)}) - \frac{1}{2^n}\sum_{x\in\{0,1\}^n} (1-\delta_{z'_x,\tilde{h}(x)}) \right\rvert\\
        &= \frac{1}{2^n}\sup_{\Tilde{h}\in\mathcal{F}}\left\lvert \delta_{z_x,\tilde{h}(x)} -\delta_{z'_x,\tilde{h}(x)} \right\rvert\\
        &= \frac{1}{2^n}.
    \end{align}
    Therefore, McDiarmid's bounded differences inequality \cite{mcdiarmid1989onthemethod} yields
    \begin{align}
        \mathbb{P}_{f\sim F(\mathcal{D})}\left[\operatorname{opt}_{f}(\mathcal{F}) > \mathbb{E}_{f\sim F(\mathcal{D})}\left[\operatorname{opt}_{f}(\mathcal{F})\right] + \frac{\varepsilon}{3}\right]
        &= \mathbb{P}\left[\xi (\{Z_b\}_{b\in\{0,1\}^n}) > \mathbb{E}\left[\{Z_b\}_{b\in\{0,1\}^n}\right] + \frac{\varepsilon}{3}\right]\\
        &\leq \exp\left( - \frac{2\cdot\left(\tfrac{\varepsilon}{3}\right)^2}{2^n\cdot \left(\tfrac{1}{2^n}\right)^2}\right)\\
        &= \exp\left( - \frac{2^{n+1} \varepsilon^2}{9}\right) .
    \end{align}
    Combining what we have so far, we have shown:
    \small
    \begin{equation}
        \mathbb{P}_{S^{\mathcal{D}}_{\tilde{m}},\mathcal{A}}\left[\operatorname{err}_{\mathcal{D}}(h_{S^{\mathcal{D}}_m,\mathcal{A}})> \operatorname{opt}_{\mathcal{D}}(\mathcal{F}) + \varepsilon\right]
        \leq \mathcal{O}(\tilde{m}^2\cdot 2^{-n}) + \frac{\delta}{2} + \exp\left( - \frac{2^{n+1} \varepsilon^2}{9}\right) + \mathbb{P}_{f\sim F(\mathcal{D}), S^{f}_{\Tilde{m}},\mathcal{A}}\left[\operatorname{err}_{\mathcal{D}}(h_{S^{f}_{\Tilde{m}},\mathcal{A}}) > \operatorname{err}_{f}(h_{S^{f}_{\Tilde{m}},\mathcal{A}})  + \frac{\varepsilon}{3}\right].
    \end{equation}
    \normalsize
    It remains to control the last summand in the above expression. To this end, observe that
    \begingroup
    \allowdisplaybreaks
    \begin{align}
        \mathbb{P}_{f\sim F(\mathcal{D}), S^{f}_{\Tilde{m}},\mathcal{A}}\left[\operatorname{err}_{\mathcal{D}}(h_{S^{f}_{\Tilde{m}},\mathcal{A}}) > \operatorname{err}_{f}(h_{S^{f}_{\Tilde{m}},\mathcal{A}})  + \frac{\varepsilon}{3}\right]
        &= \mathbb{E}_{f\sim F(\mathcal{D}), S^{f}_{\Tilde{m}},\mathcal{A}}\left[\mathds{1}_{\{\operatorname{err}_{\mathcal{D}}(h_{S^{f}_{\Tilde{m}},\mathcal{A}}) > \operatorname{err}_{f}(h_{S^{f}_{\Tilde{m}},\mathcal{A}})  + \frac{\varepsilon}{3}\}}\right]\\
        &\leq \mathbb{E}_{f\sim F(\mathcal{D}), S^{f}_{\Tilde{m}},\mathcal{A}}\left[\mathds{1}_{\{\exists \tilde{h}\in\mathcal{F}: \operatorname{err}_{\mathcal{D}}(\tilde{h}) > \operatorname{err}_{f}(\tilde{h})  + \frac{\varepsilon}{3}\}}\right]\\
        &= \mathbb{E}_{f\sim F(\mathcal{D})}\left[\mathds{1}_{\{\exists \tilde{h}\in\mathcal{F}: \operatorname{err}_{\mathcal{D}}(\tilde{h}) > \operatorname{err}_{f}(\tilde{h})  + \frac{\varepsilon}{3}\}}\right]\\
        &= \mathbb{P}_{f\sim F(\mathcal{D})}\left[\exists \tilde{h}\in\mathcal{F}: \operatorname{err}_{\mathcal{D}}(\tilde{h}) > \operatorname{err}_{f}(\tilde{h})  + \frac{\varepsilon}{3}\right]\\
        &= \mathbb{P}_{f\sim F(\mathcal{D})}\left[\exists \tilde{h}\in\mathcal{F}: \mathbb{E}_{f\sim F(\mathcal{D})}[\operatorname{err}_{f}(\tilde{h})] > \operatorname{err}_{f}(\tilde{h})  + \frac{\varepsilon}{3}\right]\\
        &\leq \sum_{\tilde{h}\in\mathcal{F}}\mathbb{P}_{f\sim F(\mathcal{D})}\left[\mathbb{E}_{f\sim F(\mathcal{D})}[\operatorname{err}_{f}(\tilde{h})] > \operatorname{err}_{f}(\tilde{h})  + \frac{\varepsilon}{3}\right] .
    \end{align}
    \endgroup
    In a similar way, we can use McDiarmid \cite{mcdiarmid1989onthemethod} to show: For every $\tilde{h}\in\mathcal{F}$,
    \begin{equation}
        \mathbb{P}_{f\sim F(\mathcal{D})}\left[\mathbb{E}_{f\sim F(\mathcal{D})}[\operatorname{err}_{f}(\tilde{h})] > \operatorname{err}_{f}(\tilde{h})  + \frac{\varepsilon}{3}\right]
        \leq \exp\left( - \frac{2^{n+1} \varepsilon^2}{9}\right) .
    \end{equation}
    Thus, we have shown:
    \begin{equation}
        \mathbb{P}_{f\sim F(\mathcal{D}), S^{f}_{\Tilde{m}},\mathcal{A}}\left[\operatorname{err}_{\mathcal{D}}(h_{S^{f}_{\Tilde{m}},\mathcal{A}}) > \operatorname{err}_{f}(h_{S^{f}_{\Tilde{m}},\mathcal{A}})  + \frac{\varepsilon}{3}\right]
        \leq \lvert\mathcal{F}\rvert\cdot \exp\left( - \frac{2^{n+1} \varepsilon^2}{9}\right) .
    \end{equation}
    Altogether, we have proven that
    \begin{equation}
        \mathbb{P}_{S^{\mathcal{D}}_{\tilde{m}},\mathcal{A}}\left[\operatorname{err}_{\mathcal{D}}(h_{S^{\mathcal{D}}_m,\mathcal{A}})> \operatorname{opt}_{\mathcal{D}}(\mathcal{F}) + \varepsilon\right]
        \leq \mathcal{O}(\tilde{m}^2\cdot 2^{-n}) + \frac{\delta}{2} + \exp\left( - \frac{2^{n+1} \varepsilon^2}{9}\right) + \lvert\mathcal{F}\rvert\cdot \exp\left( - \frac{2^{n+1} \varepsilon^2}{9}\right)
    \end{equation}
    which is the claimed result.
\end{proof}


\section{Additional Auxiliary Results and Proofs}\label{appendix:proofs}

\begin{corollary}\label{corollary:functional-exact-fourier-sparse-learning}
    Let $\mathcal{D}$ be a probability distribution over $\mathcal{X}_n\times\{0,1\}$ with $\mathcal{D}=(\mathcal{U}_n, f)$ for some deterministic labeling function $f:\{0,1\}^n\to\{0,1\}$ such that $g$ is Fourier-$k$-sparse. 
    Let $\delta\in (0,1)$.
    Then, there exists a quantum algorithm that, given $\mathcal{O}\left(k^4\log(\nicefrac{k^2}{\delta})\right)$ copies of $\ket{\psi_{\mathcal{D}}}$, uses $\mathcal{O}\left(nk^4\log(\nicefrac{k^2}{\delta})\right)$ single-qubit gates, classical computation time $\tilde{\mathcal{O}}\left(nk^4\log(\nicefrac{k^2}{\delta})\right)$, and classical memory of size $\tilde{\mathcal{O}}\left(nk^4\log(\nicefrac{k^2}{\delta})\right)$, and outputs a succinct description of $f$.
\end{corollary}
\begin{proof}
    Let $\tilde{\varepsilon} = \nicefrac{\varepsilon}{2 k}$.
    First, via \Cref{corollary:quantum-approximation-fourier-spectrum}, we can use $\mathcal{O}\left(\tfrac{\log(\nicefrac{1}{\delta\tilde{\varepsilon}^2})}{\tilde{\varepsilon}^4}\right)$ copies of $\ket{\psi_{\mathcal{D}}}$, $\mathcal{O}\left(n\tfrac{\log(\nicefrac{1}{\delta\tilde{\varepsilon}^2})}{\tilde{\varepsilon}^4}\right)$ single-qubit gates, classical computation time $\tilde{\mathcal{O}}\left(n\tfrac{\log(\nicefrac{1}{\delta \tilde{\varepsilon}^2})}{\tilde{\varepsilon}^4}\right)$, and classical memory of size $\tilde{\mathcal{O}}\left(n\tfrac{\log(\nicefrac{1}{\delta\tilde{\varepsilon}^2})}{\tilde{\varepsilon}^4}\right)$ to, with success probability $\geq 1-\delta$, output a succinctly represented vector $\Tilde{g}$ such that $\norm{\Tilde{g}-\hat{g}}_\infty\leq\tilde{\varepsilon}$ and $\norm{\Tilde{g}}_0\leq\tfrac{4}{\tilde{\varepsilon}^2}$. For the rest of the proof, we condition on that success event.
    Define $\tilde{h}:\{0,1\}^n\to\{0,1\}$, $\tilde{h}(x)=\sum_{s:\tilde{g}(s)\neq 0} \Tilde{g}(s) \chi_s(x)$.
    Now, observe that, by Parseval,
    \begin{align}
        \norm{g-\tilde{h}}_2^2
        &= \sum_{s:\hat{g}(s)\neq 0} (\hat{g}(s) - \tilde{g}(s))^2 + \sum_{s:\hat{g}(s)= 0} (\hat{g}(s) - \tilde{g}(s))^2\\
        &\leq k\cdot \tilde{\varepsilon}^2 + \sum_{s:\hat{g}(s)= 0} \tilde{g}(s)^2\\
        &= k\cdot \tilde{\varepsilon}^2, 
    \end{align}
    where we used that the procedure from \Cref{corollary:quantum-approximation-fourier-spectrum} leads to an output $\tilde{g}$ such that $\hat{g}(s)=0~\Rightarrow ~\tilde{g}(s)=0$. 
    (This can, e.g., be seen in the proof of \Cref{corollary:distributional-agnostic-quantum-approximation-fourier-spectrum}: If $\hat{g}(s)=0$, then that means $\hat{\phi}(s)=0$ in the notation of that proof, so that $s\not\in L$ and therefore $\tilde{g}(s)=\tilde{\phi}(s)=0$.) 
    Moreover, since $g$ is by assumption Fourier-$k$-sparse and since $\hat{g}(s)=0~\Rightarrow ~\tilde{g}(s)=0$, we get
    \begin{align}
        \norm{g-\tilde{h}}_\infty
        &= \max_{x\in\{0,1\}^n} \left\lvert \sum_{s:\hat{g}(s)\neq 0}  (\hat{g}(s) - \tilde{g}(s)) (-1)^{s\cdot x} \right\rvert\\
        &\leq \sum_{s:\hat{g}(s)\neq 0} \left\lvert \hat{g}(s) - \tilde{g}(s)\right\rvert\\
        &\leq \norm{\hat{g}-\tilde{g}}_1\\
        &\leq \sqrt{k}\cdot \norm{\hat{g}-\tilde{g}}_2\\
        &= \sqrt{k}\cdot \norm{g-\tilde{h}}_2, 
    \end{align}
    where the second-to-last step uses Cauchy-Schwarz.
    Combining this with the previous inequality, we see that 
    \begin{equation}
        \norm{g-\tilde{h}}_\infty
        \leq k\cdot\tilde{\varepsilon}
        = \frac{1}{2}. 
    \end{equation}
    Since $g$ is $\{-1,1\}$-valued, this implies $g = \operatorname{sgn}(\tilde{h})$.
    Accordingly, we get that $f=\tfrac{1}{2}(1 - \operatorname{sgn}(\tilde{h}))$, and the vector $\tilde{g}$ can serve as a succinct representation of $f$.
\end{proof}

\begin{proof}[Proof of \Cref{lemma:noisy-functional-quantum-Fourier-sampling-v1}]
    We first compute:
    \small
    \begin{align}
        H^{\otimes (n+1)}\rho_{(\mathcal{U}_n, f),\eta}H^{\otimes (n+1)}
        &= (1-\eta)H^{\otimes (n+1)} \ket{\psi_{(\mathcal{U}_n,f)}}\bra{\psi_{(\mathcal{U}_n,f)}} H^{\otimes (n+1)} + \eta H^{\otimes (n+1)} \ket{\psi_{(\mathcal{U}_n,f\oplus 1)}}\bra{\psi_{(\mathcal{U}_n,f\oplus 1)}} H^{\otimes (n+1)}\\
        &= \frac{1-\eta}{2} \left(\ket{0}^{\otimes (n+1)} + \sum_{s\in\{0,1\}^n} \hat{g}(s)\ket{s,1} \right) \left(\bra{0}^{\otimes (n+1)} + \sum_{s\in\{0,1\}^n} \hat{g}(s)\bra{s,1} \right)\\
        &\hphantom{=}~~ + \frac{\eta}{2}\left(\ket{0}^{\otimes (n+1)} - \sum_{s\in\{0,1\}^n} \hat{g}(s)\ket{s,1} \right) \left(\bra{0}^{\otimes (n+1)} - \sum_{s\in\{0,1\}^n} \hat{g}(s)\bra{s,1} \right)\\
        &= \frac{1}{2}(\ket{0}\bra{0})^{\otimes (n+1)} + \frac{1}{2}\left(\sum_{s\in\{0,1\}^n} \hat{g}(s)\ket{s,1}\right)\left(\sum_{s\in\{0,1\}^n} \hat{g}(s)\bra{s,1}\right) \\
        &\hphantom{=}~~ + \left(\tfrac{1}{2}-\eta\right)\left(\ket{0}^{\otimes (n+1)}\left(\sum_{s\in\{0,1\}^n} \hat{g}(s)\bra{s,1}\right) + \left(\sum_{s\in\{0,1\}^n} \hat{g}(s)\ket{s,1}\right)\bra{0}^{\otimes (n+1)}\right)\, .
    \end{align}
    \normalsize
    Noting that the third summand is off-diagonal w.r.t.~the computational basis, we can read off the claimed probabilities.
\end{proof}

\begin{proof}[Proof of \Cref{lemma:noisy-functional-quantum-Fourier-sampling-v2}]
    Applying $H^{\otimes (n+1)}$ to a copy of$\ket{\psi_{\mathcal{D}_\eta}}$ leads to the quantum state
    \begin{align}
        H^{\otimes (n+1)}\ket{\psi_{\mathcal{D}}}
        &= \frac{1}{\sqrt{2^n}}\sum_{x\in\{0,1\}^n} \left(\sqrt{1-\eta}H^{\otimes (n+1)}\ket{x,f(x)}+\sqrt{\eta}H^{\otimes (n+1)}\ket{x,f(x)\oplus 1}\right)\\
        &= \frac{1}{\sqrt{2^n}}\sum_{x\in\{0,1\}^n} \left(\sqrt{1-\eta} \left(\frac{1}{\sqrt{2^{n}}}\sum_{s\in\{0,1\}^n} (-1)^{s\cdot x} \left(\ket{s}\otimes \frac{\ket{0}+(-1)^{f(x)}\ket{1}}{\sqrt{2}}\right)\right) \right.\\
        &\hphantom{\frac{1}{\sqrt{2^n}}\sum_{x\in\{0,1\}^n} \left(\right.~~~~~}\left.+\sqrt{\eta}\left(\frac{1}{\sqrt{2^{n}}}\sum_{s\in\{0,1\}^n} (-1)^{s\cdot x} \left(\ket{s}\otimes\frac{\ket{0}-(-1)^{f(x)}\ket{1}}{\sqrt{2}}\right)\right)\right)\\
        &= \frac{\sqrt{1-\eta}+\sqrt{\eta}}{\sqrt{2}}\ket{0}^{\otimes (n+1)} + \frac{\sqrt{1-\eta}-\sqrt{\eta}}{\sqrt{2}} \sum_{s\in\{0,1\}^n}\mathbb{E}_{x\sim\mathcal{U}_n}[\chi_s(x) (-1)^{f(x)}]\ket{s,1}\\
        &= \frac{\sqrt{1-\eta}+\sqrt{\eta}}{\sqrt{2}}\ket{0}^{\otimes (n+1)} + \frac{\sqrt{1-\eta}-\sqrt{\eta}}{\sqrt{2}} \sum_{s\in\{0,1\}^n}\hat{g}(s)\ket{s,1}\, .
    \end{align}
    From this, we can now easily deduce the two claims about the outcomes of the procedure, simply observing that 
    \begin{equation}
        \left( \frac{\sqrt{1-\eta}\pm\sqrt{\eta}}{\sqrt{2}}\right)^2
        = \frac{1}{2}\pm\sqrt{(1-\eta)\eta}\, .
    \end{equation}
\end{proof}

\begin{proof}[Proof of \Cref{corollary:quantum-noise-rate-learning-v2}]
    The case $\eta_b=0$ is trivial. So, we now assume $\eta_b>0$ for the remainder of the proof.
    Using \Cref{lemma:noisy-functional-quantum-Fourier-sampling-v2} (i), a single copy of $\ket{\psi_{\mathcal{D}_\eta}}$ allows to sample from a Bernoulli distribution with mean $\tfrac{1}{2}-\sqrt{(1-\eta)\eta}$. For ease of notation, write $\xi(\eta) = \sqrt{(1-\eta)\eta}$.
    So, using Hoeffding's inequality for empirical mean estimation,
    $\mathcal{O}\left(\tfrac{\log (1/\delta )}{\tilde{\varepsilon}^2}\right)$ copies of $\ket{\psi_{\mathcal{D}_\eta}}$ suffice to obtain an estimate $\hat{\xi}$ of $\xi(\eta)$ such that, with success probability $\geq 1-\delta$, we have $\lvert \hat{\xi}-\xi\rvert\leq\tilde{\varepsilon}$. Noting that enforcing the cutoff condition $\hat{\xi}\in [0, \sqrt{(1-\eta_b)\eta_b}]$ can only improve the estimate, this proves the first part of the statement.
    
    For the second part, define $\hat{\eta} \coloneqq \tfrac{1}{2}\left(1 - \sqrt{1-4\hat{\xi}^2}\right)$.
    Then, we have
    \begin{align}
        \lvert \hat{\eta} - \eta\rvert
        &= \frac{1}{2} \left\lvert \sqrt{1-4\hat{\xi}^2} - \sqrt{1-4\xi^2}\right\rvert\\
        &\leq \frac{L_1}{2} \left\lvert 1-4\hat{\xi}^2 - (1-4\xi^2)\right\rvert\\
        &\leq 2L_1 L_2 \left\lvert\hat{\xi}-\xi\right\rvert\\
        &\leq 2L_1 L_2\cdot\tilde{\varepsilon}\\
        &= \frac{2\sqrt{(1-\eta_b)\eta_b}}{\sqrt{1 - 4(1-\eta_b)\eta_b}}\cdot\tilde{\varepsilon}\, ,
    \end{align}
    where we used that the function $[1 - 4(1-\eta_b)\eta_b, 1]\ni x\mapsto\sqrt{x}$ is Lipschitz-continuous with Lipschitz constant $L_1=\tfrac{1}{2\sqrt{1 - 4(1-\eta_b)\eta_b}}$ and that the function $[0, \sqrt{(1-\eta_b)\eta_b}]\ni x\mapsto x^2$ is Lipschitz-continuous with Lipschitz constant $L_2=2\sqrt{(1-\eta_b)\eta_b}$.
    Thus, if we choose $\tilde{\varepsilon} = \frac{\sqrt{1 -4(1-\eta_b)\eta_b}}{2\sqrt{(1-\eta_b)\eta_b}}\cdot\varepsilon$, then a $\tilde{\varepsilon}$-accurate estimate $\hat{\xi}$ for $\xi(\eta)$ leads to an $\varepsilon$-accurate estimate of $\eta$. 
    Again, enforcing the cutoff $\hat{\eta}\in [0,\eta_b]$ can only improve the estimate. To finish the proof, note that this choice of $\tilde{\varepsilon}$ leads to
    \begin{equation}
        \frac{1}{\tilde{\varepsilon}^2}
        = \frac{1}{\varepsilon^2}\cdot \frac{4(1-\eta_b)\eta_b}{1 - 4(1-\eta_b)\eta_b}
        =\frac{1}{\varepsilon^2}\cdot \frac{(1-\eta_b)\eta_b}{\left(\eta_b-\tfrac{1}{2}\right)^2}
        \leq \frac{1}{\varepsilon^2}\cdot \frac{\eta_b}{\left(\eta_b-\tfrac{1}{2}\right)^2} , 
    \end{equation}
    which can then be plugged into the previously derived complexity bounds.
\end{proof}

\begin{lemma}\label{lemma:helpful}
    For any $a,b\in\mathbb{R}_{\geq 0}^n$, we have $\norm{b^\downarrow - a^\downarrow}_\infty\leq \norm{a-b}_\infty$.
\end{lemma}
\begin{proof}
    We do a proof by induction over $n$.
    W.l.o.g., we can assume that $b=b^\downarrow$. (Otherwise, let $\tau\in S_n$ be a permutation such that $b^\downarrow = b_\tau$, and work with $\tilde{a}=a_\tau$ and $\tilde{b}=b_\tau = b^\downarrow$ instead of $a$ and $b$, and with $\pi\circ\tau$ instead of $\pi$.)
    The case $n=1$ is trivial and the case $n=2$ is easy to check.
    So, take as induction hypothesis that the statement holds true for some $n$.
    Now for the induction step from $n$ to $n+1$. Let $\pi\in S_{n+1}$ be a permutation such that $a^\downarrow = a_\pi$. Equivalently, $a = a^\downarrow_{\pi^{-1}}$. 
    If $\pi$ (and thus $\pi^{-1}$) has a fixed point $1\leq j\leq n+1$ with $\pi(j)=j$, then
    \begin{align}
        \norm{b^\downarrow - a^\downarrow}_\infty
        &= \max\{\lvert b^\downarrow_j-a^\downarrow_j\rvert, \norm{b^\downarrow_{\{j\}^c} - a^\downarrow_{\{j\}^c}}_\infty\}\\
        &= \max\{\lvert b_j-a_j\rvert, \norm{b^\downarrow_{\{j\}^c} - a^\downarrow_{\{j\}^c}}_\infty\}\\
        &\leq \max\{\lvert b_j-a_j\rvert, \norm{b_{\{j\}^c} - a_{\{j\}^c}}_\infty\}\\
        &= \norm{b-a}_\infty ,
    \end{align}
    where the inequality holds by the induction hypothesis. 
    If $\pi$ and thus $\pi^{-1}$ does not have a fixed point, write $\pi^{-1}(1)=m$ and $\pi^{-1}(k)=1$. 
    Consider a permutation $\tilde{\pi}^{-1}\in S_{n+1}$ defined as
    \begin{equation}
        \tilde{\pi}^{-1}(i)
        =\begin{cases} \pi^{-1}(i)\quad &\textrm{ if } i\neq k, 1 \\
        m &\textrm{ if } i=k\\ 
        1 &\textrm{ if } i=1 \end{cases} .
    \end{equation}
    As $b_1\geq b_k$ and $a^\downarrow_1\geq a^\downarrow_m$, the case $n=2$ tells us that
    \begin{equation}
        \max\{\lvert b_1 - a^\downarrow_1 \rvert, \lvert b_k - a^\downarrow_m\rvert\}
        \leq \max\{\lvert b_1 - a^\downarrow_m \rvert, \lvert b_k - a^\downarrow_1\rvert\}.
    \end{equation}
    Thus, 
    \begin{align}
        \norm{b-a}_\infty
        &= \norm{b-a^\downarrow_{\pi^{-1}}}_\infty\\
        &= \max\{\lvert b_1 - a^\downarrow_{\pi^{-1}(1)}\rvert, \lvert b_k - a^\downarrow_{\pi^{-1}(k)}\rvert, \norm{b_{\{1,k\}^c} - (a^\downarrow_{\pi^{-1}})_{\{\pi^{-1}(1),\pi^{-1}(k)\}^c}}_\infty\}\\
        &= \max\{\lvert b_1 - a^\downarrow_{m}\rvert, \lvert b_k - a^\downarrow_{1}\rvert, \norm{b_{\{1,k\}^c} - (a^\downarrow_{\pi^{-1}})_{\{\pi^{-1}(1),\pi^{-1}(k)\}^c}}_\infty\}\\
        &\geq \max\{\lvert b_1 - a^\downarrow_1 \rvert, \lvert b_k - a^\downarrow_m\rvert, \norm{b_{\{1,k\}^c} - (a^\downarrow_{\pi^{-1}})_{\{\pi^{-1}(1),\pi^{-1}(k)\}^c}}_\infty\}\\
        &= \max\{\lvert b_1 - a^\downarrow_{\tilde{\pi}^{-1}(1)} \rvert, \lvert b_k - a^\downarrow_{\tilde{\pi}^{-1}(k)}\rvert, \norm{b_{\{1,k\}^c} - (a^\downarrow_{\tilde{\pi}^{-1}})_{\{\tilde{\pi}^{-1}(1),\tilde{\pi}^{-1}(k)\}^c}}_\infty\}\\
        &= \norm{b-a^\downarrow_{\tilde{\pi}^{-1}}}_\infty .
    \end{align}
    As $a^\downarrow_{\tilde{\pi}^{-1}}$ differs from $a^{\downarrow}$ only by the permutation $\tilde{\pi}$ and as $\tilde{\pi}$ has a fixed point by construction, we get 
    \begin{equation}
        \norm{b-a^\downarrow_{\tilde{\pi}^{-1}}}_\infty
        \geq \norm{b^\downarrow - a^\downarrow}_\infty
    \end{equation}
    from the previous case. Thus, also in the case of no fixed point we have $\norm{b^\downarrow - a^\downarrow}_\infty\leq \norm{b-a}_\infty$, which finishes the induction step.
\end{proof}

\begin{lemma}\label{lemma:technical}
    For any $a,b\in\mathbb{R}_{\geq 0}^n$, if $\pi\in S_n$ is a permutation such that $a^\downarrow = a_\pi$, then the corresponding permutation $b_\pi$ of $b$ satisfies
    \begin{equation}
        \norm{b^\downarrow - b_\pi}_\infty
        \leq 2 \norm{a-b}_\infty .
    \end{equation}
\end{lemma}
\begin{proof}
    By triangle inequality, we have
    \begin{equation}
        \norm{b^\downarrow - b_\pi}_\infty
        \leq \norm{b^\downarrow - a^\downarrow}_\infty + \norm{a_\pi - b_\pi}_\infty
        = \norm{b^\downarrow - a^\downarrow}_\infty +\norm{a-b}_\infty
        \leq 2 \norm{a-b}_\infty ,
    \end{equation}
    where the last step is by \Cref{lemma:helpful}.
\end{proof}

\begin{lemma}\label{lemma:fully-uniform-vs-random-noisy-parity}
    Let $\eta\in [0,\nicefrac{1}{2})$.
    Any (classical or quantum) algorithm that, given $m$ classical examples from an unknown distribution $\mathcal{D}\in \{\mathcal{U}_{n+1}\}\cup \{(\mathcal{U}_n, (1-2\eta) \chi_s)\}_{s\in\{0,1\}^n}$, decides whether $\mathcal{D} = \mathcal{U}_{n+1}$ or $\mathcal{D}\in \{(\mathcal{U}_n, (1-2\eta) \chi_s)\}_{s\in\{0,1\}^n}$ with success probability $\geq \nicefrac{7}{12}$ has to use at least $m\geq \Omega (n)$ random examples.
\end{lemma}
\begin{proof}
    An algorithm as in the statement of the Lemma can in particular distinguish between $\mathcal{U}_{n+1}$ and $(\mathcal{U}_n, (1-2\eta) \chi_s)$ for a uniformly random $s\in\{0,1\}^n$ with success probability $\geq \nicefrac{7}{12}$ using $m$ random examples.
    As the optimal success probability for this distinguishing task is determined by the TV distance via
    \begin{equation}
        p_{\mathrm{success}}
        = \frac{1}{2}\left(1 + d_{\mathrm{TV}} \left( \mathcal{U}_{n+1}^{\otimes m} , \mathbb{E}_{s\sim\mathcal{U}_n} \left[(\mathcal{U}_n, (1-2\eta)\chi_s)^{\otimes m}\right]\right) \right) ,
    \end{equation}
    it suffices to show that the respective TV distance is at most $o(1)$ for $m\leq o(n)$.
    
    We first argue that it suffices to prove such a TV distance upper bound for $\eta =0$. To this end, assume $\eta>0$ and let $\mathcal{N}_\eta$ be the classical noise channel that adds label noise of strength $\eta$. That is, for $x\in\mathcal{X}_n$ and $y\in\{0,1\}$, we have $\mathcal{N}_\eta (x,y)=(x,y)$ with probability $1-\eta$ and $\mathcal{N}_\eta (x,y)=(x,1\oplus y)$ with probability $\eta$. Then, 
    \begin{equation}
        \mathcal{N}_\eta (\mathcal{U}_{n+1}) 
        = \mathcal{U}_{n+1}
    \end{equation}
    as well as
    \begin{equation}
        \mathcal{N}_\eta ((\mathcal{U}_n, \chi_s))
        = (\mathcal{U}_n, (1-2\eta)\chi_s)
    \end{equation}
    for every $s\in\{0,1\}^n$. Thus, we can see that
    \begin{align}
        d_{\mathrm{TV}} \left( \mathcal{U}_{n+1}^{\otimes m} , \mathbb{E}_{s\sim\mathcal{U}_n} \left[(\mathcal{U}_n, (1-2\eta)\chi_s)^{\otimes m}\right]\right)
        &= d_{\mathrm{TV}} \left( \mathcal{N}_\eta^{\otimes m} \left(\mathcal{U}_{n+1}\right)^{\otimes m} , \mathcal{N}_\eta^{\otimes m}\left( \mathbb{E}_{s\sim\mathcal{U}_n} \left[(\mathcal{U}_n, (1-2\eta)\chi_s)^{\otimes m}\right]\right)\right)\\
        &\leq d_{\mathrm{TV}} \left( \mathcal{U}_{n+1}^{\otimes m} , \mathbb{E}_{s\sim\mathcal{U}_n} \left[(\mathcal{U}_n, \chi_s)^{\otimes m}\right]\right) ,
    \end{align}
    where the inequality holds because the TV distance is non-increasing under classical noise channels.
    
    So, we now focus on upper bounding the TV distance $d_{\mathrm{TV}} \left( \mathcal{U}_{n+1}^{\otimes m} , \mathbb{E}_{s\sim\mathcal{U}_n} \left[(\mathcal{U}_n, \chi_s)^{\otimes m}\right]\right)$ by a direct computation:
    \begingroup
    \allowdisplaybreaks
    \begin{align}
        &d_{\mathrm{TV}} \left( \mathcal{U}_{n+1}^{\otimes m} , \mathbb{E}_{s\sim\mathcal{U}_n} \left[(\mathcal{U}_n, \chi_s)^{\otimes m}\right]\right)\\
        &= \frac{1}{2}\sum_{(x_1,y_1),\ldots,(x_m,y_m)\in\{0,1\}^n\times\{0,1\}} \left\lvert \left(\frac{1}{2^{n+1}}\right)^m - \left(\frac{1}{2^{n}}\right)^m\mathbb{E}_{s\sim\mathcal{U}_n}\left[\prod_{i=1}^m \mathbb{P}[y_i = \chi_s(x_i)] \right] \right\rvert\\
        &= \frac{1}{2}\sum_{\substack{(x_1,y_1),\ldots,(x_m,y_m)\in\{0,1\}^n\times\{0,1\}\\ \textrm{s.t.~}\{x_1,\ldots, x_m\}\textrm{ linearly independent}}} \left\lvert \left(\frac{1}{2^{n+1}}\right)^m - \left(\frac{1}{2^{n}}\right)^m\mathbb{E}_{s\sim\mathcal{U}_n}\left[\prod_{i=1}^m \mathbb{P}[y_i = \chi_s(x_i)] \right] \right\rvert\\
        &\hphantom{=}~ +\frac{1}{2}\sum_{\substack{(x_1,y_1),\ldots,(x_m,y_m)\in\{0,1\}^n\times\{0,1\}\\ \textrm{s.t.~}\{x_1,\ldots, x_m\}\textrm{ linearly dependent}}} \left\lvert \left(\frac{1}{2^{n+1}}\right)^m - \left(\frac{1}{2^{n}}\right)^m\mathbb{E}_{s\sim\mathcal{U}_n}\left[\prod_{i=1}^m \mathbb{P}[y_i = \chi_s(x_i)] \right] \right\rvert\\
        &= \frac{1}{2}\sum_{\substack{(x_1,y_1),\ldots,(x_m,y_m)\in\{0,1\}^n\times\{0,1\}\\ \textrm{s.t.~}\{x_1,\ldots, x_m\}\textrm{ linearly independent}}} \left\lvert \left(\frac{1}{2^{n+1}}\right)^m - \left(\frac{1}{2^{n}}\right)^m \left(\prod_{i=1}^m \underbrace{\mathbb{E}_{s\sim\mathcal{U}_n}\left[\mathbb{P}[y_i = \chi_s(x_i)] \right]}_{=\nicefrac{1}{2}}\right) \right\rvert\\
        &\hphantom{=}~ + \frac{1}{2^{n m}} \sum_{\{x_1,\ldots, x_m\}\textrm{ linearly dependent}} \underbrace{\frac{1}{2} \sum\limits_{y_1,\ldots,y_m\in\{0,1\}} \left\lvert \frac{1}{2^m} - \mathbb{E}_{s\sim\mathcal{U}_n}\left[\prod_{i=1}^m \mathbb{P}[y_i = \chi_s(x_i)] \right] \right\rvert}_{ = d_{\mathrm{TV}}\left(\mathcal{U}_1^{\otimes m}, \mathbb{E}_{s\sim\mathcal{U}_n}\left[ \mathbb{P}^{(s)}_{Y_1,\ldots,Y_m|x_1,\ldots,x_m}\right] \right)\in [0,1]}\\
        &\leq \mathbb{P}_{x_1,\ldots,x_m\sim\mathcal{U}_n^{\otimes m}} \left[\{x_1,\ldots, x_m\}\textrm{ linearly dependent}\right]\\
        &= \mathbb{P}_{x_1,\ldots,x_m\sim\mathcal{U}_n^{\otimes m}}\left[(x_1 = 0^n)\lor \left(\bigvee_{ i=1}^{m-1} \left( x_i\in\mathrm{span}\{x_1,\ldots,x_i\}~\wedge ~ \{x_1,\ldots,x_i\}\textrm{ linearly independent}\right)\right) \right]\\
        &\leq \mathbb{P}_{x_1\sim\mathcal{U}_n}\left[x_1 = 0^n\right] + \sum_{i=1}^{m-1} \mathbb{P}_{x_1,\ldots,x_m\sim\mathcal{U}_n^{\otimes m}}\left[x_i\in\mathrm{span}\{x_1,\ldots,x_i\}~|~\{x_1,\ldots,x_i\}\textrm{ linearly independent}\right]\cdot\\
        &\hphantom{\leq \mathbb{P}_{x_1\sim\mathcal{U}_n}\left[x_1 = 0^n\right] + \sum_{i=1}^{m-1} \mathbb{P}}\cdot\mathbb{P}_{x_1,\ldots,x_m\sim\mathcal{U}_n^{\otimes m}}\left[\{x_1,\ldots,x_i\}\textrm{ linearly independent}\right]\\
        &\leq \frac{1}{2^m} + \sum_{i=1}^{m-1} 2^{i-n}\cdot 1\\
        &= \sum_{i=0}^{m-1} 2^{i-n}\\
        &= \frac{2^m -1}{2^n} .
    \end{align}
    \endgroup
    Here, the first step consists in plugging in definitions. In the second step, we split the summation in two. The third step holds because, if $\{x_1,\ldots,x_m\}$ are linearly independent, then the events $\{y_i=\chi_s (x_i)\}$, $i=1,\ldots,m$, and therefore the random variables $\mathbb{P}[y_i = \chi_s(x_i)]$, $i=1,\ldots,m$, are independent, so that the expectation of their product factorizes. The remaining steps are either direct rewritings or a simple union bound.
    Now, if $m\leq o(n)$, then we obtain from the above calculation that
    \begin{equation}
        d_{\mathrm{TV}} \left( \mathcal{U}_{n+1}^{\otimes m} , \mathbb{E}_{s\sim\mathcal{U}_n} \left[(\mathcal{U}_n, \chi_s)^{\otimes m}\right]\right)
        \leq \frac{2^m -1}{2^n}
        \leq o(1) . 
    \end{equation}
    This was all that remained to be shown, the proof is complete.
\end{proof}

\begin{lemma}\label{lemma:fourier-modification-reduction}
    Let $0 < a\leq b \leq 1$ with $a<1$.
    Let $\alpha = a$ and $\beta=\sqrt{b^2-a^2}$.
    Let $s\in\{0,1\}^n$.
    Define
    \begin{equation}
        \mathfrak{D}_{\alpha, \beta}^{s}
        =\left\{(\mathcal{U}_n, \varphi')~|~\exists t\in\{0,1\}^n\setminus\{s\}: \phi'=1-2\varphi' = \alpha\chi_s + \beta\chi_t\right\}.
    \end{equation}
    Let $\eta\in [0,\tfrac{1-\nicefrac{\beta}{(1-\alpha)}}{2}]$.
    The problem of distinguishing between a uniformly random element of $\mathfrak{D}_{\alpha, \beta}^{s}$ and $(\mathcal{U}_n, \alpha\chi_s)$ from samples is at least as hard as distinguishing between a uniformly random $\eta$-noisy parity (acting on uniformly random inputs) and the uniform distribution $\mathcal{U}_{n+1}$ from samples.
\end{lemma}
\begin{proof}
    We first introduce a tool for our proof.
    For $\varphi:\{0,1\}^n\to [0,1]$ and $\gamma\in [0,1]$, we define $\mathcal{P}(\varphi, \gamma)[\varphi']$ as follows:
    A sample from the distribution $(\mathcal{U}_n, \mathcal{P}(\varphi, \gamma)[\varphi'])$ over $\{0,1\}^n\times\{0,1\}$ is obtained through the following procedure:
    \begin{enumerate}
        \item Sample $(x,y)\sim(\mc U_n,\varphi')$.
        \item With probability $\gamma$, re-sample $y\sim\varphi (x)$.
        \item Return $(x,y)$. 
    \end{enumerate}
    From this description, it is easy to see that
    \begin{align}
        \widehat{\mc P(\varphi,\gamma)[\varphi']}(s) = \gamma\hat\phi(s) + (1-\gamma)\hat\phi'(s)\,.
    \end{align}
    Here, we used our usual notation for $\phi = 1-2\varphi$ and $\phi'=1-2\varphi'$.
    
    Assume that $\mathcal{A}_{\alpha,\beta}^s$ is an algorithm that distinguishes a uniformly random element of $\mathfrak{D}_{\alpha, \beta}^{s}$ from $(\mathcal{U}_n, a\chi_s)$ using random examples, with success probability $\geq 1-\delta$.
    Let $\eta\in [0,\tfrac{1-\nicefrac{\beta}{(1-\alpha)}}{2}]$.
    Then, we can distinguish a uniformly random $\eta$-noisy parity (acting on uniformly random inputs) from the uniform distribution $\mathcal{U}_{n+1}$ with a sample and time complexity overhead of $\O\qty(\frac{\log(1/\delta)}{(1-2\eta)^2})$ compared to $\mathcal{A}_{\alpha,\beta}^s$ as follows:
    Let $\mathcal{D}=(\mathcal{U}_n, \varphi')$ be the unknown distribution.
    First, use $\O\qty(\frac{\log(1/\delta)}{(1-2\eta)^2})$ many random examples from $\mathcal{D}$ to estimate, with success probability $\geq 1-\tfrac{\delta}{2}$, the Fourier coefficient $\hat{\phi}'(s)$ to accuracy $\tfrac{1-2\eta}{3}$.
    If the estimate has absolute value $\geq \tfrac{1-2\eta}{3}$, then the unknown distribution must have been $(\mathcal{U}_n, (1-2\eta)\chi_s)$, so in this case we have successfully solved the distinguishing task.
    
    Thus, assume that the estimate for $\hat{\phi}'(s)$ has absolute value $< \tfrac{1-2\eta}{3}$, in which case we know that the unknown distribution cannot have been $(\mathcal{U}_n, (1-2\eta)\chi_s)$.
    Let $\gamma = 1-\tfrac{\beta}{1-2\eta}\in [\alpha,1]$.
    If the unknown distribution is $(\mathcal{U}_n, \varphi')=(\mathcal{U}_n, (1-2\eta)\chi_t)$ for some $t\neq s$, then $\mathcal{P}(\tfrac{\alpha}{\gamma}\chi_s, \gamma )[\varphi']=\mathcal{P}(\tfrac{\alpha}{\gamma}\chi_s,\gamma )[(1-2\eta)\chi_t]= \alpha \chi_s +  \beta\chi_t$. Thus, in this case $(\mathcal{U}_n, \mathcal{P}(\tfrac{\alpha}{\gamma}\chi_s, \gamma )[\varphi'])\in \mathfrak{D}_{\alpha, \beta}^{s}$.
    Similarly, if the unknown distribution is $\mathcal{U}_{n+1}=(\mathcal{U}_n, \varphi')=(\mathcal{U}_n, \tfrac{1}{2})$, then $\mathcal{P}(\tfrac{\alpha}{\gamma}\chi_s, \gamma )[\varphi'] = \mathcal{P}(\tfrac{\alpha}{\gamma}\chi_s, \gamma )[\tfrac{1}{2}] = \alpha\chi_s$. Thus, in this case $(\mathcal{U}_n, \mathcal{P}(\tfrac{\alpha}{\gamma}\chi_s, \gamma )[\varphi']) = (\mathcal{U}_n, \alpha\chi_s)$. 
    Hence, as a single sample from $(\mathcal{U}_n, \varphi')$ suffices to generate a single sample from $(\mathcal{U}_n, \mathcal{P}(\tfrac{\alpha}{\gamma}\chi_s, \gamma )[\varphi'])$, we can now call $\mathcal{A}_{\alpha,\beta}^s$ to distinguish whether the unknown distribution is $(\mathcal{U}_n, (1-2\eta)\chi_t)$ for some $t\neq s$ or $\mathcal{U}_{n+1}$.
\end{proof}

\end{document}